%% file: main.tex
\documentclass{article}[11pt]

\usepackage{titling}
\usepackage{ amssymb }
\usepackage{amsfonts}
\usepackage{amsthm}
\usepackage{amsmath}
\usepackage{bbm}
\usepackage{setspace}
\usepackage[normalem]{ulem}
\usepackage[utf8]{inputenc}

\def\showauthornotes{0}

\def\showkeys{0}
\def\showdraftbox{0}

\def\usemicrotype{1}
\def\showfixme{0}

\usepackage{csquotes}

\usepackage[
backend=biber,
maxcitenames=50,
maxbibnames=99,
style=alphabetic,
sorting=ynt
]{biblatex}
\title{Chernoff Bounds and Reverse Hypercontractivity on HDX}
\author{Yotam Dikstein\thanks{Institute for Advanced Study, USA. email: yotam.dikstein@gmail.com. This material is based upon work supported by the National Science Foundation under Grant No. DMS-1926686.}\; and Max Hopkins\thanks{Department of Computer Science and Engineering, UCSD, CA 92092. Email: \texttt{nmhopkin@ucsd.edu}. Supported by NSF Award DGE-1650112.} }
\date{\today}

\input{macros.tex}
\usepackage[ruled,vlined,linesnumbered,noresetcount]{algorithm2e}
\begin{document}

\maketitle

\begin{abstract}
We prove optimal concentration of measure for lifted functions on high dimensional expanders (HDX). Let $X$ be a $k$-dimensional HDX. We show for any $i \leq k$ and function $f: \X[i] \to [0,1]$:
\[
\Pr_{s \in \X[k]}\left[\left|\underset{{t \subseteq s}}{\mathbb{E}}[f(t)] - \mu \right| \geq \varepsilon \right] \leq \exp\left(-\varepsilon^2 \frac{k}{i}\right).
\]
Using this fact, we prove that high dimensional expanders are \textit{reverse hypercontractive}, a powerful functional inequality from discrete analysis implying that for any sets $A,B \subset \X[k]$, the probability a $\rho$-correlated pair passes between them is at least
\[
\Pr_{s,s' \sim T_\rho}[s \in A, s' \in B] \geq \Pr[A]^{O(1)} \Pr[B]^{O(1)}.
\]
Our results hold under weak spectral assumptions on $X$. Namely we prove exponential concentration of measure for any complex below the `Trickling-Down Threshold' (beyond which concentration may be arbitrarily poor), and optimal concentration for $\sqrt{k}$-skeletons of such complexes. We also show optimal bounds for the top dimension of stronger HDX among other settings.

We leverage our inequalities to prove several new agreement testing theorems on high dimensional expanders, including a new $99\%$-regime test for \textit{subsets}, and a variant of the `Z-test' achieving \textit{inverse exponential soundness} under the stronger assumption of $\ell_\infty$-expansion. The latter gives rise to the first optimal testers beyond the complete complex and products, a stepping stone toward the use of HDX in strong soundness PCPs. 

We also give applications within expansion, analysis, combinatorics, and coding theory, including a proof that two-sided HDX have optimal geometric overlap (giving the first explicit bounded-degree construction), near-optimal double samplers, new super-exponential degree lower bounds for certain HDX, distance-amplified list-decodable and locally testable codes, a Frankl–R\"{o}dl Theorem, and more.
\end{abstract}
\newpage
\tableofcontents

\newpage
\pagenumbering{arabic}

\section{Introduction}\label{sec:intro}

Recent years have seen the emergence of \textit{high dimensional expanders} (HDX) as a core tool in theoretical computer science, with breakthrough applications in approximate sampling \cite{AnariLOV2019,anari2021spectral}, coding theory \cite{DinurELLM2022,PanteleevK22}, and quantum complexity \cite{anshu2023nlts}, and strong promise toward longstanding problems in hardness of approximation and probabilistically checkable proofs \cite{DinurK2017,DiksteinD2019,kaufman2020local,GotlibK2022,HopkinsL2022,BafnaM2024,DiksteinDL2024}. One central force behind the success of HDX in application is their \textit{concentration of measure}. Consider (as a warmup) the following fundamental question: given a $k$-uniform hypergraph $X$ and a function $f:\X[1] \to [0,1]$ on its vertices,\footnote{We denote by \(\X[i]\) the subsets of size \(i\) in (the downward closure of) \(X\).} how concentrated is $f$ to its mean across a random hyperedge?
\begin{equation}\label{intro:concentration-1}
\Pr_{\{v_1,\ldots,v_k\} \in X(k)}\left[\left|\frac{1}{k} \sum_{i=1}^k f(v_i) - \mu \right| > \varepsilon \right] \overset{?}{\leq} \beta(\varepsilon,k)
\end{equation}
When $X$ is the complete hypergraph, this was classically resolved by Chernoff and Hoeffding \cite{chernoff1952measure, hoeffding1994probability} who showed $\beta(\varepsilon,k) \leq \exp(-\varepsilon^2 k)$. Over the years, near-matching bounds have been shown for more general hypergraph families (see e.g.\ \cite{paulin2014convex,kyng2018matrix,kaufman2021scalar}), and even for certain bounded degree systems like walks on expanders \cite{ajtai1987deterministic,Gillman1998}. These objects, called \textit{sampler graphs} or \textit{extractors}, are by now a core tool in the field with applications in de-randomization, complexity, and cryptography \cite{Goldreich1997,wigderson2009randomness}. 

High dimensional expanders are known to satisfy a `Chebyshev-type' bound $\beta(\varepsilon,k) \leq \frac{1}{k\varepsilon^2}$ for \eqref{intro:concentration-1} \cite{DinurK2017}, but this is of course not what makes them powerful. Rather, many modern applications in complexity require concentration of $X$ against \textit{broader classes of functions}. Of particular interest is the extension of \eqref{intro:concentration-1} to functions $f: X(i) \to [0,1]$ sitting on \textit{$i$-sets} of $X$. Complexes satisfying such a bound are called \textit{inclusion samplers}; they arise naturally in the study of agreement and PCPs (forming consistency checks), and play a critical role in low soundness constructions toward the Sliding Scale Conjecture \cite{ImpagliazzoKW2012,dinur2011derandomized,moshkovitz2017low,DinurL2017}. Unfortunately, inclusion samplers are notoriously difficult to construct, and no bounded degree families were known for many years.

This changed with the advent of spectral HDX \cite{Garland1973,LubotzkySV2005a,EvraK2016,KaufmanM2017,DinurK2017,oppenheim2018local}. A key motivation behind the modern incarnation of these objects \cite{KaufmanM2017,DinurK2017,oppenheim2018local} was the fact that their `inclusion graphs' (the bipartite adjacency matrix between $k$-sets and $i$-sets) are spectral expanders, and therefore satisfy `Chebyshev-type' concentration $\beta(\varepsilon,i,k) \leq \frac{i}{\varepsilon^2k}$. This observation played a central role in recent progress in agreement testing \cite{DinurK2017,DinurHKLT2018,DiksteinD2019,kaufman2020local,GotlibK2022,DiksteinD2023agr, BafnaM2023}, but left an exponential gap from the Chernoff-type concentration needed for other applications.

In this work we take a major step toward closing this gap, resolving the problem completely in many regimes of interest. In particular, drawing on ideas from the concentration of measure \cite{Boucheron}, PCP \cite{ImpagliazzoKW2012}, and HDX \cite{DiksteinD2019,AlevJT2019} literature, we prove HDX are optimal samplers between every two levels.
\begin{theorem}[HDX are Optimal Inclusion Samplers (Informal)]\label{thm:intro-sampling}
    Let $X$ be a nice enough $d$-dimensional HDX.\footnote{This includes, e.g., $d$-skeletons of $\frac{1}{d^2}$-one-sided $d^2$-dimensional HDX or the top level of any $2^{-d}$-two-sided HDX. See \pref{sec:intro-background} and \pref{def:nice-complexes} for a formal description of `nice enough' complexes.} Then for any $i \leq k \leq d$ and any $f: X(i) \to [0,1]$ of expectation $\mu$:
    \[
    \Pr_{s \in X(k)}\left[\Abs{\underset{t \subset s}{\mathbb{E}}[f(t)]- \mu} \geq \varepsilon \right] \leq \exp\left(-\Omega \left (\varepsilon^2 \frac{k}{i}\right ) \right).
    \]
    Moreover, this is essentially tight---no inclusion sampler achieves better than $\exp\left(-O\left(\varepsilon^2\frac{k}{i}\right)\right)$ error.
\end{theorem}

\pref{thm:intro-sampling} is powerful when $i=o(k)$, but our corresponding lower bounds show non-trivial inclusion sampling is essentially impossible in the critical regime $i=\Theta(k)$. In fact this exact issue stood as a barrier for many years on the complete complex, where Impagliazzo, Kabanets, and Wigderson \cite{ImpagliazzoKW2012} proved \pref{thm:intro-sampling} and used it to construct sub-optimal testers and PCPs by taking $i=\Theta(\sqrt{k})$. In 2017, Dinur and Livni-Navon \cite{DinurL2017} resolved this problem by appealing to a somewhat surprising player: an underutilized tool from boolean function analysis called \textit{reverse hypercontractivity} (RHC). 

Reverse hypercontractivity is a classic functional inequality due to Borell \cite{Borell1982} that lower bounds the correlation between sets on the noisy hypercube. Roughly speaking, it implies that for any two subsets $A,B \subseteq \{0,1\}^d$, the probability a $\rho$-correlated edge crosses $A$ and $B$ is at least some power of their measure:
\[
\Pr_{(s,s') \sim T_\rho}[s \in A, s' \in B] \geq \Pr[A]^{O_\rho(1)}\Pr[B]^{O_\rho(1)}.
\]
Reverse hypercontractivity is an even rarer phenomenon than inclusion sampling---while there are some moderately de-randomized inclusion graphs (namely the Grassmann \cite{ImpagliazzoKW2012}) satisfying Chebyshev-type variants of \pref{thm:intro-sampling}, none are known to admit reverse hypercontractivity. At outset this may even seem necessary. The inequality is closely tied to tensorization and (modified) Log-Sobolev Inequalities \cite{MosselOS2013}, which inherently fail on sparse hypergraphs.

In this work, we circumvent this issue by a new combinatorial argument showing any hypergraph satisfying \pref{thm:intro-sampling} `locally' is also reverse hypercontractive. This leads to the following corollary for HDX:
\begin{theorem}[Reverse Hypercontractivity for HDX (Informal)]\label{thm:intro-rhc}
let $X$ be a nice enough $d$-dimensional HDX. Then for any $\rho \in (0,1)$, $k \leq d$, and $A,B \subset X(k)$:
\[
\Pr_{(s,s') \sim T_\rho(X)}[s \in A, s' \in B] \geq \Pr[A]^{O_\rho(1)}\Pr[B]^{O_\rho(1)}.
\]
\end{theorem}
Here $T_\rho(X)$ is the natural generalization of the noise operator to hypergraphs which, roughly speaking, generates $s'$ by re-sampling each vertex in $s$ with probability $1-\rho$. We also prove a more standard analytical version of reverse hypercontractivity for all functions by reduction to the above, as well as analog statements for the heavily studied `down-up' walks which play a critical role in almost all HDX applications.

\paragraph{Applications}
The remainder of our work is devoted to proving a variety of applications of \pref{thm:intro-sampling} and \pref{thm:intro-rhc} including new families of agreement tests, bounded degree complexes with optimal geometric overlap, double samplers with near-optimal overhead, new degree lower bounds for HDX, large distance list-decodable and locally-testable codes, and several extensions of classic combinatorial and analytic results to HDX. Here we choose to highlight just two of these applications, agreement testing and geometric overlap. We refer the reader to \pref{sec:agree-intro} and \pref{sec:applications-intro} for a more in-depth discussion of our applications.

\medskip

\textit{Agreement testing} is a general property testing paradigm that generalizes classical subroutines in PCP theory like the plane vs.\ plane test (e.g.\ \cite{AroraS1997, GoldreichS1997,dinur2007pcp}). Roughly speaking, an agreement test consists of a universe $U$ and a family of overlapping subsets $S \subset \mathbb{P}(U)$ (think of $U$ as the vertices of $X$, and $S$ as its hyperedges). Given a family of \emph{local} assignments $\{f_s: s \to \{0,1\}\}_{s \in S}$, we'd like to test whether $f_s$ actually agrees with a \emph{global} function $g: U \to \{0,1\}$ in the sense that $f_s=g|_s$ for many $s$. These systems occur in PCP reductions to ensure the provers (who answer on subsets) can't `cheat' by using answers not corresponding to a real solution of the initial problem (corresponding to global functions on $U$).

Agreement testing has two major regimes, the `$99\%$-regime' (where we'd like to infer global structure only if the test passes with high probability) and the `$1\%$-regime' (where we'd like to infer global structure even when the test only passes with \textit{non-trivial} probability). Leveraging reverse hypercontractivity, we give new tests in both regimes. In the $99\%$-regime, we give a new $2$-query test where the universe $U$ corresponds to \textit{$i$-sets} of a high dimensional expander, extending a similar result of \cite{DinurFH2019} on the complete complex. Our test, along with \pref{thm:intro-sampling}, has since appeared as an important sub-routine in the construction of the first bounded-degree $1\%$ agreement testers in \cite{dikstein2023agreement,DiksteinDL2024}. 

In the $1\%$-regime, we give a new $3$-query tester for a variant of $\ell_\infty$-expanders (a strengthened notion of HDX \cite{kaufman2021scalar}) with \textit{optimal soundness}, meaning that if the test passes with probability asymptotically better than a random function ($2^{-\Theta(k)}$ for a $k$-dimensional $\ell_\infty$-expander), we are able to infer global structure. In the context of PCPs, such a soundness guarantee is critical to ensure the alphabet size stays polynomial in the soundness. Our result gives rise to the first families of optimal testers beyond the complete complex and products \cite{DinurL2017}, including (dense) random complexes, skeletons of many well-studied spin systems, the full linear matroid, and more. We remark that while all such examples are dense, we prove an optimal `local' agreement theorem only under the assumption of local spectral expansion. We lift this guarantee to the general result using $\ell_\infty$-expansion, but it is possible this requirement could be relaxed using recent techniques based on coboundary expansion \cite{GotlibK2022,bafna2023characterizing,dikstein2023agreement,BafnaM2024,DiksteinDL2024} to give sparse optimal testers.

\medskip

The \textit{Geometric Overlap Property} is one of the earliest notions of high dimensional expansion \cite{Gromov2010,lubotzky2018high}. A $d$-dimensional complex $X$ has $c$-geometric overlap if for any embedding of $X$ into $\R^d$, there is a point $p \in \R^d$ that is covered by a $c$-fraction of $X$'s hyperedges. Geometric overlap was first proven for the complete complex in 2-dimensions by Boros and F{\"u}redi \cite{BorosF1984}, and later extended to all dimensions by B\'{a}r\'{a}ny \cite{Barany1982}. Gromov \cite{Gromov2010} asked whether there are bounded degree complexes satisfying geometric overlap. This was resolved in \cite{FoxGLNP2012}, who gave both an optimal random construction and several explicit constructions with bounded but sub-optimal overlap (including one from high dimensional expanders). Several later works \cite{ParzanchevskiRT2016,Evra2017,oppenheim2018local} continued to build the connection between high dimensional expanders and geometric overlap, but it remained open whether a construction achieving the best of both worlds (explicit and optimal) could be achieved. Leveraging a variant of \pref{thm:intro-sampling}, we resolve this problem, showing under mild assumptions that \textit{any} sufficiently strong high dimensional expander has near-optimal geometric overlap. 

\paragraph{The Trickling-Down Threshold} High dimensional expanders exhibit a \textit{phase transition} at a certain expansion parameter called the `Trickling-Down (TD)-Threshold' \cite{oppenheim2018local}. Any complex breaking this barrier immediately has `local-to-global' structure, meaning properties such as global expansion and fast-mixing can be inferred just from \textit{local} structure of the complex \cite{Garland1973,KaufmanM2017,oppenheim2018local,alev2020improved}. This barrier shows up as a major point in constructions of HDX as well. It is relatively easy to build complexes \textit{at} the threshold by tensoring an expander with the complete complex \cite{LiuMY2020,Golowich2021}, but \textit{breaking} the barrier requires complicated algebraic machinery and is considered the `gold standard' for HDX---indeed only three such constructions are known \cite{LubotzkySV2005a,KaufmanO181,Dikstein2022}.

\pref{thm:intro-sampling} and \pref{thm:intro-rhc} hold for $\sqrt{d}$-dimensional skeletons of any complex breaking the TD-Threshold. In fact, we prove that the top level of such complexes satisfy exponential concentration for all Lipschitz functions, corresponding to $\beta(\varepsilon,i,k) \leq \exp(-\varepsilon\frac{\sqrt{k}}{i})$ in the inclusion sampling setting. On the other hand, we give examples of complexes \textit{at} the TD-threshold that exhibit arbitrarily poor concentration. We view our results in this sense as a significant strengthening of HDX's `local-to-global' phase transition, and a partial explanation to the great difficulty of constructing hypergraphs beyond the TD-Threshold. Whereas graph expansion and mixing are `common' phenomenon, \pref{thm:intro-sampling} and \pref{thm:intro-rhc} imply as soon as one breaks the TD-Threshold $X$ not only satisfies exponential concentration of measure, its skeletons satisfy optimal concentration and reverse hypercontractivity---rare properties not known for any other sparse systems.

\section{Main Results and Proof Overview}
\subsection{Background}\label{sec:intro-background}
Before stating our results in somewhat more formality, we cover some basic background on high dimensional expanders, random walks, and sampling. We refer the reader to \pref{sec:preliminaries} for formal details and discussion.

\paragraph{Simplicial Complexes}
A $d$-\textbf{\maximal}\textbf{simplicial complex} $X$ consists of a $d$-uniform hypergraph $\X[d]$ together with its downward closure
\[
X=\X[1] \cup \ldots \cup \X[d],
\]
where $\X[i] \subseteq \binom{[n]}{i}$, called the `$i$-faces', are all $i$-size subsets that sit in some hyperedge in $\X[d]$. Given a face $t \in \X[i]$, the \textbf{link} of $t$ is the sub-complex induced by localizing to faces that include $t$, that is $X_t \coloneqq \{ s : s \cup t \in X\}$. We say $X$ is \textbf{connected} if the base graph $(\X[1][t],\X[2][t])$ of every link is connected. A $d$-\maximal complex is \textbf{partite} if its vertices can be partitioned into $d$ `parts' such that each top-level face has one part from each component. Finally, the \textbf{inclusion graph} $(\X[k],\X[i])$ is the bipartite graph between $k$-sets and $i$-sets of $X$ where edges are given by inclusion.

We emphasize that we have changed convention here from dimension, where $X(d)$ refers to sets of size $d+1$, to \maxsizity where $\X[d]$ consists of sets of size $d$. We will use the latter notation throughout the rest of the paper to simplify both the statements and proofs of our results.

\paragraph{High dimensional expanders} A simplicial complex $X$ is a \textbf{\(\lambda\)-one-sided HDX} (resp.\ two-sided) if for every \(s \in X\) of co-dimension\footnote{Here by co-dimension $j$ we mean sets of size $d-j$.} at least $2$ $(\X[1][t],\X[2][t])$ is a \(\lambda\)-one-sided spectral expander (resp.\ two sided). 

The weakest form of high dimensional expanders are $\lambda$-\textbf{Trickling-Down (TD) Complexes}. We call $X$ $\lambda$-TD if it is connected and every co-dimension $2$ link is a $\frac{\lambda}{d-1}$-one-sided expander.

Our results hold for a broad class of HDX we call \textbf{Nice Complexes}. We defer the formal definition to \pref{sec:nice}, and here give two basic examples of nice complexes for which our results hold:
\begin{itemize}
    \item $X$ is the $\sqrt{d}$-skeleton of a $\lambda$-TD complex
    \item $X$ is a $2^{-d}$-two-sided HDX
\end{itemize}
We note our results also hold on $2^{-d}$-one-sided \textit{partite} HDX, albeit with polynomially worse $\varepsilon$-dependence. We discuss this further in the main body (see \pref{thm:hdx-is-sampler} and \pref{sec:concentration-all-nice-hdx}).

\paragraph{High Order Random Walks} Simplicial complexes admit a variety of `high order random walks' generalizing the standard random walk on a graph. There are three natural walks critically important to our results. The first are the classical \textbf{down-up walks} of \cite{KaufmanM2017}, denoted $U_{k,d}D_{d,k}$, which walk between $s,s' \in \X[d]$ through a shared $k$-face $t \subset s,s'$. The second are the \textbf{swap walks} of \cite{DiksteinD2019,AlevJT2019}, denoted $S_{i,j}$, which walk between \textit{disjoint} faces $s \in \X[i]$ and $s' \in \X[j]$ via $s \cup s' \in \X[i+j]$. The final walk is the \textbf{noise operator} $T_\rho$, which walks between $s,s' \in \X[d]$ via a random $i$-subface of $s$ for $i \sim Bin(1-\rho,d)$.\footnote{We remark that while definition this may seem odd at outset, it is the natural extension of the noise operator to hypergraphs and recovers the standard notion on products. See \cite{BafnaHKL2022} or \pref{sec:preliminaries} for further details.}

\paragraph{Sampler Graphs} A bipartite graph $G=(L,R,E)$ is an $(\varepsilon,\beta)$-function additive sampler if \(\forall f:R \to [0,1]\)
\[
    \Prob[v \in L]{\Abs{\Ex[u \in R, u \sim v]{f(u)} - \Ex[u \in R]{f(u)}} \geq \varepsilon} < \beta.
\]
$G$ is called an $(\alpha,\beta,\delta)$-multiplicative sampler if for all $f:R \to [0,1]$ of density at least $\Ex[u \in R]{f(u)} \geq \alpha$:
\[
\Prob[v \in L]{ \Abs{\Ex[u \in R, u \sim v]{f(u)} - \Ex[u \in R]{f(u)}} \geq \delta\Ex[u \in R]{f(u)}} < \beta.
\]
It is a well-known fact (c.f. \cite{ImpagliazzoKW2012}) that one can move between additive and multiplicative samplers without much parameter loss, and moreover that one can `flip' $L$ and $R$ in the sense that if $G=(L,R,E)$ is an $(\alpha,\beta,\delta)$-sampler, then $G^{op}=(R,L,E)$ is roughly a $(\beta,\alpha,\delta)$-sampler up to slight decay in parameters. We refer the reader to \pref{sec:preliminaries} for details. 

\subsection{Sampling and Concentration}
We can now state and overview the proof of optimal sampling on HDX in somewhat more formality.
\begin{theorem}[Sampling on HDX (Informal \pref{thm:hdx-is-sampler})]\label{thm:intro-hdx-is-sampler}
    Let $X$ be a nice, $k$-\maximal HDX and $i \leq k$. The containment graph \((\X[k],\X[i])\) is a \((\varepsilon, \beta)\)-function sampler for 
    \[
    \beta \leq \exp\left(-\Omega\left(\varepsilon^2 \frac{k}{i}\right)\right).
    \]
    In other words, for any $f: \X[i] \to [0,1]$ of expectation $\mu$:
    \[
        \Prob[s \in {\X[k]}]{\Abs{\Ex[t \subseteq s]{f(t)} - \mu} \geq  \varepsilon} < \exp\left(-\Omega\left(\varepsilon^2 \frac{k}{i}\right)\right).
    \]
\end{theorem}

Since $\sqrt{d}$-skeletons of $\lambda$-TD complexes are nice, we get the following immediate corollary.
\begin{corollary}[Sampling for TD-Complexes (Informal \pref{claim:td-si-are-nice} + \pref{lem:raising})]
    Fix $\lambda<1$ and let $X$ be a $d$-\maximal $\lambda$-TD complex. Then for any $i < k \leq \sqrt{d}$, the containment graph of \((\X[k],\X[i])\) is a \((\varepsilon, \beta)\)-sampler for
    \[
    \beta \leq \exp\left(-\Omega\left((1-\lambda)^{\frac{1}{2}}\varepsilon^2 \frac{k}{i}\right)\right).
    \]
\end{corollary}
We remark that this (and all following) results also hold in some variation under the popular notion of \textit{spectral independence}, which for our purposes act similar to TD complexes. We discuss in the main body.

En route to \pref{thm:intro-hdx-is-sampler}, we also show the \textit{top level} of any $\lambda$-TD complex satisfies exponential concentration for Lipschitz functions, which may be of independent interest.
\begin{theorem}[Exponential Concentration of $\lambda$-TD Complexes (Informal \pref{cor:TD-exponential})]\label{cor:intro-TD-exp}
    For any $\lambda < 1$, $\nu >0$, and $k \in \mathbb{N}$, let $X$ be a $k$-\maximal $\lambda$-TD complex. For any $\nu$-Lipschitz function $f: \X[k] \to \R$:
    \[
    \Pr_{s \in \X[k]}[|f(s)-\mathbb{E}[f]| \geq t] \leq \exp\left(-\Omega\left(\frac{t}{\sqrt{c_\lambda\nu}}\right)\right)
    \]
    where $c_\lambda \leq O(e^{\frac{\lambda}{1-\lambda}})$.
\end{theorem}
We refer the reader to \pref{sec:prelims-concentration} for the definition of $\nu$-Lipschitz, and here just note that specialized to inclusion sampling this implies $(\X[k],\X[i])$ is roughly a $(\varepsilon,\exp(-\varepsilon\frac{\sqrt{k}}{i}))$-sampler. A similar statement holds under spectral independence, albeit for a more restricted class of functions (see \pref{cor:SI-exponential}).

\paragraph{Proof Overview} The proof of \pref{thm:intro-hdx-is-sampler} is broken into two main components: the `Chernoff' ($i=1$) setting, and a bootstrapping argument lifting Chernoff to general $i$. Here we focus just on the simplest version of the proof, the case of \textit{two-sided HDX}, and refer the reader to \pref{sec:concentration-all-nice-hdx} for the general case.

\paragraph{Part I: Chernoff-Hoeffding} Toward \pref{thm:intro-hdx-is-sampler}, we first prove a fairly general ``$k$ vs $1$'' concentration theorem for any HDX $X$ whose strength depends on $X$'s underlying quantitative expansion. Under strong assumptions ($2^{-k}$-HDX), we recover true Chernoff. Perhaps more surprisingly, even under very weak assumptions ($\frac{1}{k}$-HDX) we still recover strong concentration of the form $\exp(-\varepsilon^2 \sqrt{k})$. Taking a $\sqrt{k}$-skeleton of the latter recovers Chernoff (see \pref{sec:prelims-concentration}). We give a simplified statement here:
\begin{theorem}[Chernoff-Hoeffding for HDX (Informal \pref{thm:split-concentrate})]\label{thm:intro-split-concentrate}
    Let $X$ be a $k$-\maximal $\frac{1}{k}$-two-sided HDX. Then for any $f: \X[1] \to [0,1]$:
    \[
    \underset{\{v_1,\ldots,v_k\} \in \X[k]}{\Pr}\left[\left|\frac{1}{k}\sum\limits_{i=1}^k f(v_i) - \mu\right| > \varepsilon\right] \leq \exp(-\Omega(\varepsilon^2\sqrt{k}))
    \]
    Moreover, if $X$ is a $2^{-k}$-two-sided HDX stronger concentration holds:
    \[
    \underset{\{v_1,\ldots,v_k\} \in \X[k]}{\Pr}\left[\left|\frac{1}{k}\sum\limits_{i=1}^k f(v_i) - \mu\right| > \varepsilon\right] \leq \exp(-\Omega(\varepsilon^2k))
    \]
\end{theorem}

For simplicity let $f=1_A$ be the indicator of some $A \subset X(1)$. The proof of \pref{thm:intro-split-concentrate} follows the standard Chernoff-style method of bounding the moment generating function (MGF). Namely for any $r>0$:
\[
\underset{\{v_1,\ldots,v_k\} \in \X[k]}{\Pr}\left[\sum_{i=1}^k 1_A(v_i) \geq k(\mu+\varepsilon)\right] = \Pr\left[\exp\left(\sum_{i=1}^k r1_A(v_i)\right) \geq \exp(rk\left(\mu+\varepsilon\right))\right] \leq \frac{\mathbb{E}\left[\prod\limits_{i=1}^k e^{r1_A(v_i)}\right]}{e^{rk(\mu+\varepsilon)}}.
\]
In the classical proof of Chernoff, the variables $v_i$ are independent, so one can bound the MGF by
\begin{equation}\label{eq:intro-Chern}
\underset{\{v_1,\ldots,v_k\} \in \X[k]}{\mathbb{E}}\left[\prod\limits_{i=1}^ke^{r1_A(v_i)}\right] = \underset{v \in \X[1]}{\mathbb{E}}\left[e^{r1_A(v)}\right]^k \lesssim e^{rk(\mu + r)}
\end{equation}
and set $r=\Theta(\varepsilon)$ to get the desired bound. 

When $X$ is an HDX, the variables $\{v_1,\ldots,v_k\}$ may be extremely correlated, so the above approach breaks down naively. Instead, we take inspiration from Healy's proof \cite{Healy2008} of the Expander-Chernoff Theorem and recursively bound the MGF by `splitting' it into components on lower levels of $X$ using the swap walks. Toward this end, for any $\ell \leq k$, define the partial MGF $z_{\ell}: \X[\ell] \to \R$ on a face $t=\{v_1,v_2,\ldots,v_\ell\}$ as
\[
z_\ell(v_1,v_2,\dots,v_\ell) = \prod_{i=1}^\ell e^{r1_A(v_i)} = e^{r|t \cap A|}.
\] 
Our goal is to bound $\mathbb{E}[z_k] \approx \mathbb{E}[z_1]^k$. To set up a recursion, we appeal to the elementary observation that a $k$-set $s \in \X[k]$ can be sampled by first drawing a $\frac{k}{2}$-set $t_1 \in \X[k/2]$, then $t_2 \in \X[k/2][t_1]$ conditionally from its link. We set \(s=t_1 \dunion t_2\). This allows us to `split' $z_k$ into two correlated copies of $z_{k/2}$:
\begin{align*}
\underset{s \in \X[k]}{\mathbb{E}}[z_{k}(s)] &=\underset{s \in \X[k]}{\mathbb{E}}[e^{r|s \cap A|}]\\
&=\underset{t_1 \dunion t_2}{\mathbb{E}}[e^{r|t_1 \cap A|}e^{r|t_2 \cap A|}]\\
&= \underset{t_1 \dunion t_2}{\mathbb{E}}[z_{k/2}(t_1)z_{k/2}(t_2)].
\end{align*}
The trick is now to observe that $(t_1,t_2)$ is distributed exactly as an edge in the swap walk $S_{\frac{k}{2},\frac{k}{2}}$. Since swap walks on HDX have excellent spectral expansion \cite{AlevJT2019,DiksteinD2019,GurLL2022,alev2023sequential}, we can `de-correlate' the above and write the relation
\begin{equation} \label{eq:recurse-1}
    \ex{z_k} \leq (1-\gamma)\ex{z_{k/2}}^2 + \gamma \mathbb{E}[z_{k/2}^2]
\end{equation}
where $\gamma$ is the expansion of the swap walk. 

We could now recurse if not for the error term $\gamma \mathbb{E}[z_{k/2}^2]$, which could very well be the dominating factor. The key is show $\ex{z_{k/2}^2}$ actually can't be too much bigger than $\ex{z_{k}}$ itself:
\begin{equation}\label{eq:intro-internal-reduce}
\ex{z_{k/2}^2} \lesssim e^{r^2k}\ex{z_{k}}.
\end{equation}
Setting $r$ such that $\gamma e^{r^2k} \ll 1$,  \eqref{eq:recurse-1} becomes
\[ 
\ex{z_k} \lesssim \frac{1-\gamma}{1-e^{r^2 k} \gamma} \ex{z_{k/2}}^2 \approx \ex{z_{k/2}}^2,
\]
allowing us to recurse to show an upper bound similar to \eqref{eq:intro-Chern}. For `weak' HDX, this strategy requires setting $r=\Theta(\varepsilon k^{-1/2})$, while for strong HDX we may set $r=\Theta(\varepsilon)$. This results in the gap in the stated bounds.

To prove \eqref{eq:intro-internal-reduce} we in some sense `reverse' the sampling process above, and draw a $\frac{k}{2}$-set $t$ by first sampling a $k$-set $s \in \X[k]$, then take $t \subset s$ uniformly at random. This means we can write the error term as:
\[
\underset{t \in \X[k/2]}{\mathbb{E}}[z_{k/2}(t)^2] = \Ex[s \in {\X[k]}]{\Ex[t \subseteq s]{z_{k/2}(t)^2}} =  \Ex[s \in {\X[k]}]{\Ex[t \subseteq s]{e^{2r|t \cap A|}}}.
\]
Crucially, the inner expectation is now over the $\frac{k}{2}$-\maximal \textit{complete complex} (whose vertices are the $k$-set $s$), so standard concentration implies the exponent, $2r|t \cap A|$, is close to its expectation $r|s \cap A|$ up to a factor of roughly $r^2k$ with high probability. Thus:
\[
\Ex[s \in {\X[k]}]{\Ex[t \subseteq s]{e^{2r|t \cap A|}}} \lesssim \Ex[s \in {\X[k]}]{e^{r|s \cap A|+r^2k}} = e^{r^2k}\underset{s \in \X[k]}{\mathbb{E}}[z_k].
\]

\paragraph{Part II: Bootstrapping} We now argue one can lift Chernoff for HDX to optimal sampling for their inclusion graphs $(\X[k],\X[i])$. The idea, based on the method in \cite{ImpagliazzoKW2012} for the complete complex, is to try to reduce to a $\frac{k}{i}$-\maximal complex $Y$ whose vertices are $i$-sets in $\X[i]$, and whose $\frac{k}{i}$-faces correspond in some way to the $k$-sets in $X$. We could then hope to apply Chernoff to this system to bound the $k$ vs.\ $i$ sampling behavior in the original complex.

Somewhat more formally, we'll consider the $\frac{k}{i}$-\maximal complex generated by the following process: draw a random k-face $\{v_1,\ldots,v_{k}\} \in \X[k]$, and partition the face randomly into $\frac{k}{i}$ subsets of size $i$ denoted $\{I_1,\ldots,I_{\frac{k}{i}}\}$. The resulting complex, called the \textit{faces complex} of $X$, inherits local-spectral expansion from $X$'s swap walks so we can apply Chernoff for HDX to get concentration of the form $\exp(-\varepsilon^2\frac{k}{i})$. 

Unfortunately, this does not prove concentration for quite the right object. In particular, given a function $f: \X[i] \to [0,1]$, the above really states that a random partitioning of a $k$-face $s$ into $\frac{k}{i}$ sub $i$-faces satisfies
\[
\Pr_{s \in \X[k],~\dunion_j I_j = s}\left[\left|\frac{i}{k}\sum\limits_{j=1}^{k/i}f(I_j) - \mu\right| > \varepsilon\right] \leq \exp\left(-\varepsilon^2\frac{k}{i}\right),
\]
whereas we'd like to bound $\underset{s \in \X[k]}{\Pr}\left[\left|\underset{t \subset s}{\mathbb{E}}[f(t)] - \mu\right| > \varepsilon\right]$. Here we are saved by the fact that each $I_j$ is \textit{marginally} distributed as a random $i$-set $t \subset s$. We prove in such cases it is possible to inherit sampling from the correlated bound up to a small loss in parameters, completing the proof.

\subsection{Reverse Hypercontractivity} \label{sec:rhc-intro-results}
Leveraging \pref{thm:intro-hdx-is-sampler} \textit{locally} in $X$, we prove reverse hypercontractivity for high dimensional expanders. 
\begin{theorem}[Reverse Hypercontractivity (Informal \pref{thm:intro-rhc-real})]\label{thm:real-intro-rhc-real}
    Fix $\rho \in (0,1)$ and let \(X\) be a \(k\)-\maximal nice complex for $k$ sufficiently large. Then there exist constants $C,q$ (depending only on $\rho$) such that
    \[
    \langle f, T_\rho g \rangle \geq C\norm{f}_q\norm{g}_q.
    \]
    In other words $T_\rho$ is ``$(C,q)$-reverse hypercontractive.''
\end{theorem}

The core of \pref{thm:real-intro-rhc-real} is really a version of the result for the down-up walk which simply states any two subsets $A,B \subset \X[k]$ of non-trivial size must have correspondingly many edges between them.
\begin{lemma}[Reverse Hypercontractivity of the Down-Up Walk (Informal \pref{thm:indicator-reverse-hc})]\label{lem:intro-indicator-rhc}
    Fix $\gamma \in (0,1)$ and let $X$ be a $k$-\maximal nice complex for $k$ sufficiently large. There exist constants $c,q$ (depending only on $\gamma$) such that for any $A,B \subset \X[k]$ of measure at least $\exp(-c k)$:
    \[
    \Pr_{s,s' \sim U_{\gamma k, k}D_{k,\gamma k}}[s \in A, s' \in B] \geq \Pr[A]^q\Pr[B]^q.
    \]
\end{lemma}
Note that the assumption on the set size in \pref{lem:intro-indicator-rhc} is essentially tight. Even if one walks from \(s\) to \(s'\) while going down to a single vertex it is possible to find sets $A$ and $B$ of size roughly $\exp(-k)$ which are totally disconnected, even on the complete complex.

\paragraph{Proof Overview} The proof of \pref{lem:intro-indicator-rhc} relies on the following simple observation: since the down-up walk samples an edge by first sampling \(t \in \X[\gamma k]\) and then \textit{independently} sampling \(s,s' \in \X[k]\) containing \(t\), we can express the $\gamma$-correlated mass of $A$ and $B$ as:
    \[
    \Prob[s,s'\sim U_{\gamma k,k} D_{k, \gamma k}]{s \in A, s' \in B} = \Ex[t \in {\X[\gamma k]}]{\cProb{s}{A}{s \supseteq t} \cProb{s}{B}{s \supseteq t}}.
    \]
Now the connection with inclusion sampling becomes clear---if even a non-negligible fraction of $t\in \X[\gamma k]$ see a reasonable portion of $A$ and $B$ simultaneously, the righthand side should be large!

Naturally our first thought might be to appeal to \pref{thm:intro-hdx-is-sampler} which bounds this type of behavior, but when $i=\gamma k$ the resulting bound is far too weak to be useful. In particular, after swapping from additive to multiplicative sampling and `flipping' the graph, \pref{thm:intro-hdx-is-sampler} roughly states that $(\X[i],\X[k])$ is an $(\alpha,\beta,\delta)$-sampler for $\alpha \approx \exp(-\delta^2 \beta \frac{k}{i})$. When $i=\Theta(k)$, this only gives a guarantee for sets of constant size.

The key is to realize that because $X$ expands \textit{locally}, we don't have to sample $t \in \X[\gamma k]$ all at once. Instead, we'll sample $t$ in `steps' as $t=t_1 \cup \ldots \cup t_m$ where each $t_j$ is sampled conditionally (i.e.\ from the link of) its predecessors, and apply \pref{thm:intro-hdx-is-sampler} to each $t_j$ individually. Toward this end, assume $\Pr[A] \leq \Pr[B]$ and fix the `step size' $\ell=|t_j|$ roughly such that $\exp(-c'\frac{k}{\ell}) = \Pr[A]$ for some sufficiently small $c' >0$.

Drawing our first component $t_1$, observe for small enough $c'$ we have set our parameters such that \pref{thm:intro-hdx-is-sampler} promises $(\X[\ell],\X[k])$ is an $(\alpha,\frac{1}{4},\frac{1}{2})$-sampler for $\alpha \leq \min\{\Pr[A],\Pr[B]\}$, so with probability at least $\frac{1}{2}$ the conditional measures of $A$ and $B$ in $t_1$ are at least $\frac{1}{2}\Pr[A]$ and $\frac{1}{2}\Pr[B]$ respectively. Moreover, moving into the link of $t_1$ we can apply exactly the same process to $t_2$ (and so forth), so by the $i$-th step one can inductively show that with probability at least $2^{-i}$, the conditional mass of $A$ and $B$ are at least  $2^{-i}\Pr[A]$ and $2^{-i}\Pr[B]$ respectively. On the other hand, we only take $m=\frac{\gamma k}{\ell} = O(\log(\Pr[A]^{-1}))$ total steps by our choice of $\ell$, so at the end of the procedure we have that with probability at least $2^{-m}=\Pr[A]^{O(1)}$, the density of $A$ and $B$ in $t$ is at least $\Pr[A]^{O(1)}$ as well. Putting everything together, we have
    \begin{align*}
    \Prob[s,s'\sim U_{\gamma k,k} D_{k, \gamma k}]{s \in A, s' \in B} &= \Ex[t \in {\X[\gamma k]}]{\cProb{s}{A}{s \supseteq t} \cProb{s}{B}{s \supseteq t}}\\
    &\geq 2^{-m}\cdot (2^{-m}\Pr[A]) \cdot (2^{-m}\Pr[B])\\
    &\geq \Pr[A]^{O(1)}\Pr[B]
    \end{align*}
as desired.

Once one has \pref{lem:intro-indicator-rhc}, \pref{thm:real-intro-rhc-real} follows by passing to discrete approximations of $f$ and $g$ that carefully balance the quality of approximation (namely regarding the mass of $f$ and $g$ and pointwise closeness) and coarseness (number of discrete values in the approximation). Done correctly, one can then divide the functions into boolean level-sets and apply \pref{lem:intro-indicator-rhc} to achieve full reverse hypercontractivity at the cost of a slight decrease in the corresponding norm and an additional constant factor. 

\subsection{Optimality and the Trickling-Down Threshold}\label{sec:intro-optimality}

Oppenheim's Trickling Down Theorem is among the most fundamental results in the theory of spectral high dimensional expansion \cite{oppenheim2018local}. It states that any connected complex whose top links have expansion \textit{strictly} better than $\frac{1}{d-1}$ exhibits `local-to-global' behavior: one can immediately give quantitative bounds on the spectral expansion of \textit{all} links of $X$ (including its `global' $1$-skeleton), as well as on the down-up walks \cite{alev2020improved}. Our results are in some sense a strengthening of the local-to global behavior of Oppenheim: essentially as soon as one passes this threshold, not only can one infer `Chernoff'-type concentration on skeletons of $X$, $X$ itself actually satisfies exponential concentration for any Lipschitz function. This is in stark contrast to complexes \textit{at} the TD-Barrier, which we observe may have no concentration properties whatsoever:
\begin{theorem}[Lower Bounds at the TD-Barrier (\pref{prop:TD-barrier})]\label{thm:intro-TD-barrier}
    For every $\beta<1$ and $k \in \mathbb{N}$, there exists a family of $k$-\maximal $1$-TD complexes $\{X_n\}$ such that $(\X[k][n],\X[1][n])$ is not a $(\frac{1}{2},\beta)$-additive sampler.
\end{theorem}

\pref{thm:intro-hdx-is-sampler} also gives essentially the best possible parameters for inclusion sampling on any complex.
\begin{theorem}[Inclusion Sampling Lower Bounds (Informal \pref{thm:optimal-sampling})]\label{thm:intro-lower-bound}
    Let $i < k \in \mathbb{N}$, $\varepsilon \in (0,0.1)$, and let \(X\) be a $k$-\maximal complex, then one of the following holds:
    \begin{enumerate}
        \item $(\X[k],\X[i])$ is not an $(\varepsilon,\beta)$-sampler for $\beta \leq \exp \left(-O\left (\varepsilon^2 \frac{k}{i} \right ) \right)$.
        \item $(\X[k],\X[i])$ is not an  \((\frac{1}{10i},\frac{\varepsilon}{6i})\)-sampler.
    \end{enumerate}
\end{theorem}
Classical lower bounds on samplers (e.g.\ \cite{CanettiEG1995}) are based on \textit{degree} and are far from tight in our setting where inclusion structure is the main barrier. Nevertheless, we are able to argue our bounds are optimal by reducing to $i=1$ where the degree bound is tight. The idea is to show any set $A$ which is a counter-example to sampling of $(\X[k],\X[1])$ can be `lifted' to a counter example to $(\X[k],\X[i])$ by taking $A'$ to be the family of $i$-sets that hit $A$. This works when the measure of $A'$ is roughly $i\mu(A)$, and therefore when 1) $A$ is small, and 2) $(\X[k],\X[i])$ is a sampler. We show the original technique of \cite{CanettiEG1995} can be used to analyze sets of any density, thereby giving the desired counter-example.

\subsection{Agreement Testing} \label{sec:agree-intro}
Leveraging \pref{thm:real-intro-rhc-real}, we prove several new agreement theorems in both the $99\%$ and $1\%$-regimes. In the former we consider a general setup of \cite{DinurFH2019}. Given a complex $X$ and an ensemble of functions $\mathcal{F} = \{f_s: \binom{s}{i} \to \mathbb{F}_2\}_{s \in \X[k]}$, we'd like to test whether $\mathcal{F}$ `comes from' a global function $g: \X[i] \to \mathbb{F}_2$. To be concrete, consider $i=1$. Here we are given \(\mathcal{F} = \set{f_s:s \to \mathbb{F}_2}\) and can interpret each \(f_s\) as the indicator of some subset \(r_s \subseteq s\). We want to test whether there is a `global' subset \(R \subseteq \X[1]\) that approximately determines most local subsets in the sense that \(r_s \approx R \cap s\). For \(i=2\), a similar interpretation holds but for \emph{subgraphs} of \(X\). Namely, think of \(f_s:\binom{s}{2} \to \mathbb{F}_2\) now as a subgraph of \(s\), i.e.\ corresponding to the indicator of an edge-set \(e_s \subseteq \binom{s}{2}\). Now `coming from' a global function means there is a global subgraph \(E \subseteq \X[2]\) such that for most local graphs, \(e_s \approx E \cap s\). For \(i>2\) one interprets this similarly as a question about local vs.\ global sub-complexes of \(X\). Such `subcomplex' tests have proved useful both in boolean analysis \cite{DinurFH2019} and as a tool for analyzing the more challenging $1\%$-regime \cite{dikstein2023agreement,DiksteinDL2024}.

We study the classical $V$-test in this context, which draws a pair of $k$-faces $(s,s') \sim U_{k/2,k}D_{k,k/2}$ intersecting on $k/2$ vertices and checks if $s$ and $s'$ agree on most $i$-faces in the intersection. Let $Agree^V_\eta(\mathcal{F})$ denote the probability $(s,s')$ agree on at least a $1-\eta$ fraction of their shared $i$-faces. Similarly, for any global function $g:\X[i] \to \mathbb{F}_2$, write $f_s \overset{\eta}{\approx} g|_s$ if $f_s$ agrees with $g$ on a $1-\eta$ fraction of its $i$-faces. Then:

\begin{theorem}[The Subcomplex V-Test, $99\%$-Regime (Informal \pref{thm:agreement-99-percent})]\label{thm:intro-agreement-99-real}
     Let $X$ be a $k$-\maximal nice complex. For any $\eta,\varepsilon >0$ if $Agree^V_\eta(\mathcal{F}) \geq 1-\varepsilon$:
     \[
     \exists g:\X[1] \to \mathbb{F}_2, \; \; \Pr_{s\in \X[k]}[f_s \overset{\eta'}{\approx} g|_s] \geq 1-\varepsilon'.
     \]
     where $\eta' \leq O(\eta+\varepsilon)$ and $\varepsilon' \leq \varepsilon^{O_{\eta}(1)}+e^{-\Omega_{\eta}(k)}$.
\end{theorem}

In the `$1\%$'-regime, our goal is to infer global structure of $\mathcal{F}$ even when the test passes with small (but non-negligible) probability. We focus on the $i=1$ setting and aim to construct a test $\mathcal{T}$ that infers global structure whenever $Agree^{\mathcal{T}}(\mathcal{F}) \geq \exp(-\Omega(k))$, the best possible bound since a random function passes any such test with this probability. To this end, we consider a variant of \cite{ImpagliazzoKW2012}'s $Z$-test on high dimensional expanders which samples a triple $(\sigma,\sigma',\sigma'') \in \X[k]$ roughly such that $(\sigma,\sigma')$ is distributed as the V-test, and $\sigma''$ is drawn from the link of $\sigma' \setminus \sigma$. The $Z$-test passes if $(\sigma,\sigma')$ and $(\sigma',\sigma'')$ are consistent on their intersections. We prove that the Z-test is sound under the stronger assumption that the complex is ``$\lambda$-global'', meaning the swap walk $S_{\frac{k}{2},\frac{k}{2}}$ is $\lambda$-close to its stationary distribution in $\ell_\infty$-norm (see \pref{def:global}).
\begin{theorem}[The Z-Test, $1\%$-Regime (Informal \pref{thm:Z-test})]\label{thm:Z-test-intro}
    $\forall \lambda,\eta>0$ and large enough $k$, let $X$ be a \(\lambda\)-global $k$-\maximal nice complex. Then for any $\delta \in (8\lambda + e^{-\Omega(\eta k)},\frac{1}{8})$ if $ Agree_0^Z(\mathcal{F}) \geq \delta$:
    \[
        \exists g:\X[1] \to \mathbb{F}_2, \; \; \Pr_{s\in \X[k]}[f_s \overset{\eta}{\approx} g|_s] \geq \delta/8.
    \]
\end{theorem}

We will not give a true proof overview of our testing results (though it should be mentioned that the high level strategy closely follows \cite{ImpagliazzoKW2012,DinurL2017}), but wish to highlight that the core of \pref{thm:Z-test} is really a `weak' \textit{local} agreement theorem for the V-test on nice HDX of independent interest. The statement is quite technical, so we give a very informal version here that captures it in spirit.
\begin{theorem}[The Local Agreement Theorem (Informal \pref{thm:local-agreement})]\label{thm:intro-local-agreement}
    Let $X$ be a nice complex. Then for any $\delta \geq \exp(-\Omega(k))$ if $Agree_0^V \geq \delta$, for an $\Omega(\delta)$-fraction of $t \in \X[k/2]$ there is a `smoothing' of $\mathcal{F}$ which maintains its structure and has $(1-\delta^2)$-agreement on the link $X_t$.
\end{theorem}
In other words, any non-trivial agreement of the V-test \textit{must} come from the fact that $\mathcal{F}$ is locally consistent with a global function. It is possible \pref{thm:intro-local-agreement} could be propogated to a true Z-test under much weaker conditions than $\lambda$-globality, e.g.\ under the recent topological notions of \cite{GotlibK2022,dikstein2023agreement,bafna2023characterizing}. We leave this as an open question for the $1\%$-regime.

\subsection{Further Applications} \label{sec:applications-intro}
Finally, we discuss our applications of \pref{thm:intro-hdx-is-sampler} and \pref{thm:real-intro-rhc-real} beyond agreement testing.

\paragraph{Geometric Overlap} A complex $X$ has \((d,c)\)-geometric overlap if for every embedding $\X[1]$ into \(\mathbb{R}^{d-1}\), there is a point \(q \in \mathbb{R}^{d-1}\) that lies in the convex hulls of at least a \(c\)-fraction of $X$'s embedded \(d\)-faces. \cite{FoxGLNP2012} proved every $d$-\maximal bounded degree complex has geometric overlap at best $c_d$, where $c_d$ is the overlap of the complete complex. We prove sufficiently strong high dimensional expanders match this bound.
\begin{theorem}[Geometric Overlap (Informal \pref{thm:hdx-geometric-overlap})]
    For every $d \in \mathbb{N}$ and $\varepsilon>0$, there exists $\lambda>0$ such that any $\lambda$-two-sided HDX with uniform vertex weights has $(d,c_d-\varepsilon)$-geometric overlap.
\end{theorem}

We remark that while geometric overlap is a classical problem in mathematics (dating back to \cite{BorosF1984}), we are not aware of any applications in computer science. Nevertheless other instances of overlap theorems have proven powerful (see e.g.\ \cite{chase2023local}), so it seems plausible such results may be of future use.

\paragraph{Double Samplers} Double samplers are a powerful variant of sampler graphs used in \cite{DinurHKLT2018} to construct good list-decodable codes and in \cite{DoronW2022} for the heavy hitters problem. Roughly speaking, double samplers are `three-wise' inclusion structures $(\X[k],\X[j],\X[1])$ such that 
\begin{enumerate}
    \item $(\X[k],\X[j])$ is a $(1/2,\beta)$-additive sampler
    \item For all $s \in \X[k]$, the restriction of $(\X[j],\X[1])$ to vertices in $s$ is a $(1/2,\beta)$-additive sampler
\end{enumerate}
The best known prior construction of double samplers \cite{DinurHKLT2018} had overhead $\frac{|\X[k]|}{|\X[1]|}=\exp(\poly(\beta^{-1}))$, leading to poor rate of their resulting list-decodable codes. One of the main open questions asked in \cite{DinurHKLT2018} was to determine the optimal overhead of a double sampler.

We nearly resolve this problem for `typical' complexes $X$ which are non-contracting (i.e.\ $|\X[j]|>|\X[i]|$ whenever $j>i$), and satisfy a weak `hitting-set' type guarantee. In particular, \pref{thm:intro-hdx-is-sampler} gives double samplers with only \textit{quasi-polynomial overhead}, while for such typical complexes \pref{thm:intro-lower-bound} implies a corresponding near-matching lower bound.
\begin{theorem}[Near-Optimal Double Samplers (Informal \pref{thm:double-samplers})]
    For every $\beta>0$, there exists an infinite family of double samplers with $\exp(O(\log^6(\beta^{-1})))$ overhead. Moreover, any `typical' double sampler has overhead at least $\exp(\Omega(\log^2(\beta^{-1})))$.
\end{theorem}
We refer the reader to \pref{sec:double-samplers} for the exact restrictions on $X$ for which the lower bound holds.

\paragraph{Locally Testable and List-Decodable Codes} While it is likely applying our double sampler machinery to the arguments of \cite{DinurHKLT2018} would give list-decodable codes with substantially improved rate, this was already achieved subsequent to \cite{DinurHKLT2018}'s work using walks on expanders \cite{AlevJQST2020}. Instead, using a variant of our machinery we give a closely related application to distance amplification for locally testable codes with near-optimal distance-alphabet trade-off. The resulting codes are list-decodable and preserve testability up to a log factor in the alphabet.

\begin{theorem}[Large Distance List-Decodable LTCs (Informal \pref{thm:codes})]\label{thm:intro-codes}
    For all large enough $k \in \mathbb{N}$ and $\varepsilon>0$, there exists an explicit family of $\mathbb{F}_2$-linear codes with
    \begin{enumerate}
    \item Distance: $1-2^{-k}-\varepsilon$
    \item Alphabet Size: $2^k$
    \item Rate: $\varepsilon^{O(k)}$
\end{enumerate}
    Moreover, the codes are efficiently list decodable up to distance $1-2^{-\Omega(k)}$ and locally testable in $O(1)$ queries with soundness $\Omega(\frac{1}{k\log(1/\varepsilon)})$.
\end{theorem}
Prior distance amplification techniques for LTCs \cite{kopparty2017high} have exponentially worse distance-alphabet trade-off, but improved rate. To our knowledge the above are the first large distance list-decodable LTCs. We also give a candidate technique for amplifying distance while maintaining \textit{constant} soundness (removing the $k$ factor in the denominator above) using HDX, and prove the method works on the complete complex. We refer the reader to \pref{sec:codes} for further details on these notions. 

\paragraph{Degree Lower Bounds for HDX} One of the formational results in the study of expansion is the Alon-Boppana theorem \cite{nilli1991second}, roughly stating any family of bounded-degree $\lambda$-expanders have degree at least $\frac{2}{\lambda^2}$. Despite degree being a critical parameter in application (degree controls the `blow-up' induced by using an HDX as a gadget), in high dimensions we understand very little about optimal degree. The best known constructions of $\lambda$-TD complexes have degree $\lambda^{\tilde{O}(d^2)}$. Is this optimal? Could $\lambda^{o(d^2)}$ be achieved?

We take the first step toward answering this question. Combining our sampling results with degree lower bounds of \cite{CanettiEG1995}, we show super-exponential lower bounds for certain special classes of HDX, including polynomial skeletons of any hyper-regular $\lambda$-TD complex. 
\begin{theorem}[Degree Lower Bounds (Informal \pref{thm:hdx-degree-TD})]\label{thm:intro-deg}
Fix $\lambda < 1$, $d \in \mathbb{N}$, and $k \leq \sqrt{d}$. Let $X$ be the $k$-skeleton of a $d$-\maximal hyper-regular $\lambda$-TD complex. Then 
\[
\text{deg}(X) \geq 2^{\Omega_{\lambda}(k^2)}
\]
\end{theorem}
\pref{thm:intro-deg} exhibits yet another threshold phenomenon at the TD-barrier, since there exist $1$-TD complexes \cite{Golowich2021} which have exponential degree at any level. While the latter are not hyper-regular, they are still reasonably balanced and our techniques extend to this regime.\footnote{Namely, we only really require that the underlying graph of every link does not have a (roughly) $\exp(-\sqrt{d})$ fraction of vertices making up $1/2$ the measure. Unfortunately, constructions such as the Ramanujan complex have extremely unbalanced links, and actually do fail this property. See \pref{sec:degree} for further discussion.}

If one could prove a true Chernoff bound for $\lambda$-TD complexes (rather than our concentration of $\exp(-\sqrt{k})$), \pref{thm:intro-deg} could be improved to showing degree $2^{\Omega(d^2)}$ at the \textit{top level} of $X$, albeit still under the assumption of hyper-regularity. While bounded-degree hyper-regular HDX do exist \cite{FriedgutI2020}, their degree is substantially worse and removing this constraint from our lower bound remains an important open problem.

\paragraph{Separating MLSI and Reverse Hypercontractivity} The Modified Log-Sobolev Inequality (MLSI) is a powerful analytic inequality used to bound the mixing time of Markov chains and previous weakest sufficient condition for reverse hypercontractivity \cite{MosselOS2013}. To the best of our knowledge, it was open until this work whether MLSI was \textit{necessary}. We resolve this question for the setting allowing a leading constant $C>1$.
\begin{corollary}[Separating RHC and MLSI (Informal \pref{cor:MLSI})]
    There exist constants $C,q,q'$ such that for infinitely many $N \in \mathbb{N}$, there exist $(C,q,q')$-reverse hypercontractive operators on $N$ vertices with vanishing MLSI constant:
    \[
    \rho_{MLSI} \leq \tilde{O}\left(\frac{1}{\log(N)}\right)
    \]
\end{corollary}
This separation is essentially the strongest possible for operators with constant expansion, since any such operator has LSI (and therefore MLSI) at least $\Omega(\frac{1}{\log(N)})$ \cite{diaconis1996logarithmic}.

\paragraph{A Frankl–R\"{o}dl Theorem} The Frankl–R\"{o}dl Theorem \cite{frankl1987forbidden} is a broadly applied result in extremal combinatorics bounding the independence number of the graph on $\{0,1\}^n$ connecting strings of fixed intersection size. \pref{lem:intro-indicator-rhc} implies an analogous result for the $\gamma k$-step down-up walk on HDX.
\begin{corollary}[Frankl–R\"{o}dl for HDX (Informal \pref{claim:ind-set})]
    Fix $\gamma>0$ and let $X$ be a $k$-\maximal nice complex for $k$ sufficiently large. The $\gamma k$-step down-up walk has independence number 
    \[
    \alpha \leq \exp(-\Omega_\gamma(k))\cdot|\X[k]|
    \]
\end{corollary}
In comparison to Frankl–R\"{o}dl, the above bounds the size of sets avoiding intersections of size \textit{at most} instead of \textit{exactly} $(1-\gamma)k$. A similar issue appears when using reverse hypercontractivity toward this end in the cube and is handled by a separate Fourier analytic argument in \cite{Benabbas2012}. It would be interesting to extend our results to the exact case to give a true analog of Frankl–R\"{o}dl.

\paragraph{`It Ain't Over Till It's Over'} Reverse hypercontractivity was used to resolve Friedgut and Kalai's `It Ain't Over Till It's Over' Conjecture \cite{MosselOO2005}, which bounds the tail of random restrictions in the hypercube. The core of their proof was an analog for the closely related noise operator previously used in analysis of non-interactive correlation distillation protocols \cite{MosselOROSS2006}. We generalize the latter theorem to HDX:
\begin{theorem}[Noise Operator Tail Bounds (Informal \pref{thm:tail-bounds})]
    Fix $\rho \in (0,1)$ and $k \in \mathbb{N}$ sufficiently large. For any $k$-\maximal nice complex $X$ and $f:\X[k] \to [0,1]$ of density $\mu$:
    \[
    \Pr\left[T_\rho f \notin [\delta,1-\delta]\right] \leq \delta^{O_{\mu,\rho}(1)}
    \]
\end{theorem}

\section{Related Work}
\paragraph{High Dimensional Expanders}
Spectral high dimensional expansion was developed over a series of works \cite{EvraK2016,KaufmanM2017,oppenheim2018local,DinurK2017} building on prior notions of high dimensional expansion in topology \cite{Garland1973,LubotzkySV2005a,LinialM2006,Gromov2010,KaufmanL2014,KaufmanKL2014}. Basic higher order random walks were introduced by Kaufman and Mass \cite{KaufmanM2017}, with the swap operators later introduced independently in \cite{DiksteinD2019,AlevJT2019}. A great deal of work has been done at the intersection of high dimensional expanders and analysis, including determining the optimal spectral gap of the down-up walks \cite{KaufmanM2017,KaufmanO2020,DiksteinDFH2018,alev2020improved,kaufman2021garland}, developing a general theory of Fourier analysis \cite{kaufman2020local,DiksteinDFH2018,gaitonde2022eigenstripping,BafnaHKL2022,GurLL2022, gotlib2023fine}, and applications thereof to agreement testing \cite{DinurK2017,DiksteinD2019,kaufman2020local,GotlibK2022}, mixing of Markov Chains \cite{AnariLOV2019,anari2021spectral,chen2021optimal,chen2021rapid, liu2021coupling, blanca2022mixing, feng2022rapid, anari2023universality} (among many others), and algorithms \cite{DoronW2022,bafna2022high,beaglehole2023sampling}.

Two works \cite{KaufmanS2020,kaufman2021scalar} in the HDX literature also study Chernoff-style bounds, albeit in very different regimes from our results. Kaufman and Sharakanski \cite{KaufmanS2020} study concentration for global functions over repeated random walks on $k$-faces, which is incomparable to our setting. Kaufman, Kyng, and Song \cite{kaufman2021scalar} prove scalar and matrix Chernoff bounds under \textit{$\ell_\infty$-independence}, a strengthening of spectral HDX that only holds for dense complexes. While in this work we are primarily interested in proving concentration for \textit{bounded-degree} complexes, it is interesting to ask whether their stronger result (Chernoff for matrix-valued functions) extends to the sparse regime.

In \cite{BafnaHKL2022, GurLL2022}, the authors prove a notion called ``global'' hypercontractivity for high dimensional expanders. While it is known that reverse hypercontractivity follows from the standard hypercontractive inequality \cite{MosselOS2013}, we are not aware of a reduction from the weaker global variant. Our work also differs substantially from \cite{BafnaHKL2022, GurLL2022} in its tools and regime of application.

\paragraph{Samplers and Chernoff-Hoeffding}

Samplers are among the most classical tools in theoretical computer science, see e.g.\ \cite{Goldreich1997,shaltiel2004recent,wigderson2009randomness} for surveys of their many constructions and applications. De-randomized Chernoff bounds are an important sub-family of sampler graphs that exhibit optimal tails in the degree of the graph. Sparse examples were constructed in \cite{ajtai1987deterministic,Gillman1998} using walks on expanders, and there is a great deal of literature toward understanding what general families of hypergraphs admit such concentration, e.g.\ under limited independence \cite{schmidt1989aspects,schmidt1990analysis,schmidt1995chernoff,pemmaraju2001equitable}, negative correlation \cite{dubhashi1996balls}, for edge colorings \cite{PanconesiS1997}, under Dobrushin uniqueness \cite{paulin2014convex}, for strongly Raleigh distributions \cite{kyng2018matrix}, and most recently under $\ell_\infty$-independence \cite{kaufman2021scalar}. Such bounds have many interesting applications beyond those discussed in this work.

Inclusion samplers were first introduced in Impagliazzo, Kabanets, and Wigderson's \cite{ImpagliazzoKW2012} work on agreement testers and PCPs with strong soundness. Prior to \pref{thm:intro-hdx-is-sampler}, there were only two known systems of inclusion samplers with optimal tails: the complete complex \cite{ImpagliazzoKW2012}, and curves \cite{moshkovitz2017low}. Both examples were used toward the construction of PCPs with strong soundness, but failed to achieve optimal parameters due to their blow-up in size. Several sparse variants of inclusion samplers were known to admit weaker `Chebyshev' type bounds, including high dimensional expanders \cite{DinurK2017} and the Grassmann \cite{ImpagliazzoKW2012}. The latter also lead to PCPs with strong soundness \cite{dinur2011derandomized}, but failed to achieve optimal parameters due to its polynomial tail. \pref{thm:intro-hdx-is-sampler} is the first to achieve the best of both worlds, though there remain substantial barriers toward its use for PCPs.

\paragraph{Reverse Hypercontractivity}

Reverse hypercontractivity was first shown by Borell \cite{Borell1982} and first utilized in theoretical computer science by \cite{MosselOROSS2006} who observed the inequality implies fine-control of mixing between sets (akin to \pref{lem:intro-indicator-rhc}). The result has since found great use in hardness of approximation \cite{MosselOO2005,feige2007understanding,sherman2009breaking,Benabbas2012,KauersOTZ2014}, social choice \cite{mossel2012quantitative1,Keller2012,mossel2012quantitative2}, and extremal combinatorics \cite{Benabbas2012, KauersOTZ2014}. In 2013, \cite{MosselOS2013} extended known bounds on reverse hypercontractivity to general product spaces and settings with bounded modified Log-Sobolev Inequalities, a surprising result given the failure of standard hypercontractivity in these settings. Dinur and Livni-Navon \cite{DinurL2017} were the first to apply these results to agreement testing, where they resolved the conjecture of \cite{ImpagliazzoKW2012} regarding exponential soundness for the Z-test and built the first `combinatorial' PCPs with optimal soundness. 

We note that MLSI bounds (and therefore RHC) are known for many `dense' settings of HDX studied in the sampling literature (see e.g.\ \cite{chen2021optimal}), with the most general prior condition being the notions of entropic independence and fractional log-concavity \cite{anari2022entropic} which are significant strengthenings of spectral HDX. These methods all necessarily rely on density and cannot capture our theory. 

\paragraph{Agreement Tests}
Agreement testing, also known as direct product testing, is a powerful tool in the construction of PCPs and locally testable codes \cite{RubinfeldS1996,AroraS1997,RazS1997,GoldreichS1997,dinur2006assignment,dinur2007pcp,DiksteinDHR2020,DinurHKR2022}. Agreement tests in the $99\%$-regime were studied by \cite{GoldreichS1997} and \cite{DinurS2014} in the complete complex, and extended to high dimensional expanders in \cite{DinurK2017,DiksteinD2019,kaufman2020local}. \cite{DinurFH2019} gave the first `subcomplex' agreement theorem in the $99\%$-regime (the setting of \pref{thm:intro-agreement-99-real}) in the complete complex. Agreement testing in the $1\%$-regime was studied in the complete complex in \cite{dinur2008locally,ImpagliazzoKW2012,DinurL2017}, with the lattermost work giving optimal bounds. \cite{ImpagliazzoKW2012} also gave a de-randomized test over the Grassmann complex, which is polynomial size, but suffers from exponentially worse soundness. A different $1\%$-regime test on the Grassmann was also used to resolve the $2$-$2$ Games Conjecture \cite{KhotMS2017,DinurKKMS2018-grassman, BarakKS19}. 

Finally, \cite{GotlibK2022,bafna2023characterizing,dikstein2023agreement,BafnaM2024,DiksteinDL2024} studied the $1\%$ (or closely related \textit{list} agreement) regime on high dimensional expanders. These works identified key \textit{topological} properties any complex must exhibit to be a good tester, and subsequently constructed complexes satisfying these notions, giving the first bounded degree $1\%$-agreement testers. The methods of \cite{dikstein2023agreement,DiksteinDL2024} in particular rely on \pref{thm:intro-hdx-is-sampler} and \pref{thm:intro-agreement-99-real}. The soundness achieved by the above works, however, is inverse \textit{logarithmic} in the dimension instead of inverse exponential (or even polynomial). Achieving inverse exponential soundness as in \pref{thm:Z-test-intro} for families of bounded degree complexes remains an important open question.

\section*{Open questions}
\begin{enumerate}
    \item While we are able to show optimal concentration for a fairly broad class of high dimensional expanders, in the weakest settings (namely under spectral independence and at the TD-Threshold), we are only able to prove exponential concentration at the top level. It is unclear whether this is a fundamental or purely technical barrier: do such complexes satisfy a true Chernoff bound? As discussed above, a resolution of this question in the positive leads to better degree lower bounds for HDX, and in particular a $2^{\Omega(d^2)}$ lower bound for the top level of hyper-regular $\lambda$-TD complexes.
    \item The best known constructions of high dimensional expanders are not hyper-regular. Is it possible to remove this constraint from our degree lower bound? While our technique holds even for `reasonably balanced' complexes, it cannot handle objects like the Ramanujan complexes have extremely unbalanced links. It seems likely this is a technical rather than inherent barrier, and we conjecture some finer notion of concentration or application thereof may be able to remove this constraint.
    \item In \pref{sec:concentration-all-nice-hdx} we prove near-optimal concentration for partite HDX, but our bounds suffer polynomial loss in $\varepsilon$ over the two-sided case. The reason boils down to understanding concentration of a simple hypergraph we call the \textit{swap complex}, made up of $\frac{k}{i}$-tuples of disjoint $i$-sets $s \in \binom{[n]}{i}$. We prove a Chernoff bound for this hypergraph when $n \gtrsim \frac{k}{\varepsilon^5}$. Does such a bound hold in the regime where $n=\Theta(k)$? This would imply fully optimal concentration for partite HDX.
    \item We show that `standard' simplicial complexes used in the construction of double samplers cannot achieve better than quasipolynomial overhead. Are there non-simplicial objects better served for this purpose?
    \item Our $1\%$-regime test only holds on dense complexes due to the assumption of $\ell_\infty$-expansion. However, the main argument only requires reverse hypercontractivity and spectral gap of the down-up walk. Can the argument be completed without the assumption of $\ell_\infty$-expansion to give new sparse agreement testers in the low acceptance regime?
\end{enumerate}

\section*{Roadmap}

In \pref{sec:preliminaries} we give more detailed and formal preliminaries on high dimensional expanders, random walks, samplers, and reverse hypercontractivity. In \pref{sec:chernoff} we prove optimal inclusion sampling for two-sided HDX, and cover all other cases in \pref{sec:concentration-all-nice-hdx}.

In \pref{sec:rhc} we prove reverse hypercontractivity. In \pref{sec:agreement} and \pref{sec:an-comb-apps} we give applications of reverse hypercontractivity to agreement testing, and analysis and combinatorics respectively. 
In \pref{sec:codes} we give basic sampling lemmas for splitting trees and construct high distance list-decodable LTCs. Finally in \pref{sec:lower-bounds} we show optimality of our inclusion sampling results both in terms of spectral requirements on $X$ and strength of sampling itself.
\section{Preliminaries and Notation}  \label{sec:preliminaries}

\subsection{Graphs and Spectral Expansion}
A weighted un-directed graph \(G = (V,E,\Pr)\) consists of a finite vertex set \(V\), a set of edges \(E \subseteq \binom{V}{2}\), and a distribution on the edges \(\Pr:E \to (0,1]\). 
For a vertex \(v\) we denote by \(\prob{v} = \frac{1}{2} \sum_{uv \in E} \prob{uv}\). 
In this paper all graphs are assumed to be weighted (the weight is implicit in notation). We also assume that there are no isolated vertices.

Let \(\ell_2(V) = \set{f:V \to \RR}\). The graph $G$ induces an inner product \(\iprod{f,g} = \Ex[v]{f(v)g(v)}\) on \(\ell_2(V)\), as well as a normalized adjacency operator \(A:\ell_2(V) \to \ell_2(V)\) defined as
\[
Af(v) := \Ex[u \sim v]{f(u)} = \sum_{uv \in E}\frac{\prob{uv}}{\sum\limits_{w\sim v}\prob{wv}} f(u).
\]
It is well known that \(A\) is diagonalizable and has eigenvalues 
\[
1=\lambda_1 \geq \lambda_2 \geq \ldots \geq \lambda_{|V|} \geq -1,
\]
where the first `trivial' eigenvalue corresponds to the space of constant functions. A graph is called a spectral expander if all non-trivial eigenvalues are small.
\begin{definition}[Spectral expander]
    A graph \(G=(V,E)\) is called a \emph{\(\lambda\)-one-sided spectral expander} if \(\lambda_2 \leq \lambda\). We say that \(G\) is a \emph{\(\lambda\)-two-sided spectral expander} if \(\lambda_2 \leq \lambda\) and \(\lambda_{|V|}\geq -\lambda\).
\end{definition}

We record the following inequality for \(\lambda\)-two-sided spectral expanders, a variant of the classical \textit{expander-mixing lemma}. For every \(f,g \in \ell_2(V)\),
\begin{equation} \label{eq:basic-expanders}
    \abs{\iprod{f,A g} - \ex{f}\ex{g}} \leq \lambda \norm{f}_2 \norm{g}_2.
\end{equation}

\subsubsection{Bipartite Graphs and Bipartite Expanders}
    A bipartite graph is a graph where the vertex set can be partitioned to two independent sets \(V = L \dunion R\), called sides. We sometimes denote such graphs by \(G=(L,R,E)\).

\paragraph{The Bipartite Adjacency Operator}
    In a bipartite graph, we view each side as a separate probability space, where for any \(v \in L\) (resp. \(R\)), \(\prob{v} = \sum_{w \sim v} \prob{wv}\). We can define the \emph{bipartite adjacency operator} as the operator  \(B: \ell_2(L) \to \ell_2(R)\) by
    \[\forall f \in \ell_2(L), v \in R, \; Bf(v) = \Ex[w \sim v]{f(u)}\]
    where the expectation is taken with respect to the probability space \(L\), conditioned on being adjacent to \(v\).
There is a similar operator \(B^*: \ell_2(R) \to \ell_2(L)\) as the bipartite operator for the opposite side. As the notation suggests, \(B^*\) is adjoint to \(B\) with respect to the inner products of \(\ell_2(L), \ell_2(R).\)

    We denote by \(\lambda(B)\) the spectral norm of $B$ when restricted to \(\ell_2^0(L) = \set{\one}^\bot\), the orthogonal complement of the constant functions (according to the inner product the measure induces on \(L\)). Namely
    \[ \lambda(B) = \sup \sett{\iprod{Bf,g}}{\norm{g},\norm{f}=1, f \bot \one_L}.\]
    \begin{definition}[Bipartite Expander]\torestate{\label{def:bipartite-expander}
    Let $G$ be a bipartite graph, let $\lambda < 1$. We say $G$ is a \emph{$\lambda$-bipartite expander}, if
    $\lambda(B) \leq \lambda$.}
    \end{definition}
It is easy to show that a bipartite graph is a \(\lambda\)-bipartite expander if and only if it is a \(\lambda\)-one-sided spectral expander. So we use these terms interchangeably on bipartite graphs.

We record the following inequality for \(\lambda\)-bipartite expanders similar to \eqref{eq:basic-expanders}. For every \(f \in \ell_2(L), g\in \ell_2(R)\),
\begin{equation} \label{eq:basic-bipartite-expanders}
    \abs{\iprod{f,A g} - \ex{f}\ex{g}} \leq \lambda \norm{f}_2 \norm{g}_2.
\end{equation}

\subsection{Reverse Hypercontractivity}
Hypercontractivity and reverse hypercontractivity are powerful analytic inequalities from boolean function analysis that bound the contraction behavior of operators between normed spaces. Recall that for \(p\ne 0\) and a function \(f\) on a probability space, we denote by \(\norm{f}_p = \ex{|f|^p}^{1/p}\) (for \(p\leq 0\) this is only defined for \(f\) that are non-zero almost everywhere). We note that \(p \mapsto \norm{f}_p\) is monotone increasing.

Let \(V\) be a finite probability space and let \(\ell_2(V) = \set{f:V \to \mathbb{R}}\). A monotone operator is an operator such that for any non-negative function \(f:V \to [0,\infty)\), \(Af\) is also non-negative (i.e.\ for any \(x \in X\), \(Af(x)\geq 0\)). For example, adjacency operators of graphs are always monotone.

Let \(A\) be a monotone operator. Let \(1\leq p<q\) and let \(C>0\). \(A\) is typically called \textit{\((p,q,C)\)-hypercontractive} if for every \(f:V \to \RR\),
\[
\norm{Af}_q \leq C \norm{f}_p.
\]
It is well know that this is equivalent to `two-function hypercontractivity', that for every \(f,g:V \to \RR\)
\[
\iprod{Af,g} \leq C\norm{f}_{p}\norm{g}_{q'}
\]
where \(q' = \frac{q}{q-1}\) is \(q\)'s Hölder conjugate \cite{ODonnell2014}. \textit{Reverse hypercontractivity} `flips' this inequality in multiple ways: both in the direction of the inequality, and in the relation of $p$ and $q$.
\begin{definition}[Reverse hypercontractivity]
    Let \(q < p < 1\) such that \(p,q\ne 0\). Let \(C > 0\). Let \(A\) be a monotone operator. We say that \(A\) is \((p,q,C)\)-reverse hypercontractive if for every \(f:V\to \RR^{\geq 0}\),
    \[\norm{Af}_q \geq C\norm{f}_p\]
    or equivalently if for every \(f,g:V \to \RR^{\geq 0}\) it holds that 
    \[\abs{\iprod{Af,g}} \geq C\norm{f}_p \norm{g}_{q'}\]
    where \(q' = \frac{q}{q-1}\) is \(q\)'s Hölder conjugate (the equivalence is similar to standard hypercontractivity).
\end{definition}

It is sometimes convenient to substitute \(f':=f^p, g'=g^{q'}\) and write the two-function reverse hypercontractivity inequality as
\[
\abs{\iprod{A f'^{1/p}, g'^{1/q'}}} \geq \ex{f'}^{1/p} \ex{g'}^{1/q'}.
\]
Probably the simplest (and most useful) interpretation of this inequality is when \(f\) and \(g\) are indicators of sets \(M,N \subseteq V\) respectively, and that \(A\) is an adjacency operator of a graph (whose vertices are the points in \(V\)). Then this inequality has a combinatorial interpretation as a type of `mixing' lemma. Namely that the probability a random edge hits \(M\) and \(N\):
\[
\Prob[uv \in E]{u\in M, v \in N} \geq \prob{M}^{1/p} \prob{N}^{1/q'}.
\]
This becomes a powerful tool in the regime where the sets are smaller than can be controlled by the expander-mixing lemma.

\subsection{Simplicial Complexes}
\paragraph{Simplicial complexes} A pure \(d\)-\maximal simplicial complex \(X\) is a finite set system (hypergraph) consisting of an arbitrary collection of sets of size \(d\) together with all their subsets (individually called `faces'). The sets of size \(i\) in \(X\) are denoted by \(\X[i]\), and in particular, the vertices of \(X\) are denoted by \(\X[1]\) (we do not distinguish between vertices and the singletons of vertices). We write $X^{\leq j}$ to denote faces of $X$ up to size $j$.

We note that this notation departs somewhat from the typical convention in the literature which denotes faces of size \(i+1\) by \(X(i)\), and what we have called a \(d\)-\maximal simplicial complex is more typically referred to as a \emph{\((d-1)\)-dimensional} simplicial complex. We have chosen to adopt the former notation since most theorems in this work require a substantial amount of arithmetic on set sizes, and the latter convention quickly becomes overly cumbersome.

\paragraph{Probability over simplicial complexes}
Let \(X\) be a simplicial complex and \(\Pr_d:\X[d]\to (0,1]\) a density function on \(\X[d]\). 
This density function induces densities on lower level faces \(\Pr_k:\X[k]\to (0,1]\) by \(\Pr_k(t) = \frac{1}{\binom{d}{k}}\sum_{s \in \X[d],s \supset t} \Pr_d(s)\). Equivalently $\Pr_k$ is the density induced by drawing a $d$-face from $s \sim \Pr_d$, and a uniformly random $k$-face $t \subset s$.

When clear from the context, we omit the level of faces and just write \(\Pr[T]\) or \(\Prob[t \in {\X[k]}]{T}\) for an event \(T \subseteq \X[k]\).

\paragraph{Links of faces}
Let \(X\) be a \(d\)-\maximal simplicial complex. Let \(k < d\) and \(s \in \X[k]\). The link of \(s\) (called a $k$-link) is a \((d-k)\)-\maximal simplicial complex defined by
\(X_s = \sett{t \setminus s}{t \in X, t \supseteq s}\). Note the link of the emptyset is $X$ itself \(X_\emptyset = X\).

Let \(s \in \X[k]\) for some \(k \leq d\). The density function \(\Pr_d\) on \(X\) induces a density function \(\Pr^s_{d-k}:\X[d-k]\to (0,1]\) on the link where 
\({\Pr}_{d-k}^s[{t}] = \frac{\prob{t \cup s}}{\prob{s} \binom{d}{k}}\).
We usually omit \(s\) in the probability, and for \(T \subseteq \SC[X][k][s]\) write \(\Prob[t \in {\SC[X][k][s]}]{T}\).

\paragraph{Connected Complexes} We call a complex connected if the graph underlying every link is connected.

\paragraph{Partite complexes}
A \(d\)-partite simplicial complex is a \(d\)-\maximal complex whose vertices can be partitioned into \(d\) disjoint sets (typically called ``parts'' or ``colors'')
\[
\X[0]=X[1] \dunion X[2] \dunion \ldots \dunion X[d] \] 
such that every \(s \in \X[d]\) has exactly one vertex from each part.

Let \(X\) be a \(d\)-partite simplicial complex. The \textit{color} of a face \(t \in \X[k]\) is \(col(t) = \sett{i \in [d]}{t \cap X[i] \ne \emptyset}\). Let \(F \subseteq [d]\). We denote by \(X[F] = \sett{s \in X}{col(s)=F}\) the faces of $X$ with color $F$. The projection of the \textit{complex} $X$ onto $F$ is the sub-simplicial complex \(X^F = \sett{s \in X}{col(s) \subseteq F}\). When $X$ is endowed with a density $\Pr_d$, $X^F$ has a naturally induced density $\Pr^F$ by sampling $\sigma \sim \Pr_d$, and outputting the projection $\sigma_F$. Quantitatively, the induced distribution can be written as
\[
\mathbb{P}^F(s) = \sum_{t \in \X[d], s\subseteq t} \Pr[t].
\]
Finally, when \(F=\set{i}\) is a singleton we just write $X[i]$, $X^i$, and $x_i$ for brevity.

Given a generic complex $X$, it will often be useful to consider $X$'s `partitification' $P=P(X)$, which simply includes every possible ordering of the faces of $X$ as tuples. Formally, endow the faces of $X$ with an arbitrary order and define
\[
\SC[P][k] \coloneqq \sett{(s,\pi) \coloneqq \left\{(s_{\pi(1)},1),(s_{\pi(2)},2),\ldots,(s_{\pi(k+1)},k)\right\}}{s \in \X[k], \pi \in S_{k}},
\]
where $S_{k}$ is the group of permutations on $k$-letters and the measure of a face $(s,\pi)$ is inherited naturally as $\frac{1}{k!} \Prob[X]{s}$. Note that $P(X)$ does not depend on the choice of ordering, which is simply a notational convenience.
\paragraph{Degree of Simplicial Complexes} Degree is a critical parameter of simplicial complexes capturing the (local) `blow-up' incurred by moving to higher \maximal faces. We define the (max) degree of a complex with respect to $i$-faces as:
\[
\deg^i(X) = \max_{v \in \X[1]}|\{s \in \X[i]: v \in s \}|
\]

We write just $\deg(X) \coloneqq \deg^d(X)$ to denote degree with respect to top level faces of $X$, and $\deg^{(i)}(v)$ for the degree of a specific vertex. An infinite family of $d$-\maximal simplicial complexes $\{X_n\}$ is called \textit{bounded-degree} if there exists a (dimension-dependent) constant $C$ such that $\deg^i(X) \leq C$ for all $i \leq d$.
\paragraph{Hitting Set}
It will occasionally be useful for us to use a common variant of sampling on simplicial complexes called \textit{hitting set}.
We call a complex $(\gamma,i)$-hitting if for any $A \subset \X[1]$:
    \[
    \Pr_{\sigma \in \X[i]}[\sigma \subset A] \leq \prob{A}^{i} + \gamma
    \]
In other words, if the probability $\sigma \in \X[i]$ \textit{hits} $\X[1] \setminus A$ is at least $1-\mu(A)^{i}-\gamma$. If a complex is $(\gamma,i)$-hitting for all sizes $i$, we just call it $\gamma$-hitting.
\subsection{Higher Order Random Walks}\label{sec:random-walks}
Simplicial complexes come equipped with several natural families of random walks generalizing the standard random walk on a graph. Toward this end, let \(X\) be a \(d\)-\maximal simplicial complex and \(\ell<k\leq d\) we define the standard `averaging' or `random walk' operators that move up and down the complex:
\paragraph{Down and up operators} The \textit{down operator} \(D_{k,\ell}:\ell_2(\X[k]) \to \ell_2(\X[\ell])\) is the bipartite operator of the containment graph of \((\X[k],\X[\ell])\), that is:
\[
D_{k,\ell}f(s) = \Ex[t \supseteq s]{f(t)}.
\]
The adjoint of $D_{k,\ell}$ is the \textit{up operator}, \(U_{\ell,k}:\ell_2(\X[\ell]) \to \ell_2(\X[k])\), given by:
\[
U_{\ell,k} g(t) = \Ex[s \subseteq t]{g(s)}.
\]
We also write $U_k\coloneqq U_{k-1,k}$ and $D_k \coloneqq D_{k,k-1}$ as shorthand throughout.

The composition of the up and down operators $U_{\ell,k}D_{k,\ell}$, called the ``down-up'' walk, is the normalized adjacency matrix of the graph whose vertices are \(\X[k]\) and whose edge distribution is defined by sampling $t \sim \Pr_\ell$, then $s,s' \sim \Pr_k$ conditioned on containing $t$. We denote the corresponding (bipartite) graph by $(\X[k],\X[\ell])$. For brevity, we sometimes just write \(UD_{k,\ell}\) instead of \(U_{\ell,k}D_{k,\ell}\).

\paragraph{Noise operator} The noise operator is one the most classical objects of study in Boolean function analysis. Given $\rho \in [0,1]$, the standard noise operator $T_{k,\rho}$ operates on the hypercube $\mathbb{F}_2^k$ by re-sampling each bit uniformly with probability $1-\rho$. 

The noise operator has a natural extension to simplicial complexes \cite{BafnaHKL2022, GurLL2022}. It is convenient to define \(T_{k,\rho}: \ell_2(\X[k]) \to \ell_2(\X[k])\) by its action on a face $\sigma \in \X[k]$ by the following process:
\begin{enumerate}
    \item Sample \(\ell \sim Bin(k,\rho)\)
    \item Sub-sample $\tau \subset \sigma$ of size $\ell$
    \item sample $\sigma' \in \X[k]$ conditioned on containing $\tau$.
\end{enumerate}
It is easy to check that when $X$ is a product space,\footnote{One can always express a product space $\otimes_i \Omega_i$ as a partite simplicial complex by defining each coordinate as a part, see \cite{BafnaHKL2022} for further details.} this exactly recovers the standard noise operator. We note it is also possible to write the noise operator as a convex combination of down-up walks:
\[
T_{k,\rho} = \sum_{\ell=0}^{k} \binom{k}{\ell}\rho^\ell (1-\rho)^{k-\ell} UD_{k,\ell}.
\]

\paragraph{Swap walks} Let \(X\) be a \(d\)-\maximal simplicial complex. Let \(i,j\) be so that \(i+j\leq d\). The swap walk \(S_{i,j} = S_{i,j}(X)\) is the bipartite adjacency operator of the graph \((\X[i],\X[j],E)\). An edge \(\set{s_i,s_j}\) is chosen in this graph by first selecting a face \(t \in \X[i+j]\), and then partitioning it to \(t=s_i \dunion s_j\) where \(s_i \in \X[i]\) and \(s_j \in \X[j]\) uniformly at random. This walk was defined and studied independently by \cite{AlevJT2019} and by \cite{DiksteinD2019}. 

\paragraph{Colored swap walks} The standard swap walks are not well behaved on partite complexes, but there is a useful analog for this setting defined in \cite{DiksteinD2019}. Let \(X\) be a \(d\)-\maximal partite simplicial complex, and \(F_1,F_2 \subseteq [d]\) be two disjoint subsets. The colored swap walk \(S_{F_1,F_2} = S_{F_1,F_2}(X)\) is the bipartite adjacency operator of the graph \((X[F_1],X[F_2],E)\). An edge \(\set{s_1,s_2}\) is chosen in this graph by first selecting a face \(t \in X[F_1 \dunion F_2]\), and then partitioning it to \(t=s_1 \dunion s_2\) according to its colors. 

A useful observation is that the spectral expansion of swap walks (and color swap walks) is monotone in the following sense.
\begin{observation} \label{obs:monotonicity-of-swap-walks}
    Let \(X\) be a $d$-uniform simplicial complex. Then for every \(i\) and \(j'<j\), \(\lambda(S_{i,j'}) \leq \lambda(S_{i,j})\). The same holds for partite complexes. For every disjoint \(F_1, F_2\) and \(F_2' \subseteq F_2\), \(\lambda(S_{F_1,F_2'}) \leq \lambda(S_{F_1,F_2})\).
\end{observation}

\begin{proof}
    The proof follows from the fact that one can factor 
    \[
    S_{i,j'}(X) = D_{j,j'}S_{i,j}(X)
    \]
    and the observation that $D_{j,j'}$ contracts $2$-norms. A similar argument is true for colored swap walks in the partite case, replacing \(D_{j,j'}\) with the corresponding bipartite graph between the \(F_2\)-colored faces and the \(F_2'\)-colored faces they contain.
\end{proof}
\subsection{Sampler Graphs}
\textit{Sampler graphs} are bipartite graphs $G=(L,R,E)$ where a random vertex $v \in L$ ``sees'' any large enough set in $R$ with approximately the correct probability. Here we discuss a few classical variants of samplers and their relations, and refer the reader to \cite{Goldreich1997} for further discussion.

\begin{definition}[Multiplicative sampler]
Let \(G=(L,R,E)\) be a bipartite graph and \(\alpha,\beta,\delta>0\). \(G\) is an \emph{\((\alpha,\beta,\delta)\)-multiplicative sampler} if for every set \(A \subseteq R\) of size \(\prob{A} \geq \alpha\) it holds that
\[\Prob[v \in L]{\abs{\Prob[u \sim v]{u \in A}-\prob{A}} > \delta \prob{A}} \leq \beta.\]
\end{definition}
Note that the definition of a sampler is not a priori symmetric, \(L\) and \(R\) have different roles. We will also study a related \textit{additive} notion of samplers.
\begin{definition}[Additive sampler]
Let \(G=(L,R,E)\) be a bipartite graph and \(\beta,\varepsilon>0\). \(G\) is an \emph{\((\varepsilon,\beta)\)-additive sampler} if for every set \(A \subseteq R\) it holds that
\[\Prob[v \in L]{\abs{\Prob[u \sim v]{u \in A}-\prob{A}} > \varepsilon } \leq \beta.\]
\end{definition}

Finally, sometimes it is also useful to sample \emph{functions} instead of sets. toward this we introduce the following definition in additive notation.
\begin{definition}
    Let \(G=(L,R,E)\) be a bipartite graph and $\varepsilon,\beta > 0$. $G$ is a \((\varepsilon,\beta)\)-\emph{function additive sampler} if for every \(f:R \to [0,1]\) with expectation \(\Ex[v \in R]{f(v)} = \mu\), it holds that
    \[
    \Prob[v \in L]{\abs{\Ex[u \sim v]{f(v)} - \mu} > \varepsilon}<\beta,
    \]
    and similarly for a $(\alpha,\beta,\delta)$-\emph{function multiplicative sampler}.
\end{definition}

\subsubsection{Basic Sampler Properties}
While the definition of samplers are not symmetric with respect to $L$ and $R$, a near-tight correspondence $(L,R,E)$ and $(R,L,E)$ is given in \cite[Lemma 2.5]{ImpagliazzoKW2012}. We repeat their proof in \pref{app:sampler-proofs} in a more general setup.
\begin{claim} \label{claim:opposite-sampler}
    Let \(\beta,\delta > 0\), let \(\delta' > \delta\) and \(\alpha < \frac{\min \set{\delta,0.5}}{1+\delta}\). Then for every \((\alpha,\beta,\delta)\)-sampler \(G=(L,R,E)\), it holds that \(G_{op} \coloneqq (R,L,E)\) is a \((\frac{1-\alpha(1+\delta)}{\alpha(\delta' - \delta)}\beta, 2\alpha, \delta')\)-sampler.
\end{claim}

We also note that additive samplers and multiplicative are equivalent, at least in the weak sense as below. We prove this claim in \pref{app:sampler-proofs}.
\begin{claim} \label{claim:additive-multiplicative-sampler-equivalence}
Let \(G=(L,R,E)\) be a bipartite graph.
\begin{enumerate}
    \item If \(G\) is a \((\beta,\delta)\)-additive sampler then \(G\) is a \((C \delta,\beta,\frac{1}{C})\)-multiplicative sampler for any \(C>1\).
    \item If \(G\) is a \((\alpha,\beta,\delta)\)-multiplicative sampler for \(\alpha \leq \frac{1}{2}\). Then \(G\) is a \((\beta,\delta)\)-additive sampler, where \(\delta = \max \set{\delta,(1+\delta)(\alpha+p)}\) and \(p = \max_{v \in R} \prob{v}\).
\end{enumerate}
\end{claim}

Under a slightly stronger assumption one can also remove the dependence on \(\delta\) in the second item.
\begin{claim} \label{claim:improved-mult-to-add-sampler}
    Let \(\beta, \alpha_0 > 0\). If for every \(\alpha > \alpha_0\) it holds that \(G\) is an \((\alpha,\beta,\frac{\alpha_0}{\sqrt{\alpha}})\)-multiplicative sampler, then \(G\) is a\((\beta,2(\alpha+p))\)-additive sampler where \(p = \max_{v \in R} \prob{v}\).
\end{claim}

Below we show that every additive sampler is also a function additive sampler (albeit with worse parameters). We did not try to optimize parameters and it could be the case that a better reduction exists.
\begin{claim} \torestate{\label{claim:sampler-for-functions}
        Let \(G=(L,R,E)\) be an \((\varepsilon,\beta)\)-additive sampler such that every \(r \in R\) has degree at least \(k\) and every \(v \in L\) has probability at most \(\frac{1}{k}\). Assume that \(\exp(-0.01\varepsilon^2 k)) < \frac{1}{4}\). Then \(G\) is also a \((4\beta, 2\varepsilon)\)-function additive sampler.}
\end{claim}
We prove the claim in \pref{app:sampler-proofs}.
We note that \cite{Goldreich1997} also presents a reduction from an additive sampler to a function additive sampler. The reduction there changes the underlying graph, which is why we prove a reduction more suitable to our needs.
\subsection{Concentration of Measure}\label{sec:prelims-concentration}
Sampler graphs are a special case of the powerful concept of concentration of measure, a viewpoint we will also take throughout this work. We follow the standard notation of \cite{Boucheron} adapted in the natural way from product measures to simplicial complexes. Given a simplicial complex $X$ and a function $f: \X[k] \to \R$, let $Z=f(x_1,\ldots,x_k)$ denote the random variable distributed as $f(x)$ where $x \in \X[k]$. We are interested in understanding the concentration of $Z$ around its mean. $Z$ is said to satisfy \textit{subexponential} concentration if there exists a constant $c>0$ such that
\[
\Pr[|Z-\mathbb{E}[Z]| > t] \leq \exp(-ct),
\]
and \textit{subgaussian} concentration if
\[
\Pr[|Z-\mathbb{E}[Z]| > t] \leq \exp(-ct^2).
\]
Classical concentration bounds (e.g.\ on the cube, products) typically hold for functions $f$ satisfying certain Lipschitz-type conditions. Toward this end, define the $Z$-dependent variable $Z'_{(i)}=f(x_1,\ldots z_i, \ldots, x_k)$ where $z_i$ is sampled conditional on $x_{-i}$.\footnote{Recall that even in the non-partite setting, we assign each $k$-face an arbitrary order, making $Z'_{(i)}$ well defined.} Abusing notation slightly, we will call a function $\nu$-Lipschitz if its squared difference with respect to the down-up walk is bounded:
\begin{definition}[$\nu$-Lipschitz function]\label{def:bounded}
Let $X$ be a simplicial complex and $\nu>0$. We call a function $f: \X[d] \to \mathbb{R}$ $\nu$-Lipschitz if with probability $1$:
\[
\sum\limits_{i=1}^d (Z-Z'_{(i)})_+^2 \leq \nu,
\]
where $(z)_+ =\max\{z,0\}$. 
\end{definition}
We will occasionally rely on a stricter variant that requires the difference of $f$ on any two neighboring $k$-faces be bounded.
\begin{definition}[$\nu$-bounded difference]
    We say a function $f:\X[d] \to \mathbb{R}$ has $\nu$-bounded difference if for every $s \in \X[d]$ and neighboring $s'$ of the down-up walk:
    \[
        (f(s)-f(s'))^2 \leq \frac{\nu}{d}.
    \]
\end{definition}
Most of our work focuses on a special case of functions with $\nu$-bounded difference we call \textit{lifted} functions.
\begin{definition}[Lifted functions]
    Let $X$ be a $d$-\maximal simplicial complex and $f: \X[k] \to \R$ any function. For any $k < k' \leq d$, the $k'$-lift of $f$ is the function:
    \[
    U_{k,k'} f(s) = \mathbb{E}_{s' \subset s}[f(s')]
    \]
\end{definition}
We remark that lifted functions (sometimes called `degree' or `level'-$k$ functions) are fundamental objects in boolean analysis and heavily studied in the HDX literature \cite{DiksteinDFH2018,bafna2022high,gaitonde2022eigenstripping}. As discussed in the introduction, inclusion sampling is in one-to-one correspondence with concentration bounds for lifted functions. More generally, bipartite sampling guarantees on $G=(L,R,E)$ correspond to concentration for functions of the form $A_Gf$, where $f:R \to [0,1]$ and $A_G: \R^R \to \R^L$ is the normalized bipartite adjacency matrix of $G$.

On simplicial complexes, we will typically be interested in understanding concentration at various levels of the complex. One frequently useful fact is that once concentration has been established for $k$-faces, it can typically be `raised' or `lowered' to functions on other levels via simple reductions to concentration of the complete complex. In cases, this can even be used to \textit{improve} bounds on lower levels by bootstrapping from stronger bounds from larger set sizes. We record the following basic results to this effect here and use them throughout. Their proofs are given in \pref{app:sampler-proofs}. 

Given integers $k \leq d$, a $d$-\maximal simplicial complex $X$, and a function $f:\X[k] \to \R$, denote 
\[
\pi_{up}^{d,k,f}(t) = \max_{s \in \X[d]}\left\{\underset{r \subset s}{\Pr}\left[f(r)-\underset{r \subset s}{\mathbb{E}}[f] > t\right]\right\}
\]
and similarly 
\[
\pi_{low}^{d,k,f}(t) = \max_{s \in \X[d]} \left\{\underset{r \subset s}{\Pr}\left[f(r)-\underset{r \subset s}{\mathbb{E}}[f] < -t\right]\right\}
\]
to be the worst-case concentration of $f$ restricted to the induced complete complex on $d$-faces of $X$.
\begin{lemma}[Raising Concentration]\label{lem:raising}
    Let $X$ be a $d$-\maximal simplicial complex, $k \leq d$, and $f: \X[k] \to [0,1]$ a function satisfying
    \begin{enumerate}
        \item \textbf{Upper Tail}: $\underset{\X[k]}{\Pr}[f-\mathbb{E}[f] > t] \leq up(t)$
        \item \textbf{Lower Tail}: $\underset{\X[k]}{\Pr}[f-\mathbb{E}[f] < -t] \leq low(t)$.
    \end{enumerate}
    for some functions $up,low: \R_+ \to [0,1]$. Then the $d$-lift $U_{k,d}f: \X[d] \to \R$ satisfies:
    \begin{enumerate}
        \item \textbf{Upper Tail}: $\underset{\X[d]}{\Pr}[U_{k,d}f-\mathbb{E}[f] > t] \leq up(\frac{t}{2})(1-\pi_{low}^{d,k,f}(\frac{t}{2}))^{-1}$
        \item \textbf{Lower Tail}: $\underset{\X[d]}{\Pr}[U_{k,d}f-\mathbb{E}[f] < -t] \leq low(\frac{t}{2})(1-\pi_{up}^{d,k,f}(\frac{t}{2}))^{-1}$.
    \end{enumerate}
\end{lemma}
We remark that since concentration in the complete complex is typically quite good, for most function classes of interest the latter term is close to $1$ and very little is lost lifting the concentration to level $d$. We will also use a variant of this result for `lowering' concentration bounds specified to Lipschitz functions.
\begin{lemma}[Lowering Concentration]\label{lem:lowering}
    Let $X$ be a $d$-\maximal simplicial complex and $k \leq d$. Assume there exist functions $up(t,\nu)$ and $low(t,\nu)$ such that any $\nu$-Lipschitz $f: \X[d] \to \R$ satisfies:
    \begin{enumerate}
        \item \textbf{Upper Tail}: $\Pr[f-\mathbb{E}[f] > t] \leq up(t,\nu)$
        \item \textbf{Lower Tail}: $\Pr[f-\mathbb{E}[f] < -t] \leq low(t,\nu)$.
    \end{enumerate}
    Then any function $f':\X[k] \to \R$ with $\nu$-bounded difference satisfies:
        \begin{enumerate}
        \item \textbf{Upper Tail}: $\Pr[f'-\mathbb{E}[f'] > t] \leq up(\frac{t}{2},\frac{k}{d}\nu) + e^{-\frac{t^2}{4\nu}}$
        \item \textbf{Lower Tail}: $\Pr[f'-\mathbb{E}[f'] < -t] \leq low(\frac{t}{2},\frac{k}{d}\nu) + e^{-\frac{t^2}{4\nu}}$.
    \end{enumerate}
\end{lemma}

\subsection{High Dimensional Expanders}
In this section we define the notions of high dimensional expansion used throughout this work and dicsuss their relation with high order random walks.

\paragraph{Local Spectral Expanders} The most standard notion of spectral high dimensional expansion is `local-spectral expansion', due to \cite{DinurK2017,oppenheim2018local}. 
\begin{definition}[Local-spectral expander]\label{def:HDX}
    Let \(X\) be a \(d\)-\maximal simplicial complex. We say that \(X\) is a \(\lambda\)-one-sided (two sided) high dimensional expander if for every $i \leq d-2$ and \(s \in \X[i]\), the graph underlying \(X_s\) is a \(\lambda\)-one-sided (two sided) spectral expander.
\end{definition}
A key property of local-spectral expanders is that they imply the expansion of associated high order random walks \cite{KaufmanM2017,DinurK2017,KaufmanO2020,DiksteinDFH2018,alev2020improved}. We first state such a bound for the single and multi-step down-up walks (or more accurately their corresponding bipartite operator).
\begin{theorem}[{\cite{alev2020improved}}] \label{thm:gap-of-down-operator}
Let \(X\) be a \(\lambda\)-one-sided \(d\)-\maximal high dimensional expander. Then for every \(\ell<k\leq d\) it holds that \(\lambda(D_{k,\ell}) \leq \sqrt{\frac{\ell}{k}}\cdot (1+\lambda)^{\frac{k-\ell}{2}}.\)
\end{theorem}
In the special case of the single step operator, we will occasionally use the following refined bound.
\begin{theorem}[{\cite{alev2020improved}}] \label{thm:gap-of-down-operator-single-step}
Let \(X\) be a \(\lambda_i\)-one-sided \(d\)-\maximal high dimensional expander. Then for every \(\ell<k\leq d\) it holds that \(\lambda(D_k) \leq \sqrt{\frac{1}{k}\prod\limits_{i=0}^{k-2}(1-\lambda_i)}\)
\end{theorem}
A critical observation of \cite{AlevJT2019,DiksteinD2019} is that by removing the laziness inherent in the down-up walk, one can substantially improve expansion. We state an improvement of their bounds for these swap walks by Gur, Lifshitz, and Liu \cite{GurLL2022}:\footnote{Formally this result is stated in \cite{GurLL2022} for $\lambda$-products (defined below), but follows for $\lambda$-two-sided HDX since their partifications are $\lambda$-products.}
\begin{theorem}[{\cite{GurLL2022}}] \label{thm:swap-walks-expand}
    Let \(X\) be a $\lambda$-two-sided high dimensional expander. Then 
    \[
    \lambda(S_{\ell_1,\ell_2}) \leq \sqrt{\ell_1\ell_2}\lambda
    \]
\end{theorem}
Recently, \cite{alev2023sequential} proved tighter quantitative bounds that are non-trivial for any underlying local-spectral expansion. We will not need such fine-grain control in this work.

In the partite case, the colored swap walks also have known expansion bounds assuming partite local-spectral expansion \cite{DiksteinD2019,GurLL2022}.
\begin{theorem}[{\cite{DiksteinD2019}}]\label{thm:color-swap-walks-expand}
    Let \(X\) be a partite \(\lambda\)-one-sided high dimensional expander. Then 
    \[
    \lambda(S_{F_1,F_2}) \leq \abs{F_1} \abs{F_2}\lambda.
    \]
\end{theorem}

We will also need a bound on the spectrum of the averaging operators based on local-spectral expansion. This has been the subject of intense study \cite{KaufmanM2017,DinurK2017,DiksteinDFH2018,KaufmanO2020,alev2020improved}, with the best known bounds given by Alev and Lau \cite{alev2020improved}.

\paragraph{$\lambda$-Products} Another useful notion of partite high dimensional expansion is to look directly at the expansion of the bipartite graphs between different components. We say that \(X\) is a \(\lambda\)-product (a term coined in \cite{GurLL2022}), if for every face \(s \in X\) and every two sides \(X_s[i],X_s[j]\), the bipartite graph between the two sides is a \(\lambda\)-one-sided spectral expander. Dikstein and Dinur showed any partite local-spectral expander is a $\lambda$-product.
\begin{claim}[{\cite{DiksteinD2019}}]\label{claim:hdx-to-product}
    Let $X$ be a \(k\)-partite \(\frac{\lambda}{1+\lambda}\)-one-sided local-spectral expander. Then $X$ is also a $\lambda$-product.
\end{claim}

\paragraph{$\lambda$-Trickling-Down Complexes}

Oppenheim's trickling-down theorem \cite{oppenheim2018local} shows that in any connected complex, expansion in the top links `trickles down' to expansion at lower levels.
\begin{theorem}[The Trickling-Down Theorem {\cite{oppenheim2018local}}] \label{thm:td}
    Let \(X\) be a connected \(d\)-\maximal simplicial complex. Assume that for every \(s \in \X[d-2]\), it holds that all non-trivial eigenvalues of \(X_s\) are in \([-\tau,\lambda]\) for some \(\tau,\lambda \geq 0\). Then all non-trivial eigenvalues of \(X\) are in 
    \[
    \left[-\frac{\tau}{1+(d-2)\tau},\frac{\lambda}{1-(d-2)\lambda}\right].
    \]
\end{theorem}
As discussed in the introduction, the Trickling-Down Theorem introduces a `phase transition' in the local-to-global behavior of HDX when co-dimension $2$ links have expansion $\frac{1}{d-1}$. Past this point, expansion of the $1$-skeleton of $X$ can be inferred from expansion of the links. Before it, no guarantee is possible. This threshold-type behavior naturally suggests studying following definition of complexes `approaching' the TD-barrier.

\begin{definition}[$\lambda$-Trickling-Down Complex]
    We call a $d$-\maximal simplicial complex $\lambda$-Trickling-Down ($\lambda$-TD) if it is connected, and all co-dimension 2 links have (one-sided) expansion $\frac{\lambda}{d-1}$.
\end{definition}

\paragraph{Spectral Independence} Spectral independence is a closely related notion to local-spectral expansion introduced in the sampling literature \cite{anari2021spectral}. Here we give an equivalent version of the definition in terms of link expansion, and refer the reader to \cite{anari2021spectral} for the original definition and their equivalence.
\begin{definition}[Spectral independence]
    For $\eta>0$, a \(d\)-\maximal complex $X$ is called $\eta$-spectrally independent ($\eta$-SI) if the graph underlying every co-dimension $j$ link is a (one-sided) $\frac{\eta}{j}$-expander.
\end{definition}

\subsection{Nice Complexes} \label{sec:nice}
In the following sections of the paper we prove results for various standard notions of high dimensional expansion including two sided high dimensional expanders, partite high dimensional expanders and \(\sqrt{d}\)-skeletons of one-sided high dimensional expanders. To compactify our statements, we bundle these notions into one definition henceforth dubbed `nice complexes'.
\begin{definition}[Nice complex] \label{def:nice-complexes}
    Let \(c > 0\). A $k$-\maximal complex \(X\) is \(c\)-\emph{nice} if it meets one of the following conditions:
    \begin{enumerate}
        \item $X$ is a $2^{-ck}$-two-sided HDX.
        \item $X$ is a $2^{-ck}$-one-sided \(k\)-partite HDX (or a skeleton thereof)\footnote{A skeleton of a $2^{-ck}$-two-sided HDX is still a $2^{-ck}$-two-sided HDX, so this distinction is not meaningful in the previous condition.}.
        \item $X$ is the $k$-skeleton of a $d$-\maximal complex $Y$ with $\lambda(U_{d-1}D_d) = 1- \frac{c}{d}$ and $d \geq 2k^2$.
    \end{enumerate}
We call $X$ \emph{\(c\)-locally nice} if every link of the complex is also $c$-nice.
\end{definition}
\begin{remark}
    We note that the two-sided condition of \pref{def:nice-complexes} could be relaxed to instead assuming that for some $\alpha \in [0,1]$, $X$ is the $k$-skeleton of a $\frac{1}{d}2^{-cd^{\alpha}}$-two-sided $d$-\maximal HDX for $d \geq k^{\frac{2}{1+\alpha}}$. We prove a Chernoff bound in this regime which we believe is of independent interest, but for simplicity focus just on the $\alpha=1$ case for the rest of the paper.
\end{remark}

Note that skeletons of nice complexes are nice (see \cite[Theorem 3.5]{lee2023parallelising} for the third case), and that if $X$ is \(c\)-nice for one of the first two items, it is also \(c\)-locally nice. It is simple computation to verify that skeletons of \(\lambda\)-TD and \(\eta\)-SI complexes are locally nice.

\begin{claim} \label{claim:td-si-are-nice}
    Let $\lambda < 1$ and $\eta > 0$. For any $d$-\maxsize complex $Y$:
    \begin{enumerate}
        \item If \(Y\) is \(\lambda\)-TD, then its \(\sqrt{(1-\lambda)d}\)-skeleton is \(e^{-1}\)-locally nice.
        \item If \(Y\) is \(\eta\)-SI, then its \(\sqrt{\frac{d}{\max\{2,\eta\}}}\)-skeleton is \(e^{-2}\)-locally nice.
    \end{enumerate}
\end{claim}
\begin{proof}
    We first prove niceness. The first item follows from observing by \pref{thm:td} and \pref{thm:gap-of-down-operator-single-step} that the down-up walk on the $k=(1-\lambda)d$ skeleton of a $\lambda$-TD complex has \(\lambda(U_{k-1}D_{k}) = 1-\frac{e^{-1}}{(1-\lambda)d}\).

    For the second item, let $k'=\frac{d}{\max\{2,\eta\}}$. By definition one can check the \(k'\)-skeleton of $Y$ is a $\min \{\frac{3}{d}, \frac{\eta}{(1-1/\eta)d}\}$-one-sided local-spectral expander and \pref{thm:gap-of-down-operator-single-step} implies $\lambda(U_{k'-1}D_{k'}) \leq 1-\frac{e^{-2}}{k'}$.

    For local niceness, note that \(\eta\)-SI and \(\lambda\)-TD complexes have \(\eta\)-SI and \(\lambda\)-TD links respectively. Further, one can check directly that if \(X\) is a \(k\)-skeleton of \(d\)-\maximal \(Y\), then \(X_s\) is a \((k-|s|)\)-skeleton of \(Y_s\), which is \((d-|s|)\)-\maximalpunc. The rate \(\frac{d-|s|}{k-|s|}\geq \frac{d}{k}\) thus we can apply the same proof of niceness to every link as well.
\end{proof}

\section{Concentration on HDX} \label{sec:chernoff}
With `nice' complexes so defined, we may finally state our main concentration bound in full formality.
\begin{theorem}[Sampling on HDX]\label{thm:hdx-is-sampler}
    Fix $c >0$. There are constants $c_1,c_2>0$ such that for any $k \in \mathbb{N}$, $i < k \leq d$, and \(d\)-\maximal \(c\)-nice complex $X$, the inclusion graph \((\X[k],\X[i])\) is an \((\varepsilon, \beta)\)-function sampler for 
    \begin{enumerate}
        \item \(\beta = \frac{c_1}{\varepsilon}\exp\left(-c_2 \varepsilon^2 \frac{k}{i}\right)\) in the non-partite cases
        \item \(\beta = \frac{c_1}{\varepsilon}\exp\left(-c_2 \varepsilon^8 \frac{k}{i}\right)\) in the partite case
    \end{enumerate}
    Moreover, if \(X\) is a \(k\)-skeleton of a \(d\)-\maximal complex \(Y\), then for every \(k < k' \leq d\), the graph $(\Y[k'],\Y[i])$ is also an \((\varepsilon, O(\beta))\)-function sampler.
\end{theorem}

A few remarks are in order. First, in the $\lambda$-TD setting, it is worth noting the constant in the exponent can be taken to be $c_2=\Theta((1-\lambda)^{1/2})$. Second, we note that in the partite case, the optimal $\varepsilon^2$-dependence in the exponent can be recovered whenever $k \leq O(\varepsilon^6 d)$ (see \pref{sec:concentration-all-nice-hdx} for details). 

Finally, by applying \pref{claim:opposite-sampler} and \pref{claim:additive-multiplicative-sampler-equivalence}, we also get the following useful corollary for multiplicative sampling of the reversed inclusion graph $(\X[i],\X[k])$.
\begin{corollary} \label{cor:hdx-is-mult-sampler}
    Let \(X\) be a \(k\)-\maximal \(c\)-\emph{nice} complex.
    Then for every $\delta,\beta \in (0,1)$, $(\X[i],\X[k])$ is a $(\alpha, \beta, \delta)$-sampler for
        \begin{enumerate}
        \item \(\alpha  = \frac{c'}{\beta\delta}e^{-\Omega_c(\beta^2\delta^2\frac{k}{i})}\) in the non-partite cases
        \item \(\alpha = \frac{c'}{\beta\delta}e^{-\Omega_c(\beta^8\delta^8\frac{k}{i})}\) in the partite case.
    \end{enumerate}
\end{corollary}
We note that in application, we typically take $\beta$ and $\delta$ in the above to be constant, in which case the bounds are tight in both regimes.

In this section, we prove \pref{thm:hdx-is-sampler} only for the two-sided case, and defer the full proof to \pref{sec:concentration-all-nice-hdx}. The proof (and section) are split into two main components: a new `Chernoff-Hoeffding'-style bound for simplicial complexes based on expansion of swap walks, and a bootstrapping technique that lifts this bound to full inclusion sampling on two-sided HDX.

\subsection{Chernoff-Hoeffding}
In this section we introduce a new technique to prove `$k$ vs.\ $1$' concentration on simplicial complexes $X$ whose underlying strength depends on the expansion of the \textit{swap walks} of $X$. The general form is fairly cumbersome (we refer the reader to \pref{prop:concentration-splitting}, which bounds the moment generating function in terms of swap walk expansion), and for readability we only state the result here for $\frac{1}{k}2^{-k^{\alpha}}$-two-sided HDX, interpolating between $\frac{1}{2k}$-HDX at $\alpha=0$, and $2^{-\Omega(k)}$-HDX at $\alpha=1$.

\begin{theorem}\label{thm:split-concentrate}
    Let $\alpha \in [0,1]$ and $X$ be a $k$-\maximal $\frac{1}{k}2^{-k^{\alpha}}$-two-sided local-spectral expander. Then for any $f: \X[1] \to [0,1]$ and $\varepsilon >0$:
    \begin{enumerate}
        \item Upper tail: 
        \[
        \underset{\{v_1,\ldots,v_k\} \in \X[k]}{\Pr}\left[\frac{1}{k}\sum\limits_{i=1}^k f(v_i) - \mu > \varepsilon\right] \leq 2\min\left\{k,c_\alpha\right\}\exp\left(-\frac{1}{12}\varepsilon^2k^{\frac{1+\alpha}{2}}\right)
        \]
        \item Lower tail: 
        \[
        \underset{\{v_1,\ldots,v_k\}\in \X[k]}{\Pr}\left[\frac{1}{k}\sum\limits_{i=0}^k f(v_i) - \mu < -\varepsilon\right] \leq 2\min\left\{k,c_\alpha\right\}\exp\left(-\frac{1}{12}\varepsilon^2k^{\frac{1+\alpha}{2}}\right)
        \]
        where $c_\alpha \leq e^{\frac{(\frac{3}{\alpha})^{\frac{1}{\alpha}}}{e^{\alpha}}}$. 
    \end{enumerate}
\end{theorem}
We remark it is possible to use this approach to prove similar concentration for any $\lambda$-TD complex, but in \pref{sec:herbst} we will see a stronger bound for this specific regime based on a variant of the Herbst argument. As such we allow ourselves the additional constant factor here in expansion for convenience. Second, note that by setting $\alpha$ appropriately, \pref{thm:split-concentrate} shows for any $c>0$ that the containment graph \((\X[k],\X[1])\) of any \(2^{-ck}\)-two-sided local spectral expanders is a \((\varepsilon,\beta)\)-sampler for \(\beta = \exp(-\Omega_{c}(-\varepsilon^2 k))\).

The core of \pref{thm:split-concentrate} centers around the notion of \textit{splittability} \cite{AlevJT2019}, which allows us to de-correlate the variables in our complex and bound the moment generating function of $U_{1,k}f$. We introduce a slight variant of the standard definition convenient for our purposes.
\begin{definition}[Balanced Splittability]
    We call a $d$-\maximal complex $X$ $\{\gamma_i\}_{i \in [\lfloor \log(d) \rfloor-1]}$-balanced-splittable if for every $i$ the $2^i$-th swap walk expands:
    \[
    \lambda(S_{2^i,2^i}) \leq \gamma_i
    \]
\end{definition}
All applications of splittability in the literature lose \textit{additive} factors in the $\gamma_i$. We introduce a new approach combining splittability with an `internal' reduction to the complete complex that allows finer control over this error and instead incurs more manageable \textit{multiplicative} loss.
\begin{proposition}\label{prop:concentration-splitting}
    For any $j \in \mathbb{N}$ and $r \in [0,1]$, let $k=2^j$, and $X$ be a $k$-\maximal $\{\gamma_i\}$-balanced-splittable complex with $e^{\frac{5}{4}2^{-i}r^2k}\gamma_{j-1-i} < 1$. Then for any $f: \X[1] \to [0,1]$:
    \[
    \underset{\{v_1,\ldots, v_k\} \in \X[k]}{\mathbb{E}}\left[\prod_{i=0}^ke^{r f(v_i)}\right] \leq c_{r,\gamma}\underset{\X[1]}{\mathbb{E}}[e^{rf}]^k
    \]
    where $c_{r,\gamma} \leq \left(\prod\limits_{i=0}^{j-1}\left(\frac{1-\gamma_{j-1-i}}{1-e^{\frac{5}{4}2^{-i}r^2k}\gamma_{j-1-i}}\right)^{2^{i}}\right)$.
\end{proposition}
While the definition of \(c_{r,\gamma}\) may be somewhat hard to parse, note at least that for any fixed \(r\), \(c_{r,\gamma}\) goes to $1$ as \(\gamma\) tends to \(0\).
\begin{proof}
For any even $\ell$, denote by \(z_\ell:\X[\ell]\to \mathbb{R}\) the `partial' moment generating function:
    \[
    z_\ell(\set{v_1,\dots,v_{\ell}}) = \prod_{i=1}^{\ell} e^{rf(v_i)}.
    \]
    Our goal is to bound \(\ex{z_{k}}=\mathbb{E}[\exp(r\sum_{i=1}^k f(v_i))]\) by induction on $z_\ell$. Denote by \(z_\ell^0 = \ex{z_\ell}\one\) the projection of \(z_\ell\) to the space of constant functions and by \(z_\ell^\perp = z_\ell - z_\ell^0\) its orthogonal part. The first key observation is that we may split $z_\ell$ into two `copies' of $z_{\ell/2}$ using the swap operator:
    \[
    \Ex[{\X[\ell]}]{z_\ell} = \Ex[{\X[\ell/2]}]{z_{\ell/2} S_{\ell/2} z_{\ell/2}}.
    \]
    Now splitting $z_{\ell/2}$ itself into parallel and perpendicular components, we can bound this as
    \begin{align} \label{eq:do-again}
        \nonumber\Ex[X{(\ell/2)}]{z_{\ell/2} S_{\ell/2} z_{\ell/2}} &= \mathbb{E}[{z_{\ell/2} S_{\ell/2} (z_{\ell/2}^0 + z_{\ell/2}^\perp)}] \\
    \nonumber&= \ex{z_{\ell/2}}^2 + \iprod{z_{\ell/2}^\perp,S_{\ell/2} z_{\ell/2}^\perp}\\
    \nonumber&\leq \ex{z_{\ell/2}}^2 + \gamma_{\log(\ell)-1} \norm{z_{\ell/2}^\perp}_2^2\\
    &=(1-\gamma_{\log(\ell)-1})\ex{z_{\ell/2}}^2 + \gamma_{\log(\ell)-1} \norm{z_{\ell/2}}_2^2.
    \end{align}
    The first of these two terms can clearly be bounded by induction. The difficulty lies in the error term, as the 2-norm of the partial MGF may be much larger than its expectation. To handle this, we argue via reduction to the complete complex the second moment of $z_{\ell/2}$ is actually quite close to the original expectation of $z_\ell$.
    \begin{claim}\label{claim:Chernoff-inner-prod-bound}
        For any even $\ell$ and $r \in [0,1]$:
        \[
        \norm{z_{\ell/2}}_2^2 \leq e^{\frac{5}{4}r^2\ell}\mathbb{E}[z_\ell]
        \]
    \end{claim}
    Let's first complete the proof given this fact. Plugging the claim back into \pref{eq:do-again} and applying induction we get the following bound:
    \begin{align*}
    \mathbb{E}[z_{k}] &\leq \frac{1-\gamma_{j-1}}{1-e^{\frac{5}{4}r^2 \ell}\gamma_{j-1}}\mathbb{E}[z_{k/2}]^2\\
    &\leq \frac{1-\gamma_{j-1}}{1-e^{\frac{5}{4}r^2k}\gamma_{j-1}}\left(\prod_{i=0}^{j-2}\left(\frac{1-\gamma_{j-2-i}}{1-e^{2^{-i}\frac{5}{4}r^2\frac{k}{2}}\gamma_{j-2-i}}\right)^{2^{i+1}}\right)\mathbb{E}[z_1]^k\\
    &=\left(\prod_{i=0}^{j-1}\left(\frac{1-\gamma_{j-1-i}}{1-e^{2^{-i}\frac{5}{4}r^2(k+1)}\gamma_{j-1-i}}\right)^{2^{i}}\right)\mathbb{E}[e^{rf}]^k
    \end{align*}
    as desired.
\end{proof}
It is left to prove the claim relating $\norm{z_{\ell/2}}_2^2$ to $\mathbb{E}[z_\ell]$.
\begin{proof}[Proof of \pref{claim:Chernoff-inner-prod-bound}]
    Recall our goal is to show
    \[
    \mathbb{E}[z_{\ell/2}^2]= \underset{t \in \X[\ell/2]}{\mathbb{E}}\left[\prod_{v \in t}e^{2rf(v)}\right] \leq e^{\frac{5}{4}r^2\ell}\mathbb{E}[z_\ell]
    \]
    The key is to observe that we can draw $t \in \X[\ell/2]$ by first drawing $s \in \X[\ell]$, and then drawing $t \subset s$ uniformly at random:
    \begin{align*}
    \underset{t \in \X[\ell/2]}{\mathbb{E}}\left[\prod_{v \in t}e^{2rf(v)}\right]&=\underset{s \in \X[\ell]}{\mathbb{E}}\left[\underset{t \subset s}{\mathbb{E}}\left[\prod_{v \in t}e^{2rf(v)}\right]\right]
    \end{align*}
    The inner expectation is now within the complete complex which is negatively correlated, so may therefore `pull out' the product (e.g.\ by \cite[Theorem 4]{hoeffding1994probability}) and write:
    \[
    \underset{t \subset s}{\mathbb{E}}\left[\prod_{v \in t}e^{2rf(v)}\right] \leq \underset{v \in s}{\mathbb{E}}[e^{2rf(v)}]^{\ell/2}.
    \]
    Let $\mu_s=\mathbb{E}_{v \in s}[f(v)]$ and observe $\ell\mu_s=\sum_{v \in s}f(v)$. Combining the above, we have
    \begin{align*}
        \mathbb{E}[z_{\ell/2}^2] &\leq \underset{s \in \X[\ell]}{\mathbb{E}}\left[\underset{{v \in s}}{\mathbb{E}}[e^{2rf(v)}]^{\ell/2}\right]\\
        &\leq  \underset{s \in \X[\ell]}{\mathbb{E}}\left[\underset{{v \in s}}{\mathbb{E}}[1+2rf(v)+\frac{5}{2}r^2f(v)]^{\ell/2}\right]\\
        &= \underset{s \in \X[\ell]}{\mathbb{E}}\left[(1+2r\mu_s+\frac{5}{2}r^2\mu_s)^{\ell/2}\right]\\
        &\leq \underset{s \in \X[\ell]}{\mathbb{E}}\left[e^{(r+\frac{5}{4}r^2)\ell\mu_s}\right]\\
        &= \underset{s \in \X[\ell]}{\mathbb{E}}\left[\prod\limits_{v \in s}e^{(r+\frac{5}{4}r^2)f(v)}\right]\\
        &\leq e^{\frac{5}{4}r^2\ell}\mathbb{E}\left[z_\ell\right]
    \end{align*}
    where in the second inequality we've used that \(e^x < 1+x+\frac{5}{4}x^2\) for \(x \in [0,2]\) and that \(f^2(v) \leq f(v)\). 
\end{proof}
The proof of \pref{thm:split-concentrate} is now essentially immediate from Chernoff's method.
\begin{proof}[Proof of \pref{thm:split-concentrate}]
    We prove the upper tail. The lower tail follows from applying the upper tail to $1-f$. Let $k'$ be the largest power of $2$ such that $k' \leq k$. Fix $r=\frac{1}{2}\varepsilon k'^{\frac{\alpha-1}{2}}$, and recall by \pref{thm:swap-walks-expand} that $\gamma_i \leq \frac{2^{i}}{k'}2^{-(k')^\alpha}$. By \pref{prop:concentration-splitting}, we therefore have:
    \begin{align*}\label{eq:markov-for-z}
    \Prob[\set{v_1,\dots,v_{k'}} \in {\X[k']}]{\sum_{i=1}^{k'} f(v_i) \geq \mu k' + \varepsilon k'} = &\Prob[\set{v_1,\dots,v_{k'}} \in {\X[k']}]{\exp\left(r\sum_{i=1}^k f(v_i)\right) \geq \exp\left(rk' \mu + rk'\varepsilon\right)}\\
    \overset{Markov}{\leq} &\frac{\mathbb{E}\left[\prod_{i=0}^{k'} e^{rf(v_i)}\right]}{\exp(rk'\mu + rk'\varepsilon)}\\
    \leq& \left(\prod\limits_{i=0}^{j-1}\left(\frac{1}{1-2^{-i-1}e^{-\frac{1}{3}(k')^\alpha}}\right)^{2^{i}}\right) \exp\left(-\frac{\varepsilon^2}{3}(k')^{\frac{1+\alpha}{2}}\right)
    \end{align*}
    where in the final inequality we have applied \pref{prop:concentration-splitting} and bounded $\mathbb{E}[e^{rf}]^{k'}$ by standard manipulations. To bound the leading coefficient first note that for any $\alpha \geq 0$ we have
    \[
        \prod\limits_{i=0}^{j-1}\left(\frac{1}{1-2^{-i-1}e^{-\frac{1}{3}(k')^\alpha}}\right)^{2^{i}}= \prod\limits_{i=0}^{j-1}\left(\frac{1}{1-2^{-i-1}}\right)^{2^{i}} \leq k'
    \]
    since each individual term in the product is at most $\frac{1}{2}$ and $j=\log_2(k')$. For $\alpha>0$ we may also write
    \begin{align*}
        \prod\limits_{i=0}^{j-1}\left(\frac{1}{1-2^{-i-1}e^{-\frac{1}{3}(k')^\alpha}}\right)^{2^{i}} &\leq \prod\limits_{i=0}^{j-1}\left(\frac{1}{1-\frac{e^{-\frac{1}{3}(k')^\alpha}}{2}}\right)^{2^{i}}\\
        &\leq \left(1+e^{-\frac{1}{3}(k'+1)^\alpha}\right)^{k'}\\
        &\leq e^{\frac{k'}{e^{\frac{1}{3}(k')^\alpha}}}\\
        &\leq e^{\frac{(\frac{3}{\alpha})^{\frac{1}{\alpha}}}{e^{\alpha}}}
    \end{align*}
    where the last inequality follows from observing $e^{\frac{k'}{e^{\frac{1}{3}(k')^\alpha}}}$ is maximized at $k'=(\frac{3}{\alpha})^{-\alpha}$.

Finally we apply \pref{lem:raising} to lift this guarantee to $\X[k]$:
    \begin{align*}
    \Pr_{\X[k]}[U_{1,k}f - \mu > \varepsilon] &\leq \min\{k',c_\alpha\}\exp\left(-\frac{\varepsilon^2}{12}(k'+1)^{\frac{1+\alpha}{2}}\right)\left(1-\pi^{k,k',U^{k'}_0f}_{low}(\varepsilon/2)\right)^{-1}\\
    &\leq \min\{k,c_\alpha\}\exp\left(-\frac{\varepsilon^2}{12}(k+1)^{\frac{1+\alpha}{2}}\right)\left(1-e^{-\frac{\varepsilon^2(k+1)}{6}}\right)^{-1}\\
    &\leq 2\min\{k,c_\alpha\}\exp\left(-\frac{\varepsilon^2}{12}k^{\frac{1+\alpha}{2}}\right)
    \end{align*}
    where in the final inequality we have assumed $\varepsilon^2 \geq \frac{6}{k}$ (else the stated bound is trivial).
\end{proof}

\subsection{From Chernoff to Inclusion Sampling}\label{sec:inclusion}

We now show how to bootstrap Chernoff-type concentration to optimal sampling for the general inclusion graph $(\X[k],\X[i])$. We focus on the case of two-sided HDX, and give the more involved argument for the full theorem in \pref{sec:concentration-all-nice-hdx}. We restate the theorem in the two-sided case here for convenience.
\begin{theorem}[Sampling on two-sided HDX]\label{thm:hdx-is-sampler-two-sided}
    For any $c>0$ there exists $c'>0$ such that for any $k \in \mathbb{N}$, $i < k$, and \(k\)-\maximal $2^{-ck}$-two-sided local spectral expander $X$, \((\X[k],\X[i])\) is an \((\varepsilon, \beta)\) sampler for 
    \[
    \beta = \frac{c'}{\varepsilon}\exp\left(-\Omega_c\left(\varepsilon^2 \frac{k}{i}\right)\right).
    \]
    Moreover, if \(X\) is a \(k\)-skeleton of a \(d\)-\maximal complex \(Y\), then for every \(k < k' \leq d\), the graph $(\Y[k'],\Y[i])$ is also an \((\varepsilon,O(\beta))\)-function sampler.
\end{theorem}

The proof of \pref{thm:hdx-is-sampler-two-sided} is based on a refinement of an idea of \cite{ImpagliazzoKW2012}, who observed on the complete complex one can prove an analogous bound by partitioning $k$-sets into $i$ disjoint $\frac{k}{i}$-sets, and treating these as independent variables. In the world of general simplicial complexes, this type of subdivision is called the \textit{faces complex} of $X$ \cite{dikstein2023agreement}.
\begin{definition}[Faces Complex \cite{dikstein2023agreement}] \label{def:faces-complex}
    Let $X$ be a $d$-\maximal simplicial complex. For any $i \leq d$, the faces complex $F_X^{(i)}$ is the $\lfloor \frac{d}{i}\rfloor$-\maximal complex with vertices $\X[i]$ and top level faces:
    \[
    \SC[F^{(i)}][\left\lfloor \frac{d}{i}\right\rfloor][X] = \left\{\{s_1, \ldots, s_{\lfloor \frac{d}{i}\rfloor}\} ~\Bigg|~ \bigcup_{j=1}^{\lfloor \frac{d}{i}\rfloor} s_j \in X, ~ \bigcup_{j\neq \ell} s_j \cap s_\ell = \emptyset\right\}
    \]
    The weight of a face is proportional to the weight of its union in $X$. We drop $X$ from the notation when clear from context.
\end{definition}
In our setting, subdivided faces have extreme dependencies so \cite{ImpagliazzoKW2012}'s analysis for the complete complex fails. Nevertheless, it is in fact possible to generalize this approach to any simplicial complex through somewhat more careful analysis.
\begin{proposition}\label{prop:faces-to-inclusion}
    Let $X$ be a $d$-\maximal complex and $i \leq k \leq d$ be such that $\left (\SC[F^{(i)}][\lfloor \frac{k}{i}\rfloor],\SC[F^{(i)}][1] \right )$ is an $(\varepsilon,\beta)$-function additive sampler. Then $(\X[k],\X[i])$ is a $(2\varepsilon,\frac{2}{\varepsilon}\beta)$-function additive sampler.
\end{proposition}
The proof of \pref{prop:faces-to-inclusion} itself centers around a lemma for general bipartite samplers stating that sampling over correlated sub-divisions is sufficient to imply sampling for the general graph as long as the components are appropriately marginally distributed. 
\begin{lemma} \label{lem:weird-distribution-to-uniform}
    Given a graph \(G = (L,R,E)\), parameters \(\varepsilon, \beta > 0\) and $k \in \mathbb{N}$, and \(f:L \to [0,1]\), suppose there is a distribution \((v_1,v_2,\dots,v_k,u) \sim D\) over $L^{k} \times R$ satisfying
    \begin{enumerate}
        \item For every \(i\in [k]\), the marginal \((v_i,u)\) is distributed according to the edges of \(G\).
        \item \(\Prob[(v_1,v_2,\dots,v_k,u) \sim D]{\frac{1}{k}\sum\limits_{i=1}^k f(v_i) > \mathbb{E}[f] + \varepsilon} \leq \beta\).
    \end{enumerate}
    Then 
    \begin{equation}\label{eq:weird-distribution-to-uniform}
    \Prob[u \sim R]{\Ex[{v \in L, v \sim u}]{f(v)} > \mathbb{E}[f] + 2\varepsilon} \leq \frac{\beta}{\varepsilon}.
    \end{equation}
    The same statement applies replacing \(>\) with \(<\) and $\varepsilon$ with $-\varepsilon$ inside the probabilities.
\end{lemma}
We defer the proof to the end of the subsection and first prove \pref{prop:faces-to-inclusion}.
\begin{proof}[Proof of \pref{prop:faces-to-inclusion}]
    Let $m=\lfloor \frac{k}{i} \rfloor$ and fix a function \(f:\X[i] \to [0,1]\) with expectation \(\mu = \Ex[s \in {\X[i]}]{f(v)}\). Consider the distribution $D$ which samples \((s_1,s_2,\dots,s_m,t) \subset \X[i]^{m} \times \X[k]\) by drawing $\{s_i\}_{i=1}^m$ from the faces complex, and draws $t \supset \bigcup s_i$ from $\X[k]$. We claim $D$ satisfies the requirements of \pref{lem:weird-distribution-to-uniform}, namely:
    \begin{enumerate}
        \item $\forall i \in [m]$, \((s_i,t)\) is marginally distributed according to the inclusion graph edges $(\X[k],\X[i])$
        \item $\Pr_{\{s_1,s_2,\dots,s_m,t\} \sim D}\left[\left|\frac{1}{k}\sum\limits_{i=1}^k f(s_i)-\mu\right| > \varepsilon\right] \leq \beta$.
    \end{enumerate}
    The first fact follows from the definition of the faces complex, namely that marginally each $s_i$ is distributed as $\X[i]$, and $t$ is drawn from $\X[k]$ conditional on containing $s_i$. The second fact is immediate from our assumption that the faces complex is an $(\varepsilon,\beta)$-sampler.
    
\end{proof}

We are now ready to prove the two-sided case of \pref{thm:hdx-is-sampler}.
\begin{proof}[Proof of \pref{thm:hdx-is-sampler-two-sided} (\pref{thm:hdx-is-sampler} two sided case)]

By \pref{thm:split-concentrate} if the faces complex of $X$ is a $2^{-ck}$-two-sided HDX, then it is an $(\varepsilon,\beta)$-sampler for $\beta \leq c_1e^{\Omega_c\left( \varepsilon^2\lfloor \frac{k}{i}\rfloor \right)}$ and $c_1>0$ some constant dependent only on $c$. Then by \pref{prop:faces-to-inclusion}, $X$ would be a $(2\varepsilon,\frac{\beta}{\varepsilon})$-sampler. The proof is completed by observing $2^{-\Omega(k)}$-local-spectral expansion of the faces complex is immediate from \pref{thm:swap-walks-expand}, since the graph underlying any link of the faces complex is a swap walk within (a link of) $X$. The `moreover' statement follows directly from \pref{lem:raising}.
\end{proof}

It remains to prove \pref{lem:weird-distribution-to-uniform}.
\begin{proof}[Proof of \pref{lem:weird-distribution-to-uniform}]
    Let \(Z:R \to [0,1]\) record the probability $u_0 \in R$ $\varepsilon$-mis-samples $f$ under $D$:
    \[
    Z(u_0) \coloneqq \cProb{(v_1,v_2,\dots,v_k,u) \sim D}{\frac{1}{k}\sum_{i=1}^kf(v_i) > \mu+\varepsilon}{u=u_0}.
    \]
    By assumption, the expected mis-sampling probability is small
    \[
    \Ex[u_0]{Z(u_0)} = \Prob[(v_1,v_2,\dots,v_k,u) \sim D]{ \frac{1}{k}\sum_{i=1}^kf(v_i) > \mu+ \varepsilon} \leq \beta,
    \]
    so Markov's inequality implies the following tail bound on $Z(u_0)$:
    \[
    \Prob[u_0 \in R]{Z(u_0) > \varepsilon} \leq \frac{\beta}{\varepsilon}.
    \] 
    We argue if \(u_0\) is such that \(\Ex[v \in L, v\sim u_0]{f(v_i)} > \mu+ 2\varepsilon\) then \(Z(u_0) > \varepsilon\), thus concluding that \eqref{eq:weird-distribution-to-uniform} holds.

    The key is the simple observation that the (conditional) expectation $A$ under $D$ and $G$ are equivalent:    
    \begin{align}\label{eq:weird-to-true-expectation}
    \cEx{(v_1,v_2,\dots,v_k,u) \sim D}{f(v_i)}{u=u_0} &= \frac{1}{k}\sum_{i=1}^k \Ex[(v_i,u)]{f(v_i) | u=u_0} \nonumber \\ 
    &= \cEx{v \in L}{f(v)}{v \sim u_0}
    \end{align}
    
    where the first equality is by linearity of expectation, and the second is by our assumption that $(v_i,u)$ is distributed as $G$. The result now follows roughly from arguing the random variable \(\frac{1}{k}\sum_{i=1}^kf(v_i)\) restricted to \(u_0\), is non-trivially concentrated around its expectation \(\cEx{v \in L}{f(v)}{v \sim u_0}\), which is at least $2\varepsilon$-far from $\mu$.

    In particular, fix \(u_0\) such that \(\cEx{(v_1,v_2,\dots,v_k,u) \sim D}{f(v_i)}{u=u_0}=\mu_{u_0} > \mu + 2\varepsilon\). Then, we bound $Z(u_0)$ by its upper tail:
    \begin{align*}
        Z(u_0) &\geq \cProb{(v_1,v_2,\dots,v_k,u) \sim D}{\frac{1}{k}\sum_{i=1}^kf(v_i) > \mu_{u_0}-\varepsilon}{u=u_0}. 
    \end{align*}
    Denote the right-hand side by $p$, and as \(f(v) \in [0,1]\) we have
    \[ 
    \mu_{u_0} \leq p + (1-p)(\mu_{u_0}-\varepsilon). 
    \]
    Re-arranging gives 
    \[ 
    Z(u_0) \geq p \geq \frac{\varepsilon}{1-\mu_{u_0} - \varepsilon} \geq \varepsilon. 
    \]
    The last inequality follows from the the denominator being in \((0,1)\) (\(0<\mu_{u_0}-\varepsilon<1\) because \(\mu_{u_0}-\varepsilon \geq \mu+\varepsilon > 0\) and \(\mu_{u_0}-\varepsilon<\mu_{u_0}\leq 1\)).

    Finally, we note the usual trick (changing \(f\) in the proof to \(1-f\)) gives the bounds for 
    \[
    \Prob[u \in R]{\Ex[v \sim u, v \in L]{f(v)} < \mu - 2\varepsilon}.
    \]
\end{proof}
\section{Reverse Hypercontractivity} \label{sec:rhc}
In this section we recursively leverage optimal inclusion sampling to prove reverse hypercontractivity for high dimensional expanders. 
\begin{theorem}[Reverse Hypercontractivity]\label{thm:intro-rhc-real}
    Fix $c>0$, $\rho \in (0,1)$, and let \(X\) be a \(k\)-\maximal \(c\)-locally nice complex for $k$ sufficiently large. Then there exist constants $C,q$ (dependent only on $c,\rho$) such that
    \[
    \langle f, T_\rho g \rangle \geq C\norm{f}_q\norm{g}_q.
    \]
\end{theorem}

At its core, the proof of \pref{thm:intro-rhc-real} really only relies on $X$ and its links satisfying optimal $i$ vs $k$ sampling. With this in mind, we will instead work with the following notion of a \textit{Link up-sampler} that captures this core property:
\begin{definition}[Link up-sampler (LUS)]\label{def:lus}
    Let \(\tau \in (0,1)\) and let \(X\) be a \(k\)-\maximal simplicial complex. We say that \(X\) is a \(\tau\)-LUS if for every \(i \leq k-2\), every $j \leq k-i$, and every \(s \in \X[i]\), it holds that \((\SC[X][j][s],\SC[X][k-i][s])\) is a \((\frac{1}{\tau}\exp(-\tau \frac{k-i}{j}),0.1,0.2)\)-multiplicative sampler. 
\end{definition}
We note the constants \(0.1,0.2\) are essentially arbitrary and fixed for convenience. The notion can also be further relaxed in a number of ways such as asking for sampling for $j$ only up to a linear number of levels or only within links up to a certain level and the main results of this section would still hold.

We prove link-up samplers are reverse hypercontractive. \pref{thm:intro-rhc-real} is an immediate corollary.
\begin{theorem}\label{thm:reverse}
    For every \(\rho \in (0,1)\) and \(\tau>0\), there is some \(\ell > 1\) and \(C > 0\) such that the following holds for all sufficiently large \(k \in \NN\). Let \(X\) be \(k\)-\maximal and \(\tau\)-LUS. Then \(T_\rho\) is \((\frac{1}{\ell},\frac{\ell}{1-\ell},C)\)-reverse hypercontractive.
\end{theorem}

We split the proof of \pref{thm:reverse} below into two parts, roughly corresponding to reverse hypercontractivity for boolean functions, and a reduction from the general to boolean case.
\subsection{The Boolean Case}\label{sec:rhc-indicators}
The main workhorse behind \pref{thm:reverse} is to show that any linear-step \textit{down-up walk} on an LUS is reverse hypercontractive for sets of size up to \( \exp(-\Omega(k))\).
\begin{theorem}\torestate{\label{thm:indicator-reverse-hc}
Let \(\tau > 0\) and $\gamma \in (0,1)$. There exist constants \(c,q > 0\) such that the following holds for any sufficiently large \(k\). Let \(X\) be a \(k\)-\maximal \(\tau\)-LUS. Then for any sets \(A,B \subseteq \X[k]\) of relative size at least \( \exp \left (-c k \right )\) and $\gamma' \leq \gamma$ it holds that
\[
    \Prob[t,t'\sim UD_{k, \lfloor \gamma' k\rfloor}]{t \in A, t' \in B} \geq \prob{A}^q \prob{B}^q.
\]}
\end{theorem}

It will be convenient for us to assume at least one of the sets $A$ and $B$ is not too large. To this end, we first handle the case that both are large separately.
\begin{claim}\label{claim:big-sets-not-a-problem}
    For every \(p>0.5\) there exists \(q \geq 1\) so that the following holds. Let \(X\) be a simplicial complex and \(A,B \subseteq \X[k]\) subsets such that \(\prob{A},\prob{B} \geq p\). Then for any $\gamma' \leq \gamma$:
    \[
    \Prob[t,t'\sim UD_{k, \lfloor \gamma' k\rfloor}]{t \in A, t' \in B} \geq \prob{A}^q \prob{B}^q.
    \]
\end{claim}
The proof is largely un-enlightening computation, so we postpone it to the end of the subsection. Henceforth, assume without loss of generality that \(\prob{A} \leq \prob{B}\) and has \(\prob{A} \leq 0.8\).

\bigskip

With this out of the way, we overview the main proof strategy. Assume for the moment that $\gamma'k$ is integral. As discussed in the introduction, \pref{thm:indicator-reverse-hc} relies on the following elementary observation: since the down-up walk samples an edge by first sampling \(s \in \X[\gamma k]\) and then \textit{independently} sampling \(t,t' \in \X[k]\) containing \(s\), we can express the $\gamma'$-correlated mass of $A$ and $B$ as:
    \[
        \Prob[t,t'\sim U D_{k, \gamma' k}]{t \in A, t' \in B} = \Ex[s \in {\X[\gamma' k]}]{\cProb{t}{A}{t \supseteq s} \cdot \cProb{t}{B}{t \supseteq s}}.
    \]
This motivates us to consider the set of \(s \in \X[\gamma k]\) that \textit{simultaneously} see a large enough fraction of \(A,B\). Toward this end, for any integer \(\ell_0 \leq k\) and real number \(\delta>0\), let
\[
G^{X}(\ell_0,\delta) = \sett{s \in {\X[\ell_0]}}{\cProb{t}{A}{t \supseteq s} \geq \delta \prob{A} \ve \cProb{t}{B}{t \supseteq s} \geq \delta \prob{B}}
\] 
denote the set of elements in \(\X[\ell_0]\) that see a noticeable fraction of both \(A\) and \(B\). 

Recalling \(\prob{B} \geq \prob{A}\), to prove \pref{thm:indicator-reverse-hc} it suffices to show for any $\gamma' \leq \gamma$
\begin{equation}\label{eq:G-is-large}
\prob{G^{X}(\gamma' k,\prob{A}^{O(1)})} \geq \prob{A}^{O(1)},
\end{equation}
since (writing $G \coloneqq G^X(\gamma' k,\prob{A}^{O(1)})$) we then have
\begin{align} 
    \nonumber \Prob[t,t'\sim U D_{k, \gamma' k}]{t \in A, t' \in B}  &\geq \Prob[s \in {\X[\gamma' k]}]{G}\cdot\Pr_{s \in \X[\gamma' k]}\left[\Pr_t[t \in A|s \in G]\Pr_t[t \in B|s \in G]\right]\\
    &\geq \Pr[A]^{O(1)}\Pr[B]\label{eq:large-G-implies-rhc}
\end{align}

As such, our main goal is to prove \eqref{eq:G-is-large}. A key observation toward this end is that $G$ satisfies certain useful `composition' behavior: we can bound $G(\gamma' k, \cdot)$ by iteratively bounding $G(i, \cdot)$ for smaller $i$ within repeatedly deeper links of the complex. In particular, for sets $A$ and $B$ and any $r \in X$, define the restrictions of $A$ and $B$ to the link of $r$ as
\[
A_r \coloneqq \{t \in X_r: r \cup t \in A\} \quad \quad \text{and} \quad \quad B_r \coloneqq \{t \in X_r: r \cup t \in B\}
\]
and define the \textit{local} subset $G^{X_r}(\ell_0,\delta)$ as the set of $s \in \X[\ell_0][r]$ that see a sufficiently fraction of $A_r$ and $B_r$:
\[
G^{X_r}(\ell_0,\delta) = \sett{s \in {\X[\ell_0][r]}}{\cProb{t \in \X[k-r][r]}{A_r}{t \supseteq s} \geq \delta \prob{A_r} \ve \cProb{t \in \X[k-r][r]}{B_r}{t \supseteq s} \geq \delta \prob{B_r}}.
\] 
$G$ and $G^{X_r}$ compose in the following fashion:
\begin{observation}[Composition of \(G^X\)] \label{obs:links-goodness}
    Let \(A,B \subset \X[k]\), \(\ell_1, \ell_2 \in \mathbb{N}\), and let \(\delta_1,\delta_2,\eta_1,\eta_2 > 0\). If
    \begin{enumerate}
        \item \(\prob{G^X(\ell_1,\delta_1)} \geq \eta_1\).
        \item \(\forall r \in G^X(\ell_1,\delta_1)\): $\Prob[{\X[\ell_2][r]}]{G^{X_r}(\ell_2,\delta_2)} \geq \eta_2$
    \end{enumerate}
    Then 
    \[\prob{G^X(\ell_1 + \ell_2,\delta_1\delta_2)} \geq \eta_1 \eta_2.\]
\end{observation}
\begin{proof}
    Sample \(s \in \X[\ell_1+\ell_2]\) by sampling \(r_1 \in \X[\ell_1]\) and then \(r_2 \in \X[\ell_0][r]\) setting \(s=r_1 \dunion r_2\). Fix \(r_1 \in G^X(\ell_1,\delta_1)\). Observe that if \(r_2 \in G^{X_{r_1}}(\ell_2,\delta_2)\) then
    \begin{align*}
        \cProb{t}{A}{r_1 \dunion r_2 \subseteq t} &\geq \delta_2 \cProb{t}{A}{r_1 \subseteq t} \\
        &\geq \delta_1 \delta_2 \prob{A},
    \end{align*}
    and similarly for \(B\). Therefore in this case \(s \in G^X(\ell_1 + \ell_2,\delta_1\delta_2)\). Hence
    \begin{align*}
        \Prob[s]{G^X(\ell_1 + \ell_2,\delta_1\delta_2)} &\geq \Prob[r_1]{G^X(\ell_1,\delta_1)} \cdot \cProb{s=r_1 \dunion r_2}{r_2 \in G^{X_{r_1}}(\ell_2,\delta_2)}{r_1 \in G^X(\ell_1,\delta_1)}\\
        &\geq \eta_1 \eta_2.    
    \end{align*}
\qedhere
\end{proof}
The proof strategy (modulo details), is then simple. To bound $G(\gamma' k, \Pr[A]^{q_0})$, we divide $t \in \X[k]$ into pieces as $r_1 \cup r_2 \cup r_3 \ldots$, where each $r_i$ is viewed as being drawn from the link of $r_1\cup \ldots \cup r_{i-1}$. The local $G$-set for each $r_i$ can be lower bounded by the sampling properties of $X_{\cup_{j < i} r_j}$, and by carefully choosing the size of the $r_i$, we can ensure that composition gives the desired bound on $G$.

In practice, this approach hits some complications due to the fact that the localized restrictions of $A$ and $B$ shrink throughout the process. To handle this, we'll instead use a `two-phase' version of the above. First, we show that there exists some small absolute constant $\gamma_0$, potentially much smaller than the desired $\gamma$, for which the above approach does work. We call complexes satisfying such a bound lovely:
\begin{definition}[Lovely complex]
    We say that \(X\) is \((c,1-\gamma_0,k,q)\)-\emph{lovely}, if for all \(A,B \subset \X[k]\) with $\mathbb{E}[B] \geq \mathbb{E}[A]\geq \exp(-ck)$
    \[
    \prob{G^X(\gamma_0 k, \prob{A}^q)} \geq \prob{A}^q.
    \]
\end{definition}

We first argue any LUS complex is $\gamma'$-lovely for some $\gamma'$ depending only on $\tau$ and $\gamma$.

\begin{lemma} \label{lem:bootstrap-rhc}
    For every \(\tau > 0\) and \(\gamma < 1\) there exist \(c,q\) and \(\gamma_0 > 0\) so that following holds for large enough \(k\). Let \(X\) be \(\tau\)-LUS, then for every \(t \in X^{\leq \gamma k}\), \(X_t\) is \((c,1-\gamma',k-|t|,q)\)-lovely for any \(\gamma' \leq \gamma_0\).
\end{lemma}
Second, we give a similar argument showing that loveliness itself can be bootstrapped to higher levels.
\begin{claim} \label{claim:from-one-constant-to-any-constant}
    Let \(X\) be a \(k\)-\maximal simplicial complex. Let \(q',q''\geq 1\), and \(c',c'',\gamma_1,\gamma_2 \geq 0\). Assume that:
    \begin{enumerate}
        \item \(X\) is \((c',1-\gamma_1,k,q')\)-lovely.
        \item For every \(r \in X(\gamma_1 k)\), \(X_r\) is \((c'',1-\gamma_2,(1-\gamma_1)k,q'')\)-lovely.
    \end{enumerate}
    Then \(X\) is \((c,(1-\gamma_1)(1-\gamma_2),k,q)\)-lovely for \(q=q''q'+q''+q'\) and \(c= \min \set{\frac{c''}{(q'+1)(1-\gamma_1)},c'}\). 
\end{claim}

Combined with some definition chasing, \pref{lem:bootstrap-rhc} and \pref{claim:from-one-constant-to-any-constant} immediately imply \pref{thm:indicator-reverse-hc}.
\begin{proof}[Proof of \pref{thm:indicator-reverse-hc}]
    Recall that by \eqref{eq:large-G-implies-rhc} it suffices to show \(X\) is \((c,1-\gamma,k,q)\)-lovely for some constants \(c,q\) depending only on \(\gamma\) and $\tau$. We first prove a slightly different statement by induction:
    \[
    \forall \ell \in \left[\left\lceil \frac{\log(1-\gamma)}{\log(1-\gamma_0)}\right\rceil\right]: \quad \text{\(X\) is \((c_\ell,(1-\gamma_0)^\ell,k,q_\ell)\)-lovely}.
    \]
    for constants $c_\ell,q_\ell$ dependent only on $\tau,\gamma$. Note this implicitly assumes $(1-\gamma_0)^\ell$ is integer. This is unecessary and we discuss removing this assumption at the end of the subsection. Finally, observe that for $\ell = \left\lceil \frac{\log(1-\gamma)}{\log(1-\gamma_0)}\right\rceil$ we have \((1-\gamma_0)^\ell \leq 1-\gamma\). We will argue a slight modification gives exactly $1-\gamma''$ for any $\gamma'' \leq \gamma$ after showing the above.
    
    The base case $\ell=1$ is immediate from \pref{lem:bootstrap-rhc}. Assume the statement is true for some $\ell < \left\lceil \frac{\log(1-\gamma)}{\log(1-\gamma_0)}\right\rceil$, then \(X\) is \((c(\ell),(1-\gamma_0)^\ell,k,q(\ell))\)-lovely. Moreover, for \(j= 1-(1-\gamma_0)^\ell \leq \gamma k\), \pref{lem:bootstrap-rhc} implies every $j$-link \(X_t\) is \((c'',1-\gamma',k-j,q'')\)-lovely for any \(\gamma' \leq \gamma_0\).
    Thus by \pref{claim:from-one-constant-to-any-constant}, we have that $X$ is \((c_{\ell+1},1-(1-\gamma_0)^{\ell+1},k-|t|,q_{\ell+1})\)-lovely where $c_{\ell+1}$ and $q_{\ell+1}$ are functions of $\gamma_0$, $c(\ell)$, $q(\ell)$, $c''$, and $q''$ as described in \pref{claim:from-one-constant-to-any-constant}. Since $\gamma_0$, $c''$, and $q''$ are themselves functions only of $\tau$ and $\gamma$, this gives the desired bound.
    
    In order to prove $X$ is $(c,1-\gamma,k,q)$-lovely, simply observe that in the final step of this induction we may choose $\gamma_0' \leq \gamma_0$ such that $(1-\gamma_0)^\ell(1-\gamma_0')=1-\gamma$. Exactly the same argument as above then gives $X$ is at least \((c_\ell,1-\gamma,k,q_\ell)\)-lovely for $\ell = \left\lceil \frac{\log(1-\gamma)}{\log(1-\gamma_0)}\right\rceil$. The same strategy can be used to achieve the same (or better) constants for any $\gamma' < \gamma$ as well.
\end{proof}
It remains to prove \pref{lem:bootstrap-rhc} and \pref{claim:from-one-constant-to-any-constant}. While \pref{lem:bootstrap-rhc} is the central component, it is instructive to first prove the simpler \pref{claim:from-one-constant-to-any-constant} which is in some sense a `one-step' variant of the former.
\begin{proof}[Proof of \pref{claim:from-one-constant-to-any-constant}]
    Let \(A,B\) be sets of size at least \( \exp(-ck)\). Our goal can be rephrased as showing
    \[
    G^{X}((\gamma_1 + \gamma_2 (1-\gamma_1))k, \prob{A}^{q''q'+q''+q'}) \geq  \prob{A}^{q''q'+q''+q'}
    \]  
    since $1-(\gamma_1 + \gamma_2 (1-\gamma_1))=(1-\gamma_1)(1-\gamma_2)$. We prove this via \pref{obs:links-goodness}. In particular since $c' < c$, \((c',1-\gamma_1,k,q')\)-loveliness of \(X\) implies
    \[ 
    G^{X}(\gamma_1 k, \prob{A}^{q'}) \geq \prob{A}^{q'}.
    \]
    Moreover, for every \(r \in G^{X}(\gamma_1 k, \prob{A}^{q'})\) we have that \(\Pr_{X_{r}}[A],\Pr_{X_{r}}[B] \geq \prob{A}^{q'+1} \geq \exp(-c''(1-\gamma_1)k)\). Therefore by \((c'',1-\gamma_2,(1-\gamma_1)k,q'')\)-loveliness of \(X_{r}\) it holds that
    \[
    G^{X_{t_1}}(\gamma_2(1-\gamma_1) k, \Prob[X_r]{A}^{q''}) \geq \Prob[X_r]{A}^{q''} \geq \prob{A}^{(q'+1)q''}.
    \]
    Combining the two bounds and applying \pref{obs:links-goodness} we then have
    \[
    G^{X}(\gamma_1k + \gamma_2 (1-\gamma_1)k, \prob{A}^{q''q'+q''+q'}) \geq  \prob{A}^{q''q'+q''+q'}
    \]
    as desired.
\end{proof}

Finally, we end with the proof of \pref{lem:bootstrap-rhc}. Here, instead of decomposing \(s=r \cup r'\), we will need to more carefully decompose \(s=r_1 \cup r_2 \cup r_3 \dots\) and control in every step how much \(A\) and \(B\) shrink using the link-up-sampling property.

\begin{proof}[Proof of \pref{lem:bootstrap-rhc}]
    We start with a few simplifying assumptions. First, recall that by \pref{claim:big-sets-not-a-problem} we may assume that \(\prob{A} \leq 0.8\). Second, note it is sufficient to prove the theorem just for $X$. The result then follows for all $\gamma k$-links \(X_t\) since the link of a $\tau$-LUS complex is itself $\tau$-LUS, and the ambient dimension of each link is at least \((1-\gamma)k\) which we take to be still sufficiently large. Finally we introduce one notational simplification. Recall in the definition of a \(\tau\)-LUS that for any $i$-face $s$ the graph \((\SC[X][j][s],\SC[X][k-i][s])\) is a \((\frac{1}{\tau}\exp(-\tau \frac{k-i}{j}),0.1,0.2)\)-multiplicative sampler. In the proof of this lemma we will only consider links of faces \(s\) with size \(i \leq \gamma_0 k\) for some sufficiently small \(\gamma_0\), so these are always \((\frac{1}{\tau}\exp(-\tau(1-\gamma_0) \frac{k}{j}),0.1,0.2)\)-multiplicative samplers. Moreover, we will also always ensure \(j \leq \gamma_0 k\), allowing us to subsume the constant $\frac{1}{\tau}$ into the exponent. In particular, as long as $\gamma_0 \leq \frac{\tau}{3\ln \frac{1}{\tau}}$, we may assume all (examined) links in the proof are \((\exp(-\tau' \frac{k}{j}),0.1,0.2)\)-multiplicative samplers where $\tau' = \frac{\tau}{3}$.    
    
    \medskip
    Fix $\gamma_0 = \min\{\frac{\tau}{3\ln \frac{1}{\tau}},\gamma\}$, let \(c = \gamma_0\frac{\tau'}{3}\), and assume \(\prob{B}\geq \prob{A} \geq \exp(-ck)\). Fix $m=\frac{-\log \prob{A}}{c}$, and define the `step size' $\ell=\frac{k}{m}$. Observe that $1 \leq \ell \leq \gamma_0k$. Assume for the moment that \(\ell\) is also integer (see \pref{rem:integer}). We record the following immediate observations:
\begin{observation}\label{obs:rhc-ind}~
    \begin{enumerate}
        \item \(\prob{A} = \exp(-c m)\).
        \item \(0.8^{m} = \prob{A}^{q_0}\) where \(q_0 = -c^{-1}\log 0.8\) is a constant.
        \item For every \(i \leq \gamma_0 m\): 
        \[
        0.8^i \prob{A} \geq \exp(-\tau' m) = \exp(-\tau' \frac{k}{\ell}).
        \]
    \end{enumerate} 
\end{observation}

The first two items are by definition. The third is satisfied for any $\gamma_0,c$ such that $c(1+\gamma_0 q_0) \leq \tau'$. One can check that taking $\gamma_0 \leq \min\{\gamma,\frac{\tau}{3\ln \frac{1}{\tau}}\}$ suffices. Finally, assume for simplicity $k$ is such that $\gamma_0 k \in \mathbb{N}$.

We now move to the main argument. Let \(j \in \set{1,2,\ldots,j_{fin} = \lceil \gamma_0 m \rceil}\) and for every \(j \leq \lfloor \gamma_0 m \rfloor\) define
\[
    G_j \coloneqq G^X(j \ell,(0.8)^j)
\]
as above. If \(j_{fin} > \gamma_0 m\), additionally define
\[
    G_{j_{fin}} \coloneqq G^X(\gamma_0 k, 0.8^{j_{fin}}) \subseteq G^X(\gamma_0 k, \Prob{A}^{\gamma_0 q_0+1}).
\]
where in the last step we've used the assumption that $\Pr[A] \leq 0.8$. We will show
\begin{equation} \label{eq:inductive-conclusion-on-Gm}
    \prob{G_{j_{fin}}} \geq \left (0.8 \right )^{j_{fin}} \geq \prob{A}^{\gamma_0 q_0 +1}.
\end{equation} 
Setting \(q = \gamma_0 q_0 +1\) proves the lemma.
We prove \eqref{eq:inductive-conclusion-on-Gm} by inductively showing that for every \(j=1,2,\ldots,j_{fin}\), 
\begin{equation*}
    \prob{G_j} \geq \left (0.8 \right )^j.
\end{equation*} 
We start with the base case $j=1$, which already contains the main idea. Given $S \subset \X[k]$, define
\[
    T(S) \coloneqq \left\{ r \in \X[\ell] : \cProb{}{S}{r} < 0.8\prob{S}\right\}
\]
to be the set of `terrible' $\ell$-sets which see too little of $S$. Observe that by construction
\[
    G_1 = \X[\ell] \setminus (T(A) \cup T(B)).
\]
Since $\ell \leq \gamma_0k$, we have (\(\X[k],\X[\ell])\) is a \((\exp(-\tau'\frac{k}{\ell}),0.1,0.2)\)-sampler by our earlier discussion. On the other hand, by \pref{obs:rhc-ind}, we can bound the sizes of $A$ and $B$ as: 
\[
    \Pr[B] \geq \Pr[A] = \exp(-c m) \geq \exp(-\tau'\frac{k}{\ell}).
\]
Thus applying sampling we can bound $\prob{G_1}$ by
\[
    \prob{G_1} \geq 1- \prob{T(A)} - \prob{T(B)} \geq 1-2\cdot 0.1 = 0.8
\]
as desired.

The inductive step is similar to the base case, only performed \textit{locally} within the links of \(X\) together with \pref{obs:links-goodness}.
By induction hypothesis \(\prob{G_i} \geq 0.8^i\). For every \(s \in G_i\), by \pref{obs:rhc-ind}
\[
\Prob[X_s]{B}, \Prob[X_s]{A} \geq (0.8)^i \prob{A} \geq \exp(-\tau' m) \geq \exp(-\tau'\frac{k}{\ell}).
\]
Thus by the same argument above for \(X_s\) and \(A_s=\sett{t \setminus s}{s \subseteq t \in A}, B_s=\sett{t \setminus s}{s \subseteq t \in B}\), we have
\[
\prob{G^{X_s}(\ell, 0.8)~\Bigg|~s \in G_i} \geq 1- \prob{T(A)} - \prob{T(B)} \geq 1-2\cdot 0.1 = 0.8.
\]
The inequality \(\prob{G_{i+1}} \geq 0.8^{i+1}\) then follows by \pref{obs:links-goodness}.

Finally, we note that for the last step, going between \(j_{fin}-1\) to \(j_{fin}\) (in the case that \(j_{fin} > \gamma m\)), we follow the same procedure except that we may need to sample less than \(\ell\) new vertices. This only improves the sampling so the same analysis as above holds.
\end{proof}

\begin{remark}\label{rem:integer}
    Throughout the proof we have `cheated' in the typical way assuming a priori that various parameters (\(\ell\), $\gamma k$, $\gamma_0 k$, and $(1-\gamma_0)^i(1-\gamma')k$) are integer. Note that we did not assume above $\gamma_0m$ was integer because of the careful interplay between $m$ and $\Pr[A]$. In all other cases, however, these assumptions are benign and taking the appropriate integer floor does not substantially change the proof. The only modification is that every `step' in the process (e.g., drawing an $\lfloor \ell \rfloor$-size set) potentially moves one fewer level up the complex. In the worst case, this results in taking a constant multiplicative factor more steps than under integer assumptions, resulting in additional constant factors in $c$ and $q_0$. We omit the details.
\end{remark}

\begin{proof}[Proof of \pref{claim:big-sets-not-a-problem}]
    Denote by \(E(X,Y)\) the set of directed edges \((t,t')\) such that \(t \in X, t' \in Y\).  It is easy to see that \(\prob{E(A,B)}+\prob{E(A,B^c)}=\prob{A}\). Similarly we have that \(\prob{E(A,B^c)}+\prob{E(A^c,B^c)}=\prob{B^c}\) so \(\prob{B^c} \geq \prob{E(A^c,B^c)}\). combining these together gives that \(\prob{E(A,B)} \geq \prob{A} - \prob{B^c}\). Assume that \(\prob{A} \leq \prob{B}\), this implies that \(\prob{E(A,B)} \geq \prob{A} - \prob{A^c}=2\prob{A}-1\).
    Hence it is sufficient to show that given \(A\) such that \(\prob{A} > p\), there exists a constant $q>0$ such that \(2\prob{A}-1 \geq \prob{A}^q\) (which in turn is at least \( \prob{A}^q \prob{B}^q\)). Indeed let us consider \(f(x)=2x-1-x^q\) and show this function is non-negative in the range \([p,1]\) for large enough \(q\).

    Assume without loss of generality that \(q\) is an integer. In this case \(f(x) = (1-x)(\sum_{j=1}^{q-1}x^j - 1)\). This function is greater than \(0\) in \([p,1]\) if and only if \(\sum_{j=1}^{q-1}x^j \geq 1\) in this range. In the range \([p,0.9]\) the series \(\sum_{j=1}^{q-1}x^j\) converges uniformly to \(\frac{x}{1-x}\) which is always greater or equal \(\frac{p}{1-p}>1\). In the range \(x \in [0.9,1]\) we can take any \(q\geq 3\) and one can verify that \(\sum_{j=1}^{q-1}x^j > 0.9+0.9^2 > 1\). 
\end{proof}

\subsection{The General Case}
We now prove \pref{thm:reverse} using our restricted reverse hypercontractive inequality for indicators. This is done via the following reduction from reverse hypercontractivity for sets to the general case.
\begin{theorem}\torestate{\label{thm:generalizing}
    For every \(\ell \geq 1\) and \(\varepsilon > 0\) there exists \(\kappa' \geq \left(\frac{1}{5}\right)^{2\ell(1+\varepsilon)}\left(\frac{1}{18 + \frac{12}{\varepsilon}}\right)^{2(\ell-1)}\) such that the following holds. Let $V$ be a finite probability space and \(D\) a monotone linear operator such that for every \(A,B \subseteq V\),
    \begin{equation} 
        \langle 1_A, D1_B \rangle \geq \kappa \prob{A}^\ell\prob{B}^{\ell}.
    \end{equation}
    Then for every two \emph{arbitrary} functions \(f_1,f_2:V \to \RR_{\geq 0}\) it holds that
    \begin{equation}
    \iprod{f_1, D f_2} \geq \kappa \kappa' \norm{f_1}_{\frac{1}{\ell(1+\varepsilon)}} \norm{f_2}_{\frac{1}{\ell(1+\varepsilon)}} .   
    \end{equation}
    That is, \(D\) is \((\frac{1}{\ell(1+\varepsilon)},\frac{1}{1-\ell(1+\varepsilon)},\kappa \kappa')\)-reverse hypercontractive.}
\end{theorem}
The proof involves careful discretization and thresholding of $f_1$ and $f_2$ into level sets in a way that largely maintains the functions' moments. We defer the details to \pref{app:general-hc} and prove the main Theorem assuming this result.
\begin{proof}[Proof of \pref{thm:reverse}]
    By \pref{thm:generalizing} it is enough that we show that there exist \(\ell>1, \kappa>0\) such that for every \(A,B \subseteq \X[k]\)
    \[
    \langle 1_A,T_\rho 1_B \rangle \geq \kappa \prob{A}^{\ell} \prob{B}^{\ell}.
    \]
    Since $\gamma > \rho$, observe we can always take \(k\) sufficiently large so that \(\Prob[Y \sim Bin(\rho,k)]{Y \leq \lfloor \gamma k \rfloor} \geq \frac{1}{2}\). Let \(q,c\) be the constants promised by \pref{thm:indicator-reverse-hc} with respect to the \(UD_{\lfloor \gamma k \rfloor,k}\). By monotonicity of $q$ and $c$, these apply to all walks of intersection at most \(\lfloor \gamma k \rfloor\)). Thus if both \(\prob{A},\prob{B} \geq \exp(-ck)\) then by \pref{thm:indicator-reverse-hc}:
    \[
    \Prob[t_1,t_2 \sim T_\rho]{t_1 \in A, t_2 \in B} \geq \Prob[Y \sim Bin(\rho,k)]{Y \leq \gamma k} \prob{A}^{q} \prob{B}^{q} \geq \frac{1}{2}\prob{A}^{q} \prob{B}^{q}.
    \]
    Otherwise, it holds that \(\prob{A}\prob{B} \leq \exp(-ck)\). On the other hand, the probability of resampling \emph{all vertices} in the noise operation is \((1-\rho)^k = \exp(-ck)^{\ell''}\) for some constant \(\ell'' > 1\) that depends on \(c\) and on \(\rho\). In particular,
    \[ 
    \Prob[t_1,t_2 \sim T_\rho]{t_1 \in A, t_2 \in B} \geq (1-\rho)^k \prob{A} \prob{B} \geq \prob{A}^{\ell''+1} \prob{B}^{\ell''+1}.
    \]
    Taking \(\ell = \max \set{q,\ell''+1}\) completes the proof.
\end{proof}

\subsection{Proof of \pref{thm:intro-rhc-real}}\label{sec:LUS}
Finally, we briefly give the formal proof for reverse hypercontractivity on high dimensional expanders. This follows from the fact that any \(c\)-locally nice complex is $\tau(c)$-LUS for some appropriate setting of constants.

\begin{corollary} \label{cor:HDX-are-lus}
    For every \(c < 1\) there exists a constant \(\tau > 0\) such that any \(c\)-locally nice complex is \(\tau\)-LUS.
\end{corollary}
\begin{proof}
By definition every $i$-link \(X_s\) is a \(c\)-nice complex, so by \pref{cor:hdx-is-mult-sampler} every $(\X[j][s],\X[k-i][s])$ is a \((\frac{1}{c_1}\exp(-c_2\frac{k-i}{j})),0.1,0.2)\)-sampler for some constants $c_1,c_2$ dependent only on $c$. Taking $\tau=\min\{c_1,c_2\}$ then suffices.
\end{proof}

\begin{proof}[Proof of \pref{thm:intro-rhc-real}]
    The proof is immediate from \pref{cor:HDX-are-lus} and \pref{thm:reverse}.
\end{proof}

\section{Concentration for All Nice HDX} \label{sec:concentration-all-nice-hdx}
In this section we study concentration of measure for partite complexes and complexes whose down-up walk have an \(\Omega(\frac{1}{d})\)-spectral gap, completing the proof of \pref{thm:hdx-is-sampler}. Along the way, we'll show that complexes near the TD barrier satisfy exponential concentration for Lipschitz functions.
\subsection{Concentration from Spectral Gap}\label{sec:herbst}

In this section, we prove concentration inequalities for simplicial complexes under the (typically) weaker quantitative assumption that the down-up walk has a good spectral gap. The bounds so derived are incomparable to the techniques in \pref{sec:chernoff} for two-sided HDX: they hold for a broader family of complexes (and functions), but imply at best exponential (rather than subgaussian) concentration at the top level. We derive subgaussian bounds in this setting by moving to an appropriate skeleton of $X$.

Before stating our results, we briefly recall some notation from \pref{sec:prelims-concentration}. Given a function $f: \X[k] \to \R$, let $Z=f(x_1,\ldots,x_k)$ and $Z'_{(i)}=f(x_1,\ldots z_i, \ldots, x_k)$ where $z_i$ is sampled conditional on $x_{-i}$. We call $f$ $\nu$-Lipschitz if with probability $1$ over $Z,Z'$
\[
\sum\limits_{i=1}^k (Z-Z'_{(i)})_+^2 \leq \nu
\]
and say it has $\nu$-bounded difference if for all adjacent $(s,s')$ in the down-up walk:
\[
(f(s)-f(s'))^2 \leq \frac{\nu}{k}
\]
We prove that $\nu$-Lipschitz functions on HDX satisfy exponential concentration of measure.
\begin{theorem}\label{thm:poincare-to-concentration}
    Let $X$ be a $d$-\maximal simplicial complex such that $\lambda_2(U_{d-1}D_d) \geq 1-\frac{1}{Cd}$ for some $C>0$. Then for any $0 \leq k \leq d$ and $\nu$-Lipschitz function $f:\X[k] \to \R$:
    \begin{enumerate}
        \item Upper Tail: $\Pr[f-\mathbb{E}[f] \geq t] \leq 2e^{-t/\sqrt{(C+1)\nu}}$
        \item Lower Tail: $\Pr[f-\mathbb{E}[f] \leq -t] \leq 2e^{-t/\sqrt{(C+1)\nu}}$
    \end{enumerate}
\end{theorem}
Note that the spectral gap case of \pref{thm:hdx-is-sampler} follows as an almost immediate corollary.
\begin{proof}[Proof of \pref{thm:hdx-is-sampler}, spectral gap case]

Note the `moreover' part follows from \pref{lem:raising}. Let $Y$ be the promised $d$-\maximal complex such that $X$ is an (at most) $\sqrt{\frac{d}{2}}$-\maximal skeleton of $Y$ such that $Y$'s down-up walk has spectral gap $\lambda(U_{d-1}D_d) \leq 1-\frac{c}{d}$. For any $g: \X[i] \to [0,1]$ \pref{thm:poincare-to-concentration} implies $U_{i,d}g$ has the following exponential tail:
\[
\Pr[U_{i,d}g(s) - \mathbb{E}[g] \geq \varepsilon] \leq 2e^{-\Omega_c(\varepsilon\frac{\sqrt{d}}{i})} \leq 2e^{-\Omega_c(\varepsilon\frac{k}{i})}
\]
since $U_{i,d}g$ has $\frac{i^2}{d}$-bounded difference. By \pref{lem:lowering}, $U_{i,k}g$ therefore has tail:
\[
\Pr[U_{i,k}g(s) - \mathbb{E}[g] \geq \varepsilon] \leq 2e^{-\Omega_c(\varepsilon\frac{k}{i})}+\frac{2}{\varepsilon}e^{-\Omega(\varepsilon^2\frac{k}{i})},
\]
where the latter term is from the following concentration bound on the complete complex:
\begin{claim}
    For any $i \leq k \leq n$ and $g: \SC[\Delta][i][n] \to [0,1]$:
    \[
    \Pr_{\Delta_n(k)}[U_{i,k}g - \mathbb{E}[g] > \varepsilon] \leq \frac{2}{\varepsilon}e^{-\Omega(\varepsilon^2 \frac{k}{i})}
    \]
    and likewise for the lower tail.
\end{claim}
We prove this claim in \pref{app:complete} via a simple application of \pref{lem:weird-distribution-to-uniform}. The lower tail follows similarly.
\end{proof}

We also get the following corollaries for TD and SI complexes potentially of independent interest.
\begin{corollary}[Concentration up to the TD Barrier]\label{cor:TD-exponential}
    Let $X$ be a $d$-\maximal $\lambda$-TD complex for $\lambda < 1$. Then for any $0 \leq k \leq d$ and $\nu$-Lipschitz function $f:\X[k] \to \R$
    \begin{enumerate}
        \item \textbf{Upper tail:} $\Pr[f - \mathbb{E}[f] \geq t ] \leq 2e^{-\frac{t}{\sqrt{c_\lambda\nu}}}$
        \item \textbf{Lower tail:} $\Pr[f - \mathbb{E}[f] \leq -t ] \leq2e^{-\frac{t}{\sqrt{c_\lambda\nu}}}$
    \end{enumerate}
    where $c_\lambda \leq 1+e^{\frac{\lambda}{(1-\lambda)}}$.
\end{corollary}

\begin{corollary}[Concentration under Spectral Independence]\label{cor:SI-exponential}
    Fix $\eta>0$ and let $X$ be a $d$-\maximal $\eta$-spectrally independent complex. Then for any $0 \leq k \leq d$, $g:\X[1] \to [0,1]$, and $f=U_{1,k}g$:
    \begin{enumerate}
        \item \textbf{Upper tail:} $\Pr[f - \mathbb{E}[f] \geq \varepsilon] \leq 4e^{-c_\eta\varepsilon\sqrt{k}}$
        \item \textbf{Lower tail:} $\Pr[f - \mathbb{E}[f] \leq -\varepsilon] \leq 4e^{c_\eta\varepsilon\sqrt{k}}$
    \end{enumerate}
    where $c_\eta \leq \frac{1}{12\max\{1,\sqrt{\eta}\}}$.
\end{corollary}
We note it is possible to extend the latter to lifted Lipschitz functions from an appropriate skeleton.
\subsubsection{The Herbst Argument}
The proof of \pref{thm:poincare-to-concentration} is via a variant of the \textit{Herbst Argument}, a classical strategy for proving concentration of random variables based on applying functional inequalities like MLSI and LSI to the moment generating function. While sparse simplicial complexes cannot have bounded (modified) log-sobolev constants, a somewhat lesser known variant of the Herbst argument using spectral gap was developed by Aida and Stroock \cite{aida1994moment}, Bobkov and Ledoux \cite{bobkov1997poincare}, and Boucheron, Lugosi, and Massart \cite{Boucheron}. We will give an elementary adaption of their method to HDX achieving the above bounds.

The key to \pref{thm:poincare-to-concentration} is really the following approximate variant of the Efron-Stein inequality, sometimes called `approximate tensorization of variance'. We remark that similar inequalities for local-spectral expanders appear in the literature \cite{kaufman2020local,chen2021rapid}.
\begin{lemma}[Approximate Efron-Stein Inequality]
    Let $X$ be a $d$-\maximal simplicial complex such that $\lambda_2(U_{d-1}D_d) \leq 1-\frac{1}{Cd}$ for some $C>0$. Then for any $0 \leq k \leq d$, $f:\X[k]\to \R$, and $Z=f(X)$ we have:
    \[
    Var(f) \leq (C+1)\sum\limits_{i=1}^k\mathbb{E}\left[(Z-Z'_{(i)})_+^2\right]
    \]
\end{lemma}
\begin{proof}
    It is sufficient to instead prove that any $k$-\maximal complex whose down-up walk $U_{k-1}D_{k}$ has gap $\frac{1}{C'k}$ satisfies
    \[
    Var(f) \leq C'\sum\limits_{i=1}^{k}\mathbb{E}\left[(Z-Z'_{(i)})_+^2\right].
    \]
    The result then follows from \cite[Theorem 3.5]{lee2023parallelising}, who shows that $\lambda_2(U_{d-1}D_d) \leq 1-\frac{1}{Cd}$ implies $\lambda_2(U_{k-1}D_k) \leq 1-\frac{1}{(C+1)k}$ for any $0 \leq k \leq d$.

    With this in mind, let $E_i$ be the $i$th averaging operator:
    \[
    E_if(s) = \underset{s_i \sim X_{s_{-i}}}{\mathbb{E}}[f(s_{-i} \cup s_i)],
    \]
    where we recall $s_{-i}$ is $s$ without the $i$th vertex in its ordering. Define the $i$th Laplacian $L_i=I-E_i$, and the total laplacian operator as $L=I-UD=\frac{1}{k+1}\sum\limits L_i$. By assumption on the spectral gap of $UD$, we have:
    \[
    Var(f) \leq Ck \langle f, Lf \rangle
    \]
    On the other hand expanding out the Laplacian gives:
    \begin{align*}
    Ck\langle f, Lf \rangle &= C\sum\limits_{i=1}^k \langle f, L_if \rangle\\
    &=C\sum\limits_{i=1}^k \langle L_if, L_if \rangle\\
    &=C\sum\limits_{i=1}^k \underset{{s \sim \X[k]}}{\mathbb{E}}[Var^{(i)}_{s}(f)]
    \end{align*}
    where $Var^{(i)}_{s}(f)$ is the variance of $f$ over the link $X_{s_{-i}}$. Finally examining this expected variance, we have:
    \begin{align*}
    \underset{{s \sim \X[k]}}{\mathbb{E}}[Var^{(i)}_{s}(f)]
    &=\frac{1}{2}\mathbb{E}[(Z-Z'_{(i)})^2] =\mathbb{E}[(Z-Z'_{(i)})_+^2]
    \end{align*}
    where both equalities follow from the fact that $Z$ and $Z'_{(i)}$ are i.i.d conditioned on $s_{-i}$. Combining with the above completes the proof.
\end{proof}
We can now prove \pref{thm:poincare-to-concentration} closely following the approach of \cite{Boucheron}.
\begin{proof}[Proof of \pref{thm:poincare-to-concentration}]
    We prove only the upper tail. The lower tail follows similarly. It is sufficient to prove the following exponential integrability of the moment generating function:
    \begin{equation}\label{eq:exp-int}
    \mathbb{E}\left[e^{ \frac{1}{\sqrt{C\nu}}(f-\mathbb{E}[f])}\right] \leq 2,
    \end{equation}
    since we then have:
    \begin{align*}
        \Pr[f-\mathbb{E}[f] > t] &= \Pr\left[\exp\left(\frac{1}{\sqrt{C\nu}}(f-\mathbb{E}[f])\right) > \exp\left(\frac{1}{\sqrt{C\nu}}t\right)\right]\\
        &\leq 2\exp\left(-\frac{1}{\sqrt{C\nu}}t\right)
    \end{align*}
    by Markov.
    
    For notational simplicity, let $g=f-\mathbb{E}[f]$ and let $Z=g(X)$. Following \cite{Boucheron}, we use Efron-Stein to set up a recurrence relating $e^{\lambda Z}$ to $e^{\lambda Z/2}$ for small enough $\lambda>0$. Applying Efron-Stein to the latter gives:
    \begin{align*}
        \mathbb{E}[e^{\lambda Z}]-\mathbb{E}[e^{\lambda Z/2}]^2 &\leq C\sum\limits_{i=1}^k\mathbb{E}\left[(e^{\lambda Z/2}-e^{\lambda Z'_{(i)}/2})_+^2\right]\\
        &\leq C\sum\limits_{i=1}^k\mathbb{E}\left[e^{-\lambda Z}(1-e^{\lambda (Z'_{(i)}-Z)/2})_+^2\right]\\
        &\leq \frac{C\lambda^2}{4}\mathbb{E}\left[e^{-\lambda Z}\sum\limits_{i=1}^k(Z-Z_{(i)}')^2\right]\\
        &\leq \frac{C\lambda^2\nu}{4}\mathbb{E}\left[e^{-\lambda Z}\right]
    \end{align*}
    since $f$ (and therefore $g$) are $\nu$-bounded. Re-arranging we have the recurrence:
    \[
    \mathbb{E}[e^{\lambda Z}] \leq \frac{1}{1-\frac{\lambda^2C\nu}{4}}\mathbb{E}[e^{\lambda Z/2}]^2.
    \]
    It is then an elementary exercise to solve when $\lambda \leq \frac{1}{\sqrt{C\nu}}$ and derive \pref{eq:exp-int} (see \cite[Section 3.6]{Boucheron} for the full derivation).
\end{proof}
We now prove the corollaries for SI/TD-complexes
\begin{proof}[Proof of \pref{cor:TD-exponential}]
    By \pref{thm:td}, $X$ is a $\frac{\lambda}{(d-1)(1-\lambda)}$-one-sided local-spectral expander, so \pref{thm:gap-of-down-operator-single-step} gives $\lambda_2(U_{d-1}D_d) \leq 1-\frac{1}{d}e^{\frac{\lambda}{1-\lambda}}$. Plugging this into \pref{thm:poincare-to-concentration} gives the result.
\end{proof}
In the spectral independence regime, one must be slightly more careful since an $\eta$-spectrally independent system has $C \approx k^\eta$. We recover concentration by looking only at lifted functions and analyzing concentration on some lower skeleton.
\begin{proof}[Proof of \pref{cor:SI-exponential}]
    We prove the upper tail. The lower tail follows similarly. We break the proof into two cases. First, assume that $\eta \leq 2$ and let $k'=\lfloor \frac{d}{3} \rfloor$. The $k'$-skeleton of $X$ is a one-sided $\frac{1}{k'}$-local-spectral expander. Thus by \pref{thm:gap-of-down-operator-single-step} $\lambda(U_{k'-1}D_{k'}) \leq 1-\frac{1}{ek'}$ and \pref{thm:poincare-to-concentration} implies every $k \leq k'$ has concentration:
    \[
    \Pr[f - \mathbb{E}[f] > \varepsilon] \leq 2e^{-\frac{1}{\sqrt{e+1}}\varepsilon\sqrt{k}}.
    \]
    Further, by \pref{lem:raising} every $k > k'$ has concentration:
    \begin{align*}
    \Pr[f - \mathbb{E}[f] > \varepsilon] &\leq 2e^{-\frac{1}{2\sqrt{e+1}}\varepsilon\sqrt{k'}}(1-e^{-\frac{\varepsilon^2}{12}k})\\
    &\leq 4e^{-\frac{1}{2\sqrt{6(e+1)}}\varepsilon\sqrt{d}}\\
    &\leq 4e^{-\frac{1}{12}\varepsilon\sqrt{d}}
    \end{align*}
    assuming $k \geq 6$ and $\varepsilon^2 \geq \frac{12}{k}$ (else the bound is trivial).
    
    Otherwise, assume $\eta > 2$ and let $k'=\lfloor \frac{d}{\eta} \rfloor$. The $k'$-skeleton of $X$ is a $\frac{\eta}{(1-\frac{1}{\eta})d}$-one-sided local-spectral expander. Then by \pref{thm:gap-of-down-operator-single-step} $\lambda(U_{k'-1}D_{k'}) \leq 1-\frac{1}{e^2k'}$ and \pref{thm:poincare-to-concentration} implies for every $k \leq k'$:
    \[
    \Pr[f - \mathbb{E}[f] > \varepsilon] \leq 2e^{-\frac{1}{\sqrt{e^2+1}}\varepsilon\sqrt{k}}.
    \]
    Further, by \pref{lem:raising} every $k > k'$ has concentration:
    \begin{align*}
    \Pr[f - \mathbb{E}[f] > \varepsilon] &\leq 2e^{-\frac{1}{2\sqrt{e^2+1}}\varepsilon\sqrt{k'}}(1-e^{-\frac{\varepsilon^2}{12}k})\\
    &\leq 4e^{-\frac{1}{2\sqrt{4\eta(e^2+1)}}\varepsilon\sqrt{d}}\\
    &\leq 4e^{-\frac{1}{12\sqrt{\eta}}\varepsilon\sqrt{d}}
    \end{align*}
    assuming $\eta \leq \frac{d}{2}$ and $\varepsilon^2 \geq \frac{12}{k+1}$ (again, otherwise the stated bound is trivial).
    
\end{proof}

Finally before moving to reverse hypercontractivity, we briefly remark one a additional corollary regarding subgaussian concentration for more general function classes:

\begin{corollary}[Subgaussian concentration for skeletons]\label{cor:skeleton-subgauss-general}
    Let $X$ be a $d$-\maximal simplicial complex with $\lambda_2(U_{d-1}D_d) \leq 1-\frac{1}{Cd}$. Then for any $k \leq d$, $\nu>0$, and function $f:\X[k] \to \R$ with $\nu$-bounded difference:
    \begin{enumerate}
        \item \textbf{Upper Tail}: $\Pr[f-\mathbb{E}[f] \geq t] \leq 2e^{-\frac{t}{\sqrt{4(C+1)\nu}}\cdot\sqrt{\frac{d}{k}}} + e^{-\frac{t^2}{4\nu}}$
        \item \textbf{Lower Tail:} $\Pr[f-\mathbb{E}[f] \leq -t] \leq 2e^{-\frac{t}{\sqrt{4(C+1)\nu}}\cdot\sqrt{\frac{d}{k}}} + e^{-\frac{t^2}{4\nu}}$
    \end{enumerate}
\end{corollary}
\begin{proof}
The proof is immediate from \pref{thm:poincare-to-concentration} and \pref{lem:lowering}, setting $f_{up}(t,\nu)=f_{low}(t,\nu)= 2e^{\frac{-t}{\sqrt{(C+1)\nu}}}$.
\end{proof}
We remark it is likely possible to generalize the above strategy from bounded difference to Lipschitz functions on sufficiently strong local-spectral expanders by more carefully analyzing the eigenspaces of the Laplacian and higher order random walks associated with the $d$-lift (e.g.\ as in \cite{KaufmanO2020,DiksteinDFH2018,gaitonde2022eigenstripping}).
\subsection{Concentration for Partite HDX}\label{sec:conc-partite}

The final component of \pref{thm:hdx-is-sampler} is the partite case. We re-state the partite version of the theorem here for convenience.
\begin{theorem}[Sampling on HDX]\label{thm:hdx-is-sampler-2-cases}
    Let $X$ be a \(d\)-\maximal partite $2^{-cd}$-HDX. Then for any $i < k \leq d$ the containment graph \((\X[k],\X[i])\) is a \((\varepsilon, \beta)\) sampler for 
    \[
    \beta = \frac{64}{\varepsilon}\exp\left(-\Omega_c\left(\varepsilon^8 \frac{k}{i}\right)\right).
    \]
\end{theorem}
Note that in the special case $i=1$ (i.e. the Chernoff-Hoeffding setting), the $\varepsilon$-dependence can be improved to $\varepsilon^{-2}$ in the exponent and removed entirely from the leading coefficient. Similarly, the exponent can be improved to $\varepsilon^{-2}$ dependence by \pref{lem:lowering} whenever $k \leq \varepsilon^6 d$.

Similar to our approach in the two-sided case, we will split our analysis into a Chernoff-Hoeffding bound, and a bootstrapping method to lift Chernoff to full inclusion sampling. Both steps require significant modification from the two-sided case to handle inherent dependencies between coordinates only occuring in the partite case.
\subsubsection{Chernoff-Hoeffding}

Unfortunately, partite HDX are not `splittable' in the sense of two-sided HDX. Thus to prove a Chernoff-Hoeffding bound, we will instead rely on a more classical \textit{localization} approach, leveraging the local sampling properties between bipartite components of the complex. The resulting quantitative bounds via this technique are again incomparable with our prior methods. While they are weaker for large $\lambda$, they approach the error of true independent sampling for small enough $\lambda$ and remove the restrictive condition of two-sided expansion. The latter in particular actually leads to sparser complexes with subgaussian concentration due to better known degree bounds for partite constructions. 

Before stating the result formally, we briefly introduce some relevant notation. Given a complex $X$, let \(X^{ind}\) denote the complete complex whose vertices are \(\X[1]\), that is:
\[
X^{ind}(k) = \{\{v_1,v_2,\ldots,v_k\}: v_i \in \X[1]\}
\]
where each $v_i$ is drawn independently from $\X[1]$ \textit{with repetition}.\footnote{We remark we have abused notation somewhat as this is not formally a simplicial complex and the faces are \textit{multi}-sets, but this is irrelevant for our purposes.} Similarly for $X$ partite, let \(X^{ind,p}\) be the complete \(k\)-partite complex whose sides are \(X[1],X[2],\dots,X[k]\):
\[
X^{ind,p}(k)=\{\{v_1,v_2,\ldots,v_k\}: v_1 \in X[1],v_2 \in X[2],\dots, v_k \in X[k]\}
\]
where each $v_i$ is drawn independently from the marginal over $X[i]$.

Finally, let $\sigma(X,k,\varepsilon)$ denote the worst case `\(k\) vs \(1\)' \(\varepsilon\)-sampling error of \(X\).
\[
    \sigma(X,k,\varepsilon) = \max \sett{\Prob[s \in {\X[k]}]{\abs{\cProb{v \in \X[1]}{A}{v \in s} - \Prob[v \in {\X[1]}]{A}} > \varepsilon} }{A \subseteq \X[1]}
\]

We prove that sampling in two-sided and one-sided partite HDX match the bounds in their respective independent complex up to \(\poly(k)\lambda\)-error.

\begin{theorem} \label{thm:sampling-like-complete}
    Let \(\eta > 0, \varepsilon \in (0,\frac{1}{2}) \), and \(\lambda \leq \frac{\eta^{2.5}}{16k^4}\). Then
    \begin{enumerate}
        \item If \(X\) is $k$-\maximal \(\lambda\)-two-sided local-spectral expander:
        \[
        \sigma(X,k,\varepsilon) \leq \sigma(X^{ind},k,\varepsilon) + \eta
        \]
        \item If \(X\) is a \(k\)-partite \(\lambda\)-one-sided local-spectral expander:
        \[
        \sigma(X,k,\varepsilon) \leq \sigma(X^{ind,p},k,\varepsilon) + \eta
        \]
    \end{enumerate}
    In particular, in both cases \(X\) is then a \((\varepsilon,\beta)\)-sampler for \(\beta = \eta + \exp(-\Omega(\varepsilon^2 k))\) by Chernoff-Hoeffding.
\end{theorem}
In the two-sided case, note that \pref{thm:split-concentrate} is tighter when \(\lambda > \exp(-\Omega(k))\), while for small enough $\lambda$ \pref{thm:split-concentrate} wins out achieving parameters arbitrarily close to independent sampling. We note that here we only give a bound on sampling sets, but this can be converted to concentration for lifts of bounded functions with essentially no loss by \pref{claim:sampler-for-functions}.

The core of \pref{thm:sampling-like-complete} is a lower bound on the $k$-wise correlation across coordinates of a $\lambda$-product.
\begin{proposition} \label{prop:hitting}
    Let \(X\) be a $k$-\maximal $\lambda$-product. For every \(\eta < 1\) and functions \(f_i:X[i] \to [0,1]\):
    \[
    \Ex[\set{v_1,v_2,\dots,v_k} \in {\X[k]}]{\prod_{i=1}^k f_i(v_i)} \geq (1-\eta)^{k} \prod_{i=1}^k \mathbb{E}[f_i]
    \]
    whenever  \(\lambda \leq \frac{\eta^{1.5}\min_i\{\mathbb{E}[f_i]^2\}}{\sqrt{8k^3}}\).
\end{proposition}
\begin{proof}
    The proof relies on the following adaptation of an elementary connection between bipartite expansion and sampling observed in \cite{DinurK2017}:
    \begin{claim}\torestate{\label{claim:chebyshev-expanders}
        Let \(G=(L,R,E)\) be a \(\lambda\)-bipartite expander and let \(A\) be its bipartite adjacency operator. Let \(f:L \to [0,1]\) be a function with \(\ex{f} \geq \mu\) and let \(0 < \varepsilon < \mu\). Then
        \[
        \Pr_{u \in R}[Af(u) < \mu - \varepsilon] < \frac{\lambda^2 \mu }{\varepsilon^2}.
        \]}
    \end{claim}
    We prove this in \pref{app:sampler-proofs} for completeness. Using this fact, we prove the following more general claim. Write $\mu_i=\mathbb{E}[f_i]$ and fix \(\varepsilon_1,\varepsilon_2,\dots,\varepsilon_k >0\) such that \(\mu_i - (i-1)\varepsilon_i > 0\). For \(\varepsilon = \min_{i}\{\varepsilon_i\}\) we show
    \begin{align} \label{eq:link-product-bound-less-general}
            \Ex[\set{v_1,v_2,\dots,v_k} \in {\X[k]}]{\prod_{i=1}^k f_i(v_i)} \geq \prod_{i=1}^{k} \left( \mu_i -(i-1)\varepsilon_i - \lambda^2 \frac{k-i}{\varepsilon^2 \cdot \underset{j>i}{\min} (\mu_j - (i-1)\varepsilon_j)} \right).
    \end{align}
    Setting \(\varepsilon_i = \frac{\mu_i \eta}{2i}\) and applying our assumed bound on \(\lambda\) then gives the desired result.
    
    We prove \eqref{eq:link-product-bound-less-general} by induction on \(k\), the \maxsizity of the complex.
    For \(k=1\) the bound is trivial since \(\ex{f_1} = \mu_1\). Assume the statement holds for any \((k-1)\)-\maximal $\lambda$-product, and for each \(i=2,\dots,k\) let $T_i$ denote the vertices in $X[1]$ that under-sample $f_i$:
        \[
        T_i = \sett{v_1 \in X[1]}{\Ex[{v_i \in X_{v_1}[i]}]{f(v_i)} < \mu_i - \varepsilon_i}.
        \]
    Since each $(X[1],X[i])$ is a $\lambda$-bipartite expander, \pref{claim:chebyshev-expanders} implies
    \[
    \prob{\bigcup_{i=2}^k T_i} \leq \sum_{i=2}^k\frac{\lambda^2}{\mu_i \varepsilon_i^2} \leq \lambda^2 \frac{k-1}{\varepsilon^2\cdot \min_{j>1}\{\mu_j\}}.
    \]
    We de-correlate our expected product along the first coordinate by conditioning on \(G = X[1] \setminus \{\bigcup_{i=2}^k T_i\}\):
    \begin{align*} 
        \Ex{\prod_{i=1}^{k}f_i(v_i)} &\geq \Ex{\mathbf{1}_G(v_{0})\prod_{i=1}^{k}f_i(v_i)} \\
        &\geq \ex{\mathbf{1}_G(v_{0})f_{0}(v_{0})} \cdot \min_{v_{1} \in G} \left (\Ex[\set{v_2,v_3,\dots,v_k} \in X_{v_{1}}]{\prod_{i=2}^k f_i(v_i)} \right ).
    \end{align*}
    We bound the two terms separately. For the first:
    \begin{align}
        \ex{\mathbf{1}_G(v_{0})f_{0}(v_{0})} &\geq \ex{f_{0}(v_{0})} - \prob{\bigcup_{i=2}^k T_i}\\
        &\geq \mu_0 - \lambda^2 \frac{k-1}{\varepsilon^2\cdot \min_{j>1}\{\mu_j\}}.\label{eq:decor}
    \end{align}

To bound the minimum term, we apply the inductive hypothesis on every link \(X_{v_1}\). Toward this end, for notational simplicity re-index the sides of \(X_{v_1}\) to be from \(1\) to \(k-1\) and set:
\begin{enumerate}
    \item \(f'_i=f_{i+1}\),
    \item \(\varepsilon_i'=\varepsilon_{i+1}\),
    \item \(\mu_i' =\mu_{i+1} - \varepsilon_{i+1}\) and \(\mu' = \min_{j>1} \mu_i'\).
\end{enumerate}
Since \(\Ex[{X_{v_1}[i]}]{f'_i} \geq \mu_i'\) by assumption on the links, applying the inductive hypothesis to $X_{v_1}$ and $\{f'\}$ bounds the minimum by
\[
\prod_{i=1}^{k-1} \left( \mu_i' -(i-1)\varepsilon_i' - \lambda^2 \frac{k-i}{\varepsilon^2 \cdot \underset{j>i}{\min} (\mu_j' - (i-1)\varepsilon_j')} \right)
\]
Combining this with the prior term (corresponding to $i=0$) gives the desired bound
\[
\prod_{i=1}^k \left( \mu_i -(i-1)\varepsilon_i - \lambda^2 \frac{k-i}{\varepsilon^2 \cdot \underset{j>i}{\min} (\mu_j - (i-1)\varepsilon_j)} \right).
\]
\end{proof}

We note a more involved variant of the above applying different bounds in each level of induction implies the following finer-grained bound: 

\begin{claim}\label{claim:hitting-technical}
    Let \(X\) be a \(k\)-partite \(\lambda\)-product. Let \(f_i:X[i] \to [0,1]\) be such that  \(\Ex[v_i \in X{[i]}]{f(v_i)} \geq \mu_i\). Let \(\set{\varepsilon_{i,\ell} > 0}_{1\leq i < \ell \leq k}\) and let \(s_{i,\ell}^{(m)} = \sum_{j=\ell}^{m} \varepsilon_{j,i}\) (and for brevity \(s_{m+1,\ell}^{(m)}=0\)). Then
    \begin{equation}
        \Ex[\set{v_0,v_1,\dots,v_k} \in {\X[k]}]{\prod_{i=1}^k f_i(v_i)} \geq \prod_{i=1}^k \left( \mu_i - s_{i,i}^{(k-1)} - \lambda^2 \sum_{j=0}^{i-1} \frac{\mu_j - s_{i,j}^{(k-1)}}{\varepsilon_{i,j}^2} \right).
    \end{equation}
\end{claim}
We omit the proof which is an unenlightening technical extension of the above. 

\medskip

Using our lower bound on negative correlation, we show that for any function tuple $f=(f_1,f_2,\ldots,f_k)$, the random variables $f(X)$ and $f(X^{ind})$ are actually \textit{distributionally} close. This will allow us to immediately recover any concentration bound on $X^{ind}$ up to the distributional error. More formally, given a \(k\)-partite simplicial complex \(X\), let \(D_{ind}\) denote the distribution supported by of \(X^{ind,p}\), i.e. \(\Prob[D_{ind}]{\set{v_1,v_2,\dots,v_k}} = \prod_{i=1}^k \Prob[X{[i]}]{v_i}\). Moreover, for any distribution \(D\) over \(\Omega = X[1] \times X[2] \times \dots \times X[k]\), and a function \(\bar{f}:\Omega \to \mathbb{R}^{k}\) define the \textit{push forward} distribution \(\bar{f}_*D:\mathbb{R}^{k} \to [0,1]\) as
\[
\bar{f}_* D(\bar{x}) = \Prob[\omega \sim D]{\bar{f}(\omega) = \bar{x}}.
\]
Let $D_X=\pi_k$ denote the distribution over $k$-sets of $X$. We show the pushforward distributions $\bar{f}_* D_X$ and $\bar{f}_* D_{ind}$ are close in TV-distance.
\begin{lemma} \label{lem:total-variation}
    Let \(\varepsilon,\eta > 0\) and let \(X\) be a \(k\)-partite \(\lambda\)-product for \(\lambda \leq \frac{\eta^{2.5}}{8 k^{4}}\). Then for any \(f_i:X[i] \to \set{0,1}\) and \(\bar{f}=(f_1,f_2,\dots,f_k)\):
    \[
    d_{TV}(\bar{f}_* D_X, \bar{f}_* D_{ind}) \leq \eta.
    \]
\end{lemma}
\begin{proof}
    Note that both distributions are supported in \(\set{0,1}^{k}\) by our assumption on the domain of $\bar{f}$. Thus we need to show for every \(B \subseteq \set{0,1}^{k}\):
    \[
    \abs{\Prob[\bar{f}_* D_X]{B} - \Prob[\bar{f}_* D_{ind}]{B}} \leq \eta.
    \]
    Indeed since $\Prob[\bar{f}_* D_X]{B} - \Prob[\bar{f}_* D_{ind}]{B} = \Prob[\bar{f}_* D_{ind}]{\bar{B}} - \Prob[\bar{f}_* D_X]{\bar{B}}$ for $\bar{B}=\{0,1\}^{k}\setminus B$, it is enough to just show the one-sided bound 
    \[
    \Prob[\bar{f}_* D_X]{B} \geq \Prob[\bar{f}_* D_{ind}]{B} - \eta.
    \]
    Fix any event \(B \subseteq \set{0,1}^{k}\). We'll decompose $B$ into two parts. First, we'll look at the subset $B_1$ on which $f$ takes many `unlikely' values, and argue the difference is small just because $D_X$ and $D_{ind}$ have the same marginals over each coordinate of $\{0,1\}^{k}$. In particular, let \(I = \sett{i \in [k]}{\ex{f_i} \leq \frac{\eta}{2k}}\) and \(J = \sett{i \in [k]}{\ex{f_i} \geq 1-\frac{\eta}{2k}}\) be the sets of coordinates with extreme expectations and let
\[
B_1 = B \cap \left \{ x~\Bigg | \left(\bigvee_{i \in I} \set{f_i(x) = 1}\right) \lor \left(\bigvee_{j \in J} \set{f_j(x) = 0}\right) \right \}.
\]
Set \(B_2 = B \setminus B_1\).
    
    We first show that \(\Prob[\bar{f}_* D_X]{B_1} - \Prob[\bar{f}_* D_{ind}]{B_1} \leq \frac{\eta}{2}\). Since both distributions have the same marginals, we can write for either $D_X$ or $D_{ind}$ by a union bound:
    \[
    \prob{B_1} \leq \sum_{i\in I} \prob{f_i=1} + \sum_{j \in J}\prob{f_j = 0} \leq (|I|+|J|)\frac{\eta}{2k} \leq \frac{\eta}{2}.
    \]
    In particular the difference is no more than \(\frac{\eta}{2}\).

    Moving on to \(B_2\), we decompose the probability as a sum over $k$-wise products:
    \[
    \Prob[\bar{f}_* D_X]{B_2} = \sum_{\bar{b} \in B_2} \Prob[D_X]{\bigwedge_{i=1}^k f_i=b_i} = \sum_{\bar{b} \in B_2} \Ex[\set{v_1,v_2,\dots,v_k} \sim D_X]{\prod_{i=1}^k \abs{b_i - f_i(v_i)}}.
    \]
    We note that the functions \(\abs{b_i - f_i(v_i)}\) have values in \([0,1]\) and by assumption have expectation \(\geq \frac{\eta}{2}\) (since otherwise \(\bar{b} \in B_1\)). Then by \pref{prop:hitting} setting \(\varepsilon = \frac{\eta}{2}\) we have
    \begin{align*}
    \Pr_{D_X}[B_2] &\geq (1-\frac{\eta}{2k})^{k} \sum_{\bar{b} \in B_2} \prod_{i=1}^k \Ex[v_i \in X{[i]}]{\abs{b_i - f_i(v_i)}}\\
    &=(1-\frac{\eta}{2k})^{k} \sum_{\bar{b} \in B_2} \prod_{i=1}^k \Prob[D_{ind}]{\mathbf{1}_{f_i=b_i}}\\
    &=(1-\frac{\eta}{2k})^{k} \sum_{\bar{b} \in B_2} \prod_{i=1}^k \Prob[D_{ind}]{\mathbf{1}_{f_i=b_i}}\\
    &\geq \Prob[D_{ind}]{B_2} - \frac{\eta}{2}
    \end{align*}
    where we've used the marginal equivalence of $D_X$ and $D_{ind}$ and independence of $D_{ind}$ in the second and third steps respectively. Combining our bounds on $B_1$ and $B_2$ gives the result.
\end{proof}

We can now easily prove \pref{thm:sampling-like-complete}.
\begin{proof}[Proof of \pref{thm:sampling-like-complete}] 
    Let \(f:\X[1] \to [0,1]\). For the partite case, set \(f_i = f|_{{X[i]}}\). Consider the set of `bad' $k$-tuples that mis-sample $f$:
    \[
    B = \Set{\abs{\sum_{j=1}^k f(v_j) - \mu} > \varepsilon}.
    \]
    By \pref{lem:total-variation}, the TV-distance between the distribution over values of $\bar{f}(x)=f(x_1,x_2,\ldots,x_{k})$ over $X$ and $X^{ind,p}$ is at most $\eta$ so we have:
    \begin{align*}
    \sigma(X,k,\varepsilon) &= \max_{\bar{f}} \Pr_{\bar{f}_* D_X}[B]\\
    &\leq \max_{\bar{f}} \Pr_{\bar{f}_* D_{ind}}[B] + \eta\\
    &\leq S(X^{ind,p},k,\varepsilon) + \eta
    \end{align*}
    as desired. The two-sided case follows the same argument after taking the partitification \(P(X)\) of \(X\) and setting \(f_i(v,i)=f(v)\).
\end{proof}

\subsubsection{Partite Inclusion Sampling}\label{sec:inclusion-sampler-proof}
Broadly speaking, the proof of \pref{thm:hdx-is-sampler-2-cases} (that is the partite case of \pref{thm:hdx-is-sampler}), follows the same strategy as the two-sided case. We'd like to lift Chernoff-Hoeffding to sampling of the inclusion graph via a Chernoff-Hoeffding style bound for the faces complex. Unfortunately, unlike the two-sided case, the faces complex of a partite HDX is not necessarily an HDX (indeed its links may not even be connected). Nevertheless, lacking expansion does not necessarily mean the faces complex is a poor sampler. Indeed we will argue $F_X$ is a near-optimal sampler by splitting its concentration into two components 1) the probability of sampling a good partition of \textit{colors} 2) sampling a good subset conditioned on this partition. We show the first step may be reduced to concentration on the faces complex of the complete complex (which we'll call the \textit{swap complex}), while the second reduces to Chernoff for partite HDX.

With this in mind, before jumping into the proof we will first cover some necessary background on the swap complex, which is simply the uniform distribution over (families of) disjoint \textit{$\ell$}-sets of $[n]$.
\begin{definition}[The Swap Complex]
    Let $n \geq k \geq i \in \mathbb{N}$. The $(i,k,n)$-swap complex is the $\lfloor\frac{k}{i}\rfloor$-\maximal simplicial complex \(C=C_{i,k,n}\) whose vertices are all size-$\ell$ subsets $[n]$:
    \[
    \SC[C][1] = \binom{[n]}{i},
    \]
    and whose top-level faces are all possible pair-wise disjoint $\ell$-sets:
    \[
    \SC[C][\left\lfloor\frac{k}{i}\right\rfloor] = \left\{\set{s_1,s_2,\ldots,s_{\lfloor\frac{k}{i}\rfloor}}: \bigcup_{i\neq j} s_i \cap s_j = \emptyset\right\}
    \]
    endowed with the uniform distribution.
\end{definition}
In other words, $C_{\ell,k,n}$ is exactly $F^{i}_{\Delta_n(k)}$, the faces complex of the $k$-\maximal complete complex on $n$ vertices, and is simply the complete complex itself when $i=1$.

Thus, like the complete complex, one might reasonably hope that the disjointness requirement on the swap complex does not substantially hinder concentration. Indeed this is relatively easy to show when $n \gg k$ (e.g.\ from $\ell_\infty$-independence \cite{kaufman2021scalar}, or even just negative correlation for large enough $n$), but becomes challenging in the regime where $n=\Theta(k)$ where disjointness introduces non-trivial positive correlations between the variables. Nevertheless, in \pref{app:swap-complex}, we show these correlations can be handled whenever $n=\Theta_\varepsilon(k)$:
\begin{theorem} \torestate{\label{thm:probability-to-deviate-from-expectation-in-complete-swap-walk}
    Let \(\varepsilon > 0\), $i \leq k \in \mathbb{N}$, and $C=C_{i,k,n}$ for $n \geq \frac{1153 k \log \left (\frac{2}{\varepsilon} \right )}{\varepsilon^5}$. For any \(f:C(1)\to [0,1]\):
    \begin{enumerate}
        \item \textbf{Upper tail}: $\Pr[U_{1,\frac{k}{i}}f-\mathbb{E}[f] > \varepsilon] \leq \exp(-\Omega(\varepsilon^2 \frac{k}{i}))$
        \item \textbf{Lower tail}: $\Pr[U_{1,\frac{k}{i}}f-\mathbb{E}[f] < -\varepsilon] \leq \exp(-\Omega(\varepsilon^2 \frac{k}{i}))$
    \end{enumerate}
    }
\end{theorem}

With this in hand, we are finally ready to handle the partite case of \pref{thm:hdx-is-sampler}. 
\begin{proof}[{Proof of \pref{thm:hdx-is-sampler-2-cases} (\pref{thm:hdx-is-sampler}, partite case)}] 
Let $X$ be a $k$-\maximal skeleton of a $d$-partite $2^{-\Omega(d)}$-HDX. We will actually work with $k'$-skeleton of $X$ where $k' = \min\left\{ k, \frac{\varepsilon^5}{1153\log\left(\frac{2}{\varepsilon}\right)}d\right\}$. Note that if $i \geq \frac{\varepsilon^2k'}{100\log \frac{1}{\varepsilon}}$, the stated bound is trivial. Otherwise, by \pref{lem:raising}, it is enough to show $(X(k'),X(i))$ is an $(\varepsilon,\beta)$-sampler for
\[
\beta \leq \frac{16}{\varepsilon}\exp\left(-\Omega\left(\varepsilon^2\frac{k'}{i}\right)\right).
\]
Let $m=\lfloor \frac{k'}{i} \rfloor$. By \pref{prop:faces-to-inclusion}, it is enough to prove the faces complex $F_X^{(i)}(m)$ is an $(\varepsilon,8\exp(-\Omega(\varepsilon^2 m)))$ sampler for any $\varepsilon>0$. 

Toward this end, fix \(f: \X[i] \to [0,1]\) of expectation \(\mu\) viewed as a function on the vertices of the faces complex. Consider the random variable \(\mathcal{J} \coloneqq (col(s_1),col(s_2),\dots,col(s_m))\) corresponding to the partition generated by drawing $(s_1,\ldots,s_m)$ from the faces complex of $X$. For a fixed $\mathcal{J}=\set{I^1,I^2,\dots,I^m}$, denote the conditional expectation of $f$ under $\mathcal{J}$ as 
    \[
    \mu_\mathcal{J} \coloneqq \underset{s \sim \X[i]}{\mathbb{E}}[f|~col(s) \in \mathcal{J}].
    \]
    Further let $C_{\mathcal{J}}$ denote the complex generated by restricting to top-level faces of $C$ with partition $\mathcal{J}$. We define two `bad' events based on $\mathcal{J}$ outside which our sampling is $\varepsilon$-accurate:  
    \begin{enumerate}
        \item \textbf{(Bad Colors):} \(E_1\) the event that $\mu_\mathcal{J}$ deviates significantly from $\mu$
        \[
        \abs{\mu_\mathcal{J} - \mu} \geq \varepsilon/2.
        \]
        \item \textbf{(Bad Sampling):} \(E_2\) the event that $(s_1,\ldots,s_m) \sim C_\mathcal{J}$ samples $\mu_\mathcal{J}$ poorly
        \begin{align*}
            \abs{\frac{1}{m}\sum\limits_{i \in [m]}f(s_i) - \mu_\mathcal{J}} 
            \geq \varepsilon/2.
        \end{align*}
    \end{enumerate}
    It is then sufficient to prove the following three claims:
    \begin{enumerate}
        \item ~  \(\Prob[(s_1,s_2,\dots,s_m) \sim C]{\Abs{\frac{1}{m}\sum\limits_{i \in [m]}f(s_i) - \mu} > \varepsilon} \leq \prob{E_1} + \prob{E_2 \land \neg E_1}.\)
        \item ~ \(\prob{E_1} \leq 2\exp\left(-\Omega(\varepsilon^2m)\right)\).
        \item ~ \(\prob{E_2 \land \neg E_1} \leq 4\exp(-\Omega(\varepsilon^2 m))\).
    \end{enumerate}
    The first claim is immediate from definition and the triangle inequality, as the event $|\sum f - \mu| > \varepsilon$ only occurs if $E_1 \lor E_2$ holds, the probability of which is bounded by $\Pr[E_1] + \Pr[E_2 \land \lnot E_1]$ as desired.
    \paragraph{Bounding \(\prob{E_1}\)} Observe $\mathcal{J}$ is equidistributed with the complex $C_{i,im,d}$ and define \(g:\SC[C][1] \to [0,1]\) by $g(I^j) = \mathbb{E}_C[f~|~I^j \in col(s)]$. For $\mathcal{J}=\{I^1,\ldots,I^m\}$ we then have:
    \begin{enumerate}
        \item $\mathbb{E}[g] = \underset{I \sim \SC[C][1]}{\mathbb{E}}\left[\mathbb{E}[f~|~col(s) \in I]\right] = \mu$
        \item \(\frac{1}{m}\sum\limits_{j=1}^{m} g(I^j)=\mu_{\mathcal{J}}\).
    \end{enumerate}
Combining these observations with \pref{thm:probability-to-deviate-from-expectation-in-complete-swap-walk} we can write
    \begin{align*} \label{eq:bounding-E1}
        \prob{E_1} &= \Pr_{\mathcal{J}}\left[\left|\mu_{\mathcal{J}} - \mu\right|> \frac{\varepsilon}{2}\right]\\
        &=\Pr_{\{I^1,\ldots,I^m\} \sim \SC[C][m]}\left[\left|\frac{1}{m}\sum\limits_{j=1}^{m}g(I^j) - \mathbb{E}[g]\right|>\frac{\varepsilon}{2}\right]\\
        &\leq 2\exp\left(-\Omega(\varepsilon^2 m)\right)
    \end{align*}
    as desired.
    \paragraph{Bounding \(\prob{E_2 \land \neg E_1}\)} It is left to argue that for any fixed $\mathcal{J}=\{I^1,I^2,\ldots,I^m\}$, $C_{\mathcal{J}}$ is a partite $2^{-\Omega(k')}$-HDX. The result is then immediate from applying \pref{thm:sampling-like-complete} with $\eta=\exp(-\frac{1}{4}\varepsilon^2m)$:
    \begin{align*}
    \Pr[E_2 \land \lnot E_1] &\leq \Pr_{\{s_1,\ldots,s_m\} \sim D}[E_2 |~ \mathcal{J} \in \lnot E_1]\\
    & = \underset{\mathcal{J} \in \lnot E_1}{\mathbb{E}}\left[\Pr_{s \sim C_{\mathcal{J}}}\left[\abs{\frac{1}{m}\sum\limits_{i \in [m]} f(s_i) - \mu_\mathcal{J}} \geq \frac{\varepsilon}{2}\right]\right]\\
    &\leq 3\exp\left(-\frac{1}{4} \varepsilon^2m\right).
\end{align*}
    By construction it is easily checked $C_{\mathcal{J}}$ is a partite complex whose links are isomorphic to swap walks $S_{I_i,I_j}$ within a link of the original complex $X$. By \pref{thm:color-swap-walks-expand}, $C_{\mathcal{J}}$ is therefore a partite $\lambda$-one-sided HDX for $\lambda \leq 2^{-\Omega(d)}$ as well.
\end{proof}

\section{Agreement Testing} \label{sec:agreement}
In this section we use our reverse hypercontractive inequality to prove several new agreement testing theorems on HDX. In the first subsection, we give new agreement tests in the $99\%$-regime between any two levels of an HDX. In the second subsection, we explore the role of reverse hypercontractivity in the $1\%$-regime, and prove that HDX admit agreemment tests with optimal \textit{local} soundness (a weaker notion of soundness only requiring globalness on links), and show under the stronger assumption of $\ell_\infty$-expansion that this guarantee can be propagated to a true tester with optimal soundness.
\subsection{Background}
Suppose \(U\) is some finite set of vertices and \(S \subseteq \mathbb{P}(U)\) are subsets of \(U\). An \emph{agreement test} is a procedure that inputs a family of assignments to the subsets \(\set{f_s :s \to \set{0,1}}_{s \in S}\), and aims to check whether the assignments are consistent with a \emph{global} assignment to the ground set \(G:U \to \set{0,1}\). This is a classical setup in hardness of approximation (in particular in the construction of PCPs), where one repeats a problem in parallel and needs to check that answers across repetitions are consistent.

Let us give the simplest setup first, slightly generalizing it later on. Let \((U,S,\Sigma,D)\) be such that
\begin{enumerate}
    \item \(U\) is a finite set (``Universe'').
    \item \(S \subseteq \mathbb{P}(U)\) (``Sets'').
    \item \(\Sigma\) is another finite set (``Alphabet'').
    \item A distribution \(s_1,s_2 \sim D\) such that \(s_1,s_2 \in S\) (``Distribution'').
\end{enumerate}

An \emph{ensemble of functions} in this context is a family \(\mathcal{F} = \set{f_s:s\to \Sigma}_{s \in S}\). An ensemble is called \emph{global} if there exists a function \(G:U \to \Sigma\) such that for every \(s \in S\) it holds that \(f_s = G|_s\).

We wish to relate the following two quantities. The first quantity is the ``agreement'' of an ensemble. For \(\eta \in [0,1)\) we denote by
\begin{equation}\label{eq:agreement-1}
    Agree_\eta (\mathcal{F}) := \Prob[s_1,s_2,t \sim D]{f_{s_1}|_{s_1 \cap s_2} \overset{\eta}{\approx} f_{s_2}|_{s_1 \cap s_2}},
\end{equation}
where the notation inside the probability means that there is a set \(t' \subseteq s_1 \cap s_2\) of size \(|t'| \geq (1-\eta)|s_1 \cap s_2|\) such that \(f_{s_1}|_{t'} =f_{s_1}|_{t'}\). We will also use $\overset{\eta}{\not\approx}$ to denote the negation of this.

The second quantity is the distance from global functions. For \(\eta, \varepsilon > 0\) we say that \(\mathcal{F}\) is \((\eta,\varepsilon)\)-close to a global function if the exists some \(G:U \to \Sigma\) such that
\[\Prob[s]{f_s \overset{\eta}{\approx} G|_s} \geq 1-\varepsilon.\]
The probability distribution in which we take \(s \in S\) above is the marginal distribution of \(D\).\footnote{Formally, draw $(s,s') \sim D$, then $s$ or $s'$ with probability $1/2$ independently.}

We denote by 
\[
    \dist_{\eta}(\mathcal{F},Glob) = \min \sett{\varepsilon \geq 0}{\text{there exists \(G:U \to \Sigma\) such that \(\mathcal{F}\) is \((\eta,\varepsilon)\)-close to \(G\)}}.
\]

A \emph{high acceptance agreement theorem} (also known as a ``99\% agreement theorem'') is a theorem that states that if a set of local functions pass the test with almost perfect probability (``99\%''), then the functions are close to being global. In other words, this is a theorem that states for a specific \((U,S,\Sigma,D)\) that for every ensemble \(\mathcal{F}\),
\[ 
Agree_\eta (\mathcal{F}) = 1-\varepsilon \quad \implies \quad \dist_{\eta'}(\mathcal{F},Glob) \leq \varepsilon'.
\]
where \(\varepsilon,\varepsilon',\eta,\eta'>0\). Ideally, one wants such a statement for every \(\varepsilon,\eta\) small enough, with \(\varepsilon' = O(\varepsilon)\) and \(\eta' = O(\eta)\). But other parameter regimes are interesting in applications as well.

Agreement tests originated in low-degree testing results such as \cite{RubinfeldS1996} and were later abstracted by \cite{GoldreichS1997} to roughly the setup above. They are an important component of Dinur's proof of the PCP Theorem \cite{dinur2007pcp}. An optimal high acceptance agreement theorem was proven by \cite{DinurS2014} on the complete complex. This was extended by \cite{DinurFH2019} to another setup related to the complete complex we discuss below. Agreement theorems on sparse complexes such as high dimensional expanders are also known, pioneered by \cite{DinurK2017}.

A \emph{low acceptance agreement theorem} (also known as a ``1\% agreement theorem'') deals with the regime where it is only assumed that the functions pass the agreement test with small (sometimes even sub-constant) probability (``1\%''). This is a theorem of the form
\[ 
Agree_\eta (\mathcal{F}) = \delta \geq \varepsilon \quad \implies \quad \dist_{\eta'}(\mathcal{F},Glob) \leq 1-\poly(\delta).
\]
for some \((\varepsilon,\eta,\eta')\) (and some explicit polynomial in \(\delta\)). We note here that in this test even when \(\eta = 0\), an \(\eta' > 0\) is sometimes unavoidable (see discussions in \cite{DinurG2008, DinurL2017}).

Agreement theorems in the low acceptance regime are typically harder to prove than in the high acceptance regime, since the promised structure on the family of functions is mild. However, in applications such that a parallel repetition theorem \cite{ImpagliazzoKW2012} and low-degree testing \cite{AroraS1997}, they are often crucial. Dinur and Goldenberg \cite{DinurG2008} proved a low acceptance agreement theorem on the complete complex. Their proof was simplified by \cite{ImpagliazzoKW2012}, who also derived a more sophisticated Z-test (see below) and gave the first sparse structure supporting such agreement tests using subspaces. Work by \cite{DinurL2017} gave an improved analysis of \cite{ImpagliazzoKW2012}, asymptotically matching the lower bound of a random set of local functions. Their work heavily relied on reverse hypercontractivity, a strategy we follow below. Recently, \cite{dikstein2023agreement} and \cite{bafna2023characterizing} gave characterizations of high dimensional expanders that support agreement theorems, and \cite{DiksteinDL2024,BafnaM2024} constructed such complexes.

There are also other kinds of theorems about agreement tests that don't strictly fall into any of these categories. One example such as \cite{GotlibK2022} list agreement theorem, that inspired the former line \cite{dikstein2023agreement,bafna2023characterizing,DiksteinDL2024,BafnaM2024}.

\subsubsection{Varying the Test}
We present two simple extensions to the agreement test above. The first comes from sampling more than two sets. For example, in the Z-test by \cite{ImpagliazzoKW2012}, three sets \(s_1,s_2,s_3\) are sampled, and the test passes if \(f_{s_1} = f_{s_2}\) \emph{and} \(f_{s_2} = f_{s_3}\) (equality is with respect to the respective intersections). See the formal definition of this test below.

\medskip

Another extension to the definition of a test, mainly for technical reasons, is a test that checks \(f_{s_1}|_t = f_{s_2}|_t\) only on some \(t \subset s_1 \cap s_2\) which is also sampled by the tester. This is sometimes more convenient to analyze than the test checking agreement on the whole intersection. To summarize this discussion more formally, the distribution \(D\) in our setup samples \((\set{s_i}_{i=1}^k, \set{t_{i,j}}_{1 \leq i<j\leq k})\) such that \(\set{s_i}_{i=1}^k \subseteq S\) are sets, and for every \(i<j\), \(t_{i,j} \subseteq s_i \cap s_j\). The agreement test passes if for all \(i <j\), \(f_{s_i}|_{t_{i,j}} \overset{\eta}{\approx} f_{s_j}|_{t_{i,j}}\).

\subsection{Agreement Testing for Subsets}
We now move to the first setting of interest, agreement tests on simplicial complexes in the 99\% regime. In particular we study the following setup due to \cite{DinurFH2019}. Let \(j<k<d\) and let \(X\) be a \(d\)-\maximal simplicial complex. Let us denote by \(SS(X,\Sigma,j,k,d) \coloneqq (U,S,\Sigma,D)\) where: 
\begin{enumerate}
    \item \(U = \X[j]\).
    \item \(S = \sett{\binom{s}{j}}{s \in \X[d]}\).
    \item The distribution \(D = D_{d,k}\) samples \(s_1,s_2\) according to the down-up walk. The `test' set \(t \subseteq s_1 \cap s_2\) is a randomly chosen \(k\)-face.
\end{enumerate}
For simplicity of notation we sometimes refer to the sets as \(s\) instead of \(\binom{s}{j}\), and also index functions by \(f_s\) for \(s \in \X[d]\) (and not \(f_{\binom{s}{j}}\)). However, in this context when writing \(v \in s\), we mean a \(j\)-set instead of a vertex.

In other words, we study the following variant of the classical `V-test' in this setting:
\\
\\
\fbox{\parbox{\textwidth}{
\vspace{.1cm}
\underline{The V-test on $\X[k]$}:
\begin{enumerate}
    \item Draw $t \in \X[k]$, and $s_1,s_2 \in \X[d][t]$ independently
    \item Accept if:
    \[
    f_{t \cup s_1}|_t \overset{\eta}{\approx} f_{t \cup s_2}|_t
    \]
\end{enumerate}}}
\\
\\
When \(j=1\), this reduces to the heavily studied setting of (de-randomized) `direct product testing'.

We are now ready state the main result of the subsection, our generalized 99\%-tester:
\begin{theorem}\label{thm:agreement-99-percent}
        Let \(j,d\) be integers, let \(\gamma,\eta,c, \nu < 1\). Denote by \(k = \gamma d\). Suppose that \(j<k\), \(\eta \lfloor \frac{k}{j} \rfloor \geq 192 \log \frac{16}{\eta}\) and \(\nu < 1\). Let \(X\) be a \(d\)-\maximal \(c\)-locally nice complex such that for every link of a set \(v \in \X[j]\), \(\lambda(UD_{d-j,k-j}(X_v)) \leq \nu\). 
        Then the following holds for the set system \(SS(X,\Sigma,j,k,d)\); for every alphabet \(\Sigma\), family \(\mathcal{F}\), and \(\varepsilon > 0\)
        \begin{equation}
            Agree_{\eta}(\mathcal{F}) \geq 1- \varepsilon \quad \implies \quad \dist_{8\frac{\eta+\varepsilon}{1-\nu}}(\mathcal{F},g_{\text{maj}}) \leq \varepsilon^{O_{\gamma,c}(1)} + 2^{-\Omega_{\gamma,c}(d)}
        \end{equation}
        where $g_{\text{maj}} \coloneqq \underset{s \in \X[d]: s \supset v}{\text{plurality}} \set{f_{s}(v)}$ is the plurality decoding.\footnote{Note that ties may be broken arbitrarily and the weighting is according to the (normalized) complex weights $\pi_{d}$.}
\end{theorem}
We note that every \(c\)-nice complex has \(\lambda(UD_{d-j,k-j}) \leq \frac{k-j}{d-j} (1+\frac{1}{c'(d-j)})^{(1-\gamma) (d-j)}\), so this is not an extra requirement in most cases (see \cite[Theorem 3.5]{lee2023parallelising}).

We remark that for the setup of \(j=1\) there exist agreement theorems that achieve (exact) closeness \(1-O(\varepsilon)\) for high dimensional expanders \cite{DinurK2017,DiksteinD2019,kaufman2020local}. \pref{thm:agreement-99-percent} is incomparable to these results. On the one hand, in \pref{thm:agreement-99-percent} the resulting `globalness' is weaker than \cite{DinurK2017,DiksteinD2019,kaufman2020local} where one gets a guarantee of \(\dist_0(\mathcal{F},g_{maj}) \leq O(\varepsilon)\). On the other hand, the initial assumption on the sets of local functions in the aforementioned results is \(Agree_0 (\mathcal{F}) \geq 1-\varepsilon\), wheras this result assumes \(Agree_\eta (\mathcal{F}) \geq 1-\varepsilon\) for \(\eta > 0\). This relaxed assumption is crucial in the low-soundness regime theorem we give in the next section. 
 
For $j>1$, to our knowledge \pref{thm:agreement-99-percent} is the first agreement theorem beyond the complete complex, where Dinur, Filmus, and Harsha \cite{DinurFH2019} showed
\begin{equation} \label{eq:dfh-agreement-guarantee}
Agree_{0}(\mathcal{F}) \geq 1- \varepsilon \Rightarrow \dist_{0}(\mathcal{F},Glob) \geq 1-O_{j}(\varepsilon).
\end{equation}
It is plausible that, with more effort, \pref{eq:dfh-agreement-guarantee} could be proven directly on HDX. The advantage of our method in this sense lies in its generality and (relative) simplicity.

\begin{proof}[Proof of \pref{thm:agreement-99-percent}]    
    Given a family $\mathcal{F}=\{f_s\}$ with high agreement, we show it can be (approximately) decoded to the majority function $G=g_{\text{maj}}$. We argue this can be inferred directly from reverse hypercontractivity and the following claim bounding the \textit{average} probability of disagreement with $G$ for a random pair $v \subseteq s$ drawn from $(\X[j],\X[d])$:
    \begin{claim} \label{claim:vertex-closeness-to-G}~
        \(\Prob[v \subset s]{f_s(v)\ne G(v)} \leq \frac{\varepsilon + \eta}{1-\nu}.\)
    \end{claim}
    We first show this implies the result. Define the set of `bad' $d$-faces which noticeably disagree with $G$ as
    \[
    B \coloneqq \sett{s \in S}{\Prob[v \in s]{f_s(v) \neq G(v)} > \frac{8(\varepsilon + \eta)}{1-\nu}}.
    \]
    It is enough to show
    \(\prob{B} \leq \max\{\varepsilon^{O_{\gamma,c}(1)},\exp(-\Omega_{\gamma,c}(d))\} \}\). Toward bounding $B$, define an intermediary set $A$ of `good' faces:
    \[
    A \coloneqq \sett{s \in S}{\Prob[v \in s]{f_s(v) \neq G(v)} \leq \frac{4(\varepsilon + \eta)}{1-\nu}}
    \]
    and note that $\Pr[A] \geq \frac{1}{2}$ by \pref{claim:vertex-closeness-to-G} and Markov's Inequality.
    The idea is to argue that, on the one hand, reverse hypercontractivity implies the number of edges $(s_1,s_2)$ between $A$ and $B$ is at least some power of $\Pr[B]$, while, on the other hand, agreement implies the number of such edges is at most $O(\varepsilon)$ because they (typically) contain many $v \in s_1 \cap s_2$ that disagree and therefore fail the test.
    
    More formally, denote the event that $(s_1,s_2) \in (A \times B)$ as $E_{A,B}$. Assuming $\Pr[B] \geq \exp(-c' d)$ (for $c'$ depending on $c,\gamma$ as in \pref{thm:indicator-reverse-hc}), reverse hypercontractivity lower bounds the measure of this event by
    \[
    \Pr[E_{A,B}] \geq \left(\frac{\Pr[B]}{2}\right)^{q}
    \]
    for some constant \(q\) depending only on \(c,\gamma\). On the other hand, since the test rejects with probability $\varepsilon$, we can upper bound the measure of $E_{A,B}$ by 
    \(O(\prob{\text{Test rejects}}) = O(\varepsilon)\) we will obtain a bound of the form \(\prob{B} \leq O(\varepsilon^{1/q})\). To this end, note
    \[
    \Pr[E_{A,B}] \leq \frac{\varepsilon}{\Pr[\text{Test rejects} |~E_{A,B}]}.
    \]
    It is therefore enough to argue the conditional rejection probability is at least some constant (say $\frac{1}{2}$). Toward this end, for a face $s \in \X[d]$, define 
    \[
    T_s \coloneqq \{v \in s : f_s(v) \neq G(v)\}
    \]
    to be the $j$-faces on which $f_s$ disagrees with majority, and recall that an edge in our test is sampled by first picking $s_1 \in \X[d]$, then $t \subset s_1$ in $\X[\gamma d]$, then $s_2 \supset t$ conditionally. Given $t$, denote by $t(j) = \binom{t}{j}$ the set of sub $j$-faces in $t$, for shortness. We define two bad events outside of which the test rejects and prove they occur with vanishing probability.
\begin{enumerate}
    \item \textbf{Mis-sampling $\mathbf{A}$:} $E_1$ the event that $t$ sees `too many' $v \subset s_1$ that disagree with majority:
    \[
    E_1 \coloneqq \left\{ \frac{|t(j) \cap T_{s_1}|}{\binom{k}{j}} \geq 5 \frac{\varepsilon+\eta}{1+\gamma} \right\}
    \]
    \item \textbf{Mis-sampling $\mathbf{B}$:} $E_2$ the event that $t$ sees `too few' $v \subset s_2$ that disagree with majority:
    \[
    E_2 \coloneqq \left\{ \frac{|t(j) \cap T_{s_2}|}{\binom{k}{j}} \leq 7 \frac{\varepsilon+\eta}{1+\gamma} \right\}
    \]
\end{enumerate}
Conditioning on $\lnot E_1$ and $\lnot E_2$ the test rejects since $s_1$ and $s_2$ disagree on a $2\frac{\varepsilon+\eta}{1-\nu}$-fraction of $t$. Thus by a union bound:
\[
\Pr[\text{Test Rejects}|~E_{A,B}] \geq 1-\Pr[E_1|~E_{A,B}] - \Pr[E_2|~E_{A,B}]
\]
and it is enough to show that each $\Pr[E_i |~E_{A,B}] \leq \frac{1}{4}$
\paragraph{Bounding $\mathbf{E_1}$ and $\mathbf{E_2}$} Events $E_1$ and $E_2$ occur with vanishing probability due to the sampling properties of the complete complex $\Delta$ (\pref{claim:complete-complex-multiplicative-sampler}), namely that for any $\alpha,\delta>0$ and $j \leq k$, the inclusion graph $(\Delta(k),\Delta(j))$ is an $(\alpha, \frac{4}{\alpha \delta}\exp(\frac{\delta^2}{12}\alpha \left\lfloor \frac{k}{j}\right\rfloor),\delta)$-sampler.

With this in mind, recall $t \in \X[\gamma d]$ is drawn uniformly at random from $s_1$, so we may equivalently view $t$ as being drawn from a $\gamma d$-\maximal complete complex $\Delta$ on the $d$ vertices of \(s_1\). By assumption, $s_1 \in A$, so $T_{s_1} \subset \Delta(j)$ is of measure at most $4\frac{\varepsilon+\eta}{\gamma}$, and sampling then implies
\[
\Pr\left[\frac{|t(j) \cap T_{s_1}|}{\binom{k}{j}} > 5\frac{\varepsilon+\eta}{1+\gamma}\right] \leq \frac{4(1+\gamma)}{\varepsilon+\eta}\exp\left(-\frac{\eta}{96}\left\lfloor \frac{k }{j}\right\rfloor\right) \leq \frac{1}{4}.
\]
Since we can equivalently sample an edge $(s_1,s_2)$ by first picking $s_2$, then $t \subset s_2$, then $s_1 \supset t$, $E_2$ can be bounded using the lower tail of $\Delta$ in exactly the same fashion to get
\[
\Pr\left[\frac{|t(j) \cap T_{s_2}|}{\binom{k}{j}} < 7\frac{\varepsilon+\eta}{1+\gamma}\right] \leq \frac{4(1+\gamma)}{\varepsilon+\eta}\exp\left(-\frac{\eta}{192}\left\lfloor \frac{k}{j}\right\rfloor\right) \leq \frac{1}{4}.
\]
\end{proof}
\begin{proof}[Proof of \pref{claim:vertex-closeness-to-G}]
    Fix \(v \in \X[j]\) and let $D^{(v)}_{d,k}$ denote the conditional distribution over $(s_1,t,s_2) \sim D_{d,k}$ such that \(v \subset t\). Note that by construction, $D^{(v)}_{d,k}$ (after removing $v$ from each face) is distributed exactly as the down-up walk $D_{d-j,k-j}$ within the link of $v$. Partition the vertices of this graph into sets $S^{(v)}_\sigma \coloneqq \{s|~f_s(v)=\sigma\}_{\sigma\in \Sigma}$ by the alphabet value they assign $v$, and observe that since the marginals of $D^{(v)}_{d,k}$ are distributed as $\SC[X][d-j][v]$ we have
    \[
    1- \Prob[v \subset s]{f_s(v)\ne G(v)} = \Pr_{v \subset s}[f_s(v)=G(v)] = \Ex[v \sim {\X[j]}]{\max_{\sigma \in \Sigma}\left\{\Pr\left[S^{(v)}_\sigma\right]\right\}},
    \]
    where the inner probability is over the $v$-conditioned marginal.

    The idea, which is fairly standard (see e.g.\ \cite[Claim 5.2]{DinurK2017}), is now to relate this maximum partition size to the \textit{local} disagreement of $\mathcal{F}$ using expansion. In particular, for each $v \in \X[j]$ define the $v$-local disagreement of $\mathcal{F}$ as:
    \[
    \varepsilon_v \coloneqq \Pr_{(s_1,s_2) \sim D_{d,k}^{(v)}}[f_{s_1}(v) \neq f_{s_2}(v)].
    \]
    It is easy to see that \( \Ex[v]{\varepsilon_v} \leq \varepsilon + \eta\). This follows from observing that sampling a random $v \in \X[j]$, then $(s_1,t,s_2) \sim D_{d,k}^{(v)}$ is equivalent to \textit{first} sampling $(s_1,t,s_2) \sim D_{d,k}$, then sampling $v$ uniformly from $t$. The bound is clear from the latter interpretation since the tester rejects with probability at most $\varepsilon$, and otherwise agrees on all but an $\eta$ fraction of the $j$-sets in $t$. It is therefore enough to prove
    \[
    \max_{\sigma \in \Sigma}\left\{\Pr\left[S^{(v)}_\sigma\right]\right\} \geq 1-\frac{\varepsilon_v}{1-\nu},
    \]
as taking expectation on both sides gives the result.
    
Toward this end, define $E=E^{(v)}$ to be the set of edges $(s_1,s_2) \sim D^{(v)}_{d,k}$ that cross the alphabet partition. Since the marginal is distributed as the down-up walk in $v$'s link, this graph is a $\nu$-spectral expander by assumption. Writing $S_\sigma=S^{(v)}_\sigma$ for simplicity of notation, we have
        
    \begin{align} \label{eq:agreement-claim-calculation}
        \varepsilon_v &= \prob{E} = \sum_{\sigma \in \Sigma} \Prob[(s_1,s_2)]{s_1 \in S_\sigma, s_2 \notin S_{\sigma}} \nonumber \\
        &= \sum_{\sigma \in \Sigma} \iprod{\one_{S_\sigma}, D^{(v)}_{d,k}(\one - \one_{S_\sigma})} \nonumber\\
        &= 1 - \sum_{\sigma \in \Sigma} \iprod{\one_{S_\sigma}, D^{(v)}_{d,k}\one_{S_\sigma}}
    \end{align}
    where \(\one_{S_\sigma}\) is the indicator of \(S_\sigma\) and \(\one\) is the all-ones function and the final equality is from by stochasticity of $D^{(v)}_{d,k}$. Let $f_\sigma = \one_{S_\sigma} - \prob{S_\sigma} \one$ be the projection of $\one_{S_\sigma}$ onto $\one^\perp$. By orthogonality it holds that \(\snorm{f_\sigma} = \snorm{\one_{S_\sigma}} - \prob{S_\sigma}^2 \snorm{\one} = \prob{S_\sigma} - \prob{S_\sigma}^2\). Therefore exploiting the expansion of \(D^{(v)}_{d,k}\) we have
    \begin{align*}
      \iprod{\one_{S_\sigma}, D^{(v)}_{d,k} \one_{S_\sigma}} &= \prob{S_\sigma}^2\iprod{\one,\one} + \iprod{f_\sigma,D^{(v)}_{d,k} f_\sigma} \\
      &\leq \prob{S_\sigma}^2 + \nu \snorm{f_\sigma}\\
      &= \nu \prob{S_\sigma} + (1-\nu) \prob{S_\sigma}^2.
    \end{align*}
    Inserting this back to \eqref{eq:agreement-claim-calculation} we have that
    \begin{align*}
        \varepsilon_v &\geq 1- \sum_{\sigma \in \Sigma} \left (\nu \prob{S_\sigma} + (1-\nu) \prob{S_\sigma}^2  \right )\\
        &= (1-\nu) \left(1-\sum_{\sigma \in \Sigma} \prob{S_\sigma}^2\right)\\
        &\geq (1-\nu)\left(1-\max_{\sigma \in \Sigma}\left\{\Pr[S_\sigma]\right\}\right).
    \end{align*}
    Re-arranging gives the desired result.
\end{proof}

\subsubsection{Testing all the intersection}
We note that one could also consider the standard V-test where $s$ and $s'$ are still drawn from the down-up walk, but we test agreement over \(t=s_1 \cap s_2\) (whereas in the definition of the test $s_1$ and $s_2$ may have some elements in their intersection not in $t$). The success in these tests are related by the following claim.
    \begin{claim} \label{claim:t-test-related-to-all-vertices-test}
        Let \(\eta > 0\). Let \(D_1\) be the test above. Let \(D_2\) be the test with a similar distribution to \(D_1\), with the distinction that \(t=s_1 \cap s_2\). Then for any \(\mathcal{F}\), \(Agree_{\eta,D_2}(\mathcal{F}) \geq 1-\varepsilon\) implies that \(Agree_{2\eta,D_1}(\mathcal{F}) \geq 1-\varepsilon - 2^{-\Omega_{\gamma}(\eta^2 k/j)}\).
    \end{claim}

\begin{proof}[Proof of \pref{claim:t-test-related-to-all-vertices-test}]
        Let \(s_1,s_2\) be a pair supported by \(D_2\) such that \(\mathcal{F}\) passes the \(D_2\)-test on the pair. Let \(A\) be the set of \(j\)-faces contained in \(s_1 \cap s_2\) on which \(f_{s_1},f_{s_2}\) disagree on. By assumption \(|A| \leq \eta |s_1 \cap s_2|\). By the sampling properties of the complete complex \pref{claim:complete-complex-multiplicative-sampler} the fraction of \(t \subseteq s_1 \cap s_2\) such that \(|A \cap t| \geq 2\eta |t|\) is \(2^{-\Omega_{\gamma}(\eta^2 k/j)}\). Thus the probability of \emph{failing} the \(D_1\) test is at most the probability of failing the \(D_2\) test plus \(2^{-\Omega_{\gamma}(\eta^2 k/j)}\).
\end{proof}

\begin{remark}
    When $\eta \lesssim \varepsilon$, the `additive' error $\exp(-\Omega(d))$ can be removed from \pref{thm:agreement-99-percent} by interpolating the result with a standard averaging argument over \pref{claim:vertex-closeness-to-G}. This technique stops working as soon as $\eta$ is much larger than $\varepsilon$. While the additional $\exp(-\Omega(d))$ error is not harmful in most applications, it is nevertheless interesting to ask whether the term can be removed for general $\eta$.
\end{remark}

\subsubsection{Agreement Testing for the non-Lazy Down-Up Walk}
Using a similar technique, we can also prove an agreement test for the following distribution based on the swap walks. Define the test distribution \((s_1,s_2,t) \sim A_{d,k}\) by the following procedure:
\begin{enumerate}
    \item Sample \(t \in \X[k]\).
    \item Sample \(r_1,r_2 \in \X[d-k]\) in the swap-walk inside \(X_t\).
    \item Output \(s_i = t \dunion r_i\) and \(t\).
\end{enumerate}
The resulting edge distribution over $(s_1,s_2)$ is sometimes called the `partial' swap walk \cite{AlevJT2019}. Note this walk has the property that the intersection is fixed. The following claim shows that when the local swap walks expand, an agreement theorem for this walk follows as well. 
\begin{claim} \label{claim:agreement-for-non-lazy-swap}
        Let \(j,d\) be integers, let \(\gamma,\eta,c,\lambda, \nu < 1\). Let \(k = \gamma d\) and suppose \(j<k\), \(\eta \lfloor \frac{k}{j} \rfloor \geq 192 \log \frac{16}{\eta}\) and \(\nu < 1\). Let \(X\) be a \((2d-k)\)-\maximal \(c\)-locally nice complex such that for every \(v \in \X[j]\), \(\lambda(A_{d-j,k-j}(X_v)) \leq \nu\). Suppose additionally that for every \(t \in \X[k]\), \(S_{d-k,d-k}(X_t) \leq \lambda\). 
        Then the following holds for the set system \(SS(X,\Sigma,j,k,d)\); for every alphabet \(\Sigma\), family \(\mathcal{F}\), and \(\varepsilon > 0\)
        \begin{equation}
            Agree_{\eta}(\mathcal{F}) \geq 1- \varepsilon \quad \implies \quad \dist_{8\frac{\eta+\varepsilon}{1-\nu}}(\mathcal{F},g_{\text{maj}}) \leq \varepsilon^{O_{\gamma,C}(1)} + \lambda^{O_{\gamma,C}(1)}.
        \end{equation}
        where $g_{\text{maj}} \coloneqq \underset{s \in \X[d]: s \supset v}{\text{plurality}} \set{f_{s}(v)}$ is the plurality decoding.
\end{claim}

\begin{proof}[Proof of \pref{claim:agreement-for-non-lazy-swap}]
    As in the proof of \pref{thm:agreement-99-percent}, we define
    \[
    B \coloneqq \sett{s \in S}{\Prob[v \in s]{f_s(v) \neq G(v)} > \frac{8(\varepsilon + \eta)}{1-\nu}}.
    \]
    and
    \[
    A \coloneqq \sett{s \in S}{\Prob[v \in s]{f_s(v) \neq G(v)} \leq \frac{4(\varepsilon + \eta)}{1-\nu}}.
    \]
    As in the proof of \pref{thm:agreement-99-percent}, $\Pr[A] \geq \frac{1}{2}$ by an argument similar to \pref{claim:vertex-closeness-to-G} and Markov's inequality. More precisely, we need to repeat the proof of \pref{claim:vertex-closeness-to-G}, only instead of using the down-up walk conditioned on \(v\), we use \(A_{d,k}\) conditioned on \(v \in t\). This is just \(A_{d-j,k-j}\) inside \(X_v\).

    In addition, as in the proof of \pref{thm:agreement-99-percent}, it also holds that the relative fraction of edges crossing between \(A\) and \(B\) is small, i.e.\ \(\prob{E_{A,B}} \leq 2\varepsilon\).

    Our goal is to show that 
    \(\prob{B} \leq \max\{\varepsilon^{\Omega_{\gamma,c}(1)},\exp(-\Omega_{\gamma,c}(d)),\lambda^{\Omega_{\gamma,c}(1)} \}\).

    Assume otherwise, and that in particular reverse hypercontractivity for indicators applies for sets of size \(\prob{B}\). By \pref{thm:indicator-reverse-hc} 
    \[\Prob[t \in {\X[k]}]{\Prob[s \supseteq t]{A} \cdot \Prob[s \supseteq t]{B} \geq (\prob{A}\prob{B})^{O_{\gamma,c}(1)}} \geq \prob{A}\prob{B}^{O_{\gamma,c}(1)}.\]
    Let us denote by \(G\) the set of \(t\) satisfying the above condition.
    Suppose we sampled \(t \in G\) in the first step of our agreement distribution. Recall that now we sample and edge by going in the swap walk in the link of \(t\). The probability we sampled an edge between \(A\) and \(B\) is therefore
    \[
    \Prob[r_1,r_2 \sim S(X_t)]{t \cup r_1 \in A, t \cup r_2 \in B} \geq  (\prob{A}\prob{B})^{O_{\gamma,c}(1)} - \lambda = p \cdot \prob{B}^{O_{\gamma,c}(1)} - \lambda
    \]
    by the expansion of the swap walk and the expander mixing lemma (for some constant \(p\)).
    Thus if \(p \prob{B}^{O_{\gamma,c}(1)} \geq 2\lambda\) then we have that
    \[ 
    \Omega(\prob{B}^{O_{c,\gamma}(1)}) \leq \prob{G}\cProb{}{s_1 \in A,s_2 \in B}{G} \leq \prob{E_{A,B}} \leq 4\varepsilon
    \] 
    and the claim follows.
\end{proof}

\subsection{Low Soundness and the Z-Test}
In this section we will apply \pref{thm:agreement-99-percent} and reverse hypercontractivity to prove soundness of \cite{ImpagliazzoKW2012}'s Z-test in the $1\%$-regime under certain stronger assumptions on the complex. Before moving on, we briefly comment on the notation in this subsection. Unlike before, we will use the convention that \(s,t\) are \(\frac{k}{2}\)-sets (and not \(k\)-sets as in the previous subsections). We will typically denote \(k\)-sets by a capital letter instead. We first recall their test:
\\
\\
\fbox{\parbox{\textwidth}{\label{fig:z-test}
\vspace{.1cm}
\underline{The Z-test on $\X[k]$}:
\begin{enumerate}
    \item Draw $t \in \X[\frac{k}{2}]$, and $s,s' \in \SC[X][\frac{k}{2}][t]$ independently
    \item Draw $s'' \sim \SC[X][\frac{k}{2}][s']$
    \item Accept if: 
    \[
    f_{t \cup s} = f_{t \cup s'} \quad \text{and} \quad f_{s' \cup t} = f_{s' \cup s''}
    \]
\end{enumerate}}}
\\
\\

We prove soundness for complexes which are `$\lambda$-global', a slight strengthening of $\ell_\infty$-independence \cite{kaufman2021scalar}:
\begin{definition}[Global Complex]\label{def:global}
    A $k$-\maximal complex $X$ is called $\lambda$-global if \(\forall t \in \X[\frac{k}{2}]\): 
    \[
    \norm{S_t - \pi_{\frac{k}{2}}}_{TV} \leq \lambda
    \]
    where $S_t$ is the distribution over neighbors of $t$ in $S_{\frac{k}{2},\frac{k}{2}}$ and \(\pi_{\frac{k}{2}}\) is the induced distribution on \(\X[\frac{k}{2}]\).
\end{definition}
Globality of $X$ essentially promises that every $\frac{k}{2}$-set in $X$ `sees' most other $\frac{k}{2}$-sets in $X$, and is equivalent to bounding the matrix infinity norm of the swap walk from its stationary operator \(\frac{1}{2}\norm{S_{\frac{k}{2},\frac{k}{2}}-\Pi_{\frac{k}{2}}}_\infty \leq \lambda\). We note that the theorem below really only relies on a somewhat weaker average-case bound on the $1$-norm of the rows of $S_{\frac{k}{2},\frac{k}{2}}-\Pi_{\frac{k}{2}}$. However, we are not aware of any complexes satisfying such a notion that are not already global in the defined sense, so we stick with this more standard requirement below.

\begin{theorem}\label{thm:Z-test}
    $\forall \lambda,\eta >0$ and large enough $k$, let $X$ be a $k$-\maximalpunc, \(\lambda\)-global, \(c\)-locally nice complex. Then for any $\delta \in (8\lambda + e^{-\Omega_c(\eta k)},\frac{1}{8})$ if $ Agree_0^Z(\mathcal{F}) \geq \delta$:
    \[
        \exists g:\X[0] \to \Sigma, \; \; \Pr_{s\in \X[k]}[f_s \overset{\eta}{\approx} g(s)] \geq \delta/8.
    \]
\end{theorem}

A couple of remarks are in order. First, we mention that it is well known that no tester can do better than \(\exp(-\Omega(k))\) soundness \cite{DinurG2008,DinurL2017}, so our bound is essentially optimal for small $\lambda$. Second, we remark that many (dense) complexes of interest are indeed $\lambda$-global HDX, including basic examples such as the complete complex and random complexes, and more involved examples such as skeletons of various Ising models, independent sets, list-colorings, or more generally essentially any of the myriad complexes studied in the approximate sampling literature (see e.g.\ \cite{anari2021spectral,liu2021coupling,chen2021optimal,feng2022rapid,blanca2022mixing} among many others). \pref{thm:Z-test} shows the Z-test is sound on these families in a blackbox fashion, even up to $\exp(-\eta k)$ for the appropriate parameters. We elaborate on these examples further in the end of this section.

The proof of \pref{thm:Z-test} largely follows the strategy of \cite{ImpagliazzoKW2012,DinurL2017}. The main difference lies in replacing certain ad-hoc arguments for the complete complex with more general methods based on reverse hypercontractivity and globality.
\subsubsection{A Local Agreement Theorem}
The core of \pref{thm:Z-test} is a \textit{local} structure theorem for the V-test on high dimensional expanders. Stating this requires a few definitions. We remark that these are standard notions in the literature, and we follow the notation of \cite{DinurL2017}. We first define \emph{restrictions}. 
\begin{definition}[Restriction]
    A \emph{restriction} is a pair $\tau=(t,\sigma)$ such that $t \in \X[\frac{k}{2}]$ and $\sigma: t \to \Sigma$ is an assignment to $t$.
\end{definition}
Given a family of functions as above, there is a natural distribution over restrictions that comes from sampling a random function and a random restriction of it. More formally:
\begin{enumerate}
    \item Sample \(t \in \X[\frac{k}{2}]\) and \(s \in \SC[X][\frac{k}{2}][t]\).
    \item Output \(\tau=(t,f_{s \cup t}|_t)\).
\end{enumerate}

Whenever we mention a distribution over restrictions \(\tau\) we always mean \(\tau=(t,f_{s \cup t}|_t)\) sampled as above. Given such a restriction $\tau$, it will also be useful to have notation for the faces in $X$ that are \textit{consistent} with $\tau$.
\begin{definition}[Consistent Strings]
    Given a restriction $\tau=(t,\sigma)$, denote by $\mathcal{V}_\tau$ the set of faces $s \in \X[\frac{k}{2}][t]$ consistent with $\tau$:
    \[
    \mathcal{V}_\tau \coloneqq \left\{s \in \SC[X][\frac{k}{2}][t]: f_{s\cup t}|_t=\sigma\right\}
    \]
\end{definition}
A restriction is called \textit{good} if it has many consistent faces.
\begin{definition}[Good Restriction]
    A restriction $\tau=(t,\sigma)$ is called \(\delta\)-good if $\Prob[s \in X_t]{\mathcal{V}_\tau} \geq \frac{\delta}{2}$.
\end{definition}
When $\delta$ is clear from context, we just call such restrictions good. Finally, a restriction is called `DP' if its consistent strings agree with a global function (`DP' is for `direct product').
\begin{definition}[DP-Restriction]
    A restriction $\tau=(t,\sigma)$ is called \((\eta,\delta)\)-DP if there exists $g_\tau: \X[1][t] \to \Sigma$ such that:
    \[
    \Pr_{s \sim X_t}\left[f_{t\cup s} \overset{\eta/4}{\not\approx} g_\tau ~ \bigg| ~s \in \mathcal{V}_\tau\right] \leq \delta^2
    \]
\end{definition}
As before, when \(\eta,\delta\) are clear from context we just write `DP'.
We can now state the local agreement theorem for high dimensional expanders.
\begin{theorem}[The Local Agreement Theorem]\label{thm:local-agreement}
    For any $\eta>0$ and $k$ sufficiently large, let $X$ be a $k$-\maximal \(c\)-locally nice complex. Then for all $\delta > e^{-\Omega(\eta k)}$, if $Agree^V_0(\mathcal{F}) \geq \delta$:
    \begin{enumerate}
        \item Many restrictions are good:
        \[
        \Pr_{\tau}[\text{$\tau$ is good}] \geq \frac{\delta}{2}
        \]
        \item Almost all `good' restrictions are DP:
                \[
        \Pr_{\tau}[\text{$\tau$ is DP $|$ $\tau$ is good}] \geq 1-\delta^2
        \]
    \end{enumerate}
\end{theorem}
We remark that while \pref{thm:local-agreement} does not give a true agreement tester in its own right, it is a powerful tool independent of \pref{thm:Z-test}. It roughly states that good agreement implies a non-negligible fraction of good restrictions, \emph{and} that if \(f_s\) agrees with a restriction of a face \(t\) it is almost certainly a restriction of a global assignment on the link \(X_t\). This tool is a critical part of many of the previous works on agreement testing \cite{DinurG2008,DinurS2014,ImpagliazzoKW2012,DinurL2017}.

We prove \pref{thm:local-agreement} in the next section. We first show it implies the main theorem under the additional assumption of $\lambda$-globality. The argument is a simple adaptation of \cite{ImpagliazzoKW2012}.
\begin{proof}[Proof of \pref{thm:Z-test}]
    As in \cite{ImpagliazzoKW2012}, we replace the standard Z-test in \pref{fig:z-test} with a related proxy test where we first sample \((t,s)\) as before, but then draw \((s',s'')\) \textit{independent of} \((t,s)\). If \(s' \in X_t\) we continue as in the Z-test, and if \(s' \notin X_t\) we accept. The probability of passing the new test is only greater or equal to the probability of passing the original test, but by \(\lambda\)-globalness, the swap walk starting at \(t\) is \(\lambda\)-close to the stationary distribution hence the probability that $s'$ misses the link of $t$ is at most $\lambda$.

    With this in mind, let $\tau$ denote the restriction $(t,f_{s \cup t}|_t)$ and observe that if $\tau$ is not good, the new Z-test accepts with probability at most $\delta/2+\lambda \leq \delta$. Thus it must be the case that conditioned on $\tau$ being good, the test still accepts with probability at least $\delta$. Condition on this event, and observe that by \pref{thm:local-agreement}, $\tau$ is also a DP-restriction except with probability at most $\delta^2 \leq \delta/8$ (recall that \(\delta < \frac{1}{8}\)).

    Fix such a DP restriction \(\tau\), and denote by $B$ the set of faces in $\mathcal{V}_\tau$ on which $\mathcal{F}$ is $\eta/4$-close to the DP function $g_\tau$ promised by \pref{thm:local-agreement} (if there is more than one, take any such function).
    Assume toward contradiction that $\mathcal{F}$ is not $(\eta,\delta/8)$-close to $g_\tau$. Let \(S = s' \cup s''\) and let $H$ denote the event that \(f_S \overset{\eta}{\approx} g_\tau\).\footnote{Formally here we extend $g_\tau$ to a function on $\X[1]$ by arbitrarily assigning values to any vertex in $\X[1] \setminus \X[1][t]$.} By assumption \(\prob{H} \leq \frac{\delta}{8}\), because \(S\) is drawn as in the distribution of \(\X[k]\). We define four events, outside which the test rejects:
    \begin{enumerate}
        \item $E_1$: The event that $s' \notin X_t$,
        \item $E_2$: The event that $s' \in \mathcal{V}_\tau \setminus B$,
        \item $E_3$: The event that $S \in H$,
        \item $E_4$: The event that $S \notin H$, and $f_{S}|_{s'} \overset{\eta/2}{\approx} g_\tau|_{s'}$.
    \end{enumerate}
    We first argue the test rejects if none of the \(E_i\) occurs. In particular, note by $\overline{E_1} \cap \overline{E_2}$, $s'$ must either be in $B$, or in $X_t \setminus \mathcal{V}_\tau$. In the latter case, the first query of the proxy Z-test rejects since $s' \in X_t$ but \(f_{t \cup s'} \ne f_{t \cup s}\). On the other hand, in the former case $s' \in B$ and therefore $f_{t \cup s'}|_{s'} \overset{\eta/4}{\approx} g_\tau|_{s'}$. However, by $\overline{E_3} \cap \overline{E_4}$, we have $S \notin H$ and, moreover, that the restriction to \(s'\) has disagreement at least $f_{S}|_{s'} \overset{\frac{\eta}{2}}{\napprox} g_\tau|_{s'}$. Thus there must be a vertex on which $f_{t \cup s'}$ and $f_{S}$ disagree and the second query rejects.

    Finally by a union bound it is enough to prove $\sum \Pr[E_i] \leq \delta/2$. The first event occurs with probability at most $\delta/8$ by $\lambda$-globality where \(\lambda < \frac{\delta}{8}\). The second event occurs with probability at most $\delta^2 \leq \delta/8$ by \pref{thm:local-agreement}. The third event occurs with probability at most $\delta/8$ by assumption. The final event occurs with probability at most $e^{-\Omega(\eta k)} \leq \delta/8$ by Chernoff, since $f_S$ $\eta$-disagrees with $g_\tau$ and $s'$ is a random subset, which completes the proof.    
\end{proof}
\subsubsection{Proof of the Local Agreement Theorem}
Let us begin by proving the first property in \pref{thm:local-agreement}.
\begin{claim} \label{claim:many-good-restrictions}
    Under the assumptions of \pref{thm:local-agreement}, \(\Prob[\tau]{\tau \text{ is good}} \geq \frac{\delta}{2}\).
\end{claim}
\begin{proof}
    Let \(p = \Prob[\tau]{\tau \text{ is good}}\).
    Since \(Agree^V_0(\mathcal{F}) \geq \delta\) and we have
    \begin{align*}
    \delta &\leq Agree^V_0(\mathcal{F})\\
    &= p \cProb{s,t,s'}{f_{t \cup s} = f_{t \cup s'}}{(t,f|_{s \cup t}) \text{ is good}} + (1-p) \cProb{s,t,s'}{f_{t \cup s} = f_{t \cup s'}}{(t,f|_{s \cup t}) \text{ is not good}}.
    \end{align*}
    By definition \(\cProb{s,t,s'}{f_{t \cup s} = f_{t \cup s'}}{(t,f|_{s \cup t}) \text{ is not good}} \leq \frac{\delta}{2}\) so $\delta \leq p + \frac{\delta}{2}$ which implies the claim.
\end{proof}
The rest of this section is dedicated to proving the second property in \pref{thm:local-agreement}, which is more complicated. The proof has two main steps. First, we will need a variant of the notion of an \textit{excellent restriction}, that appeared in \cite{DinurG2008,ImpagliazzoKW2012,DinurL2017}. Roughly, an excellent restriction is a restriction \((t,\sigma)\) so that if \(s_1,s_2 \in \mathcal{V}_t\), the probability that \(f_{t \cup s_1}|_{s_1 \cap s_2} \approx f_{t \cup s_2}|_{s_1 \cap s_2}\) is \(1-\poly(\delta)\). The probability in which we choose \(s_1, s_2 \in X_t\) is stated below. We show most good restrictions are in fact excellent.

Second, we show that excellent restrictions are \(DP\). Here we use a variant of the ``smoothing'' technique introduced by \cite{DinurL2017} in order to get an ensemble of local functions \(\sett{\tilde{f}_{s}:s \to \Sigma}{s \in \SC[X][\frac{k}{2}][t]}\) that agree with high probability \(1-\poly(\delta)\) and use our $99\%$-tester (\pref{thm:agreement-99-percent}) to show the restriction is \(DP\).

Formally we require some additional notation to properly define these notions. First, denote the $\frac{k}{4}$-step down-up walk in the link of $t$ as $N_t$ and denote by \((s,t',s') \sim N_t\) the distribution where \(s,s'\) are chosen according to \(N_t\) and \(t'\) is the intermediate set chosen in the first down step (we note that \(t' \subseteq s \cap s'\) but equality doesn't necessarily hold). Second, let \(N_t^3\) denote three steps of \(N_t\). Let \((s_1,t_1,s_2,t_2,s_3,t_3,s_4) \sim N_t^3\) be the full walk sampled. Finally, let \((s_1,t_2,t',s_4) \sim N_t^3\) be such that \(s_1,t_2,s_4\) are the marginals as before, and \(t'=t_1\cap t_2\cap t_3\) the vertices that are in all sets in the sampled walk.

Fix some sufficiently small constant $c_1>0$. We can now define excellence.

\begin{definition}[Excellent Restriction]
    A good restriction $\tau=(t,\sigma)$ is called excellent if:
    \begin{enumerate}
        \item \(
                \Pr_{(s,t',s') \sim N_t}\left[(s,s' \in \mathcal{V}_\tau) \land \left(f_s|_{t'} \overset{\eta/1500}{\not\approx} f_{s'}|_{t'}\right)\right] \leq e^{-c_1\eta k}
            \) and,
        \item \(
                \Pr_{(s_1,t_2,t',s_4) \sim N_t^3}\left[(s,s' \in \mathcal{V}_\tau) \land \left(f_{s_1}|_{t'} \overset{\eta/1500}{\not\approx} f_{s_4}|_{t'}\right)\right] \leq e^{-c_1\eta k}.
            \)
    \end{enumerate}
    where we've abused notation and written $f_s$ as shorthand for $f_{s \cup t}$.
\end{definition}
We note that this definition is slightly cumbersome due to the fact that we don't measure the agreement of \(f_s,f_{s'}\) (respectively \(f_{s_1},f_{s_4}\)) on the intersection of \(s \cap s'\) (respectively \(s_1 \cap s_4\)), but instead on a fixed subset. This is for technical reasons and we encourage the readers to think of the case of \(N_t\) where \(t' = s \cap s'\) (and analogously for \(N_t^3\)).

For notational convenience, let $\mu \coloneqq e^{-c_1\eta k}$ and $\eta'=\eta/1500$. The following lemma (or variants thereof) is fairly standard \cite{impagliazzo2008uniform,ImpagliazzoKW2012,DinurL2017}:
\begin{lemma}[Good restrictions are excellent]\label{lem:excellence}
    There exists $c_2>0$ such that
        \[\Prob[\tau=(t,f_{s \cup t}|_t)]{\tau \text{ is good but not excellent}} \leq e^{-c_2\eta k}.\]
\end{lemma}

\begin{proof}
    We only prove that $N_t^3$ case, which is somewhat less standard. The $N_t$ proof is analogous. It is sufficient to show that over a random restriction \(\tau=(t,f_{s_0 \dunion t}|_{t})\):
    \begin{equation} \label{eq:average-excellency}
        \Pr_{\tau=(t,f_{s_0 \cup t}|_{t})}\left[\Pr_{(s_1,t_2,t',s_4) \sim N_t^3}\left[s_1,s_4 \in \mathcal{V}_t \land f_{s_1}|_{t'} \overset{\eta'}{\not\approx} f_{s_4}|_{t'}\right] >  e^{-c\eta k} \right] \leq e^{-\Omega(\eta k)}
    \end{equation}
    where \(s_1,t_2,t',s_4\) is chosen independently of \(s_0\). In particular if this holds then, by claim \pref{claim:many-good-restrictions}, conditioning on being good can increase the probability of this event to hold to at most $\frac{2e^{-\Omega(\eta k)}}{\delta} \leq e^{-c_2\eta k}$ for the right choice of constants.

    The key to \pref{eq:average-excellency} is to observe that this process can be sampled by the following equivalent method: sample $T \in \X[\frac{k}{2}+\frac{k}{4}]$, and randomly partition it into $T = t \dunion t_2$ so that \(|t| = \frac{k}{2}\) and \(|t_2|=\frac{k}{4}\). The faces $s_1$, $t_2$, $t_3$, and $s_4$ are now drawn independently from the "up-down-up" walk \(U_{\frac{k}{4},\frac{k}{2}} D_{\frac{k}{2},\frac{k}{4}} U_{\frac{k}{4},\frac{k}{2}}\) within $X_t$ starting from $t_2$.

    Observe that $|t'| \geq \frac{k}{32}$ with probability $e^{-\Omega(k)}$ by a standard Chernoff bound.
    Conditioning on this event, observe that if $s_1,s_4 \in \mathcal{V}_\tau$ and $f_{s_1}|_{t'} \overset{\eta'}{\not\approx} f_{s_4}|_{t'}$, then it must be the case that
    \begin{enumerate}
        \item $t_2$ has at least $\frac{\eta'k}{32}$ elements on which $f_{s_1} \neq f_{s_4}$
        \item $t$ has no elements on which $f_{s_1}|_{t} \neq f_{s_4}|_{t}$.
    \end{enumerate}
    However, $t$ and $t_2$ are a uniformly random partition of $T$, so such a split occurs with probability at most $e^{-\Omega(\eta k)}$ by Chernoff as desired.
\end{proof}

To proving excellent restrictions are DP, we'll use a variant of Dinur and Livni-Navon's `smoothing' operation that spreads the consistent strings $\mathcal{V}_\tau$ over the entire link. Instead of using the noise operator as in their work, we use the down-up walk which avoids a number of technical complications.
\begin{definition}[Smoothed Assignment]
    For every $\tau = (t,f_{t \cup s_0}|_t)$, define the smoothed assignment $\widetilde{\mathcal{F}}_\tau = \{\tilde{f}_s\}_{s \in X_t}$ as
    \[
    \tilde{f}_s(v) = \underset{(t',s') \sim N_t(s): s' \in \mathcal{V}_\tau ,t' \ni v}{\text{Plurality}}\{f_{s'}(v)\}.
    \]
    We break ties arbitrarily. If the plurality is not well defined (i.e.\ there are no such \(s' \in \mathcal{V}_\tau\) that contain \(v\)) we define it to be \(\tilde{f}_s(v)=\bot\).
\end{definition}
The main idea is to show that on excellent restrictions, $\widetilde{\mathcal{F}}_\tau$ is 1) highly consistent with the original family $\mathcal{F}$ on $\mathcal{V}_\tau$, and 2) has very high agreement on the \textit{entire} link. This reduces the problem to an (approximate) 99\%-regime test within the link of $t$, which we can solve via our tester from the previous section. More formally, the following lemmata suffice to prove \pref{thm:local-agreement}.
\begin{lemma}\label{lem:1-key-lemma}
Every excellent $\tau=(t,\sigma)$ satisfies:
\[
    \Pr_{s \sim X_t}\left[f_s \overset{10\eta'}{\not\approx} \tilde{f}_s ~\Bigg|~s \in \mathcal{V}_\tau\right] \leq \frac{\delta^2}{2},
\]
that is, $\mathcal{F}$ and $\widetilde{\mathcal{F}}$ are close on $\mathcal{V}_\tau$ (note \(f_s \overset{10\eta'}{\not\approx} \tilde{f}_s\) is over just the vertices of $s$, not $s \cup t$)
\end{lemma}

\begin{lemma}\label{lem:2-key-lemma}
There exists $c_3>0$ such that every excellent $\tau=(t,\sigma)$ satisfies
\[
        \Pr_{(s,t',s') \sim N_t}\left[\tilde{f}_s|_{t'} \overset{10\eta'}{\not\approx} \tilde{f}_{s'}|_{t'} \right] \leq e^{-c_3\eta k},
\]
that is, $\mathcal{\widetilde{F}}$ is highly agreeing on $X_t$.
\end{lemma}

\begin{proof}[Proof of \pref{thm:local-agreement}]
    The first property is proved in \pref{claim:many-good-restrictions}. Toward the second, we will prove that every excellent \(\tau\) is a \(DP\)-restriction. Then by \pref{lem:excellence} we have that
    \[
    \cProb{\tau}{\tau \text{ is DP}}{\tau \text{ is good}} \geq \cProb{\tau}{\tau \text{ is excellent}}{\tau \text{ is good}} \geq 1-\delta^2
    \]
    where we've used the assumption that \(\delta \geq e^{-\Omega(\eta k)}\).
    Fix any excellent \(\tau = (t,\sigma)\). One can check these satisfy the conditions of \pref{thm:agreement-99-percent} with $\nu=1/2$. Together with \pref{lem:2-key-lemma} and \pref{claim:t-test-related-to-all-vertices-test} that translates between the agreement guarantee in \pref{lem:2-key-lemma} to that of \pref{thm:agreement-99-percent}, we have that there exists a global function $g:\X[0][t] \to \Sigma$ such that
    \[
        \Pr_{s \sim X_t}[\tilde{f}_s \overset{320\eta'}{\not\approx} g|_{s}] \leq e^{-c_4\eta k}
    \]
    where \(c_4\) is some constant depending on \(c_3\), \(\eta\) and the \(c\)-local niceness of \(X\). 
     We note that after conditioning over \(s \in \mathcal{V}_\tau\), we still have that this is at most \(\frac{2e^{-c_4\eta k}}{\delta} \leq \frac{\delta^2}{2}\), where this holds by assumption that \(\delta\) is large enough.
    
    Combining this with \pref{lem:1-key-lemma}, we then have
    \begin{align*}
        \cProb{s \sim X_t}{f_s \overset{\eta/4}{\not\approx} g(s)}{s \in \mathcal{V}_t}
        &\leq \Pr_{s \sim X_t}\left[f_s \overset{10\eta'}{\not\approx} \tilde{f}_s ~\Bigg|~s \in \mathcal{V}_\tau\right] + \cProb{s \sim X_t}{\tilde{f}_s \overset{320\eta'}{\not\approx} g(s)}{s \in \mathcal{V}_t} \\
        &\leq \delta^2
    \end{align*}
    as desired.
\end{proof}
It is left to prove the key lemmata. This is the main place we use reverse hypercontractivity (other than the application of \pref{thm:agreement-99-percent} in the proof of \pref{thm:local-agreement}). The proofs for both properties follow the strategy of \cite{DinurL2017} adapted to our setting.
\begin{proof}[Proof of {\pref{lem:1-key-lemma}}]~
We define two families of `bad' sets outside of which $f$ and $\tilde{f}$ approximately agree. First, we look at the set of `lonely' strings that don't sufficiently see $\mathcal{V}_\tau$.
\[
    L \coloneqq \left\{s \in X_t: \Pr_{s' \sim N_t(s)}[s' \in \mathcal{V}_\tau] \leq \delta'\right\}
\]
for $\delta' = e^{-c_5\eta k}$ for some $c_5>0$ sufficiently small. We will assume below  that \(\delta' \leq \delta^4\) (which is possible due to the assumptions on \(\delta\)).
We record the following properties of \(L\) that we prove at the end of the subsection.
\begin{claim}\label{claim:most-vertices-in-L-behaive-like-expectation}~
There exists a sufficiently small constant \(c_5 > 0\) such that the following holds for \(\delta' = e^{-c_5 \eta k}\).
\begin{enumerate}
    \item \(\Prob[s \in X_t]{L} \leq \frac{\delta'^2}{2}\).
    \item For any \(s \notin L\), the fraction of vertices $v \in s$ s.t.\ \(\cProb{(t',s') \sim N_t(s)}{s' \in \mathcal{V}_\tau}{t' \ni v} \leq \delta'/3\) is at most $\eta'$.
\end{enumerate}
\end{claim}

Second, we look at the family of sets within $\mathcal{V}_\tau$ with strong disagreement:
\[
    B \coloneqq \left\{ s \in \mathcal{V}_\tau: \Pr_{(t',s') \sim N_t(s)}\left[s' \in \mathcal{V}_\tau \land f_s|_{t'} \overset{\eta'}{\not\approx} f_{s'}|_{t'}\right] \geq \frac{\delta'}{30}\right\}.
\]
We claim that both $L$ and $B$ have measure at most $\delta^2/2$ within $\mathcal{V}_\tau$. The former follows from the first item of \pref{claim:most-vertices-in-L-behaive-like-expectation} since $\Pr[L|\mathcal{V}_\tau] \leq \frac{\delta'^2}{\delta} \leq \frac{\delta^2}{2}$. For the latter we have $\Pr[B] \leq \frac{30\mu}{\delta'} \leq \frac{\delta^2}{2}$ from Markov's inequality and excellence of $\tau$ (here we also need that the \(c_5\) in the definition of \(\delta'\) is sufficiently smaller than the \(c_1\) in the definition of \(\mu\)).

It is left to observe that $f_s \overset{10\eta'}{\approx} \tilde{f}_s$ for any face $s \in \mathcal{V}_\tau \setminus (L \cup B)$. Fix such an $s$, and define $D_0$ to be the set of vertices on which $f$ and $\tilde{f}$ disagree:
\[
    D_0 \coloneqq \left\{v \in s: f_s(v) \neq \tilde{f}_s(v) \right\}.
\]
We need to argue $|D_0| \leq 10\eta' |s|$. Assume toward contradiction otherwise. Let $D'$ denote the vertices in $s$ such that \(\cProb{(t',s') \sim N_t(s)}{s' \in \mathcal{V}_\tau}{t' \ni v} \leq \delta'/3\) and write $D=D_0 \setminus D'$. By \pref{claim:most-vertices-in-L-behaive-like-expectation}, we have $|D| \geq 9\eta'|s|$.

Consider now the following probabilistic experiment. Sample \(v \in D\) uniformly, then some \(t',s' \sim N_t(s)\) conditioned on \(v \in t'\) and \(s' \in \mathcal{V}_\tau\). We denote this distribution by \(\mathcal{D}\). Note that for every fixed \(v \in D\), this is exactly the distribution over which the plurality in the definition of \(\tilde{f}_s(v)\) was defined.

With this in mind consider the event \(P = \set{f_s(v)\ne f_{s'}(v)}\) and observe that \(\Prob[\mathcal{D}]{P} \geq \frac{1}{2}\) because for every fixed \(v \in D\), by construction \(f_s(v)\) disagrees with the plurality vote on \(v\). On the other hand, we show toward contradiction that \(P\) is contained in events whose total probability is strictly less than \(\frac{1}{2}\). Namely:
\begin{enumerate}
    \item \(E_1\): the event \(\set{f_s|_{t'} \overset{\eta'}{\napprox} f_{s'}|_{t'}}\)
    \item \(E_2\): the event \(|t' \cap D| < 6\eta' |t'|\)
    \item $E_3$: the event \(P \setminus (E_1 \cup E_2)\).
\end{enumerate}
Obviously \(P \subseteq E_1 \cup E_2 \cup E_3\) so it is enough to bound the probability of the three events, starting with \(E_3\). In this case we have that \(f_s,f_{s'}\) disagree on at most an \(\eta'\)-fraction of \(t'\), but \(|t' \cap D| \geq 6 \eta' |t'|\). Conditioned on $t'$, $v$ is drawn uniformly from $|D \cap t'|$, so the probability that \(f_s(v) \ne f_{s'}(v)\) is at most \(\frac{1}{6}\).

Moving on to \(E_2\), notice that if \(D\) is an \(9\eta'\)-fraction of \(s\), a random \(t' \subset s\) contains \emph{less than a} \(6 \eta'\)-fraction of \(D\) with probability \(\exp(-\Omega(\eta'^2 k))\) by Chernoff. While the distribution \(\mathcal{D}\) doesn't sample \(t'\) uniformly (since it conditions on \(v \in t'\) and \(s' \in \mathcal{V}_\tau\)), we do know for any fixed $v \in D$, \(v \in t'\) with probability at least \(\frac{1}{2}\), and because \(s \notin L\) and $v \notin D'$ the probability that \(s' \in \mathcal{V}_\tau\) is at least \(\frac{\delta'}{3}\) even after conditioning on $v$. Hence, even after conditioning on these events the probability that \(|t' \cap D| < 6 \eta'|t'|\)-fraction of \(D\) is at most \(2\delta^{-1}\exp(-\Omega(\eta'^2 k)) < \frac{1}{10}\) for large enough \(k\).

Finally, we bound the probability of \(E_1\) by similar reasoning.
\begin{align*}
    &\Prob[\mathcal{D}]{E_1} = \cProb{v \sim D, (t',s') \sim N_t(s)}{f_s|_{t'} \overset{\eta'}{\not\approx} f_{s'}|_{t'}}{s' \in \mathcal{V}_\tau, v \in t' } \\
    &\leq \frac{3}{\delta'}\cProb{v \sim D, s' \sim N_t(s)}{s' \in \mathcal{V}_\tau \land f_s|_{t'} \overset{\eta'}{\not\approx} f_{s'}|_{t'}}{v \in t' } \\
    &\overset{s \notin B}{\leq} \frac{3}{\delta'} \cdot \frac{\delta'}{30} \cdot \prob{v \in t'}^{-1} \leq \frac{1}{5}.
\end{align*}
Thus \(\frac{1}{2} \leq \prob{P} \leq \prob{E_1} + \prob{E_2} + \prob{E_3} < \frac{1}{2}\) and a contradiction is reached.
\end{proof}
\begin{proof}[Proof of \pref{lem:2-key-lemma}]
The proof is similar to the first property. Let \(L\) be as in \pref{lem:1-key-lemma}. This time we define a set of bad \emph{triples} as follows. For a set \((s,t',s') \sim N_t\) we consider the following quantity
\[F(s,t',s') \coloneqq \cProb{(s_1,t_1,s_2,t_2,s_3,t_3,s_4) \sim N_t^3}{s_1,s_4\in \mathcal{V}_\tau \land f_{s_1}|_{t''} \overset{\eta'}{\not\approx} f_{s_4}|_{t''}}{s_2=s,t_2=t',s_3=s'}\]
where \(t''=t_2 \cap t_3 \cap t_4\). We define \(B\) to be the set of triples with large $F$-value:
\[
    B \coloneqq \left\{(s,t',s'): F(s,t',s') \geq \frac{\delta'^2}{80}\right\},
\]
and keep $L$ as in the proof of \pref{lem:1-key-lemma}. By excellence of $\tau$ and Markov's inequality $\Pr_{(s,t',s') \sim N_t}[B] \leq \frac{80\mu}{\delta'^2} \leq \exp(-\Omega( \eta k))$. Since $s$ and $s'$ are both marginally distributed as $\SC[X][\frac{k}{2}][t]$, we also have by \pref{claim:most-vertices-in-L-behaive-like-expectation} and a union bound that either one of $s \in L$ or \(s' \in L\) with probability at most $\delta'^2 = e^{-2c_5\eta k}$. It is therefore enough to show $\tilde{f}_{s}|_{t'} \overset{10\eta'}{\approx} \tilde{f}_{s'}|_{t'}$ under the assumption that $(s,t',s') \notin B$ and $s,s' \notin L$.

Toward this end, fix $(s,t',s')$ and consider the set of disagreeing vertices on the smoothings at $s$ and $s'$:
\[
D_0 \coloneqq \left\{v \in t': \tilde{f}_{s}(v) \neq \tilde{f}_{s'}(v)\right\}.
\]
Assume for the sake of contradiction that $|D_0| \geq 10\eta'|t'|$. Similar to before, we actually consider \(D = D_0 \setminus D'\) where $D'$ is the set of vertices in $t'$ such that $\Pr_{(t_3,s_4) \sim N_t(s')}[s_4 \in \mathcal{V}_\tau | t_3 \ni v]$ (likewise for $s_1,t_1 \sim N_t(s)$). By \pref{claim:most-vertices-in-L-behaive-like-expectation}, both $s$ and $s'$ have at most $\eta'|s|$ such vertices so $|D| \geq 6\eta'|t'|$. 

We now consider the following probabilistic experiment. Sample \(v \in D\) uniformly and \((s_1,t_1,s_2,t_2,s_3,t_3,s_4)\) conditioned on \(s_2 = s, t_2 = t', s_3=s'\). We then condition this experiment on:
\begin{enumerate}
    \item \(v \in t''=t_1 \cap t_2 \cap t_3\).
    \item \(s_1,s_4 \in \mathcal{V}_\tau\).
\end{enumerate}
Note that for a fixed \(v \in D\), the marginals \((s_1,t_1,s_2=s)\) and \((s_3=s',t_3,s_4)\) are precisely the distributions that were used in the definition of \(\tilde{f}_{s}(v), \tilde{f}_{s'}(v)\) respectively. As before, define the event \(P = \set{f_{s_1}(v) \ne f_{s_4}(v)}\) and observe that \(\Prob[\mathcal{D}]{P} \geq \frac{1}{2}\) because for every \(v \in D\) the plurality labeling for $v$ over $N_t$ under our conditioned process is different at $s$ and $s'$.

We now define events whose union contains \(P\) which occur with total probability less than \(\frac{1}{2}\). These are
\begin{enumerate}
    \item \(E_1\): the event \(\{f_{s_1}|_{t''} \overset{\eta'}{\napprox} f_{s_4}|_{t''}\}\)
    \item \(E_2\): the event \(|t'' \cap D| \leq 5\eta'|t''|\).
    \item $E_3$: the event \(P \setminus (E_1 \cup E_2)\).
\end{enumerate}
By the same analysis as \pref{lem:1-key-lemma} we have \(\prob{E_3} < \frac{1}{5}\) and \(\prob{E_2} \leq \frac{1}{10}\) (we omit the repeated details). As for \(E_1\) we have
\begin{align*}
    \Prob[\mathcal{D}]{E_1} &= \cProb{v \in D,s_1,t_1,t_3,s_4}{f_{s_1}|_{t''} \overset{\eta'}{\not\approx} f_{s_4}|_{t''}}{s_1,s_4 \in \mathcal{V}_\tau, v \in t'' } \\
    &\leq \frac{9}{\delta'^2}\cProb{v \sim D, s' \sim N_t(s)}{s_1,s_4 \in \mathcal{V}_\tau \land f_{s_1}|_{t''} \overset{\eta'}{\not\approx} f_{s_4}|_{t''}}{v \in t''} \\
    &\overset{(s,t',s') \notin B}{\leq} \frac{9}{\delta'^2} \cdot \frac{\delta'^2}{80} \cdot \prob{v \in t'}^{-1} \leq \frac{1}{5}.
\end{align*}
We reach a similar contradiction as before.
\end{proof}

\begin{proof}[Proof of \pref{claim:most-vertices-in-L-behaive-like-expectation}]
Let us begin with the first item. Recall under the assumptions of \pref{thm:Z-test} $N_t$ is reverse hypercontractive for indicators. Namely let $c_6$ be the constant promised by \pref{thm:indicator-reverse-hc}, and observe that by assumption $\Prob[X_t]{\mathcal{V}_\tau} \geq \delta/2 \geq e^{-c_6\eta k}$. If  $\Prob[X_t]{L} \leq 2e^{-2c_5\eta k}$ we are done, so assume otherwise. Then by definition of $L$ and \pref{thm:indicator-reverse-hc} we have the chain of inequalities:
\[
\Pr[L]\delta' \geq \Pr_{(s,s') \sim N_t}[s \in L, s' \in \mathcal{V}_\tau] \geq 2^{-q-1}\Pr[L]^q\delta^q.
\]
where $q$ is the constant given in \pref{thm:indicator-reverse-hc}. Re-arranging gives $\Pr[L] \leq (2^{q+1}\delta'\delta^{-q})^{\frac{1}{q-1}} \leq \exp(-c_7 \eta k) \leq \frac{\delta^2}{2}$ for the appropriate choice of constant $c_7>0$. Thus if \(c_5\) is sufficiently small we get the first item.

\medskip

For the second item, let $D'$ denote the set of $v \in s$ such that $\Pr_{(t',s') \sim N_t(s)}[s' \in \mathcal{V}_\tau~|~t' \ni v] \leq \frac{\delta'}{3}$. Define \(g:\binom{s}{|t'|} \to [0,1]\) to be \(g(t_0)=\cProb{(t',s') \sim N_t(s)}{s' \in \mathcal{V}_\tau}{t'=t_0}\). Observe that \(\ex{g} \geq \delta'\) since by assumption \(s \notin L\). Moreover by construction we have \(v \in D'\) if and only if \(\Ex[t_0 \ni v]{g(t_0)} \leq \frac{1}{3}\ex{g(t_0)}\).

The problem is now reduced to sampling on the complete complex. In other words let \(L = \binom{s}{|t'|}\), \(R = s\) and \(v \sim t\) if \(v \in t\). By Chernoff for every \(\alpha\), this graph is an \((\alpha,\beta,\frac{1}{2})\)-multiplicative sampler for \(\beta = \exp(-\Omega(\alpha k))\). By \pref{claim:opposite-sampler}, this implies that the opposite graph is a \((\frac{6\beta}{\alpha},2\alpha,\frac{2}{3})\)-multiplicative sampler. We take \(\alpha = \eta'/2\), and the constant \(c_5 > 0\) defined above small enough so that \(\frac{6\beta}{\alpha} \leq \delta'\). As such we have for every function \(g:L \to [0,1]\) of expectation at least $\delta'$ the fraction of vertices \(v \in s\) such that \(\Ex[t' \subseteq s, v \in t']{g} <\frac{1}{3} \Ex[t' \subseteq s]{g(t')}\) is at most \(\eta'\) as desired.
\end{proof}

\subsection{Examples of Global Complexes}
We conclude the section with several examples of `global HDX' to which \pref{thm:Z-test} applies in a blackbox manner. We start with a few basic direct examples and then argue a vast array of complexes studied in the approximate sampling literature give rise to global HDX. These are:
\begin{enumerate}
    \item The complete complex.
    \item Erdos-Renyi Hypergraphs.
    \item Skeletons of the full linear matroid over \(\mathbb{F}_q^d\).
    \item Skeletons of \(\ell_\infty\)-independent complexes, including many classic spin systems.
    \item Faces complexes of any of the above.
\end{enumerate}

The first and simplest example of global complexes (other than the complete complex) come from the classical random model of Linial and Meshulam \cite{LinialM2006}.
\begin{definition}[Erdos-Renyi Hypergraphs]
    The Erdos-Renyi hypergraph $X \sim G_k(n,p)$ is a random \(k\)-\maximal simplicial complex whose \((k-1)\)-skeleton is complete, and such that every \(k\)-face is sampled into \(X\) with probability \(p\) independently.
\end{definition}
The second direct example we'll give is based on the full linear matroid over $\mathbb{F}_q^d$, whose faces consist of linearly independent vectors over $\mathbb{F}_q^d$:
\[
\X[i] = \sett{\set{v_1,v_2,\dots,v_i}}{{v_1,v_2,\dots,v_i} \text{ are independent}},
\]
endowed with the uniform distribution.
\begin{claim}[Global complexes]
The following families are global and HDX:
    \begin{enumerate}
        \item The $k$-\maximal complete complex on $n$ vertices for large enough $n$ is:
        \begin{itemize}
            \item (one-sided) $0$-local-spectral HDX
            \item $n^{-\Omega(k)}$-global
        \end{itemize}
        \item The Erdos-Renyi hypergraph $G_k(n,p)$ for large enough $n$ is w.h.p:
        \begin{itemize}
            \item $o_n(1)$-local-spectral HDX
            \item $(1-p+o_n(1))$-global
        \end{itemize}
        \item The $k$-skeleton of the full linear matroid over $\mathbb{F}_q^d$ for large enough $q,d$ is :
        \begin{itemize}
            \item (one-sided) $0$-local-spectral HDX
            \item $q^{-\Omega(d)}$-global
        \end{itemize}
    \end{enumerate}
\end{claim}
\begin{proof} We cover each case individually:
    \paragraph{Complete Complex} All links of the complete complex are complete and therefore \(0\)-one-sided local spectral expanders. For globalness, observe that the only $\frac{k}{2}$-sets which do not appear in the link of some $t$ are those that intersect it, which are an $1-\frac{\binom{n-k}{k}}{\binom{n}{k}} \leq n^{-\Omega(k)}$ fraction.

    \paragraph{Random Hypergraphs} It is well known for fixed $p$ that this complex is an $o_n(1)$-local-spectral expander with high probability so long as $p \gg \frac{k\log(n)}{n}$ (the top links are simply random \(G(n,p)\) graphs). Since all other links are at least connected with high probability, \pref{thm:td} implies the local-spectral bound. Toward globality, fix any $\frac{k}{2}$-set $t \subset \binom{[n]}{k/2}$. By Chernoff the probability that more than $(1-p+o_n(1))$-fraction of $\frac{k}{2}$-sets disjoint from $t$ fail to appear in $X$ is much less than $n^{-\Omega(k)}$ for large enough $n$. Thus union bounding over $t$, the worst-case total variation of $S_t$ from uniform over ${\binom{n}{\frac{k}{2}}}$ is at most $1-p+o_n(1)$ as desired.\footnote{Formally, distance here should be measured from $\pi_{\frac{k}{2}}$, but this is within $o_n(1)$ of uniform in TV with high probability.}

    \paragraph{Linear Matroid} All matroids (and therefore their skeletons) are $0$-local spectral expanders \cite{AnariLOV2019} (for this matroid specifically, it is folklore). Toward globality, fix $\frac{k}{2}$ linearly independent vectors $B_1=\{v_1,\ldots,v_{k/2}\}$. The set $B_2=\{w_1,\ldots,w_{k/2}\}$ only fails to appear in the link of $B_1$ if $B_1$ and $B_2$ are linearly dependent, which happens with probability roughly $\frac{q^{k+(k-1)d}}{q^{dk}} \leq q^{-\Omega(d)}$ (ignoring low-order terms). Since $\Pi_{k/2}$ is uniform over linearly independent sets this gives the result.
\end{proof}
The sampling literature is rife with examples of ``$\ell_\infty$-independent'' complexes (a term formally coined in \cite{kaufman2021scalar}), which are closely related to global complexes. We'll consider the variant implicit in \cite{chen2021optimal}. Let $\mu$ be a distribution over $[q]^{d}$, and observe $\mu$ induces a natural $d$-partite simplicial complex $X_\mu$ with
\[
\X[d][\mu] \coloneqq \{(1,a_1),\ldots,(d,a_d): (a_1,\ldots,a_d) \in Supp(\mu)\}
\]
and $\pi_X( (1,a_1),\ldots,(d,a_d)) = \mu(a_1,\ldots,a_d)$. Let $S \subset [d]$ and $z \in [q]^S$ be such that $\Pr_{x \sim \mu}[x_S=z] > 0$ (we call such configurations `feasible'). The $(S,z)$-\textit{influence matrix} $\Psi_{z \to S}$ has entries
    \[
    \Psi_{z \to S}((u,i),(v,j)) = 
    \begin{cases}
        \Pr_{x \sim \mu}[x_i=u | x_j=v, x_S=z] - \Pr_{x \sim \mu}[x_i=u|x_S=z] & \text{if $j \neq i$}\\
        0 & \text{if $j=i$}.
    \end{cases}
    \]
A complex is called $\ell_\infty$-independent if $\norm{\Psi}_\infty$ is bounded for all feasible configurations.
\begin{definition}[$\ell_\infty$-independence]
    Fix $q,d \in \mathbb{N}$ and let $\mu$ be a distribution over $[q]^{d}$. 
    $\mu$ is called $D$-$\ell_\infty$-independent if for all feasible $S \subset [d]$ and $z \in \{0,1\}^S$:
    \[
    \norm{\Psi_{S \to z}}_\infty \leq D.
    \]
\end{definition}
We remark that in this regime, one typically thinks of $q$ as fixed or small and the \maxsizity $d$ as going to infinity. We show any sufficiently low-dimensional skeleton of an $\ell_\infty$-independent complex is a global HDX:
\begin{proposition}\label{prop:ind-to-global}
    Let $\mu$ be a $D$-$\ell_\infty$-independent distribution over $[q]^{d}$ and $3 \leq k \leq d$. Then the $k$-skeleton $X_\mu^{\leq k}$ satisfies:
    \begin{enumerate}
        \item $X_\mu^{\leq k}$ is a $\frac{D+q}{d+3-k}$-two-sided local-spectral expander
        \item $X_\mu^{\leq k}$ is $\frac{k^2}{8}\cdot\frac{D+q}{d+3-k}$-global.
    \end{enumerate}
\end{proposition}
We remark that (tighter) spectral variants of both items above are well known (see e.g.\ \cite{anari2021spectral,alev2023sequential}). \pref{prop:ind-to-global} implies essentially every spin-system studied in the recent breakthrough line of work on approximate sampling through spectral independence admits optimal 1\% agreement testers (taking $k \leq \sqrt{\log(d)}$ skeletons).\footnote{Note while spectral independence is weaker than what we require, almost all known methods for spin systems actually bound $\ell_\infty$-independence.} See \cite[Section 3]{kaufman2021scalar} for an overview including distributions with the stochastic covering property, independent sets, various Ising/Potts models, list colorings, and more.

The proof of \pref{prop:ind-to-global} relies on an intermediate notion of \cite{Hopkins2024} called local-$\ell_\infty$-expansion.
\begin{definition}[$(\lambda,\infty)$-expansion]
    A simplicial complex $X$ is called a $(\lambda,\infty)$-local expander if for every $\tau \in X$ with $|\tau| \leq d-2$:
    \[
    \norm{A_\tau - \Pi_\tau}_\infty \leq \lambda
    \]
    where $A_\tau$ is the weighted adjacency matrix of $X_\tau$'s 1-skeleton, and $\Pi_\tau$ is its corresponding stationary operator.
\end{definition}
We require two lemmas regarding this notion. First, we note that up to dependence on the alphabet $q$ (typically thought of as constant in this regime) $\ell_\infty$-expansion is essentially equivalent to $\ell_\infty$-independence.
\begin{lemma}\label{lem:ind-to-local}
Let $\mu$ be a $D$-$\ell_\infty$-independent distribution over $[q]^{d}$. Then for every feasible $S \subset [d]$ and $x_S \in \{0,1\}^S$:
\[
\norm{A_{x_S}-\Pi_{x_S}}_\infty \leq \frac{D+q}{d-|S|}.
\]
\end{lemma}
\begin{proof}
    We prove the case where $S=\emptyset$. The general statement follows from applying this argument to the links of $X_\mu$. With this in mind, observe that by definition:
        \begin{align*}
    \norm{A_\emptyset-\Pi}_\infty =& \max_{(u,i)} \sum\limits_{(v,j) \in X_{(u,i)}} \left|\frac{1}{d-1}\Pr[x_i=u|x_v=j] - \frac{1}{d}\Pr[x_i=u]\right|\\
    &\leq \frac{1}{d}\max_{(u,i)} \left(\sum\limits_{(v,j) \in X_{(u,i)}} \left|\Pr[x_i=u|x_j=v] - \Pr[x_i=u]\right| + \frac{1}{d-1}\Pr[x_i=u|x_j=v]\right)\\
    &\leq\frac{q}{d}+ \frac{1}{d}\max_{(u,i)} \sum\limits_{(v,j) \in X_{(u,i)}} \left|\Pr[x_i=u|x_j=v] - \Pr[x_i=u]\right|\\
    &=\frac{q}{d}+ \frac{1}{d}\max_{(u,i)}\sum\limits_{(v,j) \in X_{(u,i)}}\left|\Pr[x_i=u|x_j=v] - \Pr[x_i=u]\right|\\
    &\leq \frac{D+q}{d}
    \end{align*}
\end{proof}
Second, we'll need the slightly more involved fact from \cite{Hopkins2024} that any $\ell_\infty$-local expander is global. We reproduce the proof here for completeness.
\begin{lemma}\label{lem:inf-local-to-swap}
    Let $X$ be a $d$-\maximal $\lambda$-$\ell_\infty$-local expander. Then $X$ is $\frac{d^2}{8}\lambda$-global.
\end{lemma}
\begin{proof}
    Recall that it is sufficient to bound the infinity norm:
    \[
    \norm{S_{\frac{d}{2},\frac{d}{2}} - \Pi_{\frac{d}{2},\frac{d}{2}}}_\infty \leq \frac{d^2}{4}\lambda.
    \]
    We prove a slightly stronger statement by induction. For any $\tau \in X$ and feasible $i,j$ both the $\infty$-norm and $1$-norm are bounded by:
    \[
    \norm{S^\tau_{i,j} - \Pi^\tau_{i,j}}_1, \norm{S_{i,j} - \Pi_{i,j}}_\infty \leq ij\lambda.
    \]
    We induct on $i+j$. For the base case $i=j=1$, first observe the $\infty$-norm is exactly the $\lambda$-$\ell_\infty$-expansion of the $1$-skeleton. A classical consequence of H\"{o}lder duality states that for any H\"{o}lder conjugates $(p,q)$ and operator $M$
    \[
    \norm{M}_p = \norm{M^*}_q,
    \]
    where $M^*$ is $M$'s adjoint. Since $S_{1,1} - \Pi_{1,1}$ is self-adjoint we therefore have
    \[
    \norm{S^\tau_{0,0}-\Pi^\tau_{0,0}}_1 = \norm{S^\tau_{0,0}-\Pi^\tau_{0,0}}_\infty \leq \lambda
    \]
    as desired.
    
    Assume now by induction that, for some fixed $i+j>0$, all $i'+j'<i+j$ and $\tau \in X$ satisfy
    \[
    \norm{S^{\tau}_{i',j'} - \Pi^{\tau}_{i',j'}}_\infty \leq i'j'\lambda.
    \]
    Let $p \in \{1,\infty\}$ and let $\bar{p}$ denote the complement of $p$. We first argue we may assume without loss of generality that $i>1$. This is again by H\"{o}lder duality since
    \[
    \norm{S_{i,j}-\Pi_{i,j}}_p=\norm{S_{j,i}-\Pi_{j,i}}_{\bar{p}}.
    \]
    Thus we are done if we show the result for both $p \in \{1,\infty\}$ just for the case $i\geq j$ (and thus $i>1$).

    Toward this end, fix any $f: \X[j] \to \R$. As is typically the case, the idea is to localize the input by first drawing an $(i-1)$-face, then a vertex from its link:
    \begin{align*}
        \norm{S_{i,j}f - \Pi_{i,j}f}_p = \Bigg|\Bigg|\norm{(S_{i,j}f)|_{s} - (\Pi_{i,j}f)|_{s}}_{p,v \in \SC[X][1][s]}\Bigg|\Bigg|_{p,s \in \X[i-1]}
    \end{align*}
    where for any $g \in \X[i] \to \R$, $g|_{s}(v)=g(s \cup v)$ denotes the localization of $g$ to the link of $s$. For concreteness, we write out the above explicitly for the $p=1$ case in slightly less compact notation. The $p=\infty$ case is the same replacing expectations with absolute maxima:
    \begin{align*}
        \norm{S_{i,j}f - \Pi_{i,j}f}_p &= \underset{t \in \X[i]}{\mathbb{E}}[|S_{i,j}f(t)-\Pi_{i,j}f(t)|]\\
        &=\underset{{s \in \X[i-1]}}{\mathbb{E}}\Bigg[\underset{{v \in \SC[X][1][s]}}{\mathbb{E}}[|S_{i,j}f(s \cup v)-\Pi_{i,j}f(s \cup v)|]\Bigg]\\
        &=\underset{s \in \X[i-1]}{\mathbb{E}}\Bigg[\underset{v \in \SC[X][1][s]}{\mathbb{E}}[|S_{i,j}f|_s(v)-\Pi_{i,j}f|_s(v)|]\Bigg]\\
        &=\mathbb{E}_{s \in \X[i-1]}\Bigg[\mathbb{E}_{v \in \SC[X][1][s]}[|S_{i,j}f|_s(v)-\Pi_{i,j}f|_s(v)|]\Bigg]\\
        &=\mathbb{E}_{s \in \X[i-1]}\Bigg[\Bigg|\mathbb{E}_{v \in \SC[X][1][s]}[|S_{i,j}f|_s(v)-\Pi_{i,j}f|_s(v)|]\Bigg|\Bigg]\\
        &=\Bigg|\Bigg|\norm{(S_{i,j}f)|_{s} - (\Pi_{i,j}f)|_{s}}_{p,v \in \SC[X][1][s]}\Bigg|\Bigg|_{p,s \in \X[i-1]}.
    \end{align*}
    We stick to the more compact norm notation for the remainder of the proof. The trick is now to observe that as a function of $\SC[X][1][s]$, $S_{i,j}f|_s$ is exactly $S_{1,j}^sf^s$, where $f^s: \SC[X][j][s] \to \R$ is the restriction $f^s(\tau)=f(\tau)$. Then by adding and subtracting the corresponding local stationary operator we have:
    \begin{align*}
         \norm{S_{i,j}f - \Pi_{i,j}f}_p &= \Bigg|\Bigg|\norm{S^s_{1,j}f^s - \Pi^s_{1,j}f^s + \Pi^s_{1,j}f^s - (\Pi_{i,j}f)|_{s}}_{p,v \in \SC[X][1][s]}\Bigg|\Bigg|_{p,s \in \X[i-1]}\\
         &\leq \Bigg|\Bigg|\norm{S^s_{1,j}f^s - \Pi^s_{1,j}f^s}_{p,v \in \SC[X][1][s]}\Bigg|\Bigg|_{p,s \in \X[i-1]} + \Bigg|\Bigg|\norm{\Pi^s_{1,j}f^s - (\Pi_{i,j}f)|_{s}}_{p,v \in \SC[X][1][s]}\Bigg|\Bigg|_{p,s \in \X[i-1]}
    \end{align*}
    by the triangle inequality. The first term is now bounded by the inductive hypothesis applied in the link of $s$:
    \[
    \Bigg|\Bigg|\norm{S^s_{1,j}f^s - \Pi^s_{1,j}f^s}_{p,v \in \SC[X][1][s]}\Bigg|\Bigg|_{p,s \in \X[i-1]} \leq \Bigg|\Bigg|j\lambda\norm{f^s}_{p,\SC[X][j][s]}\Bigg|\Bigg|_{p,s \in \X[i-1]} = j\norm{f}_{p,\X[j]}
    \]
    Toward analyzing the second term, observe that $\Pi^s_{1,j}f^s = S_{i-1,j}f(s)$ and $(\Pi_{i,j}f)|_s(v)=\Pi_{i-1,j}f(s)$ so:
    \[
    \Bigg|\Bigg|\norm{\Pi^s_{1,j}f^s - (\Pi_{i,j}f)|_{s}}_{p,v \in \SC[X][0][s]}\Bigg|\Bigg|_{p,s \in \X[i-1]} = \norm{S_{i-1,j}f - \Pi_{i-1,j}f}_{p,s \in \X[i-1]} \leq (i-1)j\lambda \norm{f}_{p,\X[j]}
    \]
    by the inductive hypothesis. Altogether this gives
    \[
    \norm{S_{i,j}f - \Pi_{i,j}f}_p \leq \lambda j + \lambda (i-1)j = ij\lambda
    \]
    as desired.
\end{proof}
The proof of \pref{prop:ind-to-global} is now essentially immediate.
\begin{proof}[Proof of \pref{prop:ind-to-global}]
    The first fact is immediate from \pref{lem:ind-to-local} and the fact that $\norm{A_\tau-\Pi_\tau}_2 \leq \norm{A_\tau - \Pi_\tau}_\infty$. The second fact is an immediate consequence of combining \pref{lem:ind-to-local} and \pref{lem:inf-local-to-swap}.
\end{proof}

\bigskip

We conclude the subsection by the observation that if a complex \(X\) is global, then so is its faces complex (\pref{def:faces-complex}). Since the faces complex also inherits the two-sided expansion of the original complex up to factors in \maxsizity (c.f.\ the proof of \pref{thm:hdx-is-sampler}), this shows that one can also take as examples faces complexes of any of the examples above.
\begin{claim} \label{claim:globalness-of-faces-complex}
    Let \(X\) be a \(k\)-\maximal and \(\lambda\)-global complex. Let \(\ell\) be such that \(\ell\) divides \(k\), then the \(\ell\)-faces complex \(F^\ell_X\) is \(\lambda\)-global.
\end{claim}

\begin{proof}
    Let \(m=\frac{k}{\ell}\) for notational convenience. Fix \(t = \set{s_1,s_2,\dots,s_{\frac{m}{2}}} \in F^\ell_{X}(\frac{m}{2})\). Observe that the set of neighbors of \(t\) with respect to the swap walk in \(F^\ell X\) are exactly those \(t'\) such that \(\cup t\) and \(\cup t'\) are neighbors in the swap walk on \(\X[\frac{k}{2}]\) (here \(\cup t = s_1 \cup s_2\cup \dots \cup s_{\frac{m}{2}}\) and similarly for \(t'\)). Moreover, one samples a neighbor of \(t\) by sampling a random neighbor \(\tau\) of \(\cup t\), and then sampling a uniform at random partition of \(\tau\) to \(t' = \set{s_1',s_2',\dots,s_{\frac{m}{2}}}\) (the partition is chosen independent of \(t\)).

    Thus one observes that the \(TV\)-distance between \(S_t\) and \(\Pi\) can be written as
    \begin{align*}
        \frac{1}{2}\sum_{t' \in F^\ell \X[\frac{m}{2}]} |\Prob[S_t(F^\ell X)]{t'} - \Prob[\Pi]{t'}| &= \frac{1}{2}\sum_{\tau \in \X[\frac{k}{2}]} \sum_{t' : \cup t' = \tau} |\Prob[S_t(F^\ell X)]{t'} - \Prob[\Pi]{t'}| \\
        &= \frac{1}{2}\sum_{\tau \in \X[\frac{k}{2}]} \sum_{t' : \cup t' = \tau} p|\Prob[S_{\cup t}(X)]{\tau} - \Prob[\Pi(X)]{\tau}|
    \end{align*}
    where \(p\) is one over the number of partitions of \(\tau\) to \(t'\). Obviously this is equal to
    \[
    \frac{1}{2}\sum_{\tau \in \X[\frac{k}{2}]} \Prob[S_{\cup t}(X)]{\tau} - \Prob[\Pi(X)]{\tau} =d_{TV}(S_{\cup t}(X),\Pi(X)) \leq \lambda
    \]
    as desired.
\end{proof}

\section{Analytic, Geometric, and Combinatorial Applications}\label{sec:an-comb-apps}
In this section, we cover several brief applications of sampling and reverse hypercontractivity in well-studied combinatorial and analytic settings.
\subsection{New Double Samplers} \label{sec:double-samplers}
Double samplers are a strengthened notion of sampling introduced in \cite{DinurK2017,DinurHKLT2018} that consist of two ``stacked'' samplers with additional local sampling properties under `closure' of the top layer. These interesting objects have powerful applications in agreement testing \cite{DinurK2017} and list-decoding \cite{DinurHKLT2018}, and have even seen algorithmic use in the construction of space-efficient data structures for the heavy hitter problem \cite{DoronW2022}.
\begin{definition}[Double sampler, \cite{DinurHKLT2018}]
A double sampler consists of a triple \((V_2,V_1,V_0)\), where \(V_0\) is the ground set, \(V_1\) is a collection of \(i\)-subsets of \(V_0\) and \(V_2\) is a collection of \(k\)-subsets of \(V_0\), where \(k>i \in \mathbb{N}\). We say that \((V_2,V_1,V_0)\) is an \((\varepsilon,\beta,\varepsilon_0,\beta_0)\)-\emph{double sampler} if
\begin{itemize}
	\item The inclusion graphs on \((V_2,V_1)\) and \((V_2,V_0)\) are \((\varepsilon,\beta)\)-additive samplers and the inclusion graph on \((V_1,V_0)\) is an \((\varepsilon+\varepsilon_0,\beta+\beta_0)\)-additive sampler\footnote{The weights over edges in \((V_i,V_j)\) are the marginals of the probabilistic experiment where we first choose \(T \in V_2\) (according to some given weight distribution $\Pi$), and then choose \(v \in V_0\) and \(S \in V_1\) such that \(v \in S \subseteq T\), uniformly at random over all such pairs \(S,v\).} (recall the inclusion graph is defined by connecting two subsets by an edge if one contains the other).
	
	\item For every \(T\in V_2\), let \(V_1(T)= \{ S\in V_1\;:\; S\subset T\}\) be the sets in \(V_1\) that are contained in \(T\). Let \(G_{|T}\) be the bipartite inclusion graph connecting elements in \(T\) (viewed as elements in the ground set \(V_0\)) to subsets in \(V_1(T)\). We require that for every \(T\in V_2\), the graph \(G_{|T}\) is an \((\beta_0,\varepsilon_0)\)-additive sampler.
\end{itemize}
\end{definition}
Double samplers are only known to arise from high dimensional expanders. In their original work, \cite{DinurK2017,DinurHKLT2018} use the Ramanujan complexes \cite{LubotzkySV2005b} to construct explicit double samplers for all $(\varepsilon,\beta,\varepsilon_0,\beta_0)$:
\begin{theorem}[{\cite[Theorem 2.11]{DinurHKLT2018}}]\label{thm:original-double-samplers}
    For every $(\varepsilon,\beta,\varepsilon_0,\beta_0)>0$, there exists an explicit infinite family of double samplers $\{(V_0^{(n)},V_1^{(n)},V_2^{(n)})\}$ such that for all $n$:
    \[
    \frac{|V_1^{(n)}|}{|V_0^{(n)}|}, \frac{|V_2^{(n)}|}{|V_0^{(n)}|} = \exp\left(\poly\left(\frac{1}{\beta},\frac{1}{\beta_0},\frac{1}{\varepsilon},\frac{1}{\varepsilon_0}\right)\right).
    \]
\end{theorem}
One of the main open questions posed in \cite{DinurHKLT2018} is to settle the overhead of double samplers. In particular, while it is known that exponential dependence on $\varepsilon_i$ is necessary (even for standard samplers), it is plausible the dependence on the `failure probability' $\beta_i$ could be improved, perhaps even to match the corresponding polynomial dependence of optimal standard samplers. Since typical related applications (e.g.\ for low soundness agreement tests and PCPs \cite{ImpagliazzoKW2012,DinurL2017}) take $\varepsilon$ to be constant (or at most some polylog in $\beta$), this latter dependence is where the crux of the problem lies. 

To the best of our knowledge, there have been no lower bounds or improvements over \pref{thm:original-double-samplers} since \cite{DinurHKLT2018}'s original work. Leveraging our concentration bounds and their corresponding optimality (see \pref{sec:optimality}) we take a significant step toward resolving this problem: quasi-polynomial size double samplers exist and, under reasonable assumptions on the underlying complex, are the best possible.
\begin{theorem}[Quasi-Polynomial Double Samplers]\label{thm:double-samplers}
    for every $(\varepsilon,\beta,\varepsilon_0,\beta_0)>0$, there exists an explicit family of ($\beta,\varepsilon$),$(\beta_0,\varepsilon_0)$ double samplers $\{(V^{(n)}_2,V^{(n)}_1,V^{(n)}_0)\}$ of size at most:
\[
\frac{|V^{(n)}_1|}{|V^{(n)}_0|}, \frac{|V^{(n)}_2|}{|V^{(n)}_0|} \leq \exp\left(\widetilde{O}\left(\min\left\{\frac{\log^4\frac{1}{\beta}\log^4\frac{1}{\beta_0}}{\varepsilon^{8}\varepsilon_0^{4}},\frac{\log^3\frac{1}{\beta}\log^3\frac{1}{\beta_0}}{\varepsilon^{16}\varepsilon_0^{4}} \right\}\right)\right).
\]
Moreover, for $\varepsilon \in (0,0.01)$, if the underlying complex family $\{X_n\}$ is $\frac{\varepsilon}{3}$-hitting, vertex-uniform, and non-contracting,\footnote{Here we mean 1) $\pi_0(v)=|\X[1]|^{-1}$ for all $v \in \X[1]$, and 2) $|\X[i]|\geq |\X[i-1]|$ for all $2 \leq i \leq d$.} this is optimal up to polynomial factors in the exponent:
    \[
    \frac{|V_2^{(n)}|}{|V_0^{(n)}|} \geq \exp\left(\Omega\left(\frac{\log(\frac{1}{\beta})\log(\frac{1}{\beta_0})}{\varepsilon^2\varepsilon_0^2}\right)\right)
    \]
\end{theorem}
Note the above does not fully rule out the existence of simplicial poly-size double samplers. While the assumption on regularity in the lower bound can be easily relaxed, it is unclear to what extent the assumption on hitting-set is an artifact of our proof technique (see \pref{sec:optimality} for details), though it is certainly a natural property for such objects to satisfy. Moreover, our bound strongly relies on the assumption that $|V_2| \geq \exp(d)|V_0|$ (implied by $X$ being non-contracting), which is not generically true. Nevertheless, high dimensional expanders exhibit both good hitting set behavior (see \pref{prop:hitting-set}) and at least exponential degree (see \pref{prop:deg-spectral}), so it is at least accurate to say double samplers arising from HDX are at best quasi-polynomial in $\beta^{-1}$.

The proof of \pref{thm:double-samplers} relies on two further results. The first, toward our upper bound, is Lubotzky, Samuels, and Vishne's \cite{LubotzkySV2005a} classical construction of Ramanujan complexes.
\begin{theorem}[LSV-Complexes {\cite{LubotzkySV2005a}}]\label{thm:coset-complexes}
    For any $\lambda>0$ and $d \in \mathbb{N}$, there exists an explicit infinite family of $d$-partite one-sided $\lambda$-local-spectral expanders $\{X^{(n)},\Pi^{(n)}\}$ with degree $\text{deg}(X) \leq \lambda^{-O(d^2)}$. Moreover, for every $d \geq k > \ell \geq 0$, $\max\{\pi_k\} \leq o_n(1)$.
\end{theorem}
The second, toward the lower bound, is the following optimality result we prove in \pref{sec:optimality}.
\begin{corollary}[Corollary of {\pref{thm:optimal-sampling}}]\label{cor:sampling-for-hitters}
    For any $\varepsilon \in (0,0.01)$, let $X$ be a $d$-\maximalpunc, $\frac{\varepsilon}{3}$-hitting vertex-uniform complex and $i \leq k \leq d$. Then there exists a set $A \subset \X[i]$ such that:
    \[
    \Pr_{s \in \X[k]}[|U_{i,k}1_A(s) - \mu_A| > \varepsilon] \geq \exp\left(O\left(\varepsilon^2\frac{k}{i}\right)\right).
    \]
\end{corollary}
We are now ready to prove \pref{thm:double-samplers}.
\begin{proof}[Proof of \pref{thm:double-samplers}]
We split into the upper and lower bounds.
    \paragraph{Upper Bound} We prove the two bounds within the minimum separately. For the first, we use $d$-dimensional LSV-complexes with $\lambda=\frac{1}{2d}$ and $d=O\left(\frac{\log^2(1/\beta)\log^2(1/\beta_0)}{\varepsilon^4\varepsilon_0^2}\right)$ for top layer, and set $i=\Theta(\frac{\log(1/\beta_0)}{\varepsilon_0^2})$ for the middle layer. By \pref{cor:TD-exponential}, $(V_2,V_1)$ is an $(\varepsilon,\beta)$-sampler for the right choice of constants. Moreover, every closure graph is simply the $(\frac{\log(1/\beta_0)}{\varepsilon_0^2})$-\maximal complex on $k$ vertices, which is negatively correlated and therefore also a $(\varepsilon_0,\beta_0)$-sampler by standard Chernoff-Hoeffding for negatively correlated variables. The overhead bound is then immediate from the degree of LSV complexes in \pref{thm:coset-complexes}.

    For the latter, we will use $d$-dimensional LSV-complexes with $\lambda=2^{-d}$ and $d=O\left(\frac{\log(1/\beta)\log(1/\beta_0)}{\varepsilon^8\varepsilon_0^2}\right)$ for top layer, and set $i=\Theta(\frac{\log(1/\beta_0)}{\varepsilon_0^2})$ for the middle layer. The result then follows from \pref{thm:hdx-is-sampler} and the above.
    
    \paragraph{Lower Bound} Recall that by definition $V_0^{(n)}=\X[1][n]=[n]$, $V_1^{(n)}=\X[i][n]$, and $V_2^{(n)}=X_n(k)$ for some family of simplicial complexes $\{X_n\}$ and set sizes $1 < i < k$. By \pref{cor:sampling-for-hitters}, for any $\varepsilon \in (0,0.01)$, $(\X[k],\X[i])$ is at best an $(\varepsilon,e^{O(\varepsilon^2\frac{k}{i})})$-sampler, therefore any such double sampler requires:
    \[
    \frac{k}{i} \geq \Omega\left(\frac{\log(1/\beta)}{\varepsilon^2}\right).
    \]
    On the other hand, it is also required for every $T \in \X[k]$ that the induced graphs $G|_T=(V_2|_T,V_1|_T)$ are $(\varepsilon_0,\beta_0)$-samplers. By standard sampling lower bounds \cite{CanettiEG1995}, this graph is at best a $(\varepsilon_0,e^{O(\varepsilon_0^2i)})$-sampler, forcing
    \[
    i \geq \Omega\left(\frac{\log(1/\beta_0)}{\varepsilon_0^2}\right).
    \]
    Combining these bounds we have that for any valid choice of $i$ the top \maxsizity $k$ of an $(\varepsilon,\beta,\varepsilon_0,\beta_0)$-double sampler is at least
    \[
    k \geq \Omega\left(\frac{\log(\frac{1}{\beta})\log(\frac{1}{\beta_0})}{\varepsilon^2\varepsilon_0^2}\right).
    \]
    It is therefore enough to argue that the overhead is at least exponential in the \maxsizity $k$:
    \[
    \frac{|V_2|}{|V_0|} \geq 2^{\Omega(k)}.
    \]
    This is true for any non-contracting complex. In particular, since $|\X[k]| \geq |\X[k/2]|$, the average degree of every $(k/2)$-face (and hence every vertex) is at least $2^{\Omega(k)}$.
\end{proof}
We remark that it may also be possible to construct double samplers with smaller overhead by removing the \textit{simplicial} requirement from the definition (see e.g.\ the relaxed notion in \cite{DoronW2022}). It is likely, for instance, that `Grassmannian HDX' \cite{kaufman2021garland,golowich2023grassmannian,dinur2023new} based on subspace structure lead to double-samplers with different parameters-overhead trade-offs.

\subsection{A Degree Lower Bound for HDX}\label{sec:degree}
One of the most classical results in the study of expander graphs is the Alon-Boppana Theorem \cite{nilli1991second}, roughly stating that any family of bounded-degree $\lambda$-expanders must have degree at least $\frac{2}{\lambda^2}$. In higher dimensions, the situation is less understood. Not only should degree scale in some way with the local expansion $\lambda$, but dependence on \textit{dimension} becomes a critical parameter in application. In this section, we give (to our knowledge) the first systematic study of the degree of high dimensional expanders. In particular, we prove that various families of hyper-regular HDX require super-exponential degree and prove a threshold-phenomenon exists at the TD-barrier: there exist (reasonably balanced) constructions of $1$-TD complexes achieving exponential degree at every level, while any (reasonably balanced) $\lambda$-TD complex must have degree $2^{-\Omega(k^2)}$ for any level $k \leq \sqrt{d}$. 

Before applying our stronger concentration bounds, we first look at what one can infer about the degree of HDX directly from spectral methods. The following proposition, though not to our knowledge appearing anywhere in the literature, is the result of applying elementary spectral methods to the problem:
\begin{proposition}[Spectral Lower Bound]\label{prop:deg-spectral}
    Let $X$ be a $d$-\maximal $\{\lambda_i\}$-two-sided local-spectral expander. Then for every $k < d$, the degree of $X$'s $k$-skeleton is at least:
    \[
    \deg^{(k)}(X) \geq \frac{1}{(k-1)!}\prod\limits_{i=1}^{k-1} \frac{1}{\lambda_i}
    \]
\end{proposition}
\begin{proof}
    Observe the $k$-th degree of a vertex $v \in \X[1]$ can be bounded by
    \[
    \frac{1}{(k-1)!}\prod_{i=1}^{k-1} \min_{t \in \X[i]}\{|\X[1][t]|\}
    \]
     simply by counting the number of ways a face $\{v_2,\ldots,v_{k-1}\} \in \SC[X][k][v_1]$ can be sampled by choosing $v_2 \sim \SC[X][1][v_1]$, $v_2 \in \X[1][{\{v_1,v_2\}}]$ and correcting for the $\frac{1}{(k-1)!}$ possible orderings of the face. The result now follows from the elementary fact (see e.g.\ \cite[Lemma A.7]{gaitonde2022eigenstripping}) that any two-sided $\lambda$-expander has at least $\frac{1}{\lambda}$ vertices.
\end{proof}
We remark that one might reasonably think to use Alon-Boppana \cite{nilli1991second} (the classic `optimal' trade-off between expansion and degree) to prove a version of the above for one-sided expansion. This does not work for high dimensional expanders, since their links may have low diameter.

\pref{prop:deg-spectral} implies that any (hyper-regular) $o(\frac{1}{d})$-two-sided HDX must have super-exponential degree but is essentially trivial for complexes near the TD-barrier, under spectral independence, or even for arbitrarily strong partite one-sided HDX. In these cases, even if one is willing to look e.g.\ at polynomial size cutoffs (as is quite common in the literature) \pref{prop:deg-spectral} only gives a slightly super-exponential lower bound:
\[
\deg^{(k)}(X) \geq \exp(\Omega(k\log(k))).
\]
Using concentration, we improve this bound to strongly super-exponential. We start with the partite case.
\begin{theorem}\label{thm:hdx-degree-partite}
    Fix $c \in [0,1]$ and let $X$ be a $d$-\maximal hyper-regular partite $2^{-\Omega(d^c)}$-HDX. Then either:
        \begin{enumerate}
        \item $X$'s 1-skeleton is dense:\footnote{While this condition may seem strange, note it is necessary to handle dense settings. Take, e.g., $\{0,1\}^n$, which is a hyper-regular $0$-local-spectral expander whose $\sqrt{n}$-skeleton has degree $\exp(\sqrt{n}\log(n))$.}
        \[
        \deg^1(X) \geq \frac{|\X[1]|}{2}
        \]
        \item $X$ has super-exponential degree:
        \[
        \deg^{(d)}(X) \geq \exp(\Omega(\max\{d^{2c},d\}))
        \]
    \end{enumerate}
\end{theorem}
Second we prove an analogous bound for skeletons of TD and SI-complexes.
\begin{theorem}\label{thm:hdx-degree-TD}
Let $\lambda,c \in [0,1)$, $\eta>0$, and $d \in \mathbb{N}$. If $X$ is a $d$-\maximal hyper-regular $\lambda$-TD or $\eta$-SI complex then for any $k \leq d^c$:
    \begin{enumerate}
        \item $X$'s 1-skeleton is dense:
        \[
        \deg^1(X) \geq \frac{|\X[1]|}{2}
        \]
        \item $X$'s $k$-skeleton has super-exponential degree:
        \[
        \deg^{(k)}(X) \geq \exp(\Omega_{\lambda,\eta}(\min\{k^{2},k^{1/c}\}))
        \]
    \end{enumerate}
\end{theorem}
The latter setting above is tight in the sense that when $\lambda=1$, that is \textit{at} the TD-barrier, there exist (reasonably balanced) complexes where $\deg^{(k)}(X)=\exp(O(k))$ for all $k\leq d$. On the other hand, beyond the TD barrier the best known degree upper bound for a hyper-regular (or even reasonably balanced) HDX is $\exp(\exp(d))$ \cite{FriedgutI2020}, leaving a substantial gap with the above. Nevertheless, \pref{thm:hdx-degree-TD} is the first to exhibit quantitative threshold behavior for degree at the TD-barrier.

\pref{thm:hdx-degree-partite} and \pref{thm:hdx-degree-TD} leave open two critical questions:
\begin{enumerate}
    \item Can we prove super-exponential bounds for the \textit{top} level of a $\lambda$-TD/$\eta$-SI complex?
    \item Can the assumption of hyper-regularity be dropped?
\end{enumerate}
As we will see in the proof of \pref{thm:hdx-degree-TD}, the first question would be resolved if one could give \textit{any} asymptotic improvement over exponential concentration for weak HDX. In other words, any bound of the form:
\[
\Pr_{s \in \X[d]}\left[\frac{1}{d}|s \cap A| - \mu > \frac{1}{2}\right] \leq \exp(\omega(\sqrt{d}))
\]
for $X$ and its links would imply a super-exponential degree lower bound. If, for instance, we have a bound of \(\exp(-\Omega(d))\) holds, this would imply a degree lower bound of $\exp(\Omega(d^2))$ on the top level of $X$.

Regarding the second question, our proof technique extends to any complex whose links are reasonably `balanced' in the sense that there should not be a $2^{-\Omega(\sqrt{d})}$ fraction of vertices making up an $\Omega(1)$ fraction of the mass. Unfortunately, many classical HDX constructions, e.g.\ the Ramanujan complexes of \cite{LubotzkySV2005a}, actually are highly unbalanced in this sense. The links of these constructions essentially approximate an `unbalanced product': they are partite complexes which are regular on each individual part, but many of the parts are vanishingly small and therefore make up a significant portion of the mass. While known constructions of this type all have minimum part size at least exponential in $d$ (and therefore result in degree $\exp(d^2)$), it is unclear how to rule out a construction whose links look, e.g., like $\{0,1\}^{d} \times [2^d]$. Such a complex could potentially satisfy optimal inclusion sampling while maintaining exponential degree. 

We leave the possible extension of our lower bounds to the top level of unbalanced complexes (or conversely the construction of highly locally unbalanced exponential-size HDX) as one of the main open problems suggested by this work:

\begin{question}[Lower Bounds for Locally Unbalanced HDX]\label{question:degree}
    Let $\{X_n\}$ be an infinite family of bounded-degree HDX that are either $\lambda$-TD or $\eta$-SI. Is the degree of every (sufficiently large) $X_n$ super-exponential:
    \[
    \text{deg}(X_n) \geq \exp(\omega(d))?
    \]
\end{question}
\pref{question:degree} is critical from the standpoint of applying HDX in other areas of theoretical computer science where degree dependence corresponds to the `overhead' incurred by using the complex as a gadget. For instance, the lower bound in \pref{thm:double-samplers} can be achieved by HDX essentially if and only if \pref{question:degree} is false. Similarly, a variant of \pref{question:degree} where $d$ may depend on $n$ is critical to the study of agreement testers whose soundness scale with the \emph{number of vertices} of the system. Constructing polynomial size testers with inverse polynomial soundness is the combinatorial core of the so-called \emph{sliding scale conjecture} in PCP theory. If the above lower bound holds, HDX cannot give such objects: $d$ must be taken to be $o(\log(n))$ to maintain polynomial size, and the resulting soundness of any HDX-based poly-size agreement tester would be at best $2^{-o(\log(n))}$, missing the inverse polynomial mark. 

We now move to the proofs of \pref{thm:hdx-degree-partite} and \pref{thm:hdx-degree-TD} which are inspired by our \textit{localization} technique for proving concentration that samples \(s= \set{v_1,v_1,\dots,v_d}\) by iteratively sampling \(v_0 \in X\), then \(v_2 \in X_{v_1}\), \(v_3 \in X_{\set{v_1, v_2}}\)... and analyzing the local concentration at each step. Upon further inspection, this technique yields the fact that the $1$-skeleton of a complex with good sampling in $0$ and $1$-links must itself be an excellent sampler graph.
\begin{lemma} \label{lem:link-sampler}
    Fix \(\varepsilon,\delta_1,\delta_2 > 0\). Let \(X\) be a \(d\)-\maximal simplicial complex such that:
    \begin{enumerate}
        \item $(\X[d],\X[1])$ is a \((\frac{\varepsilon}{2}, \delta_1)\)-additive sampler.
        \item For every \(v \in \X[1]\), \((\SC[X][d][v],\SC[X][1][v])\) is an \((\frac{\varepsilon}{2},\delta_2)\)-additive sampler.
    \end{enumerate}
    Then the partite double cover of $X$'s $1$-skeleton is a \((\varepsilon,\frac{2\delta_1}{1-\delta_2})\)-sampler. 
\end{lemma}

\begin{proof}
    Fix \(A \subseteq \X[1]\). At a high level, the proof follows from observing that if the base graph is a poor sampler, then after drawing the first vertex $v_1$ of our $d$-set $t=\{v_1,\ldots,v_{d}\}$, there is a non-trivial chance that the neighborhood of $v_1$ (i.e.\ $\SC[X][1][v_1]$) `sees' $A$ in the wrong proportion. However, since the remainder of the face $\{v_2,\ldots, v_{d}\} \in \SC[X][d][v_1]$, if $(\SC[X][d][v_1],\SC[X][1][v_1])$ is indeed a good sampler it will maintain this incorrect proportion from the first step and violate sampling of $A$ in $(\X[d],\X[1])$.

    We now formalize this argument. Given our assumptions we'd like to show:
    \[
    \Prob[v \in {\X[1]}]{\Prob[u \in {\SC[X][1][v]}]{A} > \prob{A} + \varepsilon} \leq \frac{\delta_1}{1-\delta_2}.
    \]
    The lower tail follows from applying this on the complement set. Toward this end, denote by $B$ the set of `bad' vertices in $A$ that over-sample the base graph:
    \[
    B=\sett{v \in \X[1]}{\Prob[u \in {\SC[X][1][v]}]{A} > \prob{A} + \varepsilon}.
    \]
    We need to show $\Pr[B] \leq \frac{\delta_1}{1-\delta_2}$. We bound $\Pr[B]$ by relating it to the set of bad $d$-faces:
    \[
    C = \sett{t \in \X[d]}{\abs{\Prob[u \in t]{A} - \prob{A}} > \frac{\varepsilon}{2}},
    \] 
    which has measure at most $\Pr[C] \leq \delta_1$ by Condition (2). The key to relate $B$ and $C$ is to observe we can sample a random $d$-face $t$ by first sampling a vertex $v \in \X[1]$, then sampling $t$ from its link:
    \[
    \prob{C} = \Ex[t \in {\X[d]}]{\one_C(t)} = \Ex[v \in {\X[1]}]{\Ex[t \in {\SC[X][d][v]}]{\one_C(t)}} \geq \prob{B}\Ex[v \in {\X[1]}]{\Ex[t \in {\SC[X][d][v]}]{\one_C(t)}~\Bigg|~v \in B}.
    \]
    When \(v \in B\), we have \(\Prob[u \in {\SC[X][1][v]}]{A} \geq \prob{A} + \varepsilon\) by assumption, but then by Condition (1)
    \[
    \Ex[v \in {\X[1]}]{\Ex[t \in {\SC[X][d][v]}]{\one_C(t)}~\Bigg|~v \in B} \geq \Prob[t \in {\SC[X][d][v]}]{\cProb{}{A}{t} > \Prob[u \in {\SC[X][1][v]}]{A} - \frac{\varepsilon}{2}} \geq 1-\delta_2.
    \]
Putting everything together we have
    \[
    (1-\delta_2) \prob{C} \leq \prob{B} \leq \delta_1.
    \]
    and \(\prob{C} \leq \frac{\delta_1}{1-\delta_2}\) as desired.
\end{proof}
Combining this fact with standard degree lower bounds for sampler graphs \cite{CanettiEG1995} gives the following degree lower bound for the 1-skeletons of such complexes: 

\begin{corollary} \label{cor:sample-degree-bound}
    Let \(X\) be a $d$-\maximal simplicial complex satisfying:
    \begin{enumerate}
        \item $(\X[d],\X[1])$ is a \((\frac{1}{10},\delta_1)\)-sampler
        \item For all $v \in \X[1]$, $(\SC[X][d][v],\SC[X][1][v])$ is a \((\frac{1}{10},\delta_2)\)-sampler.
    \end{enumerate}
    Then $X$'s $1$-skeleton has high max-degree:
    \[
    \exists v\in \X[1]: \deg^{(1)}(v) \geq \min\left\{\frac{|\X[1]|}{2},\frac{2(1-\delta_2)}{5\delta_1}\right\}
    \]
\end{corollary}

\begin{proof}
    By \pref{lem:link-sampler} the underlying graph of \(X\) is a \((\frac{1}{5},\delta = \frac{2\delta_1}{1-\delta_2})\)-sampler. \cite[Theorem 2]{CanettiEG1995} states the following relation on a graph \(G=(L,R,E)\) that is an \((\varepsilon,\delta)\)-sampler of max degree \(t \leq \frac{|\X[1]|}{2}\).
    \[t|R| \geq \frac{|L|(1-2\varepsilon)}{2\delta}.\]
    Here \(L=R=\X[1]\) and this inequality gives \(t \geq \frac{0.8}{2 \delta} = \frac{2(1-\delta_2)}{5\delta_1}\).
\end{proof}
Notice this bound is substantially stronger than what can be inferred from spectral expansion alone. For instance, when $X$ is a \(\lambda = \poly(d^{-1})\)-HDX, Alon-Boppana only implies a degree lower bound of $\Omega(\lambda^{-2})=\poly(d^{-1})$, whereas combining the above with \pref{cor:TD-exponential} gives degree at least $\exp(-\sqrt{d})$. Given this fact, \pref{thm:hdx-degree-partite} and \pref{thm:hdx-degree-TD} follows by recursive application of \pref{cor:sample-degree-bound} on each link of the HDX.

\begin{proof}[Proof of \pref{thm:hdx-degree-partite} and \pref{thm:hdx-degree-TD}]
    Similar to the spectral argument, since $X$ is hyper-regular the $k$th degree of any vertex $v_0 \in \X[1]$ can be bounded by:
    \[
    \deg^{(k)}(v_0) \geq \frac{1}{(k-1)!}\prod\limits_{i=1}^{k-1}\min_{t \in \X[i]}\left\{\deg^{(1)}(X_t)\right\}.
    \]
    It is therefore sufficient to bound the degree of the 1-skeleton of the links of $X$. We split the analysis into the partite and TD/SI cases.
    \paragraph{Partite Case}
    Let $k' = \max\{\sqrt{d},d^c\}$. By \pref{thm:sampling-like-complete}, there exists a universal constant $c'>0$ such that for every $t \in X^{\leq \frac{d}{2}}$, the link $X_t$ satisfies:
    \begin{enumerate}
        \item $(\X[d-|t|][t],\X[1][t])$ is a $(\frac{1}{10}, \exp(-c'k'))$-sampler
        \item For all \(v \in \SC[X][1][t]\), $(\SC[X][d-|t|-1][t \cup \set{v}],\X[1][t \cup \set{v}])$ is a $(\frac{1}{10},  \exp(-c'k'))$-sampler.
    \end{enumerate}
     Setting $t \in \X[1]$, \pref{cor:sample-degree-bound} implies
    \[
    \deg^{(1)}(X) \geq \min\left\{ \frac{|\X[1]|}{2},\exp(-c'k')\right\}.
    \]
    Assume $\exp(-c'k') \leq \frac{|\X[1]|}{2}$ (else Condition (1) is satisfied and we are done). Under this assumption, we argue by induction that the max degree of any $i$-link is at least:
    \begin{equation}\label{eq:deg-bound}
        \forall i \leq \frac{d}{2}: \; \; \min_{t \in \X[i]}\{\deg^{(1)}(X_t)\} \geq \frac{1}{2^i}\exp(c'k').
    \end{equation}
    We remark the restriction on $i$ is to enforce that links of this level are still \((\frac{1}{10},\exp(-c'k'))\)-samplers.
    Plugging this back into our bound on degree gives $\exp(-\Omega(k'^2)) \geq \exp(-\Omega(\max\{d,d^{2c}\}))$ as desired. 
    
    To prove \pref{eq:deg-bound}, observe the base case $i=1$ is given by assumption. The inductive step simply follows from observing that the minimum number of vertices in any $i$-link is exactly the minimum degree of the 1-skeleton of any $(i-1)$-link. Thus applying \pref{cor:sample-degree-bound}, we have:
    \[
    \forall t \in \X[i]: \deg^{(1)}(X_t) \geq \min \left\{ \frac{|\X[1][t]|}{2},\exp(-c'k')\right\} \geq \frac{1}{2^i}\exp(c'k').
    \]
    as desired.
    \paragraph{TD/SI Case} When $X$ is $\lambda$-TD or $\eta$-SI, we'd like to run the same argument as above but run into a slight problem. We can still bound the $k$th degree of a vertex $v$ by:
    \[
    \deg^{(k)}(v) \geq \frac{1}{(k-1)!}\prod\limits_{i=1}^{k-1}\min_{t \in \X[i]}\left\{\deg^{(1)}(X_t)\right\},
    \]
    but because $X$ satisfies only exponential concentration, \pref{eq:deg-bound} becomes
    \begin{equation*}\label{eq:deg-bound2}
    \forall i \leq \frac{d}{2}: \min_{t \in \X[i]}\{\deg^{(1)}(X_t)\} \geq \frac{1}{2^i}\exp(c'\sqrt{d}).
    \end{equation*}
    This bound is now only non-trivial up to $i \leq O(\sqrt{d})$ levels of $X$, stopping us from proving super-exponential concentration for $X$'s top level. Nevertheless, when $k=d^c$, plugging the above into the degree calculation gives $\exp(d)=\exp(k^{1/c})$ whenever $c \geq 1/2$, and $\exp(k\sqrt{d}) \geq \exp(k^2)$ whenever $c < 1/2$.
\end{proof}

\subsection{Geometric Overlap}
The geometric overlap property is a classical notion of high dimensional expansion (see e.g.\ \cite{lubotzky2018high}) which promises every embedding of the complex has a point in space hit by at least a constant fraction of the faces.
\begin{definition}
    Let \(X\) be a \(d\)-\maximal simplicial complex and let \(k \leq d-1\). Let \(c > 0\). We say that \(X\) has \((k,c)\)-geometric overlap if for every \(\rho:\X[1] \to \mathbb{R}^k\), there exists a point \(q_0 \in \mathbb{R}^k\) such that
    \[\Prob[s \in {\X[d]}]{q_0 \in conv(\rho(s))} \geq c.\]
\end{definition}
Here \(conv(\rho(s))\) is the convex hull of the points of \(\rho(s) = \sett{\rho(v)}{v \in s}\). While we are not aware of any applications of geometric overlap in computer science, it is a popular topic in geometry and topology and related ideas (e.g.\ an overlap variant of the Borsuk-Ulam theorem) have seen recent use in the study of algorithmic stability \cite{chase2023local}.

Boros and F{\"u}redi \cite{BorosF1984} showed that the $2$-dimensional complete complex has geometric overlap, a result which was later generalized to all dimensions by B\'{a}r\'{a}ny \cite{Barany1982}. The notion was extended to general simplicial complexes by Gromov \cite{Gromov2010}, who asked whether there are bounded degree complexes that have this property. Followup work by Fox, Gromov, Lafforgue, Naor, and Pach \cite{FoxGLNP2012} gave an affirmative answer to this question both via random construction and through \cite{LubotzkySV2005b}'s Ramanujan complexes. Finally, a series of works \cite{ParzanchevskiRT2016,Evra2017,oppenheim2018local} extended their ideas to general spectral high dimensional expanders by leveraging high dimensional expander-mixing properties. We note that Gromov also defined a closely related (but stricter) notion of topological overlap \cite{Gromov2010}, which was later studied by \cite{KaufmanKL2014, EvraK2016} among others. We focus only on the geometric setting here.

While prior works have succeeded in showing optimal geometric overlap for \textit{random} bounded-degree complexes \cite[Theorem 1.7]{FoxGLNP2012}, it has been open since Gromov's work to construct such a family \textit{explicitly}. In this subsection, building on the approach of these prior works, we show sufficiently strong high dimensional expanders have optimal overlap, resolving this problem.

To state our result, it is useful to first understand the overlap properties of the complete complex. Let \(c_{k,n}\) be the largest constant so that $\SC[\Delta][k+1][n]$ has \((k,c_{k,n})\)-geometric overlap, and let \(c(k) = \lim_{n \to \infty} c_{k,n}\). By \cite{Barany1982} this limit exists and is positive. Similarly, let \(c_{d,k,n}^p\) be the largest constant so that the complete \(d\)-partite complex with \(n\)-vertices on every side has \((k,c_{d,k,n}^p)\)-geometric overlap, and let \(c_{d,k}^p = \lim_{n \to \infty} c_{d,k,n}^p\). Existence and positivity of this limit follows from \cite{Pach1998}. The former constant, $\SC[c][k]$, is known to be the best possible for bounded degree simplicial complexes \cite[Theorem 1.6]{FoxGLNP2012}.

We show that every sufficiently strong (partite) high dimensional expanders has geometric overlap nearly matching the (partite) complete complex:
\begin{theorem} \label{thm:hdx-geometric-overlap}
        For every $\varepsilon>0$ and $k<d \in \mathbb{N}$, there exist $\lambda,\lambda_p>0$, and $n_0,n_{0,p} \in \mathbb{N}$ such that
    \begin{enumerate}
        \item \textbf{Two-sided case:} If $X$ is \(d\)-\maximal $\lambda$-two-sided high dimensional expander on at least \(n_0\) vertices so that the measure on the vertices is uniform, then $X$ has \((k,c_k-\varepsilon)\)-geometric overlap.
        \item \textbf{Partite case:} If $X$ is $d$-partite $\lambda_p$-one-sided HDX on at least \(n_{0,p}\) vertices so that the measure on the vertices is uniform, then $X$ has \((k,c^p_{d,k}-\varepsilon)\)-geometric overlap.
    \end{enumerate}
\end{theorem}
We comment that this theorem applies to many of the complexes constructed in \cite{LubotzkySV2005b} (for the non-partite setting we need to take lower \maximal skeletons).

The proof of \pref{thm:hdx-geometric-overlap} combines our concentration bounds with tools developed in \cite{FoxGLNP2012}. We have not tried to find the optimal dependence on \(\lambda\) or on the degree of the complex. Finally before moving to the proof, we remark that the result really only relies on the weaker `high-dimensional expander-mixing' type inequality which actually holds for the more general family of splitting-trees such as expander walks (see \pref{sec:splitting}). Thus it is possible to obtain explicit complexes with geometric overlap approaching that of the complete (partite) complex through this more general family as well, though we omit the details.

\paragraph{Homogeneous tuples}
The core of \cite{FoxGLNP2012}'s proof of geometric overlap is the following partitioning theorem. We follow their definitions and presentation precisely. A tuple of subsets $S_1,\ldots,S_{k+1} \subseteq \mathbb{R}^k$ is said to be \emph{homogeneous} with respect to a point $q \in \mathbb{R}^k$ if either:
\begin{enumerate}
    \item All simplices with one vertex in each of the sets $S_{1},\ldots,S_{k+1}$ contain $q$.
    \item None of these simplices contain $q$.
\end{enumerate}

An \emph{equipartition} of a finite set is a partition of the set into subsets whose sizes differ by at most one.

\begin{theorem}[{\cite[Corollary 1.9]{FoxGLNP2012}}]\label{thm:discrete-structure}
    Let $k$ be a positive integer and $\epsilon>0$. There exists a positive integer $K=K(\epsilon,k)\geq k+1$ such that for any $K' \geq K$ the following statement holds. For any finite set $P \subseteq \mathbb{R}^{k}$ and for any point $q \in \mathbb{R}^{k}$, there is an equipartition $P=P_1 \dunion \dots \dunion P_{K'}$ such that all but at most an $\epsilon$-fraction of the $(k+1)$-tuples $P_{i_1},\ldots,P_{i_{k+1}}$ are homogeneous with respect to $q$.
\end{theorem}

Our proof adapts the idea in \cite[Theorem 1.7]{FoxGLNP2012}.
\begin{proof}[Proof of \pref{thm:hdx-geometric-overlap}]
We split the proof into the two-sided and partite cases.
\paragraph{Two-sided Case}
Let \(K = K(\varepsilon/16,k)\) be the constant promised in \pref{thm:discrete-structure} and \(K' \geq K\) be the smallest constant so that for any equipartition of the \(n\)-vertices to \(K'\) parts, the fraction of faces in the complete complex (over the same vertex set) that touch every part at most once is at least \(1-\frac{\varepsilon}{16}\).\footnote{One can verify as in \cite[Theorem 1.7]{FoxGLNP2012} that \(K' = \max \set{K,16(d+1)/\varepsilon}\) for large enough \(n\).} Let \(z_0 =\frac{1}{2K'}\) and let \(n_0\) be such that for any \(n \geq n_0\) and sets \(A_1,A_2,\dots,A_{k+1} \subset [n]\) of measure at least \(\frac{1}{K'}-z_0\) it holds that 
\[
    \Prob[\set{v_1,v_2,\dots,v_{k+1}} \in {\SC[\Delta][k+1][n]}]{\forall j=1,\dots,k+1:~ v_j \in A_j } \geq \prod_{j=1}^{k+1} \prob{A_j} - \frac{\varepsilon}{16}.
\]

Let \(X\) be as in the statement and \(\rho:\X[1] \to \mathbb{R}^k\) an embedding. There is a point \(q_0 \in \R^k\) so that \(q_0\) is contained in \(c_k-\frac{\varepsilon}{16}\) of the faces of $\SC[\Delta][k+1][n]$ embedded via $\rho$. 

Let \(P_1,P_2,\dots,P_{K'}\) be the partition we obtain from \pref{thm:discrete-structure} with respect to \(\rho\) and the complete complex. Let \(H = \sett{\set{P_{i_1},P_{i_2},\dots,P_{i_{k+1}}}}{\forall s \in E(P_{i_1},P_{i_2},\dots,P_{i_{k+1}}), \; q_0 \in conv(\rho(s))}\). Here \(E(P_{i_1},P_{i_2},\dots,P_{i_{k+1}})\) are the faces \(s \in \X[k+1]\) such that for every \(j=1,\dots,k+1\), \(\abs{s \cap P_{i_j}} = 1\). Recall that at most an \(\frac{\varepsilon}{16}\)-fraction of the tuples are not homogeneous and that at most \(\frac{\varepsilon}{16}\) of the faces in $\SC[\Delta][k+1][n]$ touch a part in this partition more than once. 
Thus,
\begin{equation}\label{eq:faces-in-complete-complex}
    \sum_{(P_{i_j})_{j=1}^{k+1} \in H}\prod_{j=1}^{k+1} \prob{P_{i_j}} \geq c_k - \frac{3\varepsilon}{16}.
\end{equation}

On the other hand,
\begin{equation} \label{eq:faces-touching-q0}
    \Prob[s \in {\X[k+1]}]{q_0 \in conv(\rho(s))} \geq \sum_{\set{P_{i_1},P_{i_1},\dots,P_{i_{k+1}}} \in H}{\Prob[s \in {\X[k+1]}]{E(P_{i_1},P_{i_2},\dots,P_{i_{k+1}})}}.
\end{equation}
By \pref{prop:hitting} (applied for indicators of the corresponding functions) for small enough \(\lambda=\lambda(K',\varepsilon,z_0)\), for every \(k+1\) parts \(P_{i_1},P_{i_2},\dots,P_{i_{k+1}}\) of size at least \(\frac{1}{K'}-z_0\),
    \[
        \Prob[\set{v_1,v_2,\dots,v_{k+1}} \in {\X[k+1]}]{\forall j=1,\dots,k+1 \; v_j \in P_{i_j}} \geq (1-\frac{\varepsilon}{2})\prod_{j=1}^{k+1} \prob{P_{i_j}}.
    \]
    Combining \eqref{eq:faces-in-complete-complex} and \eqref{eq:faces-touching-q0} gives
    \(\Prob[s \in {\X[k+1]}]{q_0 \in \rho(s)} \geq (1-\frac{\varepsilon}{2})(c_k-\frac{\varepsilon}{4}) \geq c_k-\varepsilon\) as desired.

\paragraph{Partite Case} The argument is similar. Let \(K = K(\varepsilon/16,k)\) be the constant promised in \pref{thm:discrete-structure}. Let \(K' \geq K\) be the smallest constant so that for any equipartition over \(K'\) parts, the fraction of faces in the complete partite complex that touch every part at most once is at least \(1-\frac{\varepsilon}{16}\).

Let \(z = \frac{1}{2K'}\) be as before. As before we first identify the points of \(X\) with the complete partite complex with the same number of vertices in every side. Then we find a point \(q_0\) such that \(c_{d,k}^{p}-\frac{\varepsilon}{16}\) of the faces of the complete partite complex contain \(q_0\). We find an equipartition \(P_1,P_2,\dots,P_{K'}\) of \(\rho(\X[1])\) to \(K'\) parts such that \(1-\frac{\varepsilon}{16}\) of the parts are homogeneous with respect to \(q_0\).

We partition further every \(P_i\) according to the colors of the vertices to \(P_{i}^j\) for \(j=1,2,\dots,d\). Then we remove from the partition every subset \(P_i^j\) whose relative size is smaller than \(\frac{\varepsilon}{16 K' \cdot d^2}\). We note that by taking only the remaining sets in the partition the fraction of vertices we removed from every color is no more than \(\frac{\varepsilon}{8}\) vertices in any side. This is because for every \(P_i\), colors \(j\) that have that \(\prob{P_i^j} \leq \frac{\varepsilon}{16 K' \cdot d^2}\) can account for no more than \(\frac{\varepsilon}{16 K' \cdot d}\) out of at least \(\frac{1}{K'}-z = \frac{1}{2K'}\). Thus we have not removed more than \(\frac{\varepsilon}{8d}\) vertices in total, or \(\frac{\varepsilon}{8}\)-fraction out of each color.

We now have subsets \(P_i^j\) that such that \(\frac{\varepsilon}{16 K' d^2} \leq \prob{P_i^j} \leq \frac{3}{2K'}\). The proof now follows as in the two-sided case.
\end{proof}

\subsection{Separating MLSI from Reverse Hypercontractivity}
As discussed in \pref{sec:intro}, prior techniques for establishing reverse hypercontractivity relied on tensorization or bounded (Modified) Log-Sobolev constant (MLSI). To our knowledge, all known reverse-hypercontractive objects prior to this work also have bounded MLSI. In this section, we give the first separation between these analytic notions in the case where the leading hypercontractive constant $C$ may be greater than $1$ (or, alternatively, for the indicator variant). To start, recall the definition of MLSI.

\begin{definition}[Modified Log-Sobolev Inequality]
    Let $M$ be a reversible markov chain with stationary distribution $\pi_s$. The modified log-sobolev constant of $M$ is: 
    \[
    \rho_{\text{\tiny{MLSI}}} = \inf_{f > 0}\frac{\langle f, (I-M)\log f\rangle}{2\text{Ent}_\pi(f^2)}
    \]
\end{definition}
MLSI is classically used to bound the \textit{mixing time} of its associated Markov chain,
\[
T_{\text{mix}}(M,\varepsilon) = \min_{t} \max_{\pi_{init}} \norm{\pi_{init}^TM^t - \pi_s}_{TV} \leq \varepsilon,
\]
due to the following upper bound:
\begin{fact}\label{fact:MLSI}
    The mixing time of a reversible Markov chain $M$ with stationary distribution $\pi$ is at most:
    \[
    T_{\text{mix}}(M,\varepsilon) \leq \frac{1}{\rho_{\text{\tiny{MLSI}}}}\log\left( \frac{\log\pi_*^{-1}}{\varepsilon}\right),
    \]
    where $\pi_*$ is the minimum stationary probability.
\end{fact}
High dimensional expanders (at least in the bounded degree setting), are \textit{slow-mixing}, so they cannot have small MLSI constant. In particular, applying this fact to the Ramanujan complexes gives an infinite family of reverse-hypercontractive operators with at least inverse logarithmic dependence on the domain size.
\begin{corollary}\label{cor:MLSI}
    There exist constants $\ell>1$ and $C>0$ and an infinite family $\{T_\rho^{(n)}\}$ of $(\frac{1}{\ell},\frac{\ell}{1-\ell},C)$-reverse hypercontractive operators with at most inverse logarithmic MLSI:
    \[
    \rho_{MLSI} \leq C_2\frac{\log^3\log(\pi_*)}{\log \pi_*}
    \]
\end{corollary}
\begin{proof}
    Set $\rho=\frac{1}{2}$, let $c$ be as in \pref{thm:reverse}, and let $\{X^{(n)}\}$ be the family of $k=k(n)$-\maximal LSV complexes promised in \pref{thm:coset-complexes} for $k=\log\log n$ and $\lambda=2^{-k}$, and recall that
    \begin{enumerate}
        \item The degree of every vertex is at most $2^{C_2k^3}$ for some universal constant $C_2>0$
        \item The stationary distribution of $T_{1/2}$ on $X$ is uniform
    \end{enumerate}
    By \pref{thm:reverse}, there exist (universal) constants $C,\ell$ such that the noise operator $T_{1/2}$ is $(\frac{1}{\ell},\frac{\ell}{1-\ell},C)$-reverse hypercontractive for every complex in the family. On the other hand, we will argue that:
    \[
    T_{\text{mix}}\left(T_{\rho},\frac{1}{2}\right) \geq \min \left\{\frac{2^k}{4},\frac{C'_2\log(n)}{k^3}\right\},
    \]
    which combined with \pref{fact:MLSI} completes the result.

    Define the $\varepsilon$-support of a distribution
    \[
\text{supp}^{\geq}_\varepsilon(\pi) \coloneqq \{x : \pi(x) \geq \varepsilon \}.
    \]
    Setting $\varepsilon=\frac{1}{4n}$, observe that any distribution $\pi$ with $|\text{supp}^{\geq}_{\varepsilon}(\pi)| < \frac{n}{4}$ has $\norm{\pi-\pi_s}_{TV} \geq \frac{1}{2}$. We will argue this holds for $\pi_{\text{init}}T_{\frac{1}{2}}^t $ whenever $\pi_{\text{init}}$ is any indicator and $t$ is sufficiently small.

    To see this, we separate out the bounded-degree and stationary components of $T_{\frac{1}{2}}$ as:
    \[
    T_{\frac{1}{2}}^+ \coloneqq T_{\frac{1}{2}} - 2^{-k}\Pi
    \]
    where $\Pi$ is the stationary operator, and note that $T_{\frac{1}{2}} = T_{\frac{1}{2}}^+ + 2^{-k}\Pi$. Observe that after $t$ steps, we can write any initial distribution $\pi$ as:
    \[
    \pi^T T_{\frac{1}{2}}^t = \pi^T (T^+_{\frac{1}{2}}+2^{-k}\Pi)^t = \pi^T(T^+_{\frac{1}{2}})^t + (1-(1-2^{-k})^t)\pi_s
    \]
    For large enough $k$, as long as $t \leq 2^{-k/4}$, the latter term contributes strictly less than $\frac{1}{4}$ to each coordinate. On the other hand for $\pi$ an indicator, the first term contributes positive weight to at most $deg(T_{\frac{1}{2}^+})^t$ coordinates, so setting $t$ such that $deg(T_{\frac{1}{2}^+})^t \leq \frac{n}{4}$ completes the result.
\end{proof}

\subsection{It Ain't Over Till It's Over}
Friedgut and Kalai's `It Ain't Over Till It's Over' Theorem is a classical result in social choice which states that even after taking a significant random restriction, any balanced function still has some uncertainty with high probability (in other words, it is a tail bound on the random restriction operator). The precursor to this Theorem was a version for the noise operator proved in \cite[Appendix C]{MosselOO2005}. We prove the result holds for any LUS, which implies the same for any $c$-nice complex.
\begin{theorem}\label{thm:tail-bounds}
    For any $\rho \in (0,1)$ and $\tau > 0$, let $X$ be a $d$-\maximal $\tau$-LUS for $d$ sufficiently large. Then there exists a constant $q=q(\rho,\tau)$ such that for any $\delta>0$ and $f:\X[d] \to [0,1]$ with density $\mu=\mathbb{E}[f]$:
    \[
    \Pr[T_\rho f > 1-\delta] \leq c_{\mu,1}\delta^{q}
    \]
    where $c_{\mu,1} = \left(\frac{8(1+\mu)^{q+1}}{(1-\mu)^{q+2}}\right)^{\frac{1}{q}}$, and
    \[
    \Pr[T_\rho f < \delta] \leq c_{\mu,2}\delta^{q}
    \]
    where $c_{\mu,2} = \left(\frac{8(2-\mu)^{q+1}}{\mu^{q+2}}\right)^{\frac{1}{q}}$.
\end{theorem}
\begin{proof}
We prove the former. The latter follows from considering $1-f$. Let $T$ of the set of elements whose neighborhoods are dense in $f$:
    \[
    T \coloneqq \{ x : T_\rho f(x) \geq 1-\delta\},
    \]
    and consider the indicator cut-off $h=\mathbf{1}_{f \leq \frac{1+\mu}{2}}$. We will argue that on the one hand the correlation of $T$ and $h$ cannot be large, since when $x \in T$, $y\sim_\rho x$ must be mostly above the cutoff. On the other hand, the correlation is at least some power of $\mu(T)$ by reverse hypercontractivity, so $\mu(T)$ must be small.

    Formally, observe that for any $x \in T$, Markov's inequality gives that $T_\rho h(x) \leq \frac{2\delta}{1-\mu}$ so
    \[
    \mathbb{E}[1_T T_\rho h] \leq \mu(T)\delta \frac{2}{1-\mu}.
    \]
    On the other hand, we also have by Markov that $\mathbb{E}[h] \geq \frac{1-\mu}{1+\mu}$, so reverse hypercontractivity for indicators (\pref{thm:indicator-reverse-hc}) implies there exists some $\ell$ such that:
    \[
    \mathbb{E}[1_TT_\rho h] \geq \frac{1}{4}\left(\mu(T) \frac{1-\mu}{1+\mu}\right)^\ell
    \]
    Combining the equations and setting $q=\ell-1$ gives the desired result.
\end{proof}
We note we have made no attempt to optimize the constants, and it may be possible to improve dependence on $\mu$. \cite{MosselOO2005} combine a variant of \pref{thm:tail-bounds} with the invariance principle to prove the classical form of the conjecture with respect to random restrictions. It is not clear what the correct form of this conjecture should be for HDX, or even for product spaces.\footnote{We remark it should be possible to give a version of the result for suitably `global' functions via tools of \cite{keevash2021forbidden,GurLL2022,BafnaHKL2022}, but it is not clear such a notion would have significant qualitative meaning} We leave this as an open question.

\subsection{A Frankl–R\"{o}dl Theorem}
The Frankl–R\"{o}dl Theorem \cite{frankl1987forbidden} is a powerful result from extremal combinatorics that states that the independence number of the graph whose vertices are $\{0,1\}^n$ and whose edges are given by fixed weight intersection $k$ (typically for some $k=\Theta(n)$) is at most $\exp(-\Omega(k))|V|$. Benabbas, Hatami, and Magen \cite{Benabbas2012} proved a variant of this Theorem via reverse hypercontractivity on the cube, later leading to several applications in hardness of approximation \cite{KauersOTZ2014}. Their method, which also bounds classical properties such as the chromatic and dominating set numbers, is based on the following claim:
\begin{claim} \label{claim:ind-set}
    Let \(p_0,q_0 > 0\) and let \(G=(V,E)\) be a weighted graph with minimum vertex weight $q_0$ such that for every \(A,B \subset V\) of relative size at least $p_0$:
    \[
    \Prob[(s_1,s_2) \sim E]{s_1 \in A, s_2 \in B} > 0.
    \]
    Then the maximal independent set in \(G\) has size at most \(2p_0+q_0\), the minimal dominating set has relative size at least \(1-(2p_0+q_0)\), and the chromatic number is at least \(\frac{1}{2p_0+q_0}\).
\end{claim}

As a corollary we have that the down-up walks in any of the discussed LUS-complexes in \pref{sec:LUS} have maximal independent sets of relative size at most \(\exp(-\Omega(k))\), a chromatic number of at least \(\exp(\Omega(k))\) and minimal dominating sets of size at least \(1-\exp(-\Omega(k))\). As discussed in the introduction, we remark that these graphs are not quite the correct analog of Frankl–R\"{o}dl, since they have edges with intersection up to some fixed $\gamma k$ instead of $\gamma k$ exactly. Benabbas, Hatami and Magen \cite{Benabbas2012} handle this by analyzing the closeness of the exact intersection versus noise operator on the cube, using tools from Fourier analysis. It would be interesting to see if a variant of this result holds on HDX. If this could be shown generally, then our reverse hypercontractivity theorem would recover Frankl–R\"{o}dl up to constants in the exponent.

\begin{proof}[Proof of \pref{claim:ind-set}]
    We prove the statement for independent sets. The chromatic number statement follows since every coloring is a partition to independent sets, and the statement for dominating sets follows since the complement of any dominating set is an independent set. Let \(I \subseteq G\) be an independent set and assume toward contradiction that  \(\prob{I} \geq 2p_0+q_0\). We partition $I$ into \(I_1,I_2\) such that \(\prob{I_1},\prob{I_2} \geq p_0\) (we can do so by greedily adding vertices to \(I_1\) until its probability is between \(p_0\) and \(p_0+q_0\)). Then by assumption \(\Prob[(s_1,s_2) \sim E]{s_1 \in I_1, s_2 \in I_2} > 0\), which gives the desired contradiction.
\end{proof}

\begin{remark}
    Although the graphs we discussed in this paper have self loops. The claim above continues to hold after removing these self loops, as is evident from the proof.
\end{remark}

\section{Codes and Splitting Trees}\label{sec:codes}
Error correcting codes, and in particular the powerful notions of \textit{list-decoding} and \textit{local-testability}, are among the earliest successful applications of inclusion samplers \cite{impagliazzo2008uniform,impagliazzo2009approximate,ImpagliazzoKW2012} and HDX \cite{DinurHKLT2018,AlevJQST2020,jeronimo2021near,PanteleevK22,DinurELLM2022}. In this section we explore the implications of our tools and related ideas over the weaker family of `splitting trees' in this classical setting. Our main application is the first construction of constant rate codes over large alphabets which simultaneously have 1) near-optimal distance 2) list-decodability and 3) local-testability. At the end of the section we introduce a (conjectural) HDX-based approach toward `lossless' distance-amplification of LTCs that experiences no decay in soundness.
\subsection{Splitting Trees}\label{sec:splitting}
We start with a fairly substantial detour into the world of splitting trees, which are a weakening of high dimensional expanders introduced in \cite{AlevJT2019} which only requires certain patterns of swap walks to expand. We will prove a high-dimensional expander mixing lemma and some basic sampling properties for general splitting trees that will be useful for our applications to coding theory.
\subsubsection{Splitting Preliminaries}
Given a binary tree $T$, let $\mathcal{L}(T)$ denote $T$'s leaves and $\mathcal{I}(T)$ its internal nodes. We drop $T$ from the notation when clear from context. For every internal node $v \in \mathcal{I}$, denote the left and right children of $v$ by $\ell_u$ and $r_u$ respectively.
\begin{definition}[Ordered Binary Tree]
    A $d$-\maximal ordered binary tree is a pair $(T,\rho)$ where \(T\) is a rooted binary tree with \(d\) leaves, and \(\rho:T \to [d]\) is a labeling such that:
    \begin{enumerate}
        \item The label of every leaf \(u \in \mathcal{L}\) is \(\rho(u)=1\).
        \item For every interior \(u \in \mathcal{I}\), it holds that \(\rho(u) = \rho(\ell_u) + \rho(r_u)\).
    \end{enumerate}
We call such a $\rho$ well-ordered with respect to $T$.
\end{definition}
A complex is called splittable if it can be decomposed by an ordered tree where every non-leaf node corresponds to an expanding swap walk.
\begin{definition}[Splitting tree] \label{def:split-tree}
    Let \(\gamma > 0\). A $d$-\maximal \emph{\(\gamma\)-splitting tree} is a triple \((X,T,\rho)\) where \(X\) is a \(d\)-\maximal simplicial complex and \((T,\rho)\) is an ordered binary tree such that for every \(u \in \mathcal{I}(T)\):
    \[
\lambda(S_{\rho(\ell_u),\rho(r_u)}) \leq \gamma,
    \]
\end{definition}
We will also work with the corresponding partite notion of splitting trees.
\begin{definition}[Partite Ordered Binary Tree]
    A $d$-\maximal partite ordered binary tree is a pair $(T,\rho)$ where \(T\) is a rooted binary tree with \(d\) leaves, and \(\rho:T \to \mathbb{P}\{[d]\}\setminus \{\emptyset\}\) is a labeling such that:
    \begin{enumerate}
        \item \(\rho|_{\mathcal{L}}\) is a bijection to the singletons of \([d]\).
        \item For every interior \(u \in \mathcal{I}\), it holds that \(\rho(u) = \rho(\ell_u) \dunion \rho(r_u)\).
    \end{enumerate}
We remark that when convenient, we may use a general set $S$ of size $d+1$ as the domain of $\rho$ (generally corresponding to a labeling of the "parts" of an associated $d$-\maximal partite complex).
\end{definition}
\begin{definition}[Tuple splitting tree] \label{def:tuple-split-tree}
    Let \(\gamma > 0\). A $d$-\maximal \emph{\(\gamma\)-tuple splitting tree} is a triple \((X,T,\rho)\) where \(X\) is a \(d\)-\maximal simplicial complex and \((T,\rho)\) is an ordered partite binary tree such that for every non-leaf vertex \(u \in T\):
    \[
\lambda(S_{\rho(\ell_u),\rho(r_u)}) \leq \gamma
    \]
\end{definition}

Since all swap walks on two-sided HDX expand, these give rise to a basic family of splitting trees.

\begin{example}[\cite{AlevJT2019,DiksteinD2019,GurLL2022}] \label{ex:hdx-splitting-tree}
    Let \(\lambda \geq 0\) and let \(X\) be a \(d\)-\maximal two-sided \(\lambda\)-high dimensional expander. Then for any binary tree \(T\) and labeling \(\rho\) that satisfies the first and second item in \pref{def:split-tree}, \((X,T,\rho)\) is a \(d \lambda\)-splitting tree. 
\end{example}

An easy consequence of \cite{DiksteinD2019} is that partite one-sided HDX also give rise to tuple-splitting trees.
\begin{example}
        Let \(\lambda \geq 0\) and let \(X\) be a \(d\)-\maximal two-sided \(\lambda\)-high dimensional expander. Then for any binary tree \(T\) and labeling \(\rho\) that satisfy the first and second item in \pref{def:tuple-split-tree}, \((X,T,\rho)\) is a \(d^2 \lambda\)-tuple splitting tree. 
\end{example}

Finally, \cite[Corollary 9.18]{AlevJQST2020} also prove that walks on expander graphs are splittable.
Let us define the complex formally. Let \(G = (V,E)\) be a graph and \(k \in \NN\). Let \(W_G\) be the following \(k\)-partite simplicial complex
    \[\SC[W][1][G] = V \times [k],\]
    \[\SC[W][k][G] = \sett{\set{(v_1,1),(v_2,2),\ldots,(v_k,k)}}{(v_1,v_2,\ldots,v_k) \text{ is a walk in } G}.\]
We choose a face in \(\SC[W][k][G]\) by choosing a vertex \(v_1 \in V\) according to the stationary distribution over vertices in \(G\) and then taking a \((k-1)\)-length random walk on that vertex.

\begin{example}
    Let \(G\) be a \(\lambda\)-expander and \(T, \rho\) an ordered partite tree. If every \(u \in \mathcal{I}\) satisfies: 
    \[
    \max \rho(\ell_u) < \min \rho (r_u),
    \]
    then \((W_G,T,\rho)\) is a \(\lambda\)-splittable tree.
\end{example}

\subsubsection{Operations on splitting trees}
We will need several further operations on splitting trees for our analysis: a \textit{pruning} method that allow us to analyze higher dimensionalT faces, and a \textit{partitification} technique that allows us to reduce standard splitting trees to the tuple setting.

We start with the latter. Recall $X$'s partitification is the complex
\[
\SC[P][k] \coloneqq \sett{(s,\pi) \coloneqq \left\{(s_{\pi(1)},1),(s_{\pi(2)},2),\ldots,(s_{\pi(k)},k)\right\}}{s \in \X[k], \pi \in S_k},
\]
This notion extends naturally to splitting trees.
\begin{definition} \label{def:partitification}
    Let \((X,T,\rho)\) be a $d$-\maximal \(\lambda\)-splitting tree, and \(\rho':T \to [d]\) a well-ordered labeling such that \(|\rho'(u)|=\rho(u)\) for every \(u \in T\). The \emph{\(\rho'\)-partitification} of \((X,T,\rho)\) is defined to be the tuple splitting tree \((P,T,\rho')\).
\end{definition}
Any un-ordered splitting tree can be embed into a tuple splitting tree by taking any partitification.
\begin{claim} \label{claim:partitification}
    Let \((X,T,\rho)\) be a \(\lambda\)-splitting tree. Every \(\rho'\)-partitification of \((X,T,\rho)\) is a \(\lambda\)-tuple splitting tree.
\end{claim}
\begin{proof}
    It suffices to show for every internal node \(u \in \mathcal{I}\) with \(\rho'(\ell(u))=A\) and \(\rho'(r(u))=B\) that \(S_{A,B}(P)\) is a \(\lambda\)-bipartite expander. This follows from the observation that \(S_{A,B}(P)\) is isomorphic to the bipartite graph with left vertices \(L=\X[|A|] \times [|A|!]\), right vertices \(R=\X[|B|] \times [|B|!]\), and edges \((s,i) \to (t,j)\) for all $i \in A, j \in B$, and $(s,t)$ such that \(s \sim t\) in \(S_{\rho(\ell_u),\rho(r_u)}\).\footnote{Formally, the edge \((s,i) \to (t,j)\) also inherits the weight $\frac{\Pr(st)}{|A||B|}$ from \(S_{\rho(\ell_u),\rho(r_u)}\).} It is a standard fact that such a bipartite construction inherits its expansion from \(S_{\rho(\ell_u),\rho(r_u)}\), and is therefore a \(\lambda\)-bipartite expander. 
\end{proof}

Finally, we note that tuple splitting trees can be pruned to create new splitting trees. Let \((X,T,\rho)\) be a tuple splitting tree with \(k\) leaves. Let \(u \in \mathcal{I}\) be an internal node whose label is \(F= \rho(u)\). 
Let \(P=P(X,F)\) be the partite complex with vertices \(\SC[P][1] = \SC[X^{[k] \setminus F}][1] \dunion X[F]\) and top level faces \(t \dunion \set{s}\) such that \(t \dunion s \in \X[k]\), inheriting the corresponding densities. In other words, we replace the vertices of $X$ in \(\SC[X^F][1]\) with the set of faces \(X[F]\). 

Let \(T(u)\) be the sub-tree of \(T\) rooted by \(u\). Let \(T'=(T \setminus T(u)) \cup \set{u}\). Let \(\rho':T' \to ([k]\setminus F) \cup \set{F}\) be the labeling that replaces the subset $F$ with a single replacement index $f$:
\[ \rho'(v) = \begin{cases} (\rho(v)\setminus F) \cup \set{f} & F \subseteq \rho(v) \\
\rho(v) & otherwise
\end{cases}.\]

\begin{definition}[Pruning Trees] \label{def:pruning}
    Let \((X,T,\rho)\) be a \(\lambda\)-tuple splitting tree with \(k\) leaves. Let \(u \in T\) be a non-leaf node whose label is \(F= \rho(u)\). The \(u\)-pruning of the tree is the triple \((P(X,F),T',\rho')\).
\end{definition}

\begin{claim} \label{claim:tree-pruning}
    Let \((X,T,\rho)\) be a \(\lambda\)-tuple splitting tree. Let \(u \in T\) be a non-leaf node. Then the \(u\)-pruning is also a \(\lambda\)-splitting tree. \(\qed\)
\end{claim}
The proof follows directly from the definitions and we therefore omit it. 

Finally, it will be convenient to introduce notation for \textit{sequential pruning}, that is given a collection of disjoint color sets $\mathcal{F} = \{F_1,\ldots,F_m\}$, we denote the sequential pruning of $X$ by $\mathcal{F}$ as $P(X,\mathcal{F}) \coloneqq P(...P(X,F_1),F_2),F_m)$. Note that the resulting complex is independent of the order of $\mathcal{F}$, so $P(X,\mathcal{F})$ is well-defined. One can similarly define a sequentially pruned splitting tree as above. It is immediate from \pref{claim:tree-pruning} that these are also $\lambda$-tuple splitting.

\subsubsection{Hitting set and Bias amplification}\label{sec:hdeml}
In this section, we prove two basic `sampling-type' properties for splitting trees: hitting set and bias amplification. We start with the former. Recall a complex is $(\gamma,i)$-hitting if for any $A \subset \X[0]$:
    \[
    \Pr_{\sigma \in \X[i]}[\sigma \subset A] \leq \mu(A)^{i} + \gamma
    \]
\begin{proposition}[Hitting Set]\label{prop:hitting-set}
    Any depth $D$ $\lambda$-(tuple) splitting tree with $k$ leaves is $(3^D\lambda,k)$-hitting.
\end{proposition}

Splitting trees also exhibit strong bias amplification, used to build balanced error correcting codes.
\begin{proposition}[Bias-Amplification] \label{prop:bias-amp}
Let \((X,T,\rho)\) be a $\lambda$-tuple splitting tree for $\lambda < \frac{1}{16}$. For any $0<\varepsilon<\frac{1}{4}$ and family of mean $\varepsilon$ functions \(\{f_i:X[i] \to \{\pm 1\}\}_{i \in [k]}\), $\prod f_i$ is a $\{\pm 1\}$-valued function with bias at most:
\begin{equation}
    \abs{\Ex[a\in {\X[k]}]{\prod_{i=1}^k f_i(a_i)}} \leq \varepsilon^k+2\lambda
\end{equation}
\end{proposition}

While \pref{prop:hitting-set} and \pref{prop:bias-amp} are elementary, to the authors' knowledge they do not appear in the literature. Prior works on HDX (e.g. \cite{DinurK2017,AlevJQST2020}) use weaker variants of these lemmas which require super-exponential overhead to achieve a given distance or bias, while the above only requires quasi-polynomial overhead. We discuss this in greater detail in \pref{sec:codes}.

These results are immediate corollaries of a functional variant of \cite{DiksteinD2019}'s ``high dimensional expander-mixing lemma'' (HD-EML) extended to splitting trees. We give the statement here and include a proof in the appendix for completeness along with the proofs of \pref{prop:hitting-set} and \pref{prop:bias-amp}.

\begin{theorem}[high dimensional expander Mixing Lemma]\label{thm:hdeml}
Let \((X,T,\rho)\) be a depth $D$ \(\lambda\)-tuple splitting tree with \(k\) leaves. Denote by \(d_i\) the depth of the leaf labeled \(i\). Then for any family of functions \(\{f_i:X[i] \to \RR\}_{i \in [k]}\):
\begin{equation}\label{eq:hdeml}
    \abs{\Ex[a \in {\X[k]}]{\prod_{i=1}^k f_i(a_i)}- \prod_{i=1}^k \Ex[{a_i \in X[i]}]{f_i(a_i)}} \leq 3^{D} \lambda \prod_{i=1}^k \norm{f_i}_{2^{d_i}}.
\end{equation}
If \((X,T,\rho)\) is instead a depth $D$ standard $\lambda$-splitting tree, we take \(\{f_i:\X[1] \to \RR\}_{i \in [k]}\) and have:
\begin{equation}\label{eq:hdeml-unordered}
    \abs{\Ex[a\in {\X[k]}, \pi \in S_k]{\prod_{i=1}^k f_i(a_{\pi(i)})}- \prod_{i=1}^k \Ex[a_i \in {\X[1]}]{f_i(a_i)}} \leq 3^{D} \lambda \prod_{i=1}^k \norm{f_i}_{2^{d_i}}.
\end{equation}
\end{theorem}
In the future, we state our results only for the partite setting unless there is a substantial difference in proof or parameters for the results. Otherwise, we collect analog statements for the unordered case at the end of \pref{app:HD-EML} for completeness.

We close the subsection with a few remarks. First it should be noted that for the special case of expander walks, better hitting set and bias amplification methods are known (indeed one can achieve polynomial overhead vs our quasipolynomial bound above). Nevertheless, it is useful to have these lemmas on structures like HDX that admit stronger properties like agreement tests. We remark it is also actually possible to use the high dimensional expander mixing lemma for splitting trees directly to prove a different variant of Chernoff-Hoeffding by reduction to the complete complex. However, this results in exponentially worse parameters than our arguments in \pref{sec:chernoff}, and we will not need the result in the following. 

\subsection{Coding Preliminaries}
A $q$-ary error correcting code $C$ over a finite size-$q$ alphabet $\Sigma$ is a subset $C \subset \Sigma^n$, where $n \in \mathbb{N}$ is called the \textit{block length}. We say a code is \textit{linear} if $\Sigma=\mathbb{F}_p$ for prime power $p$ and $C$ is a subspace, and \textit{$\mathbb{F}_p$-linear} if $\Sigma=\mathbb{F}_p^k$ for some $k \in \mathbb{N}$ and $C$ is a subspace of $(\mathbb{F}_p^k)^n$.

The \textit{rate} of a code $C$, which measures its overhead, is $r=\frac{\log_{|\Sigma|}(|C|)}{n}$, and its \textit{distance}, which captures the codes error correction capability, is the minimum normalized hamming distance between any two codewords
\[
d \coloneqq \min_{c_1,c_2 \in C}\left\{\frac{|\{i \in [n]: c_1(i) \neq c_2(i)\}|}{n}\right\}.
\]
We remark that for any $\mathbb{F}_p$-linear code, the distance is the minimum weight of a non-zero codeword. For ease of notation, we will call a $q$-ary code with rate $r$ and distance $d$ an $(r,d)_q$-code.

We will be interested in so-called ``good'' families of codes, which are infinite families with growing block length and constant rate and distance. We will focus in particular in this section on building large distance $\mathbb{F}_2$-linear codes over $\Sigma=\mathbb{F}_2^k$. In this regime there is a classical upper bound due to McEliece, Rodemich, Rumsey, and Welch \cite{mceliece1977new} stating that good $q$-ary codes have distance at most $1-\frac{1}{q}$, or more exactly:
\begin{theorem}[MRRW Bound {\cite{mceliece1977new}}]\label{thm:MRRW}
    For all $\varepsilon>0$ and $q \in \mathbb{N}$, any family of $q$-ary codes with distance $1-\frac{1}{q}-\varepsilon$ has rate at most:
    \[
    r \leq O\left(\varepsilon^2\log(1/\varepsilon)\right).
    \]
\end{theorem}

Given a code $C$, a \textit{unique decoding algorithm} with radius $\gamma \in [0,1/2]$ is a (possibly randomized) algorithm $Dec: \Sigma^n \to C$ which given any $w \in \Sigma^n$ such that $dist(w,C) \leq \gamma$ outputs the uniquely closest $y \in C$ to $w$. Note that unique decoding is only possible up to radius $\lfloor \frac{nd-1}{2n} \rfloor \leq 1/2$, even with large alphabet, as there may be no unique closest codeword beyond this point.

\textit{List-decoding} is the natural extension of unique decoding beyond the $\frac{1}{2}$ barrier. A code is said to be list-decodable with radius $t$ if there is an algorithm which outputs a \textit{list} of all codewords within radius $t$ of the given word. We first give the combinatorial definition, which simply bounds the number of codewords around any fixed point.
\begin{definition}[Combinatorial List-decoding]
    For $\gamma \in (0,1)$ and $L \in \mathbb{N}$, a code $C$ is said to be $(\gamma,L)$ list-decodable if for every $z \in Y$, there are at most $L$ codewords $x \in C$ such that $dist(x,z) < \gamma$.
\end{definition}
In this section, we will be interested in \textit{efficient algorithms} for list-decoding, typically a much more challening task than the existential definition.
\begin{definition}[Efficient List-decoding]
    For $\delta>0$, $\gamma \in (0,1)$ and $L \in \mathbb{N}$, a family of codes $\{C_n\}$ is said to be efficiently $(\gamma,L)$ list-decodable with confidence $\delta$ if it is combinatorially $(\gamma,L)$ list-decodable and there exists a (randomized) polynomial time\footnote{Here we mean polynomial in the block length.} algorithm \textsc{Dec}$: \Sigma^n \to P(C)$ such that for every word $w \in \Sigma^n$:
    \[
    \textsc{Dec}(w) = \{c \in C: \text{dist}(c,w) < \gamma \}
    \]
    except with probability $\delta$.
\end{definition}

Another classical approach to handling the high error regime is \textit{local testability}. A code $C$ is said to be locally testable if there is a (randomized, non-adaptive) algorithm $\mathcal{T}_C$ which on input of a word $w \in \Sigma^n$, queries constantly many symbols of $w$ and rejects with high probability if it is far from the code. In practice, if one receives such a word, instead of trying to list decode they could simply request re-transmission. More formally, we consider the following standard notion of `strong' local testability:
\begin{definition}[Locally Testable Code]
    For any $s>0$ and $q \in \mathbb{N}$, a code $C$ of blocklength $n$ is said to be $(q,s)$-locally testable if there exists a function $\mathcal{T}_C: [n]^q \times \Sigma^q \to \{0,1\}$ and a distribution $\mathcal{D}$ over $[n]^q$ such that on any word $w \in \Sigma^n$, the pair $(\mathcal{T}_{C},\mathcal{D})$ (called the tester) satisfies:
    \begin{enumerate}
        \item \textbf{Completeness:} If $w \in C$, then $\Pr_{J\sim \mathcal{D}}[\mathcal{T}_C(J,w_J)=1]=1$
        \item \textbf{Soundness:} $\Pr_{J \sim \mathcal{D}}[\mathcal{T}_C(J,w_J)=0] \geq s \cdot \text{dist}(w,C)$
    \end{enumerate}
    We call $q$ the locality, and $s$ the soundness. 
\end{definition}
Typically we will simply think of $\mathcal{T}_C$ as a randomized algorithm and write $\mathcal{T}_C(w)$ to denote its application on the word $w$, dropping the distribution of the tester when clear from context. We will work with testers that have one additional constraint: their queries should be marginally uniform. 
\begin{definition}[Uniform LTC]
    $(\mathcal{T},\mathcal{D})$ is called \textit{uniform} if every marginal of $\mathcal{D}$ is uniform over $[n]$.
\end{definition}
We remark that focusing on uniform LTCs is not particularly restrictive---since codes do not typically have `preferred' coordinates, most natural constructions are uniform.

Finally, though not strictly necessary, we will assume throughout that all complexes in this section are ``homogenous'' subsets $X \subset [n]^k$, meaning that their marginal distribution over each part is uniform over $[n]$. We remark that in the context of high dimensional expanders, any partite `type-regular' complex can be homogenized in this fashion while maintaining expansion and bounded degree (see \cite{FriedgutI2020} for details). As most known constructions of HDX are type-regular, this is not a particularly restrictive assumption, and it is also known to hold for other classical splitting trees such as expander walks \cite{AlevJQST2020}.

\subsection{The ABNNR Construction}
In the previous section, we proved several `agreement theorems' on simplicial complexes. These well-studied tests are actually very closely related to both local testability and list decoding of a particular family of codes called \textit{direct product codes}, corresponding exactly to the set of `global functions' discussed in \pref{sec:agreement}. At a broader level, these are themselves a special case of a well-studied family of codes introduced in the seminal work Alon, Bruck, Naor, Naor and Roth \cite{alon1992construction} called the `ABNNR'-construction. We focus on the case of $\mathbb{F}_2$ for simplicity, though our results extend naturally to larger fields.
\begin{definition}[the ABNNR Construction]
    Given a right $k$-regular bipartite graph $G=(L,R,E)$ and a base code $C \subseteq \mathbb{F}_2^{L}$, the ``ABNNR-Encoding'' of $C$ with respect to $G$ is the code $\text{Im}(E_G(C))$ where $E_G$ is an ``encoding'' map $E_G: \mathbb{F}_2^{|L|} \to (\mathbb{F}_2^{k})^{|R|}$ defined by setting for each $c \in C$ and $v \in R$:
\[
E_G(c)_v = c|_{N(v)}.
\]
\end{definition}
In other words, each right vertex concatenates the symbols of its neighbors to create a new code over $R$. One can check that in \pref{sec:agreement}, the set of global functions is exactly the ABNNR encoding over the inclusion graph $(\X[j],\X[k])$. For a general $k$-\maximal complex $X$ and $j \leq i \leq k$, we write $E^{(j,i)}_X$ to mean the ABNNR-encoding on the vertex inclusion graph $(\X[i],\X[j])$. In the special case where $i=k$ and $j=1$, we write just $E_X$ to denote the ABNNR-encoding over the inclusion graph $([n],\X[k])$, where inclusion is viewed taking $\X[k] \subseteq [n]^k$ as `ordered tuples' of $[n]$, rather than as faces of a simplicial complex. This is a minor difference, but is a bit more convenient in the setting of amplifying base codes.

Recent years have seen a great deal of work on instantiating the ABNNR construction on high dimensional expanders and splitting trees \cite{DinurK2017,DinurHKLT2018,AlevJQST2020,jeronimo2021near,blanc2022new,Jeronimo2023list}, in particular in the context of efficient list-decoding. Building on tools from these works, we give the first family of large alphabet codes with near-optimal distance that are simultaneously locally testable and list-decodable.

\begin{theorem}\label{thm:codes}
    For all large enough $k \in \mathbb{N}$ and all $\varepsilon>0$, there exists an explicit family of $\mathbb{F}_2$-linear ABNNR-Codes that are simultaneously:
    \begin{enumerate}
    \item $(\varepsilon^{O(k)},1-2^{-k}-\varepsilon)_{2^k}$-codes
        \item $(O(1),\frac{O(1)}{k\log(1/\varepsilon)})$-locally testable.
        \item $(1-2^{-\Omega(k)},2^{O(k)})$-efficiently list-decodable with confidence $2^{-\Omega(n-k)}$ whenever $\varepsilon \leq 2^{-\Omega(k)}$.
    \end{enumerate}
\end{theorem}
We remark that very recently, Jeronimo, Srivastava, and Tulsiani \cite{Jeronimo2023list} gave a direct list-decoding algorithm for the LTCs of \cite{DinurELLM2022} (which have some small constant distance $d \ll \frac{1}{2}$). To our knowledge, these are the only other examples of codes that are simultaneously $c^3$-LTC and efficiently list-decodable.

The remainder if this section is split into three parts focused respectively on distance amplification, local-testability, and list-decoding of the ABNNR construction, and a short fourth section devoted to a conjectural approach to removing the soundness decay in \pref{thm:codes}. Each section is independent, containing both the corresponding component of \pref{thm:codes} as well as more general results within the framework. To facilitate this structure we first give a brief proof overview of \pref{thm:codes} that doubles as a roadmap for the section.

\paragraph{Proof Overview} 
The high level idea behind \pref{thm:codes} is simple: starting with a binary locally testable `base code' $C$ with distance near $\frac{1}{2}$, we will argue that the ABNNR-encoding of $C$ on any $k$-partite splitting tree\footnote{Formally for list-decoding actually require a slightly stronger notion called `complete' splittability, c.f.\ \pref{def:complete-split}.} has distance near $1-2^{-k}$ (\pref{cor:dist-amp}), is list-decodable (\pref{thm:list-decode-ABNNR}), and maintains local testability (\pref{thm:LTC-ABNNR}). Instantiating this framework on \textit{expander walks} (the sparsest known family of splitting trees) gives codes with the claimed parameters. We give a brief description of each step.

\textit{Distance amplification} on splitting trees is fairly elementary, and follows from the hitting-set lemma. This is a standard argument: one observes that the two encoded symbols $E(x),E(x')$ differ in any face which hits the vertex set $1_{x \neq x'}$. Since by assumption the base code has distance roughly $1/2$, the hitting set lemma promises this set is hit except with roughly $2^{-d}$ probability (see \pref{cor:dist-amp}). \pref{thm:codes} actually uses a slightly stronger distance amplification lemma for expander-walks and is given in \pref{cor:LTCs}. The weaker amplification lemma (in particular its bias variant) appears later in the proof as a sub-component of the list-decoding algorithm. 

\textit{List-Decoding} the ABNNR construction, despite being heavily studied on HDX and splitting trees, is surprisingly tricky. Roughly speaking there are two main issues. First, known list-decoding algorithms for this setting are actually either for an $\mathbb{F}_2$-valued variant of ABNNR called the \textit{direct sum} construction \cite{AlevJQST2020,jeronimo2021near,blanc2022new}, or only hold on \textit{two-sided} HDX \cite{AlevJQST2020,DinurHKLT2018} (which have exponentially worse rate). To handle this, we build a reduction from list-decoding the ABNNR-encoding on $X$ to list-decoding the direct sum encodings of its projected sub-complexes $X^F$. The second issue with this strategy is that known list-decoding algorithms for direct sum \cite{jeronimo2021near,blanc2022new} actually requires the base code to be \textit{$\varepsilon$-biased}, a stronger condition than simply having distance near $\frac{1}{2}$. We will discuss how to construct such base codes shortly. We instantiate this framework in \pref{cor:decoding-expanders} to prove the list-decoding component of \pref{thm:codes}.

\textit{Local Testability} of the ABNNR construction, to our knowledge, has not been studied outside the special case of agreement testing. Our test follows a very simple strategy. We first check whether the given word $f$ is a direct product (in the sense of \pref{sec:agreement}) via a simple 2-query agreement test: pick a random vertex $v \in [n]$, and independently $s,s' \supset v$, and check whether $f_s$ and $f_{s'}$ match on $v$. If this test passes, we then simulate the tester on the `direct product decoding' of $f$. That is, for every location the base tester queries, we pick a random face including that vertex and feed the tester its value in the word $f$. We prove this procedure is sound in \pref{thm:LTC-ABNNR}, and instantiate the process in \pref{cor:LTCs} to prove the local testability component of \pref{thm:codes}.

\textit{The Base Code} for our construction, as discussed above, must be an $\varepsilon$-biased $c^3$-LTC. $c^3$-LTCs with distance near $\frac{1}{2}$ are now known due to the breakthrough works of \cite{DinurELLM2022,PanteleevK22} combined with LTC distance-amplification techniques of \cite{kopparty2017high}, but the latter amplification causes several core issues that make the tecnique largely useless for our setting: 1) the resulting codes make $\poly(1/\varepsilon)$ queries, 2) may not be $\varepsilon$-biased, and 3) are heavily randomized. We remedy these issues via the ABNNR-construction itself. In particular, starting with the base LTC of \cite{DinurELLM2022}, we encode it into $\log(1/\varepsilon)$-\maximal walks on an expander graph using ABNNR, and then concatenate with the Hadamard code. This code is $\varepsilon$-biased roughly because concatenating with the perfectly balanced Hadamard code converts the good distance of the ABNNR-encoding into good bias, and the resulting code is locally testable using the Hadamard's classical local decoding and local testing algorithms. We formalize this step in \pref{prop:base-code}
\\
\\
\noindent Finally we summarize the construction and proof below:
\\
\\
\fbox{\parbox{\textwidth}{
\vspace{.1cm}
\underline{Code underlying \pref{thm:codes}}:
\begin{enumerate}
    \item Start with the base $c^3$-LTC $C$ of \cite{DinurELLM2022}
    \item Amplify $C$'s distance via expander walks
    \item Reduce alphabet and bias by concatenating with Hadamard
    \item Re-amplify the resulting code via a second set of expander walks.
\end{enumerate}}}

\begin{proof}[Proof of \pref{thm:codes}]
    Bias, rate, and local-testability are proven in \pref{cor:LTCs}. List-decodability is proved in \pref{cor:decoding-expanders}.
\end{proof}

\subsection{Distance Amplification}
The ABNNR construction is one of the most classical distance amplification methods in coding theory. The standard result to this effect is that instantiating the framework on a (bounded degree) additive sampler amplifies distance while maintaining rate up to a constant factor.
\begin{lemma}[{\cite{alon1992construction}}]
    Let $G=(L,R,E)$ be a right $k$-regular bipartite graph such that $(R,L,E)$ is an $(\alpha,\beta)$-sampler. If $C \subset \mathbb{F}_2^{|L|}$ is a code of distance $\alpha$ and rate $r$, then $E_G(C)$ is a $(r\frac{|L|}{|R|},1-\beta)_{2^k}$-code.
\end{lemma}
This observation was used in \cite{DinurK2017,DinurHKLT2018} to give distance amplification for codes defined on HDX. Our sampling lemmas imply an exponential improvement in the rate-distance trade-off of such codes. We choose a fairly arbitrary setting of our parameters for simplicity, but any setup with good sampling works:
\begin{corollary}[Distance Amplification on HDX]\label{cor:hdx-amp}
    For any $\alpha,r,c>0$, let $C$ be a base code of distance $\alpha$ and rate $r$. Let $X$ be a $d$-\maximal $c$-nice complex. Then for any $i<j \leq d$ and $G=(\X[i],\X[j])$ the ABNNR code $E^{i,j}_X(C)$ is an $(r',d')_{|\Sigma|}$-code for
    \[
    r'=r\frac{|\X[i]|}{|\X[j]|}, \quad \quad d'=1 - \frac{1}{\alpha}\exp\left(\Omega_c\left(\alpha^8\frac{j}{i}\right)\right), \quad \quad |\Sigma| = 2^{\binom{j}{i}}.
    \]
\end{corollary}

Note that when $i=1$, \pref{cor:hdx-amp} has alphabet size $|\Sigma|=2^{j}$ and distance $1-\poly(|\Sigma|^{-1})$, optimal up to polynomial factors by \pref{thm:MRRW}. However, in this special case, it is actually possible achieve optimal amplification directly via the hitting set lemma (\pref{prop:hitting-set}).
\begin{corollary}[Distance Amplification for Splitting Trees]\label{cor:dist-amp}
    For any $\alpha,r>0$, let $C$ be a base code of distance $\alpha$ and rate $r$. If $X$ is a $k$-\maximal $\lambda$-(tuple)-splitting tree of depth $D$, then $E_X(C)$ is an $(r',d')_{|\Sigma|}$-code for
    \[
        r'= r\frac{n}{|\X[k]|}, \quad \quad d'= 1 - (1-\alpha)^k - 3^D\lambda, \quad \quad |\Sigma| = 2^k
    \]
    Moreover when $\alpha=\frac{1}{2}(1-\varepsilon)$, then $ d' \geq 1- \frac{1}{2^k}(1+\varepsilon)^k - 3^D\lambda$.
\end{corollary}
\begin{proof}
    Rate and alphabet size are immediate from construction. Let $x,x' \in C$ be any two distinct non-zero codewords. We wish to bound the distance between their encodings $E_X(x)$ and $E_X(x')$, which is exactly the hitting probability of $1_{x \neq x'}$ in the inclusion graph $([n],\X[k])$. By the distance of $C$ we are promised that $\mu(1_{x = x'}) \leq 1-\alpha$, so \pref{prop:hitting-set} gives the desired result.\footnote{\pref{prop:hitting-set} is formally for $(\X[1],\X[k])$, but because $X$ is homogenous the bound carries over by simply setting each partite component of $A$ to be $A_i=1_{x=x'}$.} Plugging $\alpha=\frac{1}{2}(1-\varepsilon)$ into the resulting bound gives the final line.
\end{proof}

\subsection{Local Testability of ABNNR}\label{sec:LTC}
In this section, we prove the ABNNR encoding of a locally testable code is locally testable.
\begin{theorem}[Local Testability of ABNNR]\label{thm:LTC-ABNNR}
    For any uniform $(q,s)$-LTC of blocklength $n$ and $(k_L,k_R)$-regular bipartite graph $G=(L,R,E)$ on $|L|=n$ vertices, $E_G(C)$ has a uniform $(q+2,\frac{cs}{qk_R})$-local tester.
\end{theorem}
Our proof relies on the following elementary agreement test. Given a word $f$, denote by $\mathcal{T}_G(f)$ the following process:
\begin{enumerate}
    \item Pick a uniformly random vertex $v \in L$
    \item Pick two uniformly random neighbors $w,w' \in N(v)$
    \item Accept if and only if $f_w(v) = f_{w'}(v)$
\end{enumerate}
In the agreement notation of the previous section, we'd write the word $f$ as an ensemble $\mathcal{F}=\{f_w: N(w) \to \mathbb{F}_2\}_{w \in R}$ and denote the acceptance probability of this test as $Agree_0(\mathcal{F})$. We will move between these two viewpoints interchangeably. $\mathcal{T}_G$ is clearly complete. Its soundness follows from a general result of \cite{goldenberg2019toward}:
\begin{lemma}[{\cite[Lemma 5.2]{goldenberg2019toward}}]\label{lem:GC-sound}
    There exists a universal constant $c \geq 1$ such that for all $\varepsilon>0$ and any ensemble $\mathcal{F}$:
    \[
    Agree_0(\mathcal{F})\geq 1 - \varepsilon \quad \implies \quad dist(F,g_{\text{maj}}) \leq ck_R\varepsilon
    \]
    where for $v \in L$ we define $g_{\text{maj}}(v) \coloneqq \text{plurality}_{w \in N(v)}\{f_w(v)\}$.\footnote{Ties can be broken arbitrarily.}
\end{lemma}
With this in mind, we prove \pref{thm:LTC-ABNNR} via the following natural tester: given a word $f \in (\mathbb{F}_2^{k_R})^R$, run the DP-tester on $w$, and the base code tester on the plurality decoding of $f$ if it passes:
    \begin{enumerate}
        \item Run $\mathcal{T}_{G}$ on $f$
        \item For every $v \in L$ that $\mathcal{T}_C$ would query:
        \begin{itemize}
            \item Sample a random neighbor of $v$: $w \sim N(v)$
            \item Denote the `decoded' plurality as $g_{\text{dec}}(v) =f_w(v)$
        \end{itemize}
        \item Run $\mathcal{T}_C$ on $g_{\text{dec}}$
    \end{enumerate}
\begin{proof}[Proof of \pref{thm:LTC-ABNNR}]
    Query complexity and completeness are immediate from construction. The tester is uniform since the base code is uniform, and picking a uniformly random vertex $v \in L$ and random $w \in N(v)$ is uniformly distributed over $R$.
    
    Toward analyzing soundness, fix any $\delta \in [0,1]$ and $f$ some word at distance $dist(f,E_X(C)) = \delta$. We need to show the test \textit{rejects} on $f$ with probability at least $\delta c'\frac{s}{qk_R}$ for some $c'\geq 0$. Observe that we may assume $\mathcal{T}_{G}(f)$ passes with probability at least $1-\delta\frac{1}{4c}\frac{s}{q k_R}$, since otherwise Step (1) rejects with probability at least $\delta\frac{1}{4c}\frac{s}{q k_R}$ and we are done. We claim that conditioned on this event, the following two properties hold:
    \begin{enumerate}
        \item $g_{\text{dec}}$ matches $g_{\text{maj}}$ with high probability: \label{item:decode}
        \[
        \Pr_{S \sim \mathcal{T}_C}[\exists v \in S: g_{\text{dec}}(v) \neq g_{\text{maj}}(v)] \leq \frac{s}{4ck_R}\delta
        \]
        \item $g_{\text{maj}}$ is far from the base code:\label{item:maj-far}
        \[
        dist(g_{\text{maj}},C) \geq \frac{1}{2k_R}\delta.
        \]
    \end{enumerate}
    Soundness is then immediate: as long as the base tester $\mathcal{T}_C$ receives the `correct' symbols from $g_{\text{maj}}$, it rejects with probability $s\cdot dist(g_{\text{maj}}) \geq \frac{s}{2k_R}\delta$ by \pref{item:maj-far}. Since the former occurs with probability at least $\frac{s}{4k_R}\delta$ by \pref{item:decode}, the total rejection probability is at least $\frac{s}{2k_R}\delta-\frac{s}{4k_R}\delta=\frac{s}{4k_R}\delta$ as desired.

    It is left to prove the claimed properties.
    \paragraph{Proof of \pref{item:decode}}
        Observe that since the queries are marginally uniform, by a union bound it is enough to show that the decoding of a uniformly random vertex succeeds except with probability $\frac{s}{4cqk_R}\delta$. The proof is essentially a simpler variant of \pref{claim:vertex-closeness-to-G}. For a fixed $v \in L$ define the set of `bad' neighbors as those whose labeling disagrees with the plurality decoding as:
    \[
    B_v \coloneqq \{w \in N(v): f_w(v) \neq g_{\text{maj}}(v)\}.
    \]
The failure probability of our single-query decoding procedure is exactly $\underset{v \in L}{\mathbb{E}}[B_v]$, so it is enough to show
\[
\mathbb{E}[B_v] \leq \frac{s}{4cqk_R}\delta.
\]
This follows from the fact that the test $\mathcal{T}_G$ can be written as an expectation over complete local distributions. In particular, observe that we can decompose $\mathcal{T}_G$ as:
    \[
    Agree_0(f) = \Pr_{v\sim L, w,w' \sim N(v)}[f_w(v) = f_{w'}(v)] \geq 1-\frac{s}{4cqk_R}\delta.
    \]
    This implies that the expected probability that $f_w(v)\neq f_{w'}(v)$ in the local distribution is at most:
    \[
    \underset{v\sim L}{\mathbb{E}}\left[\Pr_{w,w' \sim N(v)}[f_w(v) \neq f_{w'}(v)]\right] \leq \frac{s}{4cqk_R}\delta.
    \]
    Since $w$ and $w'$ are independent neighbors of $v$, the inner term is at least
    \[
    \Pr_{w,w' \sim N(v)}[f_w(v) \neq f_{w'}(v)] = 2\mathbb{E}[B_x]\mathbb{E}[\overline{B_v}] \geq \mathbb{E}[B_v]
    \]
    since $\mathbb{E}[\overline{B_v}] \geq \frac{1}{2}$ by construction, which completes the proof.
    \paragraph{Proof of \pref{item:maj-far}} The proof follows from the fact that conditioned on Step (1) (the agreement test) passing, $E_X(g_{\text{maj}})$ is close to a DP-codeword. Thus if $g_{\text{maj}}$ itself is too close to $C$, it's encoding will be close to $E_G(C)$ violating our original assumption on the distance. Formally, assume for the sake of contradiction that $g_{\text{maj}}$ is not $\frac{\delta}{2k_R}$-far from $C$. Then there exists $x \in C$ such that $dist(g_{\text{maj}}, x) < \frac{\delta}{2k_R}$. Since $G$ is bi-regular, this further implies a bound on the distance of the \textit{encodings} of $g_{\text{maj}}$ and $x$:
    \begin{align*}
        dist(E^k_G(g_{\text{maj}}),E_G^k(x)) &= \Pr_{w \sim R}[E_G(g_{\text{maj}})(w) \neq E_G(x)(w)]\\
        &\leq k_R\Pr_{w \sim R, v \sim N(w)}[E_G(g_{\text{maj}})(w)_v \neq E_G(x)(w)_v]\\
        &=k_Rdist(g_{\text{maj}},x)\\
        & < \frac{\delta}{2}.
    \end{align*}
    where the second equality follows from the fact that since $G$ is bi-regular, sampling a uniformly random vertex $v \in R$ and a random $s \in N(v)$ is equidistributed with sampling a random $w \in L$. Finally closeness of the encoded words contradicts the fact that $f$ is $\delta$-far from $E_G(C)$, since $E_G(x)$ is a codeword and by the triangle inequality:
    \[
    dist(f, E_G(x)) \leq dist(f, E_G(g_{\text{maj}})) + dist( E_G(g_{\text{maj}}), E_G(x)) < \delta,
    \]
    where we have used the fact that since $Agree_0(f) \geq 1-\frac{s}{4cqk_R}$ by assumption, $f$ is $\frac{s}{4q}\delta$-close to $E_G(g_{\text{maj}})$ by \pref{lem:GC-sound}. 
    \end{proof}
We remark that it is likely possible to reduce the query complexity of the above argument to just $q+1$ by re-using the first query of $\mathcal{T}_C$ for the agreement test and performing a more careful conditioned analysis on the latter passing. Since $q$ for us is some small constant, the difference between $q+1$ and $q+2$ is negligible so we only give the simpler independent version above.

To prove \pref{thm:codes}, we need a good explicit family of binary linear LTCs with near optimal distance to amplify via \pref{thm:LTC-ABNNR}. In fact, as discussed in the proof overview we will need base codes with the stronger property of having small \textit{bias}:
\[
bias(C) \coloneqq \max_{c \in C\setminus\{0\}} \left\{\left|\mathbb{E}_{i \in [n]}[(-1)^{c_i}]\right|\right\}.
\]
Note that an $\varepsilon$-biased linear code always has distance at least $\frac{1-\varepsilon}{2}$ since for every $c \in C$, $\Abs{\Ex[i \in {[n]}]{(-1)^{c_i}}} = 1-2 \dist(c,0)$. Our base codes start with the breakthrough construction of $c^3$-LTCs by Dinur, Evra, Livne, Lubotzky, and Moses \cite{DinurELLM2022}:
\begin{theorem}[{\cite{DinurELLM2022}}]
    There exists an explicit infinite family of binary, linear, uniform LTCs with constant rate, distance, locality, and soundness, and a linear time unique decoder up to constant radius.
\end{theorem}
We remark that while \cite{DinurELLM2022} do not argue their testers are uniform, it is easy to check from their construction that this is the case. As discussed in the overview, it is possible to use the distance amplification techniques of \cite{kopparty2017high} to achieve $c^3$-LTCs with distance near $1/2$ (even hitting the GV bound), but these codes may not be $\varepsilon$-biased and incur a heavy $\poly(\varepsilon^{-1})$ cost in the query complexity. We will take a different amplification strategy that avoids these issues almost entirely at the cost of a small polynomial blowup in rate. The idea is to concatenate our amplification procedure in \pref{thm:LTC-ABNNR} with the \textit{Hadamard code}:
\[
Had:\mathbb{F}_2^k \to \mathbb{F}_2^{2^k}, \; \alpha \overset{Had}{\mapsto} \iprod{\alpha, \cdot}
\]
We denote the resulting code (the image of this map), by $C_{\text{Had}}$.

The Hadamard code comes with three properties critical to our construction. First, it is easy to check that the code is binary, linear, and perfectly balanced ($0$-biased). Second, it is a classical result that the Hadamard code is an excellent LTC:
\begin{theorem}[\cite{blum1990self}]
    The Hadamard code is a uniform $(3,1)$-LTC.
\end{theorem}
Finally, it is well-known folklore (and an easy exercise) to show the Hadamard code is \textit{locally decodable}. We will not go into the formal details of locally decodable codes, but it suffices to say there is a randomized algorithm $Dec: \mathbb{F}_2^{2^k} \times [k] \to \mathbb{F}_2$ which given a word $y$ such that $d(y,C_{\text{Had}}) < \frac{1}{4}$ and a coordinate $i$, queries only $2$ bits of $y$ and outputs the correct decoding with probability proportional to its distance from the code:
\begin{equation}\label{eq:Had-decode}
\Pr[Dec(y,i) \neq x_i] \leq 2d(y,C_{\text{Had}})
\end{equation}
where $x=Had^{-1}(y)$ is the pre-image of the unique closest codeword to $y$.

Note that this is essentially the same property we used in \pref{thm:LTC-ABNNR} in decoding the ABNNR-encoding. We will show that concatenating with the Hadamard code similarly allows us to amplify the bias and distance of the base code while only incurring constant blow-up in query complexity and logarithmic loss in soundness.

\begin{proposition}[The Base Code]\label{prop:base-code}
    For every $\varepsilon>0$ there is an explicit infinite family of uniform, $\varepsilon$-biased, linear $(\poly(\varepsilon),\frac{1}{2}-\varepsilon)_2$-codes that are $(O(1),\log^{-1}(1/\varepsilon)))$-locally testable. Moreover, there exists $t<1$ such that these codes are uniquely decodable up to radius $t$ in time $n\exp(\poly(\varepsilon^{-1}))$.
\end{proposition}
\begin{proof}
    We apply \pref{thm:LTC-ABNNR} to \cite{DinurELLM2022}'s LTCs using the complex $X=W_G$ generated by length $k=O(\log(1/\varepsilon))$ random walks on any constant degree, regular $\frac{1}{2}$-spectral expander $G$ (such expanders are known to exist for any $n$, see e.g.\ \cite{alon2021explicit}). This results in a $\poly(\varepsilon)$-rate code with distance at least $1-\frac{\varepsilon}{2}$ by the classical expander hitting set lemma \cite{hoory2006expander} which states that for any  set $T \subset \X[1]$ of measure $\rho$
    \begin{equation}\label{eq:expander-hitting-set}
    \Pr_{(s_1,\ldots,s_k) \in \X[k]}[|\{s_1,\ldots,s_k\} \cap T|=0 ] \leq \rho\left(\frac{1+\rho}{2}\right)^{k-1}.
    \end{equation}
    The argument is then as in \pref{cor:dist-amp}. Finally by \pref{thm:LTC-ABNNR} this encoding has a uniform $(c,\frac{c'}{\log(1/\varepsilon)})$-local tester for some universal constants $c,c'>0$. 

    The key step is now to concatenate this resulting family with the Hadamard code:
    \[
  \mathcal{C} \coloneqq Had(E_X(C)),
    \]
    where the Hadamard encoding is applied pointwise to each symbol in $E_X(C)$. This encoding maintains linearity and decreases the rate by a factor of at most $2^k=\poly(\varepsilon)$. To compute the bias, observe that since the Hadamard code is perfectly balanced, every non-zero symbol of a word $w \in E_X(C)$ contributes 0 to the bias so in the worst case we have
    \begin{align*}
        bias(\mathcal{C}) \leq \max_{c \in E_X(C)\setminus \{0\}}\{wt(c)\cdot 0 + (1-wt(c))\cdot 1\} \leq 1-d(E_X(C)) \leq \varepsilon,
    \end{align*}
    where $wt(c)$ denotes the normalized hamming weight (fraction of non-zero elements in $c$).
    
    We now argue that this concatenation has a uniform local tester. Write $q=c$ and $s=\frac{c'}{\log(1/\varepsilon)}$ and recall $E_X(C)$ has a uniform $(q,s)$-local tester (which we'll refer to as the ``base tester'') which requests only a single bit of each queried symbol. We'll denote these requests as pairs $(x,i) \in \X[k-1] \times [k]$. Our testing procedure follows a similar form as in \pref{thm:LTC-ABNNR}. Namely given a potential word $y$:
    \begin{enumerate}
        \item Draw a uniformly random $x \in X$, and run the Hadamard test on $y(x)$.
        \item For each $(x,i)$ requested by the base tester, run the Hadamard decoder on $y(x)$ at coordinate $i$.
        \item Run the base tester on the output
    \end{enumerate}
    The test is complete by construction, and makes at most $2q+3$ queries (two for each run of the local decoder, and three for the initial Hadamard test). Soundness follows along the same lines as \pref{thm:LTC-ABNNR}. We first argue that the word $y$ must be close to a cartesian product of Hadamard codewords (else Step 1 rejects). We then argue that the decoding succeeds with high probability, and the decoded word must be far from the base code (else $y$ itself would be close to $\mathcal{C})$, so the base tester rejects with good probability.

    Formally, let $y \in (\mathbb{F}_2^{2^k})^{X}$ satisfy $dist(y,\mathcal{C})=\tau$. We argue that the above procedure rejects with probability at least $\frac{s\tau}{4q}$. To see this, first observe we can assume the word $y$ is $1-\frac{s\tau}{4q}$ close from being a product of Hadamard codewords. Otherwise Step 1 rejects with probability at least $\frac{s\tau}{4q}$ since
    \[
    \Pr[\text{Step 1 Rejects}] \geq \mathbb{E}_x[dist(y(x),C_{Had})] = dist\left(y,\prod_X C_{Had}\right).
    \]
    Fix some $y_{\text{had}} \in \prod_X C_{Had}$ which minimizes the distance from $y$ and denote the set of `corrupted' coordinates of $y$ to be
    \[
    B \coloneqq \left\{ x \in X: dist\left(y(x),C_{\text{Had}}\right)\geq \frac{1}{4}\right\}.
    \]
    Note that the measure of $B$ is at most $\frac{s\tau}{q}$, since the closeness of $y$ to $\prod C_{\text{Had}}$ is violated otherwise.

    Define set of `valid' decodings of $y$ to be strings in $(\mathbb{F}_2^{k})^X$ whose encodings match $y_{\text{had}}$ on all uncorrupted coordinates:
    \[
    Y_{\text{dec}} \coloneqq \left\{ y_{\text{dec}} \in (\mathbb{F}_2^{k})^X: \forall x \in \bar{B}, Had(y_{\text{dec}}(x))=y_{\text{had}} \right\}.
    \]
    Note that when we \textit{re-encode} any $y_{\text{dec}} \in Y_{\text{dec}}$, it will be at distance at most $\frac{s\tau}{q}$ from $y$ since:
    \begin{align*}
    dist(Had(y_{\text{dec}}),y) &\leq dist(Had(y_{\text{dec}}),y_{\text{had}}) +  dist(Had(y_{\text{had}}),y)\\
    &\leq (\Pr[B]\cdot \frac{1}{2}+\Pr[\bar{B}]\cdot 0) + \frac{s\tau}{4q}\\
    &\leq \frac{s\tau}{q}.
    \end{align*}
    We now argue that Step 2 correctly decodes to some $y_{\text{dec}} \in Y_{\text{dec}}$ with high probability. Since decoding always `succeeds' on coordinates in $B$ (since any decoded value is accepted), the worst case is when $B=\emptyset$ and $y_{dec}$ is the pre-image of the unique closest word in every coordinate. In this case, local decodability of the Hadamard code \pref{eq:Had-decode} promises the probability any fixed query $(x,i)$ fails to decode to $y_{dec}(x)_i$ is at most
    \[
    \mathbb{E}_x[2dist(y(x),C_{Had})] \leq 2dist\left(y,\prod_X C_{\text{Had}}\right) \leq \frac{s\tau}{2q}.
    \]
    Union bounding over all $q$ queries, we decode to $y_{\text{dec}}$ except with probability $\frac{s\tau}{2}$.

    Finally, we argue any $y_{\text{dec}} \in Y_{\text{dec}}$ must be $\frac{3}{4}\tau$-far from $E_X(C)$, in which case the base tester rejects with probability at least $\frac{3}{4}s\tau - \frac{s\tau}{2} = \frac{s\tau}{4}$ by a union bound. This simply follows from the fact that if $y_{\text{dec}}$ were $\frac{3}{4}\tau$-close to $E_X(C)$, then there exists $w \in E_X(C)$ such that 
    \begin{align*}
    dist(Had(w),y) &\leq dist(Had(w),Had(y_{\text{dec}})) + dist(Had(y_{\text{dec}}),y)\\
    &\leq dist(w,y) + \frac{s\tau}{q}\\
    &\leq \tau
    \end{align*}
    where the second inequality comes from observing that the encodings of $w$ and $y_{\text{dec}}$ are the same on any symbol on which they agree, and can only improve in distance when they differ.\footnote{Technically we've also used that $q\geq 4$ and $s\leq 1$ here, which follow from the parameters of \cite{DinurELLM2022}.} 

    It is left to show our concatenated family has an efficient unique decoder. It is a classical fact that concatenating codes with efficient unique decoding at radii $t_{\text{inner}}$ and $t_{\text{outer}}$ gives a code with an efficient unique decoder at radius $t_{\text{inner}}t_{\text{outer}}$ simply by running the inner decoder on every inner symbol, then running the outer decoder on the result \cite{forney1965concatenated}. In particular decoding at this distance is successful since if a true word $m$ is corrupted in at most a $t_{\text{inner}}t_{\text{outer}}$ fraction of its coordinates, no more than a $t_{\text{outer}}$ fraction of inner codewords can have a $t_{\text{inner}}$ fraction of corruptions. Thus running the inner decoder results in the correct word except in at most a $t_{\text{outer}}$ fraction of outer coordinates, which is then successfully decoded by the outer decoder. 
    
    We can decode over each Hadamard symbol up to radius $t' < \frac{1}{4}$ by brute force in $\exp(\poly(\varepsilon^{-1}))$ time (and otherwise output some fixed word, say 0, beyond this radius). $E_X(C)$ can be efficiently decoded up to some constant radius $t$ by taking the majority vote over each coordinate, then running the base decoder on $C$ (see e.g. \cite{DinurHKLT2018}). This process takes time $n\exp(\poly(\varepsilon^{-1})) + O(n)$, for a total time of $n\exp(\poly(\varepsilon^{-1}))$ as desired.
\end{proof}
Finally, we argue these base codes can be amplified to large alphabet LTCs with near optimal distance.
\begin{corollary}[Large Distance c3-LTCs]\label{cor:LTCs}
    For all large enough $k \in \mathbb{N}$ and all $\varepsilon>0$, there exists an infinite family of $\mathbb{F}_2$-linear $(r, d)_{|\Sigma|}$-codes that are $(q,s)$-locally testable for:
    \[
    r=\varepsilon^{O(k)}, \quad d=1-2^{-k}-\varepsilon, \quad |\Sigma|=2^k, \quad q=O(1), \quad s=O\left(\frac{1}{k\log(\frac{1}{\varepsilon})}\right)
    \]    
\end{corollary}
\begin{proof}
    We start with the explicit family of base codes given in \pref{prop:base-code} with $q=O(1)$ queries, distance $\frac{1}{2}-\frac{\varepsilon}{4}$, and rate $\poly(\varepsilon)$. For a given code $C$ in the family of blocklength $n$, let $X$ be the collection of length-$k$ walks on any explicit degree $\frac{\varepsilon}{2}$-expander of degree $O(\varepsilon^{-2})$ on $n$ vertices (such graphs exist for any large enough $n$ \cite{alon2021explicit}). The size of $\X[k-1]$ is $O(\frac{n}{\varepsilon^{2k}})$, and the expander hitting-set lemma (Equation \pref{eq:expander-hitting-set}) ensures the distance of $E_X(C)$ is at least
    \[
    1-\left(\frac{1}{2}+\frac{\varepsilon}{2}\right)^k \geq 1-2^{-k} - \varepsilon.
    \]
    Finally, since the base code's tester is uniform, local testability follows from \pref{thm:LTC-ABNNR}.
\end{proof}
It is worth comparing \pref{thm:LTC-ABNNR} to the distance amplification for LTCs of Kopparty, Meir, Ron-Zewi, and Saraf \cite{kopparty2017high}, who also analyzes local testability of distance-amplified codes via a related construction of Alon and Luby \cite{alon1996linear}. Their work largely focuses on achieving near-optimal \textit{rate-distance} trade-off: they achieve the singleton bound over large alphabet codes and the GV-bound over $\mathbb{F}_2$. In comparison \pref{thm:LTC-ABNNR} has worse rate than their construction, but makes up for this with exponentially better alphabet size and local testability. Quantitatively, to achieve distance $1-\beta$, \cite{kopparty2017high}'s codes require alphabet size $2^{\poly(1/\beta)}$ and lose $\poly(1/\beta)$ factors in the testing parameters, while ours have the optimal $\beta^{-1}$-dependence in the former and lose only \textit{logarithmically} in the latter.

\subsection{List-Decoding ABNNR}
In this section, we show that the codes in \pref{cor:LTCs} are efficiently list-decodable, completing the proof of \pref{thm:codes}. Efficient list-decoding of the ABNNR construction was first considered by Dinur, Harsha, Kaufman, Livni-Navon, and TaShma \cite{DinurHKLT2018}, who gave an algorithm for the construction instantiated on HDX. After \cite{DinurHKLT2018}'s work, Alev, Jeronimo, Quintana, Srivastava, and Tulsiani \cite{AlevJQST2020} showed it possible to efficiently list decode a closely related construction called a \textit{direct sum code} over some relaxed families of splitting trees (including HDX and expander walks). Using a reduction to direct sum, they also recover the parameters of \cite{DinurHKLT2018}'s ABBNR-codes on HDX as well. 

In this section, we give an efficient list-decoding algorithm for the ABNNR-construction on complexes satisfying a slightly strengthened variant of splittability that ensures the complex can be ``split'' at any coordinate. We refer to this notion as complete splittability.
\begin{definition}[Complete Splittability]\label{def:complete-split}
    A $d$-\maximal partite complex $X$ is called completely $\lambda$-splittable if for all $i \in [d]$:
    \[
    \lambda_2(S_{[1,i],[i+1,d]}) \leq \lambda
    \]
    where we recall $S_{[1,i],[i+1,d]}$ is the `swap walk' adjacency operator between $X[{[1,i]}]$ and $X[{[i+1,d]}]$.
\end{definition}
Complete splittability sits in-between tuple-splitting trees and partite HDX, and is closely related to notions of splittability defined in more recent works on direct-sum decoding \cite{jeronimo2021near,blanc2022new} (in particular it is equivalent to the definition introduced in the former, see \pref{app:split} for details). We can now state the main result of the subsection:
\begin{theorem}[List Decoding on Completely Splittable Complexes]\label{thm:list-decode-ABNNR}
    For any $t<1/2$, let $C_0$ be a binary code with blocklength $n$, $\text{bias}(C_0) \leq 1-2t$, and which is unique-decodable to radius $t$ in time $\mathcal{T}_0$. Let $X\subset [n]^k$ be a homogenous, completely $\lambda$-splittable complex, and $C=E_X(C_0)$ its ABNNR encoding. Let $d(C)=1-\varepsilon$ and $\tau>0$ be any value satisfying
    \[
    \tau \geq \max \left \{\sqrt{\varepsilon}, 4\sqrt{(1-2t)^{i} + 2\lambda},\sqrt{2^{21}\lambda k^3}, 8\left(1-t\right)^{i/2} \right \}
    \]
    for $i=\frac{k}{2}-\sqrt{k\log(\frac{2}{\tau})}$. There exists a randomized algorithm, which given $\tilde{y} \in (\mathbb{F}_2^k)^{X}$ recovers the list 
    \[
    \mathcal{L}_{\tau}(\tilde{y}) = \left\{y \in C: d(\tilde{y},y) \leq 1-\tau\right\}
    \]
with probability $1-\frac{\tau-\varepsilon}{\varepsilon(\tau^2-\varepsilon^2)}2^{-\Omega(n/t^2)}$ in time $\tilde{O}(\frac{\tau-\varepsilon}{\tau^2-\varepsilon^2}\exp(\log(1/t)\exp(k^3/\tau^2))(|X|+\mathcal{T}_0))$. Moreover, the number of such codewords is at most:
\[
|\mathcal{L}_{\tau}(\tilde{y})| \leq \frac{\tau-\varepsilon}{\tau^2-\varepsilon^2}.
\]
\end{theorem}
The proof of \pref{thm:list-decode-ABNNR} goes via reduction to direct sum decoding. The \textit{direct sum} encoding of a code $C$ on a complex $X$ simply concatenates the ABNNR-encoding with $XOR: \mathbb{F}_2^k \to \mathbb{F}_2$, that is $dsum_X(C) \coloneqq XOR(E_X(C))$ (where XOR applied pointwise, similar to in the Hadamard concatenation in the prior section). Most of the work in proving \pref{thm:list-decode-ABNNR} is already done via \cite{jeronimo2021near, blanc2022new}'s direct sum decoder for completely splittable complexes:
\begin{theorem}[List Decoding Direct Sum {\cite{jeronimo2021near,blanc2022new}}]\label{thm:direct-sum}
   For any $t<1/2$, let $C_0$ be a binary code with blocklength $n$, $\text{bias}(C_0) \leq 1-2t$, and which is unique-decodable to radius $t$ in time $\mathcal{T}_0$. Let $X\subset [n]^k$ be a homogenous, completely $\lambda$-splittable complex, and $C=\text{dsum}_X(C_0)$ be its direct sum encoding. Let $\varepsilon=bias(C)$ and $\beta>0$ be such that
    \[
    \beta \geq \max \left \{\sqrt{\varepsilon},\sqrt{2^{20}\lambda k^3}, 2\left(1-t\right)^{k/2} \right \}.
    \]
    There exists a randomized algorithm which given $\tilde{y} \in \mathbb{F}_2^X$ recovers the list 
    \[
    \mathcal{L}^{X,(sum)}_\beta(\tilde{y}) = \left\{x \in C_0: d(\tilde{y},dsum_X(x)) \leq \frac{1}{2}-\beta\right\}
    \]
    with probability $1-\frac{1}{\varepsilon}2^{-\Omega(n/t^2)}$ in time $\tilde{O}(\exp(\log(t)\exp(\frac{k^3}{\beta^2}))|X|+\mathcal{T}_0)$. Moreover, the list has at most
    \[
    |\mathcal{L}^{X,(sum)}_\beta(\tilde{y})| \leq \frac{1}{\varepsilon}
    \]
    codewords.
\end{theorem}
We remark that the requirement on the unique decoder radius in \cite{jeronimo2021near,blanc2022new} is slightly more strict as stated, but one can easily modify parameters in their proof to achieve the above. We now prove \pref{thm:list-decode-ABNNR} by reduction to \pref{thm:direct-sum}.
\begin{proof}[Proof of \pref{thm:list-decode-ABNNR}]
We appeal to a simple variant of the reduction of \cite{AlevJQST2020}. Given a potential word $\tilde{y} \in (\mathbb{F}_2^k)^X$ at distance $d(\tilde{y},C)=1-\tau$, we claim the following process is an efficient list-decoder.
\begin{enumerate}
    \item For every $F \subseteq [d]$, apply \pref{thm:direct-sum} to $XOR(\tilde{y}|_F)$  on $dsum_{X^F}(C)$ with $\beta=\tau/4$ to get:
    \[
\mathcal{L}^{pot}_{\beta}(\tilde{y}) \coloneqq\bigcup_{F \subseteq [d+1]} \mathcal{L}^{X^F,(sum)}_{\beta}(\tilde{y})
    \]
    \item Prune any decodings that are far from $\tilde{y}$: 
    \[
    \mathcal{L}_{\tau}(\tilde{y}) \coloneqq \left\{E_X(x): x \in \mathcal{L}_{\beta}^{pot}(\tilde{y}),~ d(E_X(x),\tilde{y}) \leq 1-\tau\right \}.
    \]
\end{enumerate}
To prove this procedure indeed gives a valid list-decoding we will use the fact that any projection of a splittable complex, in particular each $X^F$ above, is also splittable.
\begin{claim}\label{claim:project}
    Let $X \subset [n]^{d}$ be a homogeneous, completely $\lambda$-splittable complex. Then for any $F \subseteq [d]$, $X^F$ is homogeneous and completely $\lambda$-splittable.
\end{claim}
The proof follows from the inheritance of expansion under projection, and is defered to \pref{app:split}.

We now prove the output list $\mathcal{L}_{\tau}$ is exactly the set of codewords at distance at most $1-\tau$ from $C$. Recall that our base code is assumed to be at most $(1-2t)$-biased. By bias amplification lemma for splitting trees (\pref{prop:bias-amp}) and \pref{claim:project}, we have that for any $X^F$, the bias of $dsum_{X^F}(C)$ is at most $(1-2t)^{|F|} + 2\lambda$. As such, our choice of $\beta=\tau/4$ satisfies the conditions of \pref{thm:direct-sum} whenever $i \leq |F| \leq k$ (for $i$ as in the Theorem statement). Let $y \in E_X(C)$ be a codeword at distance at most $1-\tau$ from $\tilde{y}$, then we have
    \[
    \Pr[\tilde{y}_s=y_s] = \mathbb{E}_{s}\left[\mathbb{E}_{F \in \mathbb{F}_2^k}[\chi_F(\tilde{y}_s+y_s)]\right] \geq \tau,
    \]
where $\chi_F(y)=(-1)^{\langle F,y\rangle}$ are the standard characters. By averaging and Chernoff, there must exist $|F| \geq \frac{k}{2}-\sqrt{k\log(2/\tau))}$ such that
\[
\mathbb{E}_{s}\left[\chi_F(\tilde{y}_s)\chi_F(y_s)\right] \geq \tau/2,
\]
which is equivalent to the statement that the distance of $\tilde{y}|_{F}$ from $y|_F$ is at most $\frac{1}{2}-\frac{\tau}{4}=\frac{1}{2}-\beta$ in the direct sum encoding on $X^F$. Since $X^F$ is homogeneous and completely $\lambda$-splittable, by \pref{thm:direct-sum} we have that the decoding of $y$ appears in $\mathcal{L}_{pot}(\tilde{y})$ except with probability $\frac{1}{\varepsilon}\exp(n/t^2)$. By the Johnson bound (see e.g.\ \cite{DinurHKLT2018}[Theorem 2.15]), there can be at most $\frac{\tau-\varepsilon}{\tau^2-\varepsilon^2}$ such codewords $y$ close to $\tilde{y}$, so union bounding over these gives that all valid codewords are in the list with the desired probability. Finally, all extraneous decodings are removed in the second step, which completes correctness of the procedure. 

Running time follows immediately from the fact that we have run the direct sum decoder at most $2^k$ times, and the resulting potential list has at most $\exp(k)$ potential codewords by the Johnson bound and our choice of $\beta=\tau/4$, each of which can be pruned in $O(k|X|)$ time.
\end{proof}

Finally, we instantiate this result on expander walks to complete the proof of \pref{thm:codes}.
\begin{corollary}\label{cor:decoding-expanders}
    The family of codes in \pref{cor:LTCs} is $(1-2^{-\Omega(k)}, \frac{1}{2^{O(k)}})$ efficiently list-decodable with confidence $2^{-\Omega(n-k)}$.
\end{corollary}
\begin{proof}
    Since \cite{jeronimo2021near} prove that walks on an $\lambda$-expander are homogeneous and completely $\lambda$-splittable, it is enough to ensure the parameter $\tau$ in \pref{thm:list-decode-ABNNR} can be set to $2^{-\Omega(k)}$. Recalling that the base code of this construction is $\varepsilon$-biased and uniquely decodable in linear time up to constant radius, and that the encoding has distance at least $1-2^{-k}-\varepsilon$, as long as $\varepsilon \leq \exp(-\Omega(k))$ one can easily check it is possible to take $\tau=\exp(-\Omega(k))$ and $\frac{\tau-\varepsilon}{\tau^2-\varepsilon^2} = 2^{-O(k)}$ satisfies all conditions of \pref{thm:list-decode-ABNNR}.
\end{proof}

\subsection{Toward Lossless Amplification from HDX?}\label{sec:parallel}
Our expander-based construction experiences polylogarithmic decay in soundness with the alphabet size. This is exponentially better than prior methods, but it is plausible the decay could be avoided entirely. In this section, we propose an approach toward such `lossless' amplification via the ABNNR-Encoding on an HDX.

Instantiating \pref{thm:LTC-ABNNR} directly on an HDX does actually slightly improve soundness over \pref{thm:codes}, but the loss of $\frac{1}{k}$ in the dimension of the complex is inherent to our approach, no matter the instantiation. The reason is that our method only simulates the base tester a single time. As a result, the rejection probability will only scale with the distance of the decoded word from the \textit{base} code. On the other hand, because $X$ amplifies this distance by a factor of $k$, this results in an unavoidable loss in soundness.

There is a natural approach to fixing this issue: simulate the base tester ``\textit{$k$ times in parallel},'' analogous to Raz's parallel repetition theorem \cite{raz1995parallel} and its de-randomized variants in agreement testing \cite{ImpagliazzoKW2012,dinur2011derandomized,dikstein2023agreement, bafna2023characterizing}. Indeed on the complete complex, this approach can actually be carried out successfully:
\begin{theorem}[Parallel LTC Amplification]\label{thm:parallel}
    Let $C$ be a uniformly $(q,s)$-locally testable $(r,\frac{1-\varepsilon}{2})_2$-code of blocklength $n$ and $X=\binom{[n]}{k}$. Then $E_X(C)$ is a uniformly $(q',s')$-locally testable $(r',d')_{|\Sigma|}$-code for:
    \[
    r'=\frac{r}{n^{k-1}}, \quad d'=1-\frac{1}{2^k}\left(1+\varepsilon\right)^k - o_n(1), \quad |\Sigma|=2^k, \quad q'=q+2, \quad s' = cs
    \]
    where $c>1$ is some \textit{universal} constant (independent of $k$).
\end{theorem}
\begin{proof}
    Rate and alphabet are immediate from construction, and distance follows from (near)-independence of $\binom{[n]}{k}$. For local testability, we rely on the following agreement theorem of Dinur and Steurer:
\begin{claim}[{\cite{DinurS2014}}]\label{claim:DS}
    There exists a two query agreement test on $\binom{[n]}{k}$ and a constant $c\geq 1$ such that\footnote{We remark that formally Dinur and Steurer consider $[n]^k$, but their techniques transfer to this setting (formally, see e.g.\ \cite{DiksteinD2019} for the quoted result).}
    \[
    Agree_0(\mathcal{F}) \geq 1-\varepsilon \quad \implies \quad dist(\mathcal{F},g_{\text{maj}}) \leq c\varepsilon
    \]
\end{claim}
    With this in hand the proof is similar to \pref{thm:LTC-ABNNR}. For $k$ independent runs of the base tester $\mathcal{T}_C$, denote the requested symbols by $\{x^{(i)}_{j}\}_{j \in [q],i \in [k]}$ and write $s_j \coloneqq \{x_j^{(1)},\ldots,x_j^{(k)}\} \in \binom{[n]}{k}$. Given a word $f$, our tester performs the following two-step procedure:
    \begin{enumerate}
        \item Run the Dinur-Steurer test on $f$
        \item For every $i \in [q]$:
        \begin{itemize}
            \item For each $j \in [k]$, denote the `decoded' plurality function at $x^{(i)}_j$ as $g_{\text{dec}}(x_{j}^{(i)})=f_{s_j}(x_{j}^{(i)})$
            \item Run $\mathcal{T}_C$ on $g_{\text{dec}}(x^{(i)}_1),\ldots,g_{\text{dec}}(x^{(i)}_q)$
        \end{itemize}
    \end{enumerate}
    The tester is complete and makes $q+2$ queries by construction. The proof of soundness is very similar to \pref{thm:LTC-ABNNR}, so we give a shortened explanation here. Assuming $dist(f,E_X(C)) = \delta$, we need to show our test rejects with probability $\Omega(\delta)$. Assume Step (1) passes with probability at least $1-\frac{s\delta}{8cq}$ (else we are done). Since the base tester is marginally uniform and the $k$ runs are independent, each face $s_j \in X$ is marginally uniform, and therefore by \pref{claim:DS} exactly matches plurality except with probability $\frac{s\delta}{8q}$. Union bounding over the $q$ faces gives that $g_{\text{dec}}$ matches $g_{\text{maj}}$ on all queried points except with probability $\frac{s}{8}\delta$. By the same argument as in \pref{thm:LTC-ABNNR}, $g_{\text{maj}}$ must be at least $\frac{\delta}{2k}$-far from $C$, so conditioned on receiving the correct values, at least one of the base testers rejects with probability $1-(1-\frac{s\delta}{2k})^k \geq \frac{s\delta}{4}$ for a total rejection probability of $\frac{s\delta}{4}-\frac{s\delta}{8}=\frac{s\delta}{8}$ as desired.
\end{proof}
The issue with \pref{thm:parallel} is that the resulting code has very poor rate. It is natural to conjecture that, as in the single query variant, this can be resolved by replacing $\binom{[n]}{k}$ with a high dimensional expander. Indeed it is known that the agreement testing step (the analog of \pref{claim:DS}) works in this setting \cite{DinurK2017,DiksteinD2019,kaufman2020local}. Note that this is \textit{not} the case for other splitting trees such as expander walks, which would lose a factor of $k$. The challenge then mainly lies in extending the parallel simulation. Below, we discuss the current challenges toward porting this second half of the proof to HDX.

The first difficulty in parallel simulation is that the HDX and base code must `match' in the sense that there is a natural distribution over $q$-tuples of faces that corresponds to $k$ valid $q$-query tests for the base code. Recently, Golowich \cite{golowich2023grassmannian} showed 
how to construct HDX from certain (modified) locally testable codes, and Dinur, Liu, and Zhang \cite{dinur2023new} built LTCs (of sub-constant rate) that sit on the triangles of an HDX. It is not clear a priori how to fit either of these constructions into our framework, but it is plausible they or modifications thereof do.

The second difficulty in simulating the parallel tester is that the $k$ runs of the base test will be inherently correlated. There is a long history in the HDX literature of handling this type of behavior (e.g.\ \cite{DinurK2017,DiksteinD2019,kaufman2020local,KaufmanG2023,dikstein2023agreement,bafna2023characterizing}), including the agreement tests we discuss in \pref{sec:agreement}. Very recently, building upon tools in this work, \cite{dikstein2023agreement,bafna2023characterizing} prove that certain strong families of high dimensional expanders admit $1\%$-regime agreement tests, which are closely related to parallel repetition \cite{ImpagliazzoKW2012,dinur2011derandomized}. As a start, it would be interesting to extend \pref{thm:parallel} even to dense complexes known to admit such tests, such as the Grassmann \cite{ImpagliazzoKW2012} or spherical buildings \cite{dikstein2023agreement}. To summarize, we give the following informal conjecture:
\begin{conjecture}[Parallel LTC Amplification]
    There exists a family of high dimensional expanders and locally testable codes which admit a ``parallel'' variant of \pref{thm:LTC-ABNNR}, that is:
    \begin{enumerate}
        \item There is a $q$-step random walk or distribution on $X$ generating $k$ valid `parallel' tests
        \item The resulting parallel tester has no (asymptotic) decay in soundness.
    \end{enumerate}
\end{conjecture}
\section{Concentration Lower Bounds}\label{sec:lower-bounds}
In this section we prove our results are essentially optimal with respect to their dependence on the sampling parameters and quantitative expansion of the complex. We further show one cannot hope to achieve our more involved applications such as reverse hypercontractivity using classical samplers such as expander walks.

\subsection{Optimality for Inclusion Samplers}\label{sec:optimality}
In \pref{sec:inclusion}, we proved that the if \(X\) is a high dimensional expander, then the inclusion graphs \((\X[k],\X[i])\) are roughly \((\varepsilon,\exp(-\varepsilon^2\frac{k}{i}))\)-additive samplers. Inspecting the proof of \pref{thm:hdx-is-sampler}, it is easy to show that for any set $A$ of density $\mu$ and strong enough local-spectral expansion $\lambda \ll \mu$, one can give the following improved bound:\footnote{Though formally when \(i = \Omega(\frac{k}{\log k})\) one needs to assume \(\varepsilon\) is a large enough to subsume the $\frac{1}{\varepsilon}$ factor in front of the exponent.}
\[
\Prob[s \in {\X[k]}]{\cProb{r \in \X[i]}{A}{r \subseteq s} \geq \prob{A} + \varepsilon }\leq \exp\left(-\Omega\left(\frac{\varepsilon^2}{\prob{A}}\frac{k}{i}\right)\right).
\]

In this section we prove that for most \(i \leq k\), this is essentially optimal. Recall that a complex \(X\) is called \((\gamma,i)\)-hitting if for any \(A \subseteq \X[1]\), \(\Prob[\sigma \in {\X[i]}]{\sigma \subseteq A} \leq \prob{A}^i + \gamma\). We say that \(X\) is \((\gamma,i)\)-hitting for sets of size \(\mu\) if \(\Prob[\sigma \in {\X[i]}]{\sigma \subseteq A} \leq \prob{A}^i + \gamma\) holds for any \(A \subseteq \X[1]\) of size \(\prob{A} = \mu\).

\begin{theorem}\label{thm:optimal-sampling}
    Let \(\varepsilon  \in (0,0.1)\). Let \(i<k\) be two integers satisfying \(\frac{k}{i} \geq 3\varepsilon^{-2}\) and let \(n \geq ke^k\).
    Let \(X\) be a \(k\)-\maxsize simplicial complex such that every vertex in \(X\) has the same probability \(\frac{1}{n}\). If \(X\) is \((i,\frac{\varepsilon}{3})\)-hitting, then \((X(k),X(i))\) is not an \((\varepsilon,\beta)\)-additive sampler for \(\beta < \exp(-O(\varepsilon^2 \frac{k}{i}))\). 

    In particular, for any \(\mu \in (0, 0.2)\), if \(X\) is \((i,\frac{\min(\varepsilon,\mu)}{3})\)-hitting for sets of size \(\frac{\mu}{i}\) there exists  \(E \subseteq \X[i]\) with \(\frac{1}{2}\mu \leq \prob{E} \leq \mu\) such that
        \(
        \Prob[s \in {\X[k]}]{\cProb{r \in \X[i]}{E}{r\subseteq s} > \prob{E} + \varepsilon} \geq \exp(-O(\frac{\varepsilon^2}{\mu} \frac{k}{i})).
        \)
\end{theorem}

While the assumption that \(X\) is \((i,\frac{\varepsilon}{3})\)-hitting is not completely general, it follows from \((\X[k],\X[i])\) being a modest sampler. We prove this at the end of the subsection.
        \begin{claim} \label{claim:sampling-implies-hitting}
            Let \((\X[k],\X[i])\) be an \((\frac{1}{10i},\frac{\varepsilon}{6i})\)-additive sampler. Then \(X\) is a \((i,\frac{\varepsilon}{3})\)-hitting for all sets of size at least \(\frac{0.2}{i}\). 
        \end{claim}
This results in the variant of the bound stated in the introduction.
\begin{corollary}
        Let \(\varepsilon  \in (0,0.1)\). Let \(k\) and \(n\) be large enough (as a function of \(\varepsilon\)).
        Let \(X\) be a \(k\)-\maximal simplicial complex such that every vertex in \(X\) has the same probability \(\frac{1}{n}\). If \((\X[k],\X[i])\) is a \((\frac{1}{10i},\frac{\varepsilon}{6i})\)-additive sampler, then \((\X[k],\X[i])\) is not a \((\varepsilon,\beta)\)-additive sampler for any \(\beta < \exp(-O(\varepsilon^2 \frac{k}{i}))\).
\end{corollary}

The assumption that the vertices all have the same probability is also not general, see the discussion in the end of this subsection.

\subsubsection{Proof of \pref{thm:optimal-sampling}}
For the rest of the section we restrict our attention to ``one-sided'' additive sampling, that is:
\[
\Prob[s]{ \cProb{v \sim s}{A}{s} > \prob{A} + \varepsilon}.
\]
This is no larger than the probability that \(\abs{\cProb{v \sim s}{A}{s} - \prob{A}} < \varepsilon\) and is more convenient to handle.

Let us introduce some notation, similar to the notation in \pref{sec:conc-partite} (but for one-sided sampling in different levels of the complex). For a \(k\)-\maximal simplicial complex, \(i \leq k\) and constants \(\mu,\varepsilon > 0\) let
\[
    \pi(X,k,i,\mu,\varepsilon) = \max \sett{\Prob[s \in {\X[k]}]{\cProb{r \in \X[i]}{A}{r \subseteq s} > \prob{A} + \varepsilon}}{A \subseteq \X[i], \prob{A} = \mu}.
\]
That is the best one-sided sampling bound for when restricting to sets of relative size \(\mu\).\footnote{Here let us assume for simplicity that there exists sets of this probability. We note that this assumption is justified because as \(n\) grows larger the possible sizes of sets, \(\set{\frac{k}{n}}_{k=0}^n\), becomes dense in \([0,1]\). The use of this claim later on is to find a set of size \(\approx \mu\) anyway, so we neglect dealing with this minor point.} 

For general sampler graphs \(G=(L,R,E)\) we use similar notation:
\[\pi(G,\varepsilon,\mu) = \max \sett{\Prob[v \in R]{\cProb{u \in L}{A}{v \sim u} > \prob{A} + \varepsilon}}{A \subseteq L, \prob{A} = \mu}.\]

We also denote by \(\Delta_{n}\) the complete complex on \(n\) vertices.

\begin{proof}[Proof of \pref{thm:optimal-sampling}]
    We start by stating a general bound on \emph{all} samplers (not just inclusion samplers), via reduction to the sampling properties of the complete complex.
    \begin{claim} \label{claim:ceg-for-small-sets}
        Let \(k,n >0\). Let \(G\) be a bipartite graph such that \(|L|=n\) and such that it is \(k\)-right regular. Then
        \[\pi(G,\varepsilon,\mu) \geq \pi(\Delta_n,k,1,\varepsilon,\mu).\]
    \end{claim}
    We defer the proof, and first argue this implies \pref{thm:optimal-sampling}.

    Fix \(\mu\) as in the theorem statement. By \pref{claim:ceg-for-small-sets} there exists \(A \subseteq \X[1]\) of measure \(\prob{A}=\frac{\mu}{i}\) such that 
    \[
    \Prob[s \in {\X[k]}]{\cProb{}{A}{s} > \prob{A} + \frac{10\varepsilon}{i}} \geq \pi(\Delta_{n},k,1,\frac{\mu}{i},\frac{10\varepsilon}{i}).
    \] 
    Let \(E = N_i(A)\) be the set of all \(i\)-faces that hit \(A\) and let \(\mu' = \prob{E}\). We claim that
    \[
    \mu = i\Pr[A] \geq \Pr[E] \geq \frac{i}{2}\Pr[A] = \frac{\mu}{2}
    \]
    The lower bound follows from the hitting property, namely:
    \[
    \prob{E} \geq 1-(1-\frac{\mu}{i})^{i}-\frac{\mu}{3} \geq \frac{\mu}{2} = \frac{i}{2}\Pr[A]
    \]
    since \(\mu < 0.2\).

    The upper bound follows because the only way to sample a vertex in \(A\) is to sample an \(i\)-face in \(E\), then subsample a vertex in \(A\) in that \(i\)-face. The probability of doing so (conditioned on landing in $E$) is at least \(\frac{1}{i}\).

    The proof will follow if we can show that
    \begin{equation} \label{eq:comparing-to-complete}
        \Prob[s \in {\X[k]}]{\cProb{r \in \X[i]}{E}{r \subseteq s} > \prob{E} + \varepsilon} \geq \pi(\Delta_{n},k,1,\frac{\mu}{i},\frac{10\varepsilon}{i})
    \end{equation}
    since Chernoff's bound is tight on the complete complex.
    \begin{lemma}[Reverse Chernoff, e.g.\ {\cite[Lemma 5.2]{KleinN1999}}]
        Let \(\varepsilon, \mu\) and \(k\) be such that \(\frac{\varepsilon'^2}{\mu'} k \geq 3\), and let \(n \geq ke^{k}\), then \(\pi(\Delta_n,k,1,\mu',\varepsilon') \geq \exp(-10\frac{\varepsilon'^2}{\mu'} k)\).\footnote{Technically the lemma in \cite{KleinN1999} is for sampling with replacement, instead of sampling without replacement (which is the complete complex case), but for  \(n \geq ke^{k}\) the two are close enough in \(TV\)-distance so we get essentially the same theorem.}
    \end{lemma}
    Applying this lemma with \(\mu' \coloneqq \frac{\mu}{i}\) and \(\varepsilon' \coloneqq \frac{10\varepsilon}{i}\) and using the fact that \(\frac{\varepsilon'^2}{\mu'}k \geq 3\) by assumption, implies 
    \[
        \Prob[s \in {\X[k]}]{\cProb{r \in \X[i]}{E}{r \subseteq s} > \prob{E} + \varepsilon} \geq \pi(\Delta_{n},k,1,\frac{\mu}{i},\frac{10\varepsilon}{i}) \geq \exp(-O(\frac{\varepsilon^2}{\mu} \frac{k}{i})).
    \]

    Toward proving \eqref{eq:comparing-to-complete} observe that for a fixed $s \in \X[k]$, the probability of \(E\) inside $s$ depends only on the probability of \(A\) inside \(s\). Assume that $\Pr[A|s]-\Pr[A] > \frac{10\varepsilon}{i}$. Then
    \begin{align*}
        \cProb{r \in \X[i]}{E}{r \subseteq s} &> 1-(1-\cProb{}{A}{s})^{i}\\
        &\geq 1-(1-\Pr[A]-\frac{10\varepsilon}{i})^{i}\\
        &\geq 1-(1-\Pr[A])^{i} + 2\varepsilon\\
        &\geq \Pr[E] + \varepsilon
    \end{align*}
    where the first inequality is by observing the choice of \(r \subseteq s\) subsamples $i$ random vertices (without replacement). The probability of hitting $A$ in the latter is then at least \(1-(1-\cProb{}{A}{s})^{i}\), the hitting probability with replacement. The third inequality is true when \(i \geq 1\) and \(\varepsilon < 0.1\).
    
    Combining this with our original guarantee on sampling $A$ itself we have
    \begin{align*}
      \pi(X,k,i,\mu',\varepsilon) &\geq \Prob[s \in {\X[k]}]{\cProb{r \in \X[i]}{E}{r \subseteq s} > \prob{E} + \varepsilon} \\
      &\geq \Prob[s \in {\X[k]}]{\cProb{}{A}{v \in S} \geq \prob{A} + \frac{10\varepsilon}{i}} \\
      &\geq \pi(\Delta_{n},k,1,\frac{\mu}{i},\frac{10\varepsilon}{i})
    \end{align*}
    as desired.
\end{proof}
It remains to prove \pref{claim:ceg-for-small-sets}. The proof is basically repeating the lower bound in \cite{CanettiEG1995}. The latter show any degree-$k$ \((\varepsilon,\delta)\)-additive sampler satisfies \(\delta \geq \exp(-\Omega(\varepsilon^2 k))\). Similar to our strategy above, their proof can be interpreted as showing $(\SC[\Delta][k][n],\SC[\Delta][1][n])$ containment graph is the optimal sampler in terms of right-degree, then reducing the general case to $\Delta_n$. We give a refinement of their argument for sets of specific size, though our definition of samplers is less general than theirs.
\begin{proof}[Proof of \pref{claim:ceg-for-small-sets}]
    Let \(D_\mu\) be the uniform distribution over all sets \(A \subseteq L\) of probability \(\mu\). Then
    \[
    \pi(G,\varepsilon,\mu) \geq \Prob[r \in R,A \sim D_\mu]{\cProb{v \in L}{A}{v \sim r} > \mu + \varepsilon}.
    \]
    Observe that since we randomize over \(A\) and $G$ is right regular, the choice of \(r \in R\) doesn't affect the inner probability. As such, we may fix any \(r_0 \in R\) and instead consider
    \[
    \Prob[A \sim D_\mu]{\cProb{v \in L}{A}{v \sim r_0} > \mu + \varepsilon},
    \]
    that is the probability (over \(A\)) that \(A \cap r_0\) contains \(\geq (\mu + \varepsilon)k\) elements.

    We claim this is equal to \(\pi(\Delta_{n},k,1,\varepsilon,\mu)\). To see this, assume without loss of generality that \(L=\set{1,2,\dots,n}\) and \(r_0=\set{1,2,\dots,k}\). Since we have assumed the vertex weights are uniform, \(A\) is a uniform sample of \(\mu n\) points out of \(L\). It is convenient to think of $A$ as being sampled via the following procedure:
    \begin{enumerate}
        \item Draw a random permutation \(\sigma \in S_n\)
        \item Set $A= \sigma(A_0)$, where $A_0=\set{1,2,\dots,\mu n}$
    \end{enumerate}
    Observe \(|r_0 \cap A| \geq (\mu + \varepsilon)k\) if and only if \(|\sigma^{-1}(r_0) \cap A_0| \geq (\mu + \varepsilon)k\). But note that the random variable \(\sigma^{-1}(r_0)\) is exactly a random choice of a set \(\SC[\Delta][k][n]\). Therefore, this is equal to
    \[
    \Prob[s \in {\SC[\Delta][k][n]}]{\cProb{v \in [n]}{A_0}{v \in s} > \mu + \varepsilon}.
    \]
    Since every set of vertices of equal measure has the same sampling in $\Delta_n$ by symmetry, this is exactly \(\pi(\Delta_{n},k,1,\varepsilon,\mu)\) as desired.
\end{proof}

\subsubsection{Sampling implies hitting}
We briefly prove \pref{claim:sampling-implies-hitting}.
\begin{proof}[Proof of \pref{claim:sampling-implies-hitting}]
    Let \(A \subset \X[1]\) be a set of size \(\frac{\mu}{i}\) and \(E \subseteq \X[i]\) its set of neighbors. We need to show 
    \[
    \prob{E} \geq 1- (1-\prob{A})^{i} - \frac{\varepsilon}{3}.
    \]
     Let \(f:\X[i] \to [0,1]\) be \(f(s) = \frac{|s \cap A|}{i}\) and observe that \(\Ex[s \in {\X[i]}]{f(s)} = \frac{\mu}{i}\). Let $B$ be the set of bad $k$-faces that over-samples $f$:
     \[
     B = \sett{t \in \X[k]}{\Ex[s \subseteq t]{f(s)} \geq \frac{\mu}{i} + \frac{1}{10i}},
     \]
     then by our sampling guarantee \(\prob{B} < \frac{\varepsilon}{6i}\). 

     Fix \(t \notin B\). Then it holds that \(\cProb{v \in \X[1]}{A}{v \in  t} \leq \frac{\mu+1}{i} \leq \frac{1.1}{i}\) because \(\cProb{v \in \X[1]}{A}{v \in  t} = \Ex[s \subseteq t]{f(s)}\). Moreover, by the hitting properties of the complete complex, we have that \(\cProb{}{E}{t} \geq 1-(1-\cProb{}{A}{t})^{i}\) (using the fact that inside the complete complex, the sampling is negatively correlated). Thus 
    \begin{align*}
        \prob{E} &= \Ex[t]{\cProb{}{E}{t}} \\
        &\geq \prob{\neg B}\Ex[t \notin B]{\cProb{}{E}{t}} \\
        &\geq (1-\frac{\varepsilon}{6})(1-\Ex[t \notin B]{(1-\cProb{}{A}{t})^{i}}) \\
        &\geq (1-\frac{\varepsilon}{6})(1-\Ex[t \notin B]{(1-\cProb{}{A}{t} - \frac{\varepsilon}{6i})^{i}}) \\
        &\geq (1-\frac{\varepsilon}{6})(1-(1-\prob{A})^{i}-\frac{\varepsilon}{6})\\
        &= 1-(1-\prob{A})^{i} -\frac{\varepsilon}{3}.
    \end{align*}
\end{proof}

\subsubsection{Discussion}
We conclude the section with a few remarks.
\begin{remark}~
    \begin{enumerate}
        \item \pref{thm:optimal-sampling} doesn't follow from the tight bounds depending on degree or on \(\frac{|\X[k]|}{|\X[i]|}\) \cite{CanettiEG1995}. The degree of a $k$-face in \((\X[k],\X[i])\) is \(\binom{k}{i}\), and the degree of a set in \(\X[i]\) can be arbitrarily large, as in the complete complex case. The lower bounds in \cite{CanettiEG1995}, imply a lower of \(\exp(-O(\binom{k}{i}))\) instead of \(\exp(-O(\frac{k}{i}))\). In addition, the ratio \(\frac{|\X[k]|}{|\X[i]|}\), which is another source for lower bounds in \cite{CanettiEG1995}, can be arbitrarily large. That being said, we will \emph{reduce} to the bound in \cite{CanettiEG1995} (only for small sets).
        \item While most high dimensional expanders can be symmetrized to satisfy the assumed vertex uniformity \cite{FriedgutI2020}, it is worth noting the condition can likely be relaxed to a bound on \(\max_{s \in \X[k]}\Prob[v \in {\X[1]}]{v\in s}\) with some additional effort. Further, some assumption of this sort is necessary. To see this, take the complete \(k\)-partite complex \(X\) where \(X[i]=\set{x_i}\) are singletons for all \(i <k\), and \(X[k] = \set{y_1,y_2,\dots,y_n}\). That is \(\X[k] = \sett{\set{x_0,x_1,\dots,x_{k-1}} \cup \set{y_j}}{j=1,2,\dots,n}\) are all possible faces between the different parts. For any \(f:\X[1]\to [0,1]\) we note that \(\ex{f} = \frac{\sum_{i=1}^{k-1}f(x_0)}{k} + \frac{1}{k}\Ex[{y_i \in \X[k]}]{f(y_i)}\). Thus for every \(s\), \(\abs{\Ex[v \subseteq s]{f(v)} - \ex{f}} \leq \frac{1}{k}\). In particular, for every \(\varepsilon > \frac{1}{k}\), the inclusion graph is an \((\varepsilon,0)\)-sampler.
        \item We note that a similar bound can be achieved by essentially the same proof assuming only \(\frac{\mu}{3}\)-hitting, albeit at the cost of worse dependence on $\varepsilon$. In particular, assuming only that $X$ is $\Theta(1)$-hitting, one can still prove a lower bound against $(\varepsilon,\exp(-O(\varepsilon\frac{k}{i})))$-sampling.
    \end{enumerate}
\end{remark}

\subsection{High Dimensional Expansion}\label{sec:td-lower}
Having shown that the concentration bounds of \pref{thm:hdx-is-sampler} cannot be quantitatively improved, we turn to understanding the extent to which our requirements on the \textit{expansion} of the underlying complex are necessary for strong concentration. We give two lower bounds to this effect. First, we argue that concentration is not implied by local-spectral expansion at or below the TD-barrier.
\begin{proposition}[Lower Bounds at the TD-Barrier]\label{prop:TD-barrier}
    For every $\beta<1$ and $d \in \mathbb{N}$, there exists a family of $1$-TD complexes $\{X_n\}$ such that $(\SC[X][d][n],\SC[X][1][n])$ is not a $(\frac{1}{2},\beta)$-additive sampler.
\end{proposition}
Thus in this sense our results giving concentration bounds for $\lambda$-TD complexes for any $\lambda<1$ are tight. Note that by \pref{lem:raising}, the above also implies failure of sampling for these complexes for any (large enough) $k \leq d$ (this can also be shown directly by the same method). 

The proof of \pref{prop:TD-barrier} is based on the `product' complexes of Golowich \cite{Golowich2021}, whose construction we quickly describe. Given a weighted graph $G$ defines the product-complex $X_G$ as
\[
X_G(d) = \left\{ \{(v_1,s_1),\ldots,(v_d,s_d)\}: \{v_i\}_{i \in [d]} \in E, \{s_i\}_{i \in [d]} \in \binom{[n]}{d}\right\}.
\]
The measure of each face $\sigma = \{(v_1,s_1),\ldots,(v_d,s_d)\}$ corresponding to an edge $\{s,t\}$ is proportional to $w(s,t)\cdot f(j)$, where $f$ is some weighting function dependent only on the number $v_i=s$.

\begin{proof}[Proof of \pref{prop:TD-barrier}]
    Let $G = K^{(1)}_{n/2} \amalg K^{(1)}_{n/2} \amalg \{e\}$, that is two disjoint copies of $K_{n/2}$ with an additional edge $e$ passing between them, and take the weights of every edge to be uniform. \cite{Golowich2021} shows that there exists a weighting function $f$ such that $X_G$ is connected and has expansion $\frac{1}{d}$ in every $(d-2)$-link. By construction, the 1-skeleton of $X$ is a $(d+1)$-cover of $G$ with edges $(\{v_0,s_0\},\{v_1,s_1\})$ whenever $\{v_0,v_1\}$ is an edge in $E$, and $s_0 \neq s_1$. With this in mind, define the set $A$ to be the indicator of the first clique in each part:
    \[
    A \coloneqq \{ (v,s): v \in K^{(1)}_{n/2} \}
    \]
    $A$ clearly has measure $\frac{1}{2}$ by construction, while on the other hand any $\sigma \in \X[d]$ corresponding to an edge $e' \neq \{e\}$ is either entirely contained in or entirely misses $A$. Since all edges are evenly weighted, the probability that a $d$-face corresponds to $\{e\}$ is at most $O(n^{-2})$, which gives the desired bound for large enough $n$.
\end{proof}
The above proof extends immediately to a lower bound on sampling of $X_G$ under any weight function and any graph $G$ simply by taking $A$ to be any (nearly) balanced subset in $G$ with expansion (near) $\frac{1}{2}$. It is an easy exercise to show that choosing $A$ randomly suffices (and namely that such a set always exists).
\medskip

\indent Our second lower bound studies a slightly different regime. $1$-TD complexes typically have only constant local-spectral expansion at low level links, so \pref{prop:TD-barrier} does not, for instance, rule out showing $o_d(1)$-local-spectral HDX satisfy strong concentration. Using similar ideas to our degree lower bounds, we show at least some inverse polynomial local-spectral expansion is required to have strong concentration.
\begin{proposition}\label{prop:poly-lower-bound}
    For any $c < \frac{1}{2}$ and large enough $d \in \mathbb{N}$, there exists a family $\{X_n\}$ of $d^{-c}$-two-sided local spectral expanders such that for any $\beta<1$, $(\SC[X][d][n],\SC[X][1][n])$ is not a $(d^{-c},\beta)$-additive sampler.
\end{proposition}
The proof of \pref{prop:poly-lower-bound} relies on another construction of Golowich \cite{golowich2023grassmannian} building inverse polynomial HDX from Cayley complexes.
\begin{theorem}[{\cite{golowich2023grassmannian}}]
    For every $c<\frac{1}{2}$ and large enough $d \in \mathbb{N}$, there exists a family of \(d\)-\maximalpunc, $2^{\Omega_d(\sqrt{\log(n)})}$-degree, \(\frac{1}{d^c}\)-two-sided local spectral expanders $\{X_n\}$ whose $1$-skeletons are Cayley graphs over $\mathbb{F}_2^n$ with two-sided expansion \(\lambda = \Theta(\frac{1}{d^c})\).
\end{theorem}
We can now use \pref{lem:link-sampler} and the Cayley structure of $X$'s 1-skeleton to prove \pref{prop:poly-lower-bound}.
\begin{proof}[Proof of \pref{prop:poly-lower-bound}]
Since the $1$-skeleton of $X$ is a Cayley graph over \(\mathbb{F}_2^n\), there is an eigenvector \(g_0:\mathbb{F}_2^n \to \set{\pm 1}\) whose eigenvalue is \(\lambda\). Let us consider the indicator \(g\) of the set \(A = \sett{v \in \mathbb{F}_2^n}{g_0(v) = 1}\) whose measure is \(\prob{A} = \frac{1}{2}\). This indicator is equal to \(g=\frac{1}{2}+\frac{1}{2}g_0\). We note that for every \(v\),
\[
\Prob[u \sim v]{A} = \Ex[u \sim v]{g(u)} = \frac{1}{2} +\frac{1}{2}\Ex[u \sim v]{g_0(u)}.
\]
The function \(g_0\) is an eigenvalue so \(\frac{1}{2}\Ex[u \sim v]{g_0(u)} = \frac{\lambda}{2} g_0(v) = \pm \frac{\lambda}{2}\). Thus, \(\Prob[u \sim v]{A} - \prob{A} = \pm \frac{\lambda}{2}\) so for every \(\varepsilon \leq \frac{\lambda}{2}\) and any \(\beta < 1\), the underlying graph \emph{is not} a \((\varepsilon,\beta)\)-sampler. By \pref{lem:link-sampler}, this implies that either the containment graph is not a \((2\varepsilon,\beta)\)-sampler, or the containment graph of some link of a vertex is not a \((2\varepsilon,\frac{1}{2})\)-sampler. Since the vertex links of $\{X_n\}$ also give an infinite family of \(\frac{1}{d^c}\)-two-sided local spectral expanders this implies the result.
\end{proof}
This counterexample poses the following question: For  \(\varepsilon\in [d^{-c},d^{-c/2}]\) what is the worst \(\lambda\) which still promises that the containment graph is a \((\varepsilon,\exp(-d^c))\)-sampler for some \(c>0\)?

\subsection{Expander-Walks}\label{sec:expander-lower}
Finally we take a step back from concentration itself and look at necessity of \textit{local} concentration for our applications (that is concentration of the links of $X$ and in particular the locally nice property). We argue this is at least in some sense necessary: classical samplers such as expander walks which behave poorly under restriction fail reverse hypercontractivity. Recall the expander walk complexes \(W_G\) defined in \pref{sec:splitting}.

\begin{proposition}\label{prop:splitting-lower}
    For any $0 < \gamma < 1$ and $\lambda >0$, there are infinitely many pairs $k,n \in \mathbb{N}$ with a corresponding $\lambda$-expanders $G=([n],E)$ and balanced subset $A \subset W_G(k)$ satisfying:
    \[
    \Prob[t,t' \sim UD_{k,\gamma k}]{t' \in \X[1] \setminus A \ve t \in A} \leq 2^{-\Omega_\gamma(k)}
    \]
\end{proposition}
\begin{proof}
    Take $k$ odd and $n$ even. It suffices to take any regular graph $G$ with girth$(G)>k$. Define $A$ to be the set of walks whose center vertex is in $\set{1,2,\dots,\frac{n}{2}}$. Notice that $A$ is indeed a balanced function, as every element is the center of the same number of random walks in a regular graph (see e.g.\ \cite{AlevJQST2020} for details).

    On the other hand, it is easy to see $A$ has very poor expansion due to the girth of $G$. In particular, the only way to leave $A$ (resp. $\X[1] \setminus A$) is if the walk re-samples all of $[k/2+1]$, or all of $[k] \setminus [k/2]$. If neither event occurs, there exist indices $i < \lceil\frac{k}{2}\rceil < j$ which were not re-sampled, and the girth of the graph ensures there is only one option for the $\lceil\frac{k}{2}\rceil$th element (namely its starting position). One can compute directly the probability either event occurs is at most $2^{-\Omega_\gamma(k)}$.

    To complete the proof, we need an infinite family of regular $\lambda$-expanders for any $\lambda>0$ with super-constant girth. Many such constructions exist in the literature, including the classical Ramanujan graphs of \cite{lubotzky1988ramanujan}.
\end{proof}
Note that for any fixed $\gamma,C,\ell>0$, taking $k$ in \pref{prop:splitting-lower} sufficiently large means we can always find subsets $A,B \subset \X[k]$ such that
\[
\Pr_{t,t' \sim UD_{k,\gamma k}}[t \in A \land t' \in B] \leq C\Pr[A]^\ell\Pr[B]^\ell,
\]
violating reverse hypercontractivity for all parameter settings. We remark that the failure of expander walks in agreement testing is well known, and can be derived e.g.\ from \cite{goldenberg2019toward}.

\section*{Acknowledgements}
We thank Irit Dinur, Tali Kaufman, and Shachar Lovett for many helpful discussions on high dimensional expanders, sampling, and agreement testing. We thank Swastik Kopparty and Madhur Tulsiani for helpful discussion of list-decoding and amplification techniques for locally testable codes. We thank Vedat Alev for helpful discussions on concentration of measure and for pointing out an error in the proof of the partite case of \pref{thm:hdx-is-sampler} in a preliminary version of this manuscript. We thank Noam Lifshitz for pointers to work applying reverse hypercontractivity in extremal combinatorics. We thank Alex Lubotzky for pointing out the connection between high dimensional expanders and geometric overlap. Finally the authors thank the Simons Institute for the Theory of Computing for graciously hosting them for part of the duration of this work, and MH would further like to thank Irit Dinur and the Weizmann Institute for the same.

\printbibliography

\appendix
\section{Concentration for the Complete Complex}\label{app:complete}
In this section we show near-optimal concentration bounds for the complete complex.
\begin{claim} \label{claim:complete-complex-multiplicative-sampler}
    Let \(X\) be the complete complex on \(n\) vertices, that is \(\X[k] = \binom{[n]}{k}\). Let \(\alpha, \delta > 0\). Then for any $\ell \leq k$, the containment graph \(G = (\X[k],\X[\ell])\) is an \((\alpha,\frac{4}{\alpha \delta} \exp(-\frac{\delta^2}{12} \alpha \lfloor \frac{k}{\ell} \rfloor),\delta)\)-multiplicative sampler and \((\varepsilon,\frac{2}{ \varepsilon} \exp(-\frac{\varepsilon^2}{8} \lfloor \frac{k}{\ell} \rfloor))\)-additive sampler for any $\varepsilon>0$.
\end{claim}
\begin{proof}
    Fix and \(A \subseteq \X[\ell]\) such that \(\prob{A} \geq \alpha\). We apply \pref{lem:weird-distribution-to-uniform} in the following distribution. The distribution we use samples \((s_1,s_2,\dots,s_m,t) \sim D\) such that \((s_1,s_2,\dots,s_m)\) are independent and \(t \in \X[k]\) is a uniform face conditioned on containing all \(s_i\). Obviously, \((s_i,t)\) is a uniform pair of \(\ell\) and \(k\) faces where \(s_i \subseteq t\). Moreover, by Chernoff's bound for independent sampling
    \[
    \Prob[(s_1,s_2,\dots,s_m,t) \sim D]{\abs{\frac{|s_i \in A|}{m} - \prob{A}} > \frac{\delta}{2} \prob{A}} \leq 2\exp(-\frac{\delta^2}{12}\alpha m).
    \]
    The claim follows from \pref{lem:weird-distribution-to-uniform}. The additive bound follows by the same argument applying Hoeffding's inequality.
\end{proof}
Note that we do not restrict the number of vertices $n$ in this claim. Since the \(k\)-skeleton of the complete complex is only an \(\frac{1}{n-k}\)-two sided spectral expander, this claim does not follow from \pref{thm:hdx-is-sampler}.
\section{Concentration for the Swap Complex}\label{app:swap-complex}
In this appendix we prove \pref{thm:probability-to-deviate-from-expectation-in-complete-swap-walk}, which we restate here for convenience.
\restatetheorem{thm:probability-to-deviate-from-expectation-in-complete-swap-walk} The main technical claim we prove is the following. We reparametrize \(k=im\).
\begin{claim} \label{claim:complete-swap-walk-hoeffding}
    Let \(\varepsilon > 0\), $i,m \in \Z_+$, \(n \geq \frac{1153 im \log \left (\frac{2}{\varepsilon} \right )}{\varepsilon^5}\), and denote $C=C_{i,i m,n}$. Then for any \(A \subseteq C(1)\):
        \[
        \Prob[t=\set{s_1,s_2,\ldots,s_m} \in C(m)]{\cProb{}{A}{t} \geq \Pr[A] + \varepsilon} \leq \exp \left (-\frac{1}{9}\varepsilon^2 m \right ).
        \]
    The same lower bound applies when replacing \(\geq \Pr[A] + \varepsilon\) with \(\leq \Pr[A] - \varepsilon\).
\end{claim}

\pref{thm:probability-to-deviate-from-expectation-in-complete-swap-walk} directly follows from \pref{claim:complete-swap-walk-hoeffding} together with \pref{claim:sampler-for-functions}.

We will prove this claim by reducing to the following result of Panconesi and Srinivasan.
\begin{theorem}[{\cite[Theorem 3.2]{PanconesiS1997}}]\label{thm:Panco-Chernoff}
    Let \(f_1,f_2,\ldots,f_m\) be functions on a probability space with \(\mu = \sum_{\ell=1}^m \ex{f_i}\), and such that \(a_\ell \leq f_\ell \leq b_\ell\). Assume that for every \(n_1,n_2,\ldots,n_m \in \NN\)
    \[
        \Ex{\prod_{\ell=1}^m f_\ell^{n_\ell}} \leq \eta \prod_{\ell=1}^m \ex{f_\ell^{n_\ell}}.
    \]
    Then
    \[
        \prob{\sum_{\ell=1}^m f_\ell \geq \mu + t} \leq \eta \exp \left ( \frac{-2t^2}{\sum_{\ell=1}^m (b_\ell-a_\ell)^2} \right ).
    \]
\end{theorem}

\begin{proof}[Proof of \pref{claim:complete-swap-walk-hoeffding}]
    We prove the bound for \(\geq \prob{A} + \varepsilon\); the \(\leq \prob{A} - \varepsilon\) follows from the original inequality applied on the complement \(\SC[1] \setminus A\). We also assume without loss of generality that \(\prob{A} \geq \frac{\varepsilon}{2}\), and prove that 
    \begin{equation}\label{eq:swap-ind-Chern}
        \Prob[\set{s_1,s_2,\ldots,s_m} \in C(m)]{\frac{\abs{s_\ell \in A}}{m} \geq \mu + \varepsilon} \leq \exp \left (-\varepsilon^2 m \right ).
    \end{equation}
    This suffices since, by our assumption that $n \geq \frac{1153}{\varepsilon^5}$, for any \(\prob{A} \leq \frac{\varepsilon}{2}\) we may always find \(B \supseteq A\) with \(\prob{B} \in  \left ( \frac{\varepsilon}{2}, \frac{2\varepsilon}{3} \right )\) and deduce
    \[
        \Prob[\set{s_1,s_2,\ldots,s_m} \in C(m)]{\frac{\abs{s_\ell \in A}}{m} \geq \mu + \varepsilon} \leq \Prob[\set{s_1,s_2,\ldots,s_m} \in C(m)]{\frac{\abs{s_\ell \in B}}{m} \geq \prob{B} + \frac{\varepsilon}{3}} \leq \exp \left (-\frac{\varepsilon^2}{9} m \right )
    \]
    where we've applied \pref{eq:swap-ind-Chern} to the appropriately sized set \(B\). \medskip
    
    Intuitively, we wish to prove \pref{eq:swap-ind-Chern} by applying the theorem of \cite{PanconesiS1997} where each \(f_\ell:C(m)\to \set{0,1}\) is the indicator that \(s_\ell \in A\) (ignoring the rest of the coordinates), $t=\varepsilon m$, and $\eta=\exp(\varepsilon^2 m)$. Then since $a_\ell=0$ and $b_\ell=1$, we'd have
    \[
    \Pr\left[\frac{|s_\ell \in A|}{m} \geq \mu + \varepsilon\right] \leq \eta\exp(-2\varepsilon^2 m) = \exp(-\varepsilon^2 m)
    \]
    as desired.
    
    Unfortunately, since $C(m)$ is unordered, this is not quite well-defined. Instead, we will apply \pref{thm:Panco-Chernoff} as above to the \emph{partification} of $C(m)$, denoted $\widetilde{C}(m)$, the distribution over \emph{ordered} $m$-tuples $(s_1,\ldots,s_m)$ such that each $s_\ell \in \binom{[n]}{ i}$ and all $s_\ell$ are pairwise disjoint. Since the $\sum\limits_{\ell=1}^m f_\ell=|s_\ell \in A|$ is independent of the ordering of the face, proving \pref{eq:swap-ind-Chern} for the partification clearly implies the desired unordered bound as well.
    
    Moving back now to the statement of \pref{thm:Panco-Chernoff}, for any $I \subseteq [m]$, let $s_I \in A$ be shorthand for the event that $\forall \ell \in I, s_\ell \in A$. The expectations \(\Ex{\prod_{\ell=1}^m f_\ell^{n_\ell}}\) in our setting are then of the form \(\Prob[(s_1,s_2,\ldots,s_m) \in \widetilde{C}(m)]{s_I \in A}\), so we need to show $\forall I \subseteq [m]$:
        \[
            \Prob[(s_1,s_2,\ldots,s_m) \in \widetilde{C}(m)]{s_I \in A}\leq \exp(\varepsilon^2 m) \prob{A}^{|I|},
        \]
    We will prove the slightly stronger bound by induction on $|I|$:
    \[
        \Prob[(s_1,s_2,\ldots,s_m) \in \widetilde{C}(m)]{s_I \in A}\leq (1+\varepsilon^2)^{|I|} \prob{A}^{|I|}.
    \]
    The base case \(|I|=1\) clearly holds since the marginals $s_\ell \sim  \binom{[n] }{ i}$ are uniformly random. Assume then the result holds for all sets of size \(r\). Let \(I\) be a set of size $r+1$ and write \(I = J \dunion \set{\ell}\). 
    
    The idea is to split the analysis based on whether or not the marginal $s_J$ samples $A$ well. Toward this end, define the set of $s_J$ that over-sample $A$ as
    \[
    T_J \coloneqq \left\{s_J : \Pr_{s_I}[s_i \in A | s_J] \geq (1+\frac{1}{2}\varepsilon^2)\prob{A}\right\}.
    \]
    We can bound the probability that $s_I \in A$ by
    \begin{align*}
        \Prob[(s_1,s_2,\ldots,s_m) \in \widetilde{C}(m)]{ s_I \in A} &\leq  \Prob{T_J} + \Prob[(s_1,s_2,\ldots,s_m) \in \widetilde{C}(m)]{\overline{T_J} \ve s_I \in A}\\
        &\leq \prob{T_J} + (1+\frac{1}{2}\varepsilon^2)\prob{A}\cdot \prob{s_J \in A} & \text{(Def of $T_J$)}  \\
        &\leq \prob{T_J} + (1+\frac{1}{2}\varepsilon^2)(1+\varepsilon^2)^{|J|} \prob{A}^{|J|+1} & \text{(Ind.\ Hypothesis)}
    \end{align*}
    Thus it remains to argue that \(\prob{T_J} \leq \frac{\varepsilon^2}{2} (1+\varepsilon^2)^{|J|}\prob{A}^{|J|+1}\). In particular, by our assumption that $\Pr[A] \geq \frac{\varepsilon}{2}$, it suffices to show 
    \[
    \prob{T_J} \leq \frac{\varepsilon^{m+2}}{2^{m+1}}.
    \]
    We prove this by reducing analysis of $\prob{T_J}$ to the probability of mis-sampling $A$ on a higher level of the complete complex, an event we've already shown has a subgaussian tail in \pref{app:complete}.

    Toward this end, let $t= \bigcup_{j \in J} s_j$ denote the disjoint union of the $\ell$-sets in $s_J$, viewed as an element in $\binom{[n]}{ i|J|}$, and write $\bar{t}=[n] \setminus t$. Observe that the conditional distribution of $s_\ell$ given $s_J$ is simply uniform over $\binom{\bar{t}}{ i}$ --- in other words, it depends \textit{only} on the union $t$. As such, the event \(\set{s_j}_{j \in J} \in T_J\) also depends only on the union, namely \(\set{s_j}_{j \in J} \in T_J\) if and only if
    \[
    \cProb{s_\ell \cup s_J}{s_\ell \in A}{s_\ell \subseteq \overline{t}} > (1+\frac{1}{2}\varepsilon^2) \prob{A}.
    \]
    Since the distribution over $\bar{t}$ generated as above is uniform over $\binom{[n]}{n-i|J|}$, we may rewrite the condition as:
    \[
    \Pr[T_J] = \Prob[\overline{t}]{\cProb{s_\ell}{A}{s_\ell \subseteq \overline{t}} > (1+\frac{1}{2}\varepsilon^2) \prob{A}} .
    \]
    The righthand side is now exactly the probability of (multiplicatively) mis-sampling $A$ in the inclusion graph of the complete complex over \(n\) vertices $(\Delta_{n}(n-i|J|),\Delta_n(\ell))$. In \pref{claim:complete-complex-multiplicative-sampler}, we proved this graph is a \((\frac{\varepsilon}{2},\frac{16}{\varepsilon^3} \exp \left (-\frac{\varepsilon^5}{96}\lfloor \frac{|\overline{t}|}{i} \rfloor \right ) ,\frac{\varepsilon^2}{2})\)-multiplicative sampler. Since $\Pr[A] \geq \frac{\varepsilon}{2}$ by assumption, we therefore have
    \[
    \Pr[T_J] = \Prob[\overline{t}]{\cProb{s_\ell}{A}{s_\ell \subseteq \overline{t}} > (1+\frac{1}{2}\varepsilon^2) \prob{A}} \leq \frac{16}{\varepsilon^3} \exp \left (-\frac{\varepsilon^5}{192}\frac{|\overline{t}|}{i}  \right ).
    \]
    The size of \(\overline{t}\) is at least \(n-im \geq \frac{1152im  \log \left (\frac{2}{\varepsilon} \right )}{\varepsilon^5}\), so 
    \[
    \Prob{T_J} \leq \frac{16}{\varepsilon^3} \exp \left (6m \log \left ( \frac{\varepsilon}{2} \right )  \right ) \leq \frac{\varepsilon^{m+2}}{2^{m+1}}
    \]
    for all $m \geq 1$ as desired.
\end{proof}
\section{Outstanding Proofs on Samplers and Concentration} \label{app:sampler-proofs}
\subsection{Sampling}
\begin{claim}[{\pref{claim:opposite-sampler}} Restated]
    Let \(\beta,\delta > 0\), let \(\delta' > \delta\) and \(\alpha < \frac{\min \set{\delta,0.5}}{1+\delta}\). Then for every \((\alpha,\beta,\delta)\)-sampler \(G=(L,R,E)\), it holds that \(G_{op} \coloneqq (R,L,E)\) is a \((\frac{1-\alpha(1+\delta)}{\alpha(\delta' - \delta)}\beta, 2\alpha, \delta')\)-sampler.
\end{claim}
\begin{proof}[Proof of \pref{claim:opposite-sampler}]
    Let \(A \subseteq L\) be such that \(\prob{A} \geq \frac{1-\alpha(1+\varepsilon)}{\alpha(\varepsilon' - \varepsilon)}\beta\). Let \[B_S = \sett{v \in R}{\cProb{u \in L}{u \in A}{u \sim v} < (1-\varepsilon')\prob{A}}\] and let \[B_B = \sett{v \in R}{\cProb{u \in L}{u \in A}{u \sim v} > (1+\varepsilon')\prob{A}}.\]
    Showing that \(\prob{B_B},\prob{B_S} \leq \alpha\) will prove the claim. Let us begin with \(B_S\). Assume toward contradiction that \(\prob{B_S} > \alpha\). We denote by \(\one_A:L \to \set{0,1}\) and \(\one_{B_S}:R \to \set{0,1}\) the indicators of \(A,B_S\) respectively. On the one hand
    \begin{equation} \label{eq:B_S-lower-bound}
        \Ex[uv \in E]{\one_A(u) \one_{B_S}(v)} < \prob{B_S} (1-\varepsilon') \prob{A},
    \end{equation}
    since the product of indicators is one iff \(v \in B_S\), in which case, there are at most a \((1-\varepsilon')\prob{A}\) fraction of \(u \in A\) adjacent to \(v\) by the definition of \(B_S\). On the other hand, by the sampling properties of \((L,R,E)\) there is at most a \(\beta\)-fraction of \(u \in L\) such that \(\cProb{v \in R}{v \in B_S}{v \sim u} < (1-\varepsilon) \prob{B_S}\). Hence 
    \begin{equation} \label{eq:B_S-upper-bound}
        \Ex[uv \in E]{\one_A(u) \one_{B_S}(v)} \geq \prob{B_S} (1-\varepsilon) (\prob{A} - \beta).\footnote{\(\prob{A}>\beta\) by assumption.}
    \end{equation}
Combining \eqref{eq:B_S-lower-bound} with \eqref{eq:B_S-upper-bound} yields
\[\prob{B_S}\prob{A}(\varepsilon' - \varepsilon) < \beta \prob{B_S}(1-\varepsilon) \]
or
\[ \prob{A} < \frac{1-\varepsilon}{\varepsilon' - \varepsilon} \beta \overset{(\alpha<\frac{\varepsilon}{1+\varepsilon})}{\leq} \frac{1-\alpha(1+\varepsilon)}{\alpha(\varepsilon' - \varepsilon)}\beta,\]
a contradiction to the lower bound on the size of \(A\).

Let us bound \(B_B\). Let \(\one_{B_B}:R \to \set{0,1}\) be the indicator of \(B_B\). Similar to \eqref{eq:B_S-lower-bound} we have that
    \begin{equation} \label{eq:B_B-upper-bound}
        \Ex[uv \in E]{\one_A(u) \one_{B_S}(v)} > \prob{B_B} (1+\varepsilon') \prob{A}.
    \end{equation}
By the sampling properties of \((L,R,E)\) there is at most a \(\beta\)-fraction of \(u \in L\) such that \(\cProb{v \in R}{v \in B_B}{v \sim u} > (1+\varepsilon) \prob{B_B}\). So similar to \eqref{eq:B_S-upper-bound} we have that
    \begin{equation} \label{eq:B_B-lower-bound}
        \Ex[uv \in E]{\one_A(u) \one_{B_S}(v)} \leq \beta + \prob{B_B} (1+\varepsilon) (\prob{A} - \beta).
    \end{equation}
Combining \eqref{eq:B_B-upper-bound} with \eqref{eq:B_B-lower-bound} yields
\[\prob{B_B}\prob{A}(\varepsilon - \varepsilon') < \beta \left (1 - \prob{B_B}(1+\varepsilon) \right ) \]
or
\[ \prob{A} < \frac{1- \prob{B_B}(1+\varepsilon)}{\prob{B_B}(\varepsilon' - \varepsilon)} \beta \leq \frac{1-\alpha(1+\varepsilon)}{\alpha(\varepsilon'-\varepsilon)}\beta,\]
where the last inequality is because the function \(x \mapsto \frac{1-(1+\varepsilon)x}{x(\varepsilon'-\varepsilon)}\beta\) is monotone increasing when \(x \in [0,\frac{0.5}{1+\varepsilon}]\) and \(\alpha\) is in this interval.
This is a contradiction to the lower bound on the size of \(A\).
\end{proof}
\begin{claim}[{\pref{claim:additive-multiplicative-sampler-equivalence}} Restated]
Let \(G=(L,R,E)\) be a bipartite graph.
\begin{enumerate}
    \item If \(G\) is a \((\beta,\delta)\)-additive sampler then \(G\) is a \((C \delta,\beta,\frac{1}{C})\)-multiplicative sampler for any \(C>1\).
    \item If \(G\) is a \((\alpha,\beta,\delta)\)-multiplicative sampler for \(\alpha \leq \frac{1}{2}\). Then \(G\) is a \((\beta,\delta)\)-additive sampler, where \(\delta = \max \set{\delta,(1+\delta)(\alpha+p)}\) and \(p = \max_{v \in R} \prob{v}\).
\end{enumerate}
\end{claim}
\begin{proof}[Proof of \pref{claim:additive-multiplicative-sampler-equivalence}]
    The first item follows immediately from the definition of multiplicative samplers. If \(\prob{A} \geq C\varepsilon\) then \(\frac{1}{C} \prob{A} \geq \varepsilon\) so by the promise of additive sampling it holds that
    \[\Prob[v \in L]{\abs{\cProb{u \in R}{u \in A}{u \sim v}-\prob{A}} > \frac{1}{C}\prob{A} } \leq \Prob[v \in L]{\abs{\cProb{u \in R}{u \in A}{u \sim v}-\prob{A}} > \varepsilon } \leq \beta.\]

    We turn to the second item. Let \(A \subseteq R\). For every \(u \in L\) it holds that 
    \[\abs{\cProb{u \in R}{u \in A}{u \sim v}-\prob{A}} = \abs{1-\cProb{u \in R}{u \in A}{u \sim v}-(1-\prob{A})} = \abs{\cProb{u \in R}{u \in R \setminus A}{u \sim v}-\prob{R \setminus A}},\]
    so we assume that without loss of generality \(\prob{A} \leq \frac{1}{2}\). We need to show that 
    \begin{equation}
        \Prob[v \in L]{\abs{\cProb{u \in R}{u \in A}{u \sim v}-\prob{A}} > \delta} \leq \beta.
    \end{equation}
    
    If \(\prob{A} \geq \alpha\) then this holds from the multiplicative sampling guarantee. Hence it suffices to show that this inequality \emph{always holds} for sets of relative size smaller that \(\delta\).

    First we note that if \(u \in L\) is such that \(\cProb{u \in R}{u \in A}{u \sim v} < \prob{A}\) then 
    \(\abs{\cProb{u \in R}{u \in A}{u \sim v}-\prob{A}} \leq \prob{A} \leq \delta\)
    so 
    \begin{align*}
        \Prob[v \in L]{\abs{\cProb{u \in R}{u \in A}{u \sim v}-\prob{A}} > \delta} =
        \Prob[v \in L]{\cProb{u \in R}{u \in A}{u \sim v} >\delta + \prob{A}}.
    \end{align*}
    In this case we find a subset \(B \supseteq A\) such that \(\prob{B} \leq  \alpha + p\) (we can do so by adding vertices to \(A\) one-by-one which is where \(p\) comes into play). Then
    \begin{align*}
    \Prob[v \in L]{\cProb{u \in R}{u \in A}{u \sim v} > \delta  + \prob{A}} &\leq
    \Prob[v \in L]{\cProb{u \in R}{u \in B}{u \sim v} > \delta  + \prob{A}} \\
    &\leq \Prob[v \in L]{\cProb{u \in R}{u \in B}{u \sim v} > (1+\varepsilon)\prob{B}}\\
    &\leq \beta.
    \end{align*}
    The claim follows. 
\end{proof}
One can also remove the dependence on \(\delta\) in the second item.
\begin{claim}[{\pref{claim:improved-mult-to-add-sampler}} Restated]
    Let \(\beta, \alpha_0 > 0\). If for every \(\alpha > \alpha_0\) it holds that \(G\) is an \((\alpha,\beta,\frac{\alpha_0}{\sqrt{\alpha}})\)-multiplicative sampler, then \(G\) is a\((\beta,2(\alpha+p))\)-additive sampler where \(p = \max_{v \in R} \prob{v}\).
\end{claim}
\begin{proof}[Proof of \pref{claim:improved-mult-to-add-sampler}]
    The proof is the same as the proof of the second item in  \pref{claim:additive-multiplicative-sampler-equivalence}. the only difference is that for \emph{large} sets, we still have that
     \begin{equation}
        \Prob[v \in L]{\abs{\cProb{u \in R}{u \in A}{u \sim v}-\prob{A}} > 2\alpha_0+p} \leq \beta
    \end{equation}
    because if \(\prob{A}=\alpha\) then \(\frac{\alpha_0}{\sqrt{\alpha}}\alpha \leq \alpha_0\).
\end{proof}
\restateclaim{claim:sampler-for-functions}

\begin{proof}[Proof of \pref{claim:sampler-for-functions}]
Fix \(f:L \to [0,1]\). We prove a \((2\beta, 2\varepsilon)\) upper tail bound. The lower bound is similar. Let \(A \sim \mathcal{P}(L)\) be a random subset where every vertex is inserted into \(A\) independently with probability \(p_v = f(v)\). Since $G$ is a sampler, for every possible outcome \(A\) we have \(\Prob[r \in R]{\Prob[v \sim r]{v \in A}  - \prob{A} > \varepsilon} < \beta\). Let \(\one(A,r)\) be the indicator for this event. In particular, we have that \(\Ex[A,r \sim R]{\one(A,r)} \leq \beta\). By Markov, the fraction of \(r_0 \in R\) such that \(\Ex[A]{\one(A,r_0)} > \frac{1}{2}\) is at most \(2\beta\). On the other hand, we will show that if
\[
\Ex[v \sim r_0]{f(v)} \geq \Ex[v \in L]{f(v)} + 2\varepsilon
\]
then \(\Ex[A]{\one(A,r_0)} > \frac{1}{2}\). Indeed, fix such an \(r_0\). Observe that if \(\Prob[v \in L]{A} < \Ex[v \in L]{f} + \frac{1}{2}\varepsilon\) and \(\Prob[v \sim r_0]{A} > \Ex[v \sim r_0]{f(v)} - \frac{1}{2}\varepsilon\), then 
\[
\Prob[v \sim r_0]{A} > \Ex[v \sim r_0]{f(v)} - \frac{1}{2}\varepsilon \geq \Ex[v \in L]{f(v)} + \frac{3}{2}\varepsilon \geq \prob{A} + \varepsilon.
\]
It is direct to check that \(\Ex[A]{\prob{A}}=\ex{f}\) and that \(\Ex[A]{\Prob[v \sim r_0]{A}} = \Ex[v \sim r_0]{f(v)}\). Hoeffding's inequality therefore bounds the probability that either \(\Prob[v \in L]{A} < \Ex[v \in L]{f} + \frac{1}{2}\varepsilon\) or \(\Prob[v \sim r_0]{A} > \Ex[v \sim r_0]{f(v)} - \frac{1}{2}\varepsilon\) by \(2\exp(-0.01\varepsilon^2 k) < \frac{1}{2}\) and \(\Ex[A]{\one(A,r_0)} > \frac{1}{2}\) follows as desired.
\end{proof}
Finally we prove the basic Chebyshev-type sampling bound from expansion.
\restateclaim{claim:chebyshev-expanders}
\begin{proof}[Proof of \pref{claim:chebyshev-expanders}]
Dinur and Kaufman observe the following bound
    \[
        \Pr_{u \sim R}[|Af(u) - \mu| > \varepsilon ] < \frac{\lambda^2 Var(f)}{\varepsilon^2}.
    \]
    Let \(f^\perp = f-\mu\). Then \(Af(v) - \mu = Af^{\perp}\). By Chebyshev's inequality,
    \(\prob{T} \leq \frac{\ex{(Af^{\perp})^2}}{\varepsilon^2}\)
    and by \(\lambda\)-expansion \(\ex{(Af^{\perp})^2} \leq \lambda^2 \ex{(f^\perp)^2} = \lambda^2 Var(f)\). The claim follows.
    Now by the above we can write
    \[
    \prob{Af(u) < \mu - \varepsilon} \leq \prob{\abs{Af(u)-\ex{f}} > \ex{f}-\mu + \varepsilon} \leq \frac{\lambda^2 Var(f)}{(\varepsilon + (\ex{f}-\mu))^2}.
    \]
    Noting that \(Var(f) \leq \ex{f} = \mu+ (\ex{f}-\mu)\) and denoting \(x = \ex{f}-\mu\), this is at most \(\frac{\lambda^2 (\mu + x)}{(\varepsilon + x)^2} \leq \frac{\lambda^2 \mu }{\varepsilon^2}\). 
\end{proof}
\subsection{Concentration}
\begin{lemma}[{\pref{lem:raising}} Restated]
    Let $X$ be a $d$-\maximal simplicial complex, $k \leq d$, and $f: \X[k] \to [0,1]$ a function satisfying
    \begin{enumerate}
        \item \textbf{Upper Tail}: $\underset{\X[k]}{\Pr}[f-\mathbb{E}[f] > t] \leq f_{up}(t)$
        \item \textbf{Lower Tail}: $\underset{\X[k]}{\Pr}[f-\mathbb{E}[f] < -t] \leq f_{low}(t)$.
    \end{enumerate}
    for some functions $f_{up},f_{low}: \R_+ \to [0,1]$. Then the $d$-lift $U_{k,d}f: \X[d] \to \R$ satisfies:
    \begin{enumerate}
        \item \textbf{Upper Tail}: $\underset{\X[d]}{\Pr}[U_{k,d}f-\mathbb{E}[f] > t] \leq f_{up}(\frac{t}{2})(1-\pi_{low}^{d,k,f}(\frac{t}{2}))^{-1}$
        \item \textbf{Lower Tail}: $\underset{\X[d]}{\Pr}[U_{k,d}f-\mathbb{E}[f] < -t] \leq f_{low}(\frac{t}{2})(1-\pi_{up}^{d,k,f}(\frac{t}{2}))^{-1}$.
    \end{enumerate}
\end{lemma}
\begin{proof}
    We give the argument for the upper tail. The argument for the lower tail is analogous. Assume for the sake of contradiction that
    \[
    \underset{\X[d]}{\Pr}\left[U_{k,d}f - \mu > 2t\right] > f_{up}\left(t\right)\left(1-\pi_{low}^{f,k,d}\left(t\right)\right)^{-1}.
    \]
    Since a $k$-face can be drawn by first drawing a $d$-face, then a uniformly random $k$-subface, we can use concentration of the complete complex to derive a contradiction. In particular, let $E_1$ denote the event that $U_{k,d}f - \mu > 2t$. Then:
    \begin{align*}
\underset{\{v_0,\ldots,v_{k}\} \in \X[k]}{\Pr}\left[U_{k,d} f - \mu > t\right] &\geq \Pr[E_1]\Pr_{r \subset s}[f(r)-\mu > t  ~|~s \in E_1]\\
&\geq\Pr[E_1](1-\Pr_{r \subset s}[f(r)-\mathbb{E}[f] < t ~|~s \in E_1])\\
&\geq \Pr[E_1](1-\Pr_{r \subset s}[f(r)- U_{k,d} < -t ~|~s \in E_1])\\
&> f_{up}(t)
    \end{align*}
\end{proof}
\begin{lemma}[{\pref{lem:lowering}} Restated]
    Let $X$ be a $d$-\maximal simplicial complex and $k \leq d$. Assume there exist functions $f_{up}(t,\nu)$ and $f_{low}(t,\nu)$ such that any $\nu$-Lipschitz $f: \X[d] \to \R$ satisfies:
    \begin{enumerate}
        \item \textbf{Upper Tail}: $\Pr[f-\mathbb{E}[f] > t] \leq f_{up}(t,\nu)$
        \item \textbf{Lower Tail}: $\Pr[f-\mathbb{E}[f] < -t] \leq f_{low}(t,\nu)$.
    \end{enumerate}
    Then any function $f':\X[k] \to \R$ with $\nu$-bounded difference satisfies:
        \begin{enumerate}
        \item \textbf{Upper Tail}: $\Pr[f'-\mathbb{E}[f'] > t] \leq f_{up}(\frac{t}{2},\frac{k}{d}\nu) + e^{-\frac{t^2}{4\nu}}$
        \item \textbf{Lower Tail}: $\Pr[f'-\mathbb{E}[f'] < -t] \leq f_{low}(\frac{t}{2},\frac{k}{d}\nu) + e^{-\frac{t^2}{4\nu}}$.
    \end{enumerate}
\end{lemma}
\begin{proof}
    Let $f:\X[k] \to \R$ be a function with $\nu$-bounded difference and $g$ its $d$-lift $g=U_{k,d} f$. Since a face $r \in \X[k]$ can equivalently be sampled by first drawing a face $s \in \X[d]$, then subsampling $r$ uniformly from $s$, we can bound the concentration of $f$ as:
    \begin{align*}
        \Pr[f(r) - \mathbb{E}[f] \geq t] &\leq \Pr_{r \subset s}[g(s) - \mathbb{E}[f] \geq t/2 \lor f(r) - g(s) \geq t/2]\\
        &\leq \Pr_{r \subset s}[g(s) - \mathbb{E}[g] \geq t/2] + \Pr_{r \subset s}\left[f(r) - \underset{r \subset s}{\mathbb{E}}[f(r)] \geq t/2\right]
    \end{align*}
    where we have used the fact that $\mathbb{E}[g]=\mathbb{E}[f]$. Informally, the idea is then to argue the first term is small due to exponential concentration of $\X[d]$ and the fact that the lift $g$ is itself more concentrated than $f$, and the second term is small by subgaussian concentration of the complete complex.

    We first argue that if $f$ has $\nu$-bounded difference, $g$ is $\frac{k}{d}\nu$-Lipschitz. Given $s \in \X[d]$ and $i \in [d]$, let $s^{(i)}$ denote a copy of $s$ with the $i$th coordinate re-sampled. Expanding $\nu_g$ we have: 
    \begin{align*}
    \nu_g &\leq 
    \sum\limits_{i=1}^d \underset{s,s^{(i)}}{\mathbb{E}}\left[\left( g(s)-g(s^{(i)})\right)_+^2\right]\\
    &= \sum\limits_{i=1}^d \underset{s,s^{(i)}}{\mathbb{E}}\left[\left( \underset{t \subset s}{\mathbb{E}}[f] - \underset{t \subset s^{(i)}}{\mathbb{E}}[f]\right)_+^2\right]\\
    &=\frac{1}{\binom{d}{k}^2}\sum\limits_{i=1}^d \underset{s,s^{(i)}}{\mathbb{E}}\left[\left( \sum\limits_{T \subseteq [d]} f(s_T) - f(s^{(i)}_T)\right)_+^2\right]
    \end{align*}
    where we recall $s_T$ is the $k$-face consisting of the elements of $s$ indexed by $T$. Observe that the inner term is non-zero only when $i\in T$, since otherwise $s_T = s^{(i)}_T$. Thus by Cauchy-Schwarz, we have:
    \begin{align*}
     \nu_f &\leq \frac{\binom{d-1}{k-1}}{\binom{d}{k}^2}\sum\limits_{i=1}^d \underset{s,s^{(i)}}{\mathbb{E}}\left[ \sum\limits_{T \subseteq [d]} \left(f(s_T) - f(s^{(i)}_T)\right)^2\right]\\
    &\leq \frac{\binom{d-1}{k-1}}{\binom{d}{k}^2}\sum\limits_{i=1}^d \underset{s,s^{(i)}}{\mathbb{E}}\left[ \sum\limits_{i \in T \subseteq [d]} \frac{\nu}{k}\right]\\
    &=\frac{k}{d}\nu
    \end{align*}  
    where we have used the fact that $s(T)$ and $s'_{(i)}(T)$ are neighboring $k$-faces in the down-up walk on level $k$. We can now bound the first term of our upper tail by assumption as
    \[
    \Pr_{r \subset s}[g(s) - \mathbb{E}[g] \geq t/2] \leq f_{up}\left(\frac{t}{2},\frac{k}{d}\nu\right).
    \]
    For the second term, observe that since $f$ has $\nu$-bounded difference as a function of $\X[k]$, the restriction of $f$ to $\binom{s}{k} \subset \X[k]$ also has $\nu$-bounded difference (since any edge $(r,r')$ in the down-up walk restricted to $s$ is also an edge in the walk on $\X[k]$). Thus using subgaussian concentration for Lipschitz functions on the complete complex (see e.g.\ \cite[Corollary 2]{cryan2019modified}) we also have
    \[
    \Pr_{r \subset s}\left[f(r) - \underset{r \subset s}{\mathbb{E}}[f(r)] \geq t/2\right] \leq e^{-\frac{t^2}{4\nu}}.
    \]
    Combining the two gives the result.
\end{proof}

\section{High Dimensional Expander-Mixing Lemma}\label{app:HD-EML}
In this section, we give the proof of our variant of the high-dimensional expander mixing lemma and its various corollaries.
\begin{theorem}[high dimensional expander Mixing Lemma (\pref{thm:hdeml} Restated)]
Let \((X,T,\rho)\) be a depth $D$ \(\lambda\)-tuple splitting tree with \(k\) leaves. Denote by \(d_i\) the depth of the leaf labeled \(i\). Then for any family of functions \(\{f_i:X[i] \to \RR\}_{i \in [k]}\):
\begin{equation*}
    \abs{\Ex[a \in {\X[k-1]}]{\prod_{i=1}^k f_i(a_i)}- \prod_{i=1}^k \Ex[{a_i \in X[i]}]{f_i(a_i)}} \leq 3^{D} \lambda \prod_{i=1}^k \norm{f_i}_{2^{d_i}}.
\end{equation*}
If \((X,T,\rho)\) is instead a depth $D$ standard $\lambda$-splitting tree, we take \(\{f_i:\X[1] \to \RR\}_{i \in [k]}\) and have:
\begin{equation*}
    \abs{\Ex[a\in {\X[k-1]}, \pi \in S_k]{\prod_{i=1}^k f_i(a_{\pi(i)})}- \prod_{i=1}^k \Ex[a_i \in {\X[1]}]{f_i(a_i)}} \leq 3^{D} \lambda \prod_{i=1}^k \norm{f_i}_{2^{d_i}}.
\end{equation*}
\end{theorem}

\begin{proof}
We prove the partite case directly, then argue \pref{eq:hdeml-unordered} follows from partitification. It is convenient to introduce the following notation. Let \(\bar{f} = (f_1,f_2,\ldots,f_k)\) such that \(f_i:X[i] \to \RR\). Let \(S \subseteq [k]\). Let \((T,\rho)\) be an ordered tree. Without loss of generality we identify every leaf \(u \in \mathcal{L}\) with its label \(\rho(u)\).
\begin{itemize}
    \item Denote by \(\bar{f}_S\) the sub-tuple containing only functions whose index is in \(S\).
    \item Let \(\pi_{\bar{f}_S}:X[S] \to \RR\) be the product of functions in $\bar{f}_S$, that is \(\pi_{\bar{f}_S}(s) = \prod_{i \in S} f_i(a_i)\) (where \(s = \set{a_i}_{i \in S}\)). We note that when \(S_1,S_2\) are disjoint then \(\pi_{\bar{f}_{S_1}} \cdot \pi_{\bar{f}_{S_2}} = \pi_{\bar{f}_{S_1 \dunion S_2}}\).
    \item Let \(E_{\bar{f}_S} = \prod_{i \in S} \ex{f_i}\).
    \item Let \(T_L\) and \(T_R\) be the sub-trees rooted by the left and right children of the root of \(T\), respectively. Let \(L,R\) be the leaves of \(T_L, T_R\) respectively.
    \item For \(T,T_L,T_R\) and tuple of functions \(\bar{f}\) we denote by 
    \[\mathcal{E}(\bar{f}_{[k]},T) = \prod_{i=1}^k\norm{f_i}_{2^{d_i}}, \; \mathcal{E}(\bar{f}_{L},T_L) = \prod_{i\in L}\norm{f_i}_{2^{d_i-1}}, \; \mathcal{E}(\bar{f}_{R},T_R) = \prod_{i\in R}\norm{f_i}_{2^{d_i-1}}.\]
\end{itemize}

In this notation we need to show: 
\begin{equation*} \label{eq:hdeml-what-we-need-to-prove}
\abs{\ex{\pi_{\bar{f}_{[k]}}} - E_{\bar{f}_{[k]}}} \leq 3^D \lambda \mathcal{E}(\bar{f}_{[k]},T).
\end{equation*}

The proof is by induction on the depth $D$. The base case (depth \(1\)) is trivial. Assume that the theorem holds for any set of functions on any tree of depth \(\leq D-1\). By adding and subtracting $\ex{\pi_{\bar{f}_{L}}} \cdot \ex{ \pi_{\bar{f}_{R}}}$, observe that
\begin{equation} \label{eq:hdeml-triangle-inequality}
    \abs{\ex{\pi_{\bar{f}_{[k]}}} - E_{\bar{f}_{[k]}}} \leq 
     \underbrace{\abs{\ex{\pi_{\bar{f}_{[k]}}} - \ex{\pi_{\bar{f}_{L}}} \cdot \ex{ \pi_{\bar{f}_{R}}} }}_{\text{I}} + 
    \underbrace{\abs{\ex{\pi_{\bar{f}_{L}}} \cdot \ex{ \pi_{\bar{f}_{R}}} - E_{\bar{f}_{[k]}}}}_{\text{II}}
\end{equation}
by the triangle inequality. We bound I and II separately, starting with I. Since the swap walk \(S_{L,R}\) is a \(\lambda\)-bipartite expander,  applying \eqref{eq:basic-bipartite-expanders} to \(\pi_{\bar{f}_L},\pi_{\bar{f}_R}\) we get that the first term in the right-hand side of \eqref{eq:hdeml-triangle-inequality} is 
\begin{equation*}
    \abs{\ex{\pi_{\bar{f}_{[k]}}} - \ex{\pi_{\bar{f}_L}} \ex{\pi_{\bar{f}_R}} } 
    = \abs{\ex{\pi_{\bar{f}_L} \cdot \pi_{\bar{f}_R}} - \ex{\pi_{\bar{f}_L}} \ex{\pi_{\bar{f}_R}} } 
    \leq \lambda \norm{\pi_{\bar{f}_L}}_2\norm{\pi_{\bar{f}_R}}_2 \leq \lambda \prod_{i}\norm{f_i}_{2^{d_i}}
\end{equation*}
where we have used the fact that
\begin{equation}\label{eq:HDEML-CS}
 \norm{\pi_{\bar{f}_L}}_2 \leq \prod_{i \in L} \norm{f_i}_{2^{d_i}}, \quad \text{and}  \quad \norm{\pi_{\bar{f}_R}}_2 \leq \prod_{i \in R} \norm{f_i}_{2^{d_i}}
\end{equation}
by repeatedly applying Cauchy-Schwarz along the splitting trees of L and R respectively.

It remains to bound II. By the triangle inequality, we can factor this term into components dependent on the left and right subtrees:

\begin{align} \label{eq:hdeml-bound-on-rightmost-term}
\abs{\ex{\pi_{\bar{f}_{R}}} \cdot \ex{ \pi_{\bar{f}_{L}}} - E_{\bar{f}_{[k]}}} &\leq \abs{\ex{\pi_{\bar{f}_{L}}} \cdot \ex{ \pi_{\bar{f}_{R}}} - \ex{\pi_{\bar{f}_L}} \cdot E_{\bar{f}_R}} + \abs{\ex{\pi_{\bar{f}_L}} \cdot E_{\bar{f}_R} - E_{\bar{f}_{[k]}}}\\
    \nonumber&= \abs{\ex{\pi_{\bar{f}_{L}}}} \abs{\ex{ \pi_{\bar{f}_{R}}} -  E_{\bar{f}_R}} + \abs{E_{\bar{f}_{R}}} \abs{\ex{ \pi_{\bar{f}_{L}}} -  E_{\bar{f}_L}}.
\end{align}
Since \((X^L,T_L,\rho|_{T_L}), (X^R,T_R,\rho|_{T_R})\) are both \(\lambda\)-tuple splitting trees of depth \(\leq D-1\),\footnote{Recall that \(X^L,X^R\) are the simplicial sub-complexes induced by the faces in \(X[L]\) and \(X[R]\) respectively.} the inductive hypothesis implies
\[
\abs{\ex{ \pi_{\bar{f}_{R}}} -  E_{\bar{f}_R}} \leq 3^{D-1}\lambda \mathcal{E}(\bar{f}_R,T_R), \quad \text{and} \quad \abs{\ex{ \pi_{\bar{f}_{L}}} -  E_{\bar{f}_L}} \leq 3^{D-1}\lambda \mathcal{E}(\bar{f}_L,T_L)
\]
so we can upper bound \eqref{eq:hdeml-bound-on-rightmost-term} by
\begin{align*}\label{eq:hdeml-bound-on-rightmost-term2}
    3^{D-1}\lambda \left(\abs{\ex{\pi_{\bar{f}_{L}}}} \mathcal{E}(\bar{f}_R,T_R) +
    \abs{E_{\bar{f}_R}}\mathcal{E}(\bar{f}_L,T_L)\right).
\end{align*}

Finally, observe that 
\[
\abs{E_{\bar{f}_R}} = \prod\limits_{i \in R}\abs{\mathbb{E}[f_i]} \leq \prod\limits_{i \in R}\norm{f_i}_1 \leq \prod\limits_{i \in R}\norm{f_i}_{2^{d_i}}
\]
by monotonicity of expectation norms and similarly
\[
\abs{\ex{\pi_{\bar{f}_{L}}}} \leq \norm{\pi_{\bar{f}_{L}}}_2 \leq \prod\limits_{i \in L}\norm{f_i}_{2^{d_i}}
\]
by \eqref{eq:HDEML-CS}. Altogether, this gives an upper bound of
\[
\text{I} + \text{II} \leq \lambda \mathcal{E}(\bar{f}_{[k]},T) + 2\cdot 3^{D-1}\lambda \mathcal{E}(\bar{f}_{[k]},T) \leq \lambda3^D\mathcal{E}(\bar{f}_{[k]},T)
\]
as desired.

For the unordered case, take an arbitrary partitification \((X',T,\rho')\) as in \pref{def:partitification}. \((X',T,\rho')\) is a \(\lambda\)-tuple splitting tree by \pref{claim:partitification}. Define \(\tilde{f}_i:X[\rho'(u_i)] \to \RR\) by \(\tilde{f}_i((v,j)) = f_i(v)\). Viewed as random variables, $\tilde{f}_i$ and $f_i$ are equidistributed, so $\mathbb{E}[f_i]= \mathbb{E}[\tilde{f}_i]$ and $\norm{f_i}_{2^{d_i}}=\norm{\tilde{f}_i}_{2^{d_i}}$. Moreover, by construction
    \[
        \Ex[a\in {\X[k-1]}, \pi \in S_k]{\prod_{i=1}^k f_i(a_{\pi(i)})} = \Ex[(a,\pi) \in {\SC[X'][k-1]} ]{\prod_{i=1}^k \tilde{f}_i(a_{\pi(i)},\pi(i))},
    \]
so can write
\begin{align*}
    \abs{\Ex[a\in {\X[k]}, \pi \in S_k]{\prod_{i=1}^k f_i(a_{\pi(i)})}- \prod_{i=1}^k \Ex[a_i \in {\X[1]}]{f_i(a_i)}} &= \abs{\Ex[(a,\pi) \in {\SC[X'][k]}]{\prod_{i=1}^k \tilde{f}_i((a,\pi)_i)}- \prod_{i=1}^k \Ex[{(a,\pi)_i \in \SC[X'][i]}]{\tilde{f}_i(a_i,i)}}\\
    &\leq 3^D\lambda \prod\limits_{i=1}^k \norm{\tilde{f}_i}_{2^{d_i}}\\
    &=3^D\lambda \prod\limits_{i=1}^k \norm{f_i}_{2^{d_i}}
\end{align*}
as desired.
\end{proof}
We now record a few useful corollaries from which hitting set and bias amplification are essentially immediate.
\begin{corollary} \label{cor:hdeml-indicators}
    Let \(X\) be as in \pref{thm:hdeml} and \(A_i \subseteq X[i]\). Then
        \begin{equation}\label{eq:hdeml-indicators}
            \abs{\Prob[a\in {\X[k-1]}]{\bigwedge_{i=1}^k a_i \in A_i}- \prod_{i=1}^k \prob{A_i}} \leq 3^{D} \lambda \prod_{i=1}^k \prob{A_i}^{2^{-d_i}} \leq 3^{D} \lambda.
    \end{equation}
\end{corollary}
\begin{corollary} \label{cor:hdeml-higher-dim-faces}
    Let \((X,T,\rho)\) be a \(\lambda\)-tuple splitting tree. Let \(u_1,u_2,\ldots,u_\ell \in T\) be nodes such that for every \(i \ne j\), \(u_i\) is not an ancestor of \(u_j\), and let \(d_i\) be the depth of \(u_i\) in \(T\). Let \(f_1,f_2,\ldots,f_\ell\) be functions such that \(f_i:\X[\rho(u_i)] \to \RR\). Then 
    \begin{equation}\label{eq:hdeml-higher-dim-faces}
    \abs{\Ex[a\in {\X[k-1]}]{\prod_{i=1}^\ell f_i(a_{\rho(u_i)})}- \prod_{i=1}^\ell \ex{f_i(a_{\rho(u_i)})}} \leq 3^{D} \lambda \prod_{i=1}^\ell \norm{f_i}_{2^{d_i}},
\end{equation}
where \(D = \max \set{d_i}\).
\end{corollary}

\begin{proof}
    We sequentially prune the tuple-splitting tree using nodes \(u_1,u_2,\ldots,u_\ell\) as in \pref{def:pruning}. If there are nodes of depth greater than \(D\) we also prune their depth \(D\) ancestor. By \pref{claim:tree-pruning}, this results in a pruned \(\lambda\)-tuple splitting tree \((X',T',\rho')\) where by construction \(u_1,u_2,\ldots,u_\ell\) are leaves in \(T'\) corresponding to \(X'[i] = X[\rho(u_i)]\) (we note $X'$ may have other parts not associated with these leaves as well). Setting $f_j=1$ for the latter, applying \pref{thm:hdeml} to \((X',T',\rho')\) gives the desired bound.
    
\end{proof}

\begin{proposition}[Hitting Set (\pref{prop:hitting-set} Restated)]
    Any depth $D$ $\lambda$-(tuple) splitting tree with $k$ leaves is $(3^D\lambda,k)$-hitting.
\end{proposition}
\begin{proof}
    Any subset $A \subset \X[1]$ can be divided into partite components $A = A_1 \amalg \ldots \amalg A_k$ where $A_i \subset X[i]$. \pref{cor:hdeml-indicators} then implies
    \[
    \Pr_{\sigma \in \X[k]}[\sigma \subset A] \leq \prod_{i=1}^k \mu(A_i)+3^D\lambda \leq \mu(A)^k+3^D\lambda
    \]
    where the righthand inequality follows from the fact that $\frac{1}{k}\sum_{i=1}^k \mu(A_i) = \mu(A)$.
\end{proof}

\begin{proposition}[Bias-Amplification (\pref{prop:bias-amp} Restated)]
Let \((X,T,\rho)\) be a $\lambda$-tuple splitting tree for $\lambda < \frac{1}{16}$. For any $0<\varepsilon<\frac{1}{4}$ and family of mean $\varepsilon$ functions \(\{f_i:X[i] \to \{\pm 1\}\}_{i \in [k]}\), $\prod f_i$ is a $\{\pm 1\}$-valued function with bias at most:
\begin{equation}
    \abs{\Ex[a\in {\X[k-1]}]{\prod_{i=1}^k f_i(a_i)}} \leq \varepsilon^k+2\lambda
\end{equation}
\end{proposition}
\begin{proof}
    The proof is similar to \pref{thm:hdeml}, and we adopt the same notational conventions. The base case $k=1$ is trivial. For the inductive step, we have
    \begin{align*}
        |\mathbb{E}[\pi_{\bar{f}_{[k]}}]| &\leq  \abs{\ex{\pi_{\bar{f}_{[k]}}} - \ex{\pi_{\bar{f}_{L}}} \cdot \ex{ \pi_{\bar{f}_{R}}} } + \ex{\pi_{\bar{f}_{L}}} \cdot \ex{ \pi_{\bar{f}_{R}}}\\
        &\leq \lambda + (\varepsilon^{|L|}+2\lambda)(\varepsilon^{|R|}+2\lambda)\\
        &\leq \varepsilon^k + 2\lambda
    \end{align*}
    by our assumptions on $\lambda$ and $\varepsilon$.
\end{proof}
Finally we list for completeness the unordered analogs of these results, which follow immediately from partitification as in \pref{thm:hdeml}.

\begin{corollary}[Expander-Mixing (\pref{cor:hdeml-indicators})]
    Let \(X\) be a depth $D$ $\lambda$-splitting tree. Then for any $A_1,\ldots,A_k \subset \X[1]$:
        \begin{equation}
            \abs{\Prob[a\in {\X[k]},\pi \in S_k]{\bigwedge_{i=1}^k a_i \in A_i}- \prod_{i=1}^k \prob{A_i}} \leq 3^{D} \lambda \prod_{i=1}^k \prob{A_i}^{2^{-d_i}} \leq 3^{D} \lambda. \qed
    \end{equation}
\end{corollary}
\begin{corollary}[Extended HD-EML (\pref{cor:hdeml-higher-dim-faces})]
    Let \((X,T,\rho)\) be a \(\lambda\)-splitting tree, \(u_1,u_2,\ldots,u_\ell \in T\) nodes such that for every \(i \ne j\), \(u_i\) is not an ancestor of \(u_j\), and let \(d_i\) be the depth of \(u_i\) in \(T\). Let \(f_1,f_2,\ldots,f_\ell\) be functions such that \(f_i:\X[\rho(u_i)] \to \RR\). Then 
    \begin{equation}
    \abs{\Ex[a\in {\X[k-1]},\pi \in S_k]{\prod_{i=1}^\ell f_i(a_{\rho(u_i)})}- \prod_{i=1}^\ell \ex{f_i(a_{\rho(u_i)})}} \leq 3^{D} \lambda \prod_{i=1}^\ell \norm{f_i}_{2^{d_i}},
\end{equation}
where \(D = \max \set{d_i}\).
\end{corollary}

\begin{corollary}[Hitting Set (\pref{prop:hitting-set})]
    Any depth $D$ $\lambda$ splitting tree with $k$ leaves is $(3^D\lambda,k)$-hitting.
\end{corollary}

\begin{corollary}[Bias-Amplification (\pref{prop:bias-amp})]
Let \((X,T,\rho)\) be a $\lambda$ splitting tree for $\lambda < \frac{1}{16}$. For any $0<\varepsilon<\frac{1}{4}$ and family of mean $\varepsilon$ functions \(\{f_i:\X[1] \to \{\pm 1\}\}_{i \in [k]}\), $\prod f_i$ is a $\{\pm 1\}$-valued function with bias at most:
\begin{equation}
    \abs{\Ex[a\in {\X[k-1]},\pi \in S_k]{\prod_{i=1}^k f_i(a_{\pi(i)})}} \leq \varepsilon^k+2\lambda
\end{equation}
\end{corollary}

\section{Reverse Hypercontractivity: From Boolean to All Functions} \label{app:general-hc}
We now give a generic reduction from reverse hypercontractivity for general functions to the Boolean case.

\restatetheorem{thm:generalizing}

Before moving directly to the proof, we give a brief overview of the main idea. Starting with arbitrary $f_1,f_2$, our goal is to build discrete approximations of the $f_i$ that remain close to the original functions in expectation and don't take too many unique values. Once this is the case, we can write the functions as a weighted sum over indicators for each value, and apply the boolean result without too much loss. 

The key to the proof really lies in these approximations, which we build by discretizing the functions to powers of 2 and applying a careful iterative zero-ing procedure to components that fall too far from the mean. Naively this approach seems problematic, as $f_i$ is arbitrary and may not have a well-behaved tail to truncate. We handle this by exploiting the fact that no function can have more than a small constant fraction of its mass beyond its expectation times (an appropriate power of) its $(1{\to}(1+\varepsilon))$-norm.

More formally, the proof of \pref{thm:generalizing} is split into two intermediate reductions, corresponding to the `approximation step' and the `indicator step' above. The first of these is a reduction from general reverse hypercontractivity to a slightly richer set of functions that are \textit{balanced} and \textit{discrete}.

\begin{definition}[Nice Functions]
For any $\alpha \in [0,1]$, and $\beta \in [1,\infty)$, we say the function $f:V \to \mathbb{R}_{\geq 0}$ is ($\alpha,\beta$)-nice if it is:
\begin{enumerate}
    \item \textbf{Balanced:} all non-zero outputs of $f$ satisfy
    \[
    \alpha \mathbb{E}[f] \leq f(x) \leq \beta \mathbb{E}[f].
    \]
    \item \textbf{Discrete:} the range of $f$ is entirely contained in $\{0\} \cup \{2^j: j \in \mathbb{Z}\}$.
\end{enumerate}
\end{definition}
To state the reduction, we require a setting of parameters $\alpha$ and $\beta$ based on the $(1{\to}(1+\varepsilon))$-norm of the functions in question due to the tail truncation strategy described above. In particular, for a fixed non-negative function $f_i$, define
\[
\eta_i = \frac{\ex{f_i^{1+\varepsilon}}}{\ex{f_i}^{1+\varepsilon}} =  \norm{f_i}^{1+\varepsilon}_{1{\to}(1+\varepsilon)}.
\]
With this in mind, we state our two intermediate reductions and prove \pref{thm:generalizing}.
\begin{proposition}[From General to Nice]\label{prop:general-to-nice}
    Let $\ell,\varepsilon, \kappa, V$, and $D$ be as in \pref{thm:generalizing}. Suppose that for any two functions \(f_1,f_2:V \to \RR_{\geq 0}\) that are $(\frac{1}{16\eta_i},(25\eta_i)^{\frac{2(1+\varepsilon)}{\varepsilon}})$-nice
    \begin{equation} \label{eq:assumption-on-nice}
        \iprod{f_1^{\ell(1+\varepsilon)}, D f_2^{\ell(1+\varepsilon)}}\geq \kappa \ex{f_1}^{\ell(1+\varepsilon)} \ex{f_2}^{\ell(1+\varepsilon)}.
    \end{equation}
    Then for every two functions \(f_1,f_2:V \to \RR_{\geq 0}\) it holds that
    \begin{equation} \label{eq:nice-generalizing-hc-equation}
    \iprod{f_1^{\ell(1+\varepsilon)}, D f_2^{\ell(1+\varepsilon)}} \geq \left(\frac{1}{5}\right)^{2\ell(1+\varepsilon)} \kappa\ex{f_1}^{\ell(1+\varepsilon)} \ex{f_2}^{\ell(1+\varepsilon)}.   
    \end{equation}
\end{proposition}
The `indicator step' then reduces reverse hypercontractivity for nice functions to the Boolean case.
\begin{proposition}[From Nice to Boolean]\label{prop:nice-to-bool}
     Let $\ell,\varepsilon,V$, and $D$ be as in \pref{thm:generalizing}. If for every \(A,B \subseteq V\),
    \begin{equation} \label{eq:bool-assumption-on-sets}
        \langle 1_A, D1_B \rangle \geq \kappa \prob{A}^\ell\prob{B}^{\ell},
    \end{equation}
    then every two $(\frac{1}{16\eta_i},(25\eta_i)^{\frac{2(1+\varepsilon)}{\varepsilon}})$-nice functions \(f_1,f_2:V \to \RR_{\geq 0}\) satisfy
    \begin{equation} \label{eq:bool-to-nice-generalizing-hc-equation}
    \iprod{f_1^{\ell(1+\varepsilon)}, D f_2^{\ell(1+\varepsilon)}} \geq \left(\frac{1}{18 + \frac{12}{\varepsilon}}\right)^{2(\ell-1)}\kappa \ex{f_1}^{\ell(1+\varepsilon)} \ex{f_2}^{\ell(1+\varepsilon)}.   
    \end{equation}
\end{proposition}
The proof of \pref{thm:generalizing} is now essentially immediate.
\begin{proof}[Proof of \pref{thm:generalizing}]
    We have by assumption that all $A,B \subseteq V$ satisfy $\langle 1_A, D1_B \rangle \geq \kappa \prob{A}^\ell\prob{B}^{\ell}$,
    so we may apply \pref{prop:nice-to-bool} to get that every two $(\frac{1}{16\eta_i},(25\eta_i)^{\frac{2(1+\varepsilon)}{\varepsilon}})$-nice functions \(f_1,f_2:V \to \RR_{\geq 0}\) satisfy
    \begin{equation*}
    \iprod{f_1^{\ell(1+\varepsilon)}, D f_2^{\ell(1+\varepsilon)}} \geq \left(\frac{1}{18 + \frac{12}{\varepsilon}}\right)^{2(\ell-1)}\kappa \ex{f_1}^{\ell(1+\varepsilon)} \ex{f_2}^{\ell(1+\varepsilon)}.   
    \end{equation*}
    Applying \pref{prop:general-to-nice} then implies \textit{arbitrary} non-negative $g_1,g_2$ satisfy
    \[
    \iprod{g_1^{\ell(1+\varepsilon)}, D g_2^{\ell(1+\varepsilon)}} \geq \left(\frac{1}{5}\right)^{2\ell(1+\varepsilon)}\left(\frac{1}{18 + \frac{12}{\varepsilon}}\right)^{2(\ell-1)}\kappa \ex{g_1}^{\ell(1+\varepsilon)} \ex{g_2}^{\ell(1+\varepsilon)}.
    \]
    Finally, for $f_1,f_2$ arbitrary non-negative functions applying the above to $g_i = f_i^{\frac{1}{\ell(1+\varepsilon)}}$ gives the result.
\end{proof}

We now move to the formal proofs.
\begin{proof}[Proof of \pref{prop:general-to-nice}]
    Given arbitrary \(f_1,f_2: V \to \mathbb{R}_{\geq 0}\) we will construct corresponding $f_1',f_2'$ satisfying
    \begin{enumerate}
        \item $0 \leq f_i' \leq f_i$ \label{eq:item-1}
        \item $\mathbb{E}[f_i'] \geq \frac{1}{5}\mathbb{E}[f]$ \label{item:expectation-bound}
        \item $f_i'$ is $(\frac{1}{16\eta_i},(25\eta_i)^{\frac{2(1+\varepsilon)}{\varepsilon}})$-nice \label{item:regularity}
    \end{enumerate}
Combining the above with monotonicity of $D$ and reverse hypercontractive inequality for nice functions gives
\begin{align*}
    \langle f_1, Df_2 \rangle &\geq \langle f'_1, Df'_2 \rangle& \text{(Item \eqref{eq:item-1} and Monotonicity)}\\
    &\geq \kappa \mathbb{E}[f'_1]^{\ell(1+\varepsilon)}\mathbb{E}[f'_2]^{\ell(1+\varepsilon)} & \text{(Item \eqref{item:regularity} and Equation \eqref{eq:assumption-on-nice})}\\
    &\geq \left(\frac{1}{5}\right)^{2\ell(1+\varepsilon)}\kappa \mathbb{E}[f'_1]^{\ell(1+\varepsilon)}\mathbb{E}[f'_2]^{\ell(1+\varepsilon)} & \text{(Item \eqref{item:expectation-bound})}
\end{align*}
as desired.

We construct $f_i'$ in two separate steps: a simple discretization procedure, and a more involved balancing procedure that iteratively zeroes out unbalanced parts of the function. Starting with the former, given a non-negative $f$ define its \textit{rounding down} \(\tilde{f}\) by setting for $j \in \mathbb{Z}$:
    \[
    \tilde{f}(x)=
    \begin{cases}
        2^j & \text{for } f(x) \in [2^j,2^{j+1})\\
        0 & \text{if } f(x)=0.
    \end{cases}
    \]
Note that, by construction, $\frac{f_i}{2} \leq \tilde{f}_i \leq f_i$.

    We now define an iterative balancing procedure outputting a sequence of functions \(\tilde{f}_i=f_i^{(0)}\geq  f_i^{(1)}\geq \ldots \geq f_i^{(k_i)}=f'_i\) for $k_i \in \mathbb{N}$ some finite stopping time. For these (soon-to-be-defined) functions, define \(\eta_i^{(j)} = \frac{\ex{(f_i^{(j)})^{1+\varepsilon}}}{\ex{(f_i^{(j)})}^{1+\varepsilon}}\) and denote the set of `balanced' inputs of $f^{(j)}_i$ as:
    \[
    B_i^{(j)} = \sett{x}{\frac{1}{4\eta_i^{(j)}} \ex{f_i^{(j)}} \leq f_i^{(j)}(x) \leq \left (4\eta_i^{(j)} \right )^{\frac{2(1+\varepsilon)}{\varepsilon}} \ex{f_i^{(j)}}}.
    \]
    If \(f_i^{(0)}\) is already $(\frac{1}{16\eta_i},(25\eta_i)^{\frac{2(1+\varepsilon)}{\varepsilon}})$-nice, set $k_i=0$ and output $f'_i=\tilde{f}_i$. Otherwise for \(j=1,2,\ldots\) define \(f_i^{(j)}\) by zeroing out all non-balanced inputs:
    \[
    f_i^{(j)} := f_i^{(j-1)} \cdot \one_{B_i^{(j-1)}},
    \]
    and define the stopping time $k_i$ to be the first index $j$ such that \(\eta_i^{(j)} \geq \frac{1}{4}\eta_i^{(j-1)}\). This procedure terminates (i.e.\ $k_i$ is finite) since each \(\eta_i^{(j)} \geq 1\) by convexity. Every step the process doesn't stop decreases \(\eta_i^{(j)}\)  by a factor of \(4\), so there can be a total of $O(\log(\eta_i^{(0)}))$ steps.
    
    It is immediate from definition that $0 \leq f_i' \leq f_i$, so we just need to show \pref{item:expectation-bound} and \pref{item:regularity}. The key is to observe that our balancing process cannot remove too much mass in each step:
    \begin{equation} \label{eq:expectation-is-similar}
        \ex{f_i^{(j-1)}} \geq \ex{f_i^{(j)}} \geq \left( 1-\frac{1}{2\eta_i^{(j-1)}} \right) \ex{f_i^{(j-1)}}.
    \end{equation}    
    We defer the proof, and first show \pref{item:expectation-bound} and \pref{item:regularity} given this assumption.
\paragraph{Proof of \pref{item:expectation-bound}} The proof is essentially immediate by the exponential decay of $\eta_i^{(j)}$. Namely by iterated application of \eqref{eq:expectation-is-similar} we have
        \[
        \frac{\ex{f_i^{(k)}}}{\mathbb{E}[f_i^{(0)}]} \geq \prod_{j=0}^{k-1} \left ( 1- \frac{1}{2\eta_i^{(k-1-j)}} \right ) \geq \prod_{j=0}^{\infty} \left ( 1- \frac{1}{2\cdot4^j} \right ) \geq \frac{2}{5}
        \]
        since by construction \(\eta_i^{(j)}> 4^{k-1-j} \eta_i^{(k-1)} \geq 4^{k-1-j}\). Combining with the fact $\tilde{f}_i=f_i^{(0)} \geq \frac{f_i}{2}$ gives the result.
\paragraph{Proof of \pref{item:regularity}} Recall that for any $j$ we have by construction the following `pseudo'-balance condition:
        \[
        \forall f_{i}^{(j)}(x) \neq 0: \frac{1}{4\eta_i^{(j-1)}} \ex{f_i^{(j-1)}} \leq f_i^{(j)}(x) \leq \left (4\eta_i^{(j-1)} \right )^{\frac{2(1+\varepsilon)}{\varepsilon}} \ex{f_i^{(j-1)}}.
        \]
        When $j=k_i$ is the stopping time, observe that by \pref{eq:expectation-is-similar} and construction respectively we have
        \begin{enumerate}
            \item $\mathbb{E}[f_{i}^{(k_i)}] \leq \mathbb{E}[f_{i}^{(k_i-1)}] \leq 2\mathbb{E}[f_{i}^{(k_i)}]$
            \item \(4\eta_i^{(k_i)}\geq \eta_i^{(k_i-1)}\).
        \end{enumerate}
        Substituting into the above gives
        \[
        \frac{1}{16\eta_i^{(k_i)}} \ex{f_i^{(k_i)}} \leq f_i^{(k_i)}(x) \leq \left (25\eta_i^{(k_i)} \right )^{\frac{2(1+\varepsilon)}{\varepsilon}} \ex{f_i^{(k_i)}}
        \]
        as desired.

\paragraph{Proof of \eqref{eq:expectation-is-similar}} The lefthand inequality is by construction. Toward the righthand, fix \(f=f_i^{(j-1)}\) and $\eta=\eta^{(j-1)}_i$ for notational simplicity and define the set of `terrible' inputs zeroed in the $j$th step as:
    \[
    T_1 \coloneqq \sett{x}{f(x) < \frac{1}{4\eta} \ex{f}} \quad \quad \quad \quad T_2 \coloneqq \sett{x}{f(x) > (4\eta)^\frac{2(1+\varepsilon)}{\varepsilon}\ex{f}}
    \]
    Since \(T_1 \cup T_2 = V \setminus B_i^{(j-1)}\) and \(f_i^{(j)} = \one_{B_i^{(j-1)}} f\), it is enough to show that
    \begin{equation}
        \ex{f \cdot \one_{T_1}}, \ex{f \cdot \one_{T_2}} \leq \frac{1}{4\eta} \ex{f}.
    \end{equation}
    It is obvious that \(\ex{f \cdot \one_{T_1}} \leq \frac{1}{4\eta}\ex{f}\) since $f$ is upper bounded by this value within \(T_1\). To bound the upper tail we use Hölder and Markov:
    \begin{align*}
            \ex{f \cdot \one_{T_2}} &\leq \norm{f}_{1+\varepsilon} \prob{f \geq (4\eta)^{\frac{2(1+\varepsilon)}{\varepsilon}}\ex{f}}^{\frac{\varepsilon}{1+\varepsilon}}& \text{(Hölder's inequality)}\\
            &\leq \norm{f}_{1+\varepsilon} \frac{1}{16\eta^2}& \text{(Markov's inequality)}\\
            &\leq  \left(\frac{1}{4}\eta^{-\frac{\varepsilon}{1+\varepsilon}}\right) \frac{1}{4\eta}\ex{f}& \text{($\norm{f}_{1+\varepsilon}=\eta^{\frac{1}{1+\varepsilon}} \ex{f}$)} \\
            &\leq  \frac{1}{4\eta}\ex{f}& \text{($\eta^{-\frac{\varepsilon}{1+\varepsilon}} \leq 1$)}
    \end{align*}
    as desired.
\end{proof}

Finally we prove reverse hypercontractivity for `nice' functions (\pref{prop:nice-to-bool}).
\begin{proof}[Proof of \pref{prop:nice-to-bool}]
    Since $f_i$ is $(\frac{1}{16\eta_i},(25\eta_i)^{\frac{2(1+\varepsilon)}{\varepsilon}})$-nice, it can attain at most
    \[
    n_i = \left(3+\frac{2}{\varepsilon}\right)\log \eta_i+\frac{10}{\varepsilon}+15
    \] 
    non-zero values. With this in mind, decompose $f_i=\sum\limits_{j \in \mathbb{Z}}2^j\mathbf{1}_{f_i(x)=2^j}$ into level sets and observe that, by linearity of $D$, we can decompose the inner product itself into boolean sub-components and apply our assumed reverse hypercontractive inequality (Equation \pref{eq:bool-to-nice-generalizing-hc-equation}):    
    \begin{align*}
    \iprod{f_1^{(1+\varepsilon)\ell},Df_2^{(1+\varepsilon)\ell}} &= \sum_{i,j \in \mathbb{Z}}2^{(1+\varepsilon)\ell j} 2^{(1+\varepsilon)\ell i} \iprod{\one_{f_1(x)=2^j}, D\one_{f_2(x)=2^i}} \\
        &\geq\kappa \sum_{i,j \in \mathbb{Z}}2^{(1+\varepsilon)\ell j} 2^{(1+\varepsilon)\ell i} \prob{f_1(x)=2^j}^\ell \prob{f_2(x)=2^j}^\ell\\
          &=\kappa\left(\sum_{j\in \mathbb{Z}}(2^{(1+\varepsilon)j} \prob{f_1(x)=2^j})^\ell\right) \cdot \left(\sum_{j \in \mathbb{Z}}(2^{(1+\varepsilon)j} \prob{f_2(x)=2^j})^\ell\right). 
    \end{align*}
    Now since the sums have only \(n_1,n_2\) non-zero terms respectively, applying $(\ell,\frac{\ell}{\ell-1})$-H{\"o}lder's inequality lower bounds this quantity by
    \begin{align*}
        &\kappa(n_1 n_2)^{1-\ell}\left(\sum\limits_{j\in \mathbb{Z}}2^{(1+\varepsilon)j} \prob{f_1(x)=2^j}\right)^\ell \cdot \left(\sum\limits_{j\in \mathbb{Z}}2^{(1+\varepsilon)j}\prob{f_1(x)=2^j}\right)^\ell\\
        =&\kappa n_1^{1-\ell}\ex{f_1^{1+\varepsilon}}^\ell \cdot n_2^{1-\ell}\ex{f_2^{1+\varepsilon}}^\ell \\
        \geq & \kappa \cdot (n_1^{1-\ell} \eta_1^\ell) \cdot (n_2^{1-\ell} \eta_2^\ell) \cdot \ex{f_1}^{\ell(1+\varepsilon)}\ex{f_2}^{\ell(1+\varepsilon)}. 
    \end{align*}
It can be checked directly that
\[
\min_{\eta_i \geq 1} n_i\eta_i \geq \frac{1}{18+\frac{12}{\varepsilon}}
\]
which gives the desired result.
\end{proof}

\section{Splittability}\label{app:split}
In this section we discuss the relations between our notion of complete splittability, the splittability notion of \cite{jeronimo2021near}, and the "$\tau$-sampling" gaurantee of \cite{DoronW2022}, and the inheritance of splittability under projection. We start by defining the latter two notions.
\begin{definition}[JST-Splittability {\cite{jeronimo2021near}}]
    $X \subset [n]^{d}$ is said to be $\lambda$-JST-splittable if for all $1 \leq a \leq t \leq b \leq d$:
    \[
    \lambda_2(S_{[a,t],[t+1,b]}) \leq \lambda
    \]
\end{definition}
In their later work, \cite{DoronW2022} introduced the weaker notion of $\tau$-sampling which suffices for list-decoding.
\begin{definition}
    $X \subset [n]^{d}$ is said to be $\tau$-sampling if for $t \in [d]$, $S \subset [n]$ and $W \subset [n]^t$:
    \[
    \underset{w\sim X}{\text{Cov}}\left(1[w_{t+1} \in S],1[w_{[t]} \in W]\right) \leq \tau.
    \]
\end{definition}
Before arguing relations between these notions, we give the key lemma that underlies all proofs in this section: the expansion of swap walks is inherited under projection.
\begin{lemma}[Swap Walk Inheritence]\label{lem:swap-project}
    Let $X$ be a partite $d$-\maximal complex, and $A$ and $B$ disjoint subsets of $[d]$, and $F$ any subset of $[d]$. Then:
    \[
    \lambda_2(S_{A \cap F,B \cap F}) \leq \lambda_2(S_{A,B})
    \]
\end{lemma}
\begin{proof}
    By the variational characterization it is enough to show
    \[
    \max_{\mathbb{E}[f^F]=0} \frac{\norm{S_{A \cap F,B \cap F} f^F}}{\norm{f^F}_2} \leq \lambda_2(S_{A,B})
    \]
    Given any such function $f^F: X[A \cap F] \to \R$ define its extension to $X[A]$ by $f(x) \coloneqq f^F(x_F)$. Here $x_F = \{v \in x: col(v) \in F\}$ is the projection into \(X^F\). and observe that this extension satisfies:
    \begin{enumerate}
        \item $\mathbb{E}_{X[A]}[f]=\mathbb{E}_{X[A\cap F]}[f^F]=0$
        \item $\norm{f}_{2,X[A]} = \norm{f^F}_{2,X[A \cap F]}$.
    \end{enumerate}
    The result then follows from convexity, in particular:
    \begin{align*}
        \norm{S_{A \cap F,B \cap F}f^F}_2&=\underset{y \sim X[A \cap F]}{\mathbb{E}}\left[\left(\underset{x \in \X[d]}{\mathbb{E}}[f^F(x_{B \cap F}) | x_{A \cap F} = y]\right)^2\right]^{\frac{1}{2}}&\\
        &=\underset{y \sim X[A \cap F]}{\mathbb{E}}\left[\left(\underset{x \in \X[d]}{\mathbb{E}}[f(x_{B}) | x_{A \cap F} = y]\right)^2\right]^{\frac{1}{2}}&\\
        &=\underset{y \sim X[A \cap F]}{\mathbb{E}}\left[\left(\underset{y' \sim X_y[A \setminus F]}{\mathbb{E}}\left[\underset{x \in \X[d]}{\mathbb{E}}[f(x_{B} | x_{A \cap F} = y \land x_{A \setminus F} = y']\right]\right)^2\right]^{\frac{1}{2}}&\\
        &\leq \underset{y \sim X[A]}{\mathbb{E}}\left[\left(\underset{x \in \X[d]}{\mathbb{E}}[f(x_{B} | x_A = y]\right)^2\right]^{\frac{1}{2}} &\text{(Jensen's Inequality)}\\
        &=\norm{S_{A,B}f}_2\\
        &\leq \lambda_2(S_{A,B})\norm{f}_2 & \text{($\mathbb{E}[f] = \mathbb{E}[f^F]=0$)}\\
        &= \gamma\norm{f^F}_2
    \end{align*}
\end{proof}
With this in mind, we prove the following relations between complete splittability, JST-splittability, and $\tau$-sampling.
\begin{proposition}
    $X \subset [n]^t$ is $\lambda$-JST-splittable if and only if it is completely $\lambda$-splittable. Moreover any completely $\lambda$-splittable complex is $\lambda$-sampling.
\end{proposition}
\begin{proof}
If $X$ is $\lambda$-JST-splittable it is clearly completely $\lambda$-splittable setting $a=1,b=d$. If  $X$ is completely $\lambda$-splittable, then for any $1 \leq a \leq i \leq b \leq d$ we have:
\[
S_{[a,i],[i+1,b]} \leq S_{[1,i],[i+1,d]} \leq \lambda
\]
as desired. Finally, if $X$ is completely $\lambda$-splittable, then for any $S \subset [n]$ and $X \subset [n]^t$ write $1_S \coloneqq \mathbf{1}[w_{t+1} \in S]$ and $1_X \coloneqq \mathbf{1}[w_{[t]} \in X]$ 
\begin{align*}
\text{Cov}_{w\sim X}\left(1_S,1_X\right) &= \underset{y \sim X^{t+1}}{\mathbb{E}}[1_S(y) S_{t+1,[t]}1_X(y)] - \mathbb{E}[1_S]\mathbb[E][1_X]\\ 
&\leq \lambda_2(S_{t+1,[t]}) \norm{1_S}\norm{1_X}\\ 
&\leq \lambda
\end{align*}
where we have used that by H\"{o}lder duality and \pref{lem:swap-project} $\lambda_2(S_{t+1,[t]}) = \lambda_2(S_{[t],t+1}) \leq \lambda_2(S_{[t],t+1}) \leq \lambda$.
\end{proof}

Finally, we show that both complete splittability and standard tuple splittability are inherited under projections. Recall that given a $d$-partite complex $X$ and $F \subseteq [d]$, $X^F$ is the projection of $X$ given by drawing a $d$-face from $X$, and projecting the resulting face onto the parts in $F$. 
\begin{proposition}[\pref{claim:project} Restated]
    Let $X \subset [n]^{d}$ be a homogeneous, completely $\lambda$-splittable complex. Then for any $F \subseteq [d]$ $X^F$ is homogeneous and completely $\lambda$-splittable.
\end{proposition}
\begin{proof}
    The proof is essentially immediate from definition and \pref{lem:swap-project}. Recall $X$ is homogeneous if the projection onto each part is uniform on $[n]$. Since projecting onto $X^F$ and then a coordinate $i \in F$ is the same as just projecting onto $i$, homogeneity is inherited. Toward splittability, write $F = \{i_1,\ldots,i_{|F|}\}$ where $i_j < i_{j+1}$. Then for any $i_j$, applying \pref{lem:swap-project} gives
    \begin{align*}
        \lambda_2(S^{X^F}_{\{i_1,\ldots, i_{j}\},\{i_{j+1},\ldots,i_{|F|}}\}) &\leq \lambda_2(S^X_{\{i_1,\ldots, i_{j}\},\{i_{j+1},\ldots,i_{|F|}}\})\\
        &\leq \lambda_2(S_{[i_j],[i_j+1,d]}) \leq \lambda
    \end{align*}
    where we have used the fact that $\lambda_2(S^{X^F}_{\{i_1,\ldots, i_{j}\},\{i_{j+1},\ldots,i_{|F|}\}})=\lambda_2(S^{X}_{\{i_1,\ldots, i_{j}\},\{i_{j+1},\ldots,i_{|F|}\}})$ (indeed they are the same operator).
\end{proof}

Finally, while not strictly necessary for the results in this work, we show for completeness that inheritance of splittability also holds in the setting of general tuple-splitting trees as well. 
\begin{proposition}[Projected Splitting Trees]\label{prop:splitting-project}
    Let X be any $d$-\maximal partite complex with a $\lambda$-tuple splitting tree of depth $D$. For every $F \subseteq [d]$, $X^F$ also has a $\lambda$-tuple splitting tree of depth at most $D$.
\end{proposition}
The proof of \pref{prop:splitting-project} relies on a simple inductive algorithm for projecting partite orderered binary trees $(T,\rho)$ onto a subset of their coordinates. In particular, given such a tree and a subset $F \subseteq [d]$, we define the projection operation via the following inductive algorithm that modifies the tree $T$ in place starting from a specific node $v \in T$.
\begin{algorithm}[H]
\label{alg:practice}
\SetAlgoLined
\KwResult{Projected partite ordered binary tree $(T_v^F,\rho_v^F)$ rooted at $v$}
 \uIf{$v$ is a leaf}{\textbf{return}}
 \uElseIf{$\rho(\ell_{v}) \cap F = \emptyset$}{
    Delete($T(\ell_v)$)\;
    Merge($v$, $r_v$)\footnote{Merge($v$,$r_v$) deletes the node $v$ and replaces $r_v$ as the respective child of $v$'s parent.}\;
    Project$(T,\rho,F,r_v)$
 } 
 \uElseIf{$\rho(r_{v}) \cap F = \emptyset$}{
    Delete($T(r_v)$)\;
    Merge($v$, $\ell_v$)\;
    Project$(T,\rho,F,\ell_v)$
 }
 \uElse{
    Project$(T,\rho,F,r_v)$\;
    Project$(T,\rho,F,\ell_v)$
 } 
 \caption{Project$(T,\rho,F,v)$}
\end{algorithm}
In other words, \pref{alg:practice} is the result of intersecting $\rho$ with $F$ and deleting the resulting nodes with empty or repeated labelings. We define the $F$-projection of a tree by applying this process at the root.
\begin{definition}[Projected Trees]
    Let $(T,\rho)$ be a $d$-\maximal partite ordered binary tree and $F \subset [d+1]$ any coordinate subset. The $F$-projection of $(T,\rho)$, denoted $(T^F,\rho^F)$ is given by
    \[
    T^F = \text{Project}(T,\rho,F,\text{root}(T)), \quad \quad \rho^F = (\rho \cap F)|_{T^F}.
    \]
\end{definition}
We are now ready to prove \pref{prop:splitting-project}.
\begin{proof}[Proof of \pref{prop:splitting-project}]
    We will show that $(X^F,T^F,\rho^F)$ is a $\lambda$-tuple splitting tree of depth at most $D$. We first argue that $(X^F,T^F,\rho^F)$ is a partite ordered binary tree of depth at most $D$. First, observe that every operation in \pref{alg:practice} maps binary trees to binary trees, and cannot increase depth. In particular, the only modification of the tree structure occurs when a branch is deleted and the relevant root is contracted with its other child, maintaining the invariant that every internal node has one parent and two children, and that the depth is at most $D$. To see that the leaves are in bijection with $F$, observe that all leaves labeled by $F$ survive the projection, all leaves labeled by $[d] \setminus F$ are deleted, and no new leaves are introduced by construction. The first two claims hold since a node $v$ is deleted if and only if either $\rho(v) \cap F = \emptyset$, or the labeling is repeated later down the tree. Finally, to see that the children of any internal node in $T^F$ partition its labeling in $\rho^F$, observe that this property holds trivially for $T$ under the intersected labeling $\rho \cap F$ (albeit the partition may be trivial), and that this invariant is preserved by every operation in \pref{alg:practice}.

   Finally we argue the projected tree inherits the splittability of $X^F$. By \pref{lem:swap-project}, we have that for any internal node $u$:
    \[ \lambda_2(S_{\rho^F(\ell_u),\rho^F(r_u)}) =\lambda_2(S_{\rho(\ell_u) \cap F,\rho(r_u) \cap F}) \leq \lambda_2(S_{\rho(\ell_u) ,\rho(r_u)}) \leq \lambda
    \]
    as desired.
\end{proof}
\end{document}

%% file: macros.tex
\providecommand{\RR}{\mathbb{R}}
\providecommand{\NN}{\mathbb{N}}

\usepackage[l2tabu, orthodox]{nag}

\usepackage{xspace,enumerate}

\usepackage[dvipsnames]{xcolor}

\usepackage[T1]{fontenc}
\usepackage[full]{textcomp}

\usepackage[american]{babel}

\usepackage{mathtools}

\usepackage{amsthm}
\usepackage{amsmath}

\newtheorem{theorem}{Theorem}[section]
\newtheorem*{theorem*}{Theorem}

\newtheorem{proposition}[theorem]{Proposition}
\newtheorem*{proposition*}{Proposition}
\newtheorem{lemma}[theorem]{Lemma}
\newtheorem*{lemma*}{Lemma}
\newtheorem{corollary}[theorem]{Corollary}
\newtheorem*{corollary*}{Corollary}
\newtheorem*{conjecture*}{Conjecture}
\newtheorem{fact}[theorem]{Fact}
\newtheorem*{fact*}{Fact}

\newtheorem*{hypothesis*}{Hypothesis}
\newtheorem{conjecture}[theorem]{Conjecture}

\theoremstyle{definition}
\newtheorem{definition}[theorem]{Definition}
\newtheorem*{definition*}{Definition}

\newtheorem{example}[theorem]{Example}
\newtheorem{question}[theorem]{Question}

\theoremstyle{remark}
\newtheorem{claim}[theorem]{Claim}
\newtheorem*{claim*}{Claim}
\newtheorem{remark}[theorem]{Remark}
\newtheorem*{remark*}{Remark}
\newtheorem{observation}[theorem]{Observation}
\newtheorem*{observation*}{Observation}

 \usepackage[letterpaper,
 top=1in,
 bottom=1in,
 left=1in,
 right=1in]{geometry}

\usepackage[varg]{pxfonts} 

\usepackage[sc,osf]{mathpazo}
\usepackage{lmodern}

\ifnum\showkeys=1
\usepackage[color]{showkeys}
\fi

\usepackage[
colorlinks=true,
urlcolor=blue,
linkcolor=blue,
citecolor=OliveGreen,
]{hyperref}

\usepackage{prettyref}

\newcommand{\savehyperref}[2]{\texorpdfstring{\hyperref[#1]{#2}}{#2}}

\newrefformat{eq}{\savehyperref{#1}{\textup{(\ref*{#1})}}}
\newrefformat{lem}{\savehyperref{#1}{Lemma~\ref*{#1}}}
\newrefformat{obs}{\savehyperref{#1}{Observation~\ref*{#1}}}
\newrefformat{def}{\savehyperref{#1}{Definition~\ref*{#1}}}
\newrefformat{thm}{\savehyperref{#1}{Theorem~\ref*{#1}}}
\newrefformat{cor}{\savehyperref{#1}{Corollary~\ref*{#1}}}
\newrefformat{cha}{\savehyperref{#1}{Chapter~\ref*{#1}}}
\newrefformat{sec}{\savehyperref{#1}{Section~\ref*{#1}}}
\newrefformat{app}{\savehyperref{#1}{Appendix~\ref*{#1}}}
\newrefformat{tab}{\savehyperref{#1}{Table~\ref*{#1}}}
\newrefformat{fig}{\savehyperref{#1}{Figure~\ref*{#1}}}
\newrefformat{hyp}{\savehyperref{#1}{Hypothesis~\ref*{#1}}}
\newrefformat{alg}{\savehyperref{#1}{Algorithm~\ref*{#1}}}
\newrefformat{rem}{\savehyperref{#1}{Remark~\ref*{#1}}}
\newrefformat{item}{\savehyperref{#1}{Item~\ref*{#1}}}
\newrefformat{step}{\savehyperref{#1}{step~\ref*{#1}}}
\newrefformat{conj}{\savehyperref{#1}{Conjecture~\ref*{#1}}}
\newrefformat{fact}{\savehyperref{#1}{Fact~\ref*{#1}}}
\newrefformat{prop}{\savehyperref{#1}{Proposition~\ref*{#1}}}
\newrefformat{prob}{\savehyperref{#1}{Problem~\ref*{#1}}}
\newrefformat{claim}{\savehyperref{#1}{Claim~\ref*{#1}}}
\newrefformat{relax}{\savehyperref{#1}{Relaxation~\ref*{#1}}}
\newrefformat{red}{\savehyperref{#1}{Reduction~\ref*{#1}}}
\newrefformat{part}{\savehyperref{#1}{Part~\ref*{#1}}}
\newrefformat{ex}{\savehyperref{#1}{Example~\ref*{#1}}}
\newrefformat{para}{\savehyperref{#1}{Paragraph~\ref*{#1}}}
\newrefformat{question}{\savehyperref{#1}{Question~\ref*{#1}}}

\newcommand{\Sref}[1]{\hyperref[#1]{\S\ref*{#1}}}

\usepackage{nicefrac}

\ifnum\usemicrotype=1
\usepackage{microtype}
\fi

\ifnum\showauthornotes=1
\newcommand{\Authornote}[2]{{\sffamily\small\color{red}{[#1: #2]}}}
\newcommand{\Authornotecolored}[3]{{\sffamily\small\color{#1}{[#2: #3]}}}
\newcommand{\Authorcomment}[2]{{\sffamily\small\color{gray}{[#1: #2]}}}
\newcommand{\Authorstartcomment}[1]{\sffamily\small\color{gray}[#1: }

\newcommand{\Authorfnote}[2]{\footnote{\color{red}{#1: #2}}}
\newcommand{\Authorfixme}[1]{\Authornote{#1}{\textbf{??}}}
\newcommand{\Authormarginmark}[1]{\marginpar{\textcolor{red}{\fbox{\Large #1:!}}}}
\else
\newcommand{\Authornote}[2]{}
\newcommand{\Authornotecolored}[3]{}
\newcommand{\Authorcomment}[2]{}
\newcommand{\Authorstartcomment}[1]{}

\newcommand{\Authorfnote}[2]{}
\newcommand{\Authorfixme}[1]{}
\newcommand{\Authormarginmark}[1]{}
\fi

\ifnum\showfixme=0

\fi

\usepackage{boxedminipage}

\newcommand{\brac}[1]{[#1]}
\newcommand{\Brac}[1]{\left[#1\right]}

\newcommand{\abs}[1]{\left\lvert#1\right\rvert}
\newcommand{\Abs}[1]{\left\lvert#1\right\rvert}

\newcommand\sett[2]{\left\{ #1 \left| \; \vphantom{#1 #2} \right. #2  \right\}}
\newcommand{\set}[1]{\{#1\}}
\newcommand{\Set}[1]{\left\{#1\right\}}

\newcommand{\norm}[1]{\lVert#1\rVert}

\newcommand{\snorm}[1]{\norm{#1}^2}

\newcommand{\iprod}[1]{\langle#1\rangle}

\newcommand{\Esymb}{\mathbb{E}}
\newcommand{\Psymb}{\mathbb{P}}

\DeclareMathOperator*{\E}{\Esymb}

\DeclareMathOperator*{\ProbOp}{\Psymb}

\renewcommand{\Pr}{\ProbOp}
\def\one{{\mathbf{1}}}

\newcommand{\prob}[1]{\Pr \left[ {#1} \right] }
\newcommand{\Prob}[2][]{\Pr_{{#1}}\left[#2\right]} 
\newcommand{\cProb}[3]{\Pr_{{#1}}\left[ #2 \left| \; \vphantom{#2 #3} \right. #3  \right]} 

\newcommand{\ex}[1]{\E\brac{#1}}
\newcommand{\Ex}[2][]{\E_{{#1}}\Brac{#2}}
\newcommand{\cEx}[3]{\E_{{#1}}\left[ #2 \left| \; \vphantom{#2 #3} \right. #3  \right]}

\newcommand{\ve}{\;\hbox{and}\;}

 \usepackage{dsfont}
\usepackage{mathrsfs}

\newcommand{\textparen}[1]{\text{(#1)}}

\ifx\because\undefined
\newcommand{\because}[1]{\textparen{because #1}}
\else
\renewcommand{\because}[1]{\textparen{because #1}}
\fi

\newcommand\bdot\bullet

\DeclareMathOperator{\poly}{poly}

\DeclareMathOperator{\dist}{dist}

\newcommand{\Z}{\mathbb Z}

\newcommand{\R}{\mathbb R}

\renewcommand{\leq}{\leqslant}

\renewcommand{\geq}{\geqslant}

\ifnum\showdraftbox=1

\else

\fi

\let\epsilon=\varepsilon

\numberwithin{equation}{section}

\newcommand{\MYstore}[2]{%
  \global\expandafter \def \csname MYMEMORY #1 \endcsname{#2}%
}

\newcommand{\MYload}[1]{%
  \csname MYMEMORY #1 \endcsname%
}

\newcommand{\MYnewlabel}[1]{%
  \newcommand\MYcurrentlabel{#1}%
  \MYoldlabel{#1}%
}

\newcommand{\MYdummylabel}[1]{}

\newcommand{\torestate}[1]{%
  \let\MYoldlabel\label%
  \let\label\MYnewlabel%
  #1%
  \MYstore{\MYcurrentlabel}{#1}%
  \let\label\MYoldlabel%
}

\newcommand{\restatetheorem}[1]{%
  \let\MYoldlabel\label
  \let\label\MYdummylabel
  \begin{theorem*}[Restatement of \prettyref{#1}]
    \MYload{#1}
  \end{theorem*}
  \let\label\MYoldlabel
}

\newcommand{\restatelemma}[1]{%
  \let\MYoldlabel\label
  \let\label\MYdummylabel
  \begin{lemma*}[Restatement of \prettyref{#1}]
    \MYload{#1}
  \end{lemma*}
  \let\label\MYoldlabel
}

\newcommand{\restateprop}[1]{%
  \let\MYoldlabel\label
  \let\label\MYdummylabel
  \begin{proposition*}[Restatement of \prettyref{#1}]
    \MYload{#1}
  \end{proposition*}
  \let\label\MYoldlabel
}

\newcommand{\restateclaim}[1]{%
  \let\MYoldlabel\label
  \let\label\MYdummylabel
  \begin{claim*}[Restatement of \prettyref{#1}]
    \MYload{#1}
  \end{claim*}
  \let\label\MYoldlabel
}

\newcommand{\restatecorollary}[1]{%
  \let\MYoldlabel\label
  \let\label\MYdummylabel
  \begin{corollary*}[Restatement of \prettyref{#1}]
    \MYload{#1}
  \end{corollary*}
  \let\label\MYoldlabel
}

\newcommand{\restatefact}[1]{%
  \let\MYoldlabel\label
  \let\label\MYdummylabel
  \begin{fact*}[Restatement of \prettyref{#1}]
    \MYload{#1}
  \end{fact*}
  \let\label\MYoldlabel
}

\newcommand{\restatedefinition}[1]{
\let\MYoldlabel\label 
\let\label\MYdummylabel 
\begin{definition*}[Restatement of \prettyref{#1}] 
    \MYload{#1} 
\end{definition*} 
\let\label\MYoldlabel 
} 

\newcommand{\restate}[1]{%
  \let\MYoldlabel\label
  \let\label\MYdummylabel
  \MYload{#1}
  \let\label\MYoldlabel
}

\let\origparagraph\paragraph
\renewcommand{\paragraph}[1]{\origparagraph{#1.}}

\allowdisplaybreaks

\newcommand{\dunion}{\mathbin{\mathaccent\cdot\cup}}

\let\pref=\prettyref

\newcommand{\bproof}[1]{\begin{proof}[Proof of \pref{#1}]}
\newcommand{\eproof}{\end{proof}}

\NewDocumentCommand{\SC}{ O{X} O{0} O{} }{#1_{#3}{\left ( #2 \right ) }}
\NewDocumentCommand{\X}{ O{} O{} }{\SC[X][#1][#2]}
\NewDocumentCommand{\Y}{ O{} O{} }{\SC[Y][#1][#2]}

\newcommand{\maximalpunc}{uniform{}}
\newcommand{\maximal}{uniform\ }
\newcommand{\maxsize}{uniform\ }
\newcommand{\maxsizity}{uniformity\ }